\UseRawInputEncoding
\pdfoutput=1
\documentclass[11pt]{article}
\usepackage{amsmath}
\usepackage{amssymb}
\usepackage{amsthm}
\usepackage{mathrsfs}
\usepackage{geometry}
\usepackage{subfigure}
\usepackage[colorlinks,linkcolor=blue,citecolor=blue]{hyperref}
\newtheorem{thm}{Theorem}[section]

\newtheorem{pro}{Proposition}[section]
\newtheorem{lem}{Lemma}[section]
\newtheorem{cor}{Corollary}[section]

\newtheorem{Ap}{Assumption}[section]
\theoremstyle{definition}

\newtheorem{remark}{Remark}[section]
\usepackage{graphicx}
\usepackage{epstopdf}
\usepackage{enumerate}
\usepackage{color}
\usepackage{mathdots} 
\usepackage{booktabs}
\usepackage[all]{xy}

\usepackage{xcolor}
\usepackage[T1]{fontenc}
\usepackage{babel}
\usepackage{a4wide}
\usepackage{pdfpages}

\newcommand{\mfn}{\mathfrak{n}}

\newcommand{\mfc}{\mathfrak{c}}
\newcommand{\mfa}{\mathfrak{r}}
\newcommand{\mfb}{\mathfrak{b}}

\newcommand{\mfm}{\mathfrak{m}}

\newcommand{\mfv}{\mathfrak{v}}
\newcommand{\mfw}{\mathfrak{w}}
\newcommand{\mfz}{\mathfrak{z}}

\newcommand{\bbA}{\boldsymbol{A}}
\newcommand{\bbB}{\boldsymbol{B}}

\newcommand{\bbD}{\boldsymbol{D}}

\newcommand{\bbF}{\boldsymbol{F}}

\newcommand{\bbH}{\boldsymbol{H}}
\newcommand{\bbI}{\boldsymbol{I}}

\newcommand{\bbM}{\boldsymbol{M}}

\newcommand{\bbP}{\boldsymbol{P}}

\newcommand{\bbR}{\boldsymbol{R}}
\newcommand{\bbS}{\boldsymbol{S}}
\newcommand{\bbT}{\boldsymbol{T}}
\newcommand{\bbU}{\boldsymbol{U}}
\newcommand{\bbV}{\boldsymbol{V}}
\newcommand{\bbW}{\boldsymbol{W}}
\newcommand{\bbX}{\boldsymbol{X}}
\newcommand{\bbY}{\boldsymbol{Y}}

\newcommand{\bba}{\boldsymbol{a}}

\newcommand{\bbe}{\boldsymbol{e}}

\newcommand{\bbi}{\boldsymbol{i}}

\newcommand{\bbm}{\boldsymbol{m}}

\newcommand{\bbq}{\boldsymbol{q}}
\newcommand{\bbr}{\boldsymbol{r}}

\newcommand{\bbv}{\boldsymbol{v}}
\newcommand{\bbw}{\boldsymbol{w}}
\newcommand{\bbx}{\boldsymbol{x}}
\newcommand{\bby}{\boldsymbol{y}}
\newcommand{\bbz}{\boldsymbol{z}}

\newcommand{\bbve}{\boldsymbol{\varepsilon}}

\newcommand{\bbzeta}{\boldsymbol{\zeta}}
\newcommand{\bbtau}{\boldsymbol{\tau}}

\newcommand{\bbGa}{\boldsymbol{\Gamma}}

\newcommand{\bbPi}{\boldsymbol{\Pi}}
\newcommand{\bbXi}{\boldsymbol{\Xi}}

\newcommand{\bbSig}{\boldsymbol{\Sigma}}
\newcommand{\bbOme}{\boldsymbol{\Omega}}

\newcommand{\mbC}{\mathbb{C}}
\newcommand{\mbE}{\mathbb{E}}
\newcommand{\mbH}{\mathbb{H}}
\newcommand{\mbP}{\mathbb{P}}
\newcommand{\mbM}{\mathbb{M}}
\newcommand{\mbN}{\mathbb{N}}
\newcommand{\mbT}{\mathbb{T}}
\newcommand{\mbR}{\mathbb{R}}

\newcommand{\mbZ}{\mathbb{Z}}
\newcommand{\mbU}{\mathbb{U}}

\newcommand{\mrO}{\mathrm{O}}
\newcommand{\mro}{\mathrm{o}}

\newcommand{\mcA}{\mathcal{A}}
\newcommand{\mcB}{\mathcal{B}}
\newcommand{\mcC}{\mathcal{C}}

\newcommand{\mcN}{\mathcal{N}}

\newcommand{\mcV}{\mathcal{V}}
\newcommand{\mcU}{\mathcal{U}}
\newcommand{\mcE}{\mathcal{E}}
\newcommand{\mcF}{\mathcal{F}}
\newcommand{\mcD}{\mathcal{D}}
\newcommand{\mcG}{\mathcal{G}}
\newcommand{\mcI}{\mathcal{I}}

\newcommand{\mcS}{\mathcal{S}}
\newcommand{\mcT}{\mathcal{T}}
\newcommand{\mcR}{\mathcal{R}}

\newcommand{\hla}{\hat{\lambda}}
\newcommand{\hSig}{\hat{\boldsymbol{\Sigma}}}
\newcommand{\hR}{\hat{\boldsymbol{R}}}
\newcommand{\hF}{\hat{F}}

\newcommand{\fM}{\mathcal{M}}
\newcommand{\fE}{\mathcal{E}}

\newcommand{\mfC}{\mathfrak{C}}

\newcommand{\fs}{\mathfrak{s}}

\newcommand{\fe}{\mathfrak{e}}

\newcommand{\tisig}{\tilde{\sigma}}

\newcommand{\tichi}{\tilde{\chi}}
\newcommand{\tiF}{\tilde{F}}
\newcommand{\tiW}{\tilde{\boldsymbol{W}}}
\newcommand{\tsigma}{\beta}
\newcommand{\tF}{\tilde{F}}
\newcommand{\talpha}{\tilde{\alpha}}

\newcommand{\mtm}{\mathtt{m}}

\newcommand{\tr}{\operatorname{Tr}}
\newcommand{\diag}{\operatorname{diag}}
\newcommand{\Cov}{\operatorname{Cov}}
\newcommand{\Var}{\operatorname{Var}}

\allowdisplaybreaks

\title{Testing High-dimensional Nonstationary Time Series}
\author{Ruihan Liu\thanks{rhliu@connect.hku.hk}\quad and\quad Chen Wang\thanks{stacw@hku.hk}\\
	\small{School of Computing and Data Science, The University of Hong Kong}}
\date{}
\begin{document}
\maketitle
\abstract{In this article, we first establish the joint central limit theorem (CLT) for the extreme eigenvalues of the sample correlation matrix of high-dimensional random walks with cross-sectional dependence. We further investigate the asymptotic spectral properties of the sample correlation matrix of high-dimensional autoregressive processes. To apply our theoretical results, we propose a novel high-dimensional  unit root test and develop a forward sequential test to determine the number of unit roots in high-dimensional time series data. Finally, we conduct an empirical study of the purchasing power parity (PPP) hypothesis in high-dimensional settings.}

\vspace{6mm}
\newpage
\tableofcontents
\clearpage
\section{Introduction}\label{Main Sec of introduction}
Sample covariance matrices and sample correlation matrices are popular tools in many statistical inference problems. As the entries of the sample correlation matrices are standardized, sample-correlation based statistics have important advantages when dealing with certain high-dimensional problems. For example, \cite{gao2017high} and \cite{jiang2004asymptotic} argued that the advantage of using sample correlation matrices over sample covariance matrices is that the former does not require the first two population moments, which are usually unknown in real applications. Recently, \cite{fan2022estimating} showed that the estimation of number of factors in high-dimensional factor models using eigenvalues of sample covariance matrices is generally inconsistent due to heterogeneous sample variances. To this end, they developed an adjusted method using eigenvalues of sample correlation matrices. 

Due to technical difficulties, only a few of existing literature focus on high-dimensional sample correlation matrices of \emph{cross-sectional dependent} data; that is, data with non-diagonal population covariance matrices. Specifically, assuming that the population covariance matrices of finite rank factors are non-diagonal, \cite{morales2021asymptotics} studied the extreme eigenvalues of the sample correlation matrices of high-dimensional spiked covariance models. \cite{jiang2019determinant} and \cite{zheng2019test} tested whether the population correlation matrix is a specific matrix. \cite{yin2021spectral,yin2023central} comprehensively studied asymptotic spectral properties and CLT for linear spectral statistics of rescaled sample correlation matrices, assuming a general structure of population covariance matrices. However, assuming cross-sectional independence is rather restrictive for real applications. As mentioned in \cite{o1998overvaluation} and \cite{breitung2008unit}, ignoring the cross-sectional dependence would cause the PPP over-valued.

In this article, we establish the asymptotic behavior of extreme eigenvalues of the sample correlation matrix of high-dimensional nonstationary time series data. To our best knowledge, this is the first theoretical investigation on such topic. Let \(\bbX=[X_1,\cdots,X_T]\in\mbR^{n\times T}\) be the observed data matrix, where the data dimension \(n\) and the number of samples \(T\) tend to infinity proportionally. Our main theoretical contributions are as follows.
\begin{enumerate}
    \item When \(X_t\) is a random walk, we establish the joint CLT for the first \(K\in\mbN^+\) largest eigenvalues of the sample correlation matrix allowing to be cross-sectionally dependent.
    \item We investigate the asymptotic spectral properties of the sample correlation matrix of \(\bbX\) when \(X_t\) is generated by more general high-dimensional autoregressive (AR) processes.
\end{enumerate}    
To establish the asymptotic behavior of extreme eigenvalues of the sample correlation matrix, we need to overcome four main technical challenges.
\begin{enumerate}
    \item \textit{High-dimensionality.} For the fixed dimension scenario, the sample correlation matrices can be estimated entry-wise. This method does not work for the high-dimensional case, since the estimation error could be significant as the curse of dimensionality.
    \item \textit{Nonlinearity.} The entries of correlation matrices are standardized as ratios of quadratic forms of random variables. Such nonlinear structure dramatically complicates the theoretical analysis, as one needs to deal with both the numerator and the denominator simultaneously.
    \item \textit{Dependence.} Besides the cross-sectional dependence among \(X_t\)'s entries, there is also temporal dependence when $X_t$ generated by AR processes, which further complicates our analysis.
    \item \textit{Nonstationarity.} Since \(X_t\) is nonstationary, the spectral norm of the population covariance matrix of \(X_t\) will tend to infinity as \(t\to\infty\). This is distinct from the standard framework of random matrix theory (RMT), so the general RMT tools do not work for our situation.
\end{enumerate}
    Notably, all existing literature of high-dimensional sample correlation matrices assumes that \(X_t\)'s are independent and the population covariance matrices of \(X_t\) are identical, which are essentially different from our technical challenges 3 and 4. Moreover, our technical framework is quite general in the sense that it can deal with not only sample correlation matrices generated by random walks, but also sample correlation matrices generated by more general AR processes.

    For time series generated by AR process, one research interest is to test its stationarity. Readers may refer to \cite{pesaran2015time,choi2015almost} and references therein for a comprehensive literature review of various unit root tests. In spite of diverse established results, testing nonstationarity of high-dimensional time series still needs further investigations. First, as pointed out in \cite{pesaran2007simple}, many existing unit root tests are valid only when the data dimension is fixed. Moreover, the null hypothesis in most of existing literature are simple random walks only. However, the study of general nonstationary time series, e.g. the AR process with more than one roots on the unit circle, or even outside the unit circle, is still scarce. 
    
    Based on asymptotic spectral behaviors of sample correlation matrices, we propose the following applications.
    \begin{enumerate}
        \item We develop a novel high-dimensional unit root test based on the CLT of the largest eigenvalue of the sample correlation matrix. 
        \item We further develop a forward sequential test to determine the number of unit roots in high-dimensional time series data. To our knowledge, there is no established and rigorously justified procedure. Importantly, the power of our sequential test tends to 1 and our estimation number of unit roots is consistent. 
        \item We propose a criterion to determine whether the characteristic polynomial of observed high-dimensional time series data has roots inside, on or outside the unit circle.
    \end{enumerate}
    In the above applications, one key reason for using sample correlation matrices instead of sample covariance matrices is that the extreme eigenvalues of the sample correlation matrices have more stable asymptotic behaviors than those of sample covariance matrices. Roughly speaking, let $X_t\in\mbR^n$ be a high-dimensional nonstationary AR process, the divergence rate of the largest eigenvalue of $X_t$'s the sample covariance matrix depends on the roots of $X_t$'s characteristic polynomial. By contrast, the divergence rate of the largest eigenvalue of $X_t$'s the sample correlation matrix is always the same as the data dimension $n$.
    
    The rest of this article is organized as follows. In Section \ref{Main Sec of CLT correlation}, we establish the joint CLT for the first \(K\) largest eigenvalues of the sample correlation matrix of high-dimensional random walks. In Section \ref{Main Sec of Robusticity}, we investigate the asymptotic spectral behaviors of the sample correlation matrix of the high-dimensional AR processes. We propose a new unit root test and a forward sequential test to determine the number of unit roots in Section \ref{Main Sec of applications}. Several numerical experiments are conducted in Section \ref{Main Sec of numerical}. An empirical study of PPP for high-dimensional data is provided in Section \ref{Main Sec of PPP}. The proofs of all our results are included in the Supplementary Materials (Appendix).

    We end this section by listing some useful notations.
    \begin{enumerate}
		\item \(C_a\) represents a positive constant that depends on some parameter \(a\).
		\item For two real sequences \(\{a_n\}\) and \(\{b_n\}\), we denote $a_n\asymp\mrO(b_n)\quad\Longleftrightarrow\quad M_1b_n\leq a_n\leq M_2b_n$ for some positive constants \(M_1,M_2\). Moreover, if $\{a_n\}$ and $\{b_n\}$ are sequences of random variables, $a_n\asymp\mrO_{\mbP}(b_n)\quad\Longleftrightarrow\quad\lim_{n\to\infty}\mbP(M_1b_n\leq a_n\leq M_2b_n)=1$.
		\item Given any integrable random variable/vector \(X\), \(X^{\circ}:=X-\mbE[X]\) denotes its centered version.
		\item The \(L^2\) convergence, the convergence in probability and in distribution are denoted by \(\overset{L^2}{\longrightarrow},\overset{\mbP}{\longrightarrow}\) and \(\overset{d}{\longrightarrow}\), respectively.
        \item Given a matrix \(\bbA=[A_{i,j}]_{n\times n}\), \(\tr(\bbA)=\sum_{i=1}^nA_{i,i}\), \(\bbA'\) denotes the transpose of \(\bbA\), and \(\diag(\bbA)\) is the diagonal matrix consisting of the main diagonal of \(\bbA\). Moreover, \(\Vert\bbA\Vert\) denotes the spectral norm of \(\bbA\).
	\end{enumerate}    

    \section{CLT for extreme eigenvalues of the sample correlation matrix of high-dimensional random walks}\label{Main Sec of CLT correlation}
    In this section, we establish the joint CLT of the first $K\in\mbN^+$ largest eigenvalues of the sample correlation matrix of an $n$-dimensional random walk $X_t$ defined as follows:
	\begin{align}
		X_t=X_{t-1}+e_t,\quad e_t=\bbGa\sum_{k=0}^{\infty}\Psi_k\varepsilon_{t-k},\quad\varepsilon_t\overset{i.i.d.}{\sim}\mcN(\boldsymbol{0},\bbI_n).\label{Main Eq of nonpanel Xt}
	\end{align}
	Here $\{\Psi_k:k\in\mbN^+\}$ is a sequence of $n\times n$ diagonal matrices satisfying
    \begin{Ap}\label{Main Ap of panel lag polynomial}
        All \(\{\Psi_k:=\diag(\varphi_{1,k},\cdots,\varphi_{n,k})\in\mbR^{n\times n}\},k\in\mbN\) are diagonal matrices and there exist two positive constants \(b,B\) such that 
        $$\sum_{k=0}^{\infty}(1+k)^2\Vert\Psi_k\Vert\leq B\quad\text{and}\quad \min_{1\leq j\leq n}\inf_{x\in[-\pi,\pi]}\left|\sum_{k=0}^{\infty}\varphi_{j,k}e^{{\rm i}kx}\right|\geq b.$$
    \end{Ap}
    Moreover, the \emph{cross-sectional matrix} $\bbGa\in\mbR^{n\times n}$ satisfies the following condition.
    \begin{Ap}\label{Main Ap of nonpanel}
        There exist two positive constants \(m_0,M_0\) such that $m_0\leq\lambda_{\min}(\bbGa\bbGa')\leq\lambda_{\max}(\bbGa\bbGa')\leq M_0$, where \(\lambda_{\max}(\bbGa\bbGa')\) and \(\lambda_{\min}(\bbGa\bbGa')\) are the largest and smallest eigenvalue of \(\bbGa\bbGa'\), respectively.
    \end{Ap}
    The condition \(\sum_{k=0}^{\infty}(1+k)^2\Vert\Psi_k\Vert\leq B\) in Assumption \ref{Main Ap of panel lag polynomial} is widely used in the time series literature (e.g. \cite{buhlmann1997sieve} and \cite{park2002invariance}), which ensures the stationarity of the linear process \(\sum_{k=0}^{\infty}\Psi_k\varepsilon_{t-k}\), including both \({\rm MA}(\infty)\) and \({\rm AR}(1)\) models. Moreover, since $\varepsilon_t$ are i.i.d. $\mcN(\boldsymbol{0},\bbI_n)$, we can conclude that $e_t\sim\mcN\big(\boldsymbol{0},\bbGa(\sum_{k=0}^{\infty}\Psi_ke^{{\rm i}\pi kt/T})(\sum_{k=0}^{\infty}\Psi_ke^{-{\rm i}\pi kt/T})'\bbGa'\big)$ by Theorem 13 in \cite{hannan2009multiple}. Hence, Assumptions \ref{Main Ap of panel lag polynomial} and \ref{Main Ap of nonpanel} ensure that the covariance matrix of $e_t$ is positive semidefinite with bounded spectral norm.
    
    Given observations \(\bbX=[X_1,\cdots,X_T]\) generated by \eqref{Main Eq of nonpanel Xt}, the following high-dimensionality regime is assumed.
    \begin{Ap}\label{Main Ap of highdimensionality}
        As the dimension \(n\to\infty\), the number of observations \(T\) also tends to infinity such that \(\lim_{n\to\infty}n/T=c\in(0,\infty)\).
    \end{Ap}
	Let \(\bbM:=\bbI_T-\boldsymbol{1}_{T\times T}/T\), where \(\bbI_T\) is the identity matrix with a size of \(T\times T\) and \(\boldsymbol{1}_{T\times T}\) is a \(T\times T\) matrix whose entries are all \(1\). Then we have $\bbX-\bar{\bbX}=\bbX\bbM$, where $\bar{\bbX}=[\bar{X},\cdots,\bar{X}]$ and $\bar{X}=T^{-1}\sum_{t=1}^TX_t$ is the sample mean. Note that $\bbM^2=\bbM$, so the sample correlation matrix of \(\bbX\) is $\bbD^{-1/2}\bbX\bbM\bbX'\bbD^{-1/2}$, where \(\bbD:=\diag(\bbX\bbM\bbX')\). Since we only focus on the extreme eigenvalues of the sample correlation matrix, and the nonzero eigenvalues of
	\begin{align}
		\hR:=\bbM\bbX'\bbD^{-1}\bbX\bbM,\label{Main Eq of correlation matrix}
	\end{align}
    and $\bbD^{-1/2}\bbX\bbM\bbX'\bbD^{-1/2}$ are coincide. Therefore, we regard the matrix $\hR$ as the sample correlation matrix of $\bbX$.


    \subsection{Limit of the convergence in probability}
	Let $\hla_1\geq\cdots\geq\hla_K$ be the first $K$ largest eigenvalue of $\hR$. We first establish the limit of $n^{-1}\hla_k$, $k=1,\cdots,K$. To characterize this limit, we define the following random variable: 
     \begin{align}
		\fM_{k,l}(x):=\frac{(kl)^{-x}\mfz_k\mfz_l}{\sum_{t=1}^{\infty}t^{-2x}\mfz_t^2},\quad x\in[1,+\infty),1\leq k,l\leq K,\label{Main Eq of fM}
	\end{align}
    where $\{\mfz_t\overset{i.i.d.}{\sim}\mcN(0,1):t\in\mbN^+\}$.
    \begin{pro}\label{Main Thm of convergence in probability}
        Under Assumptions \ref{Main Ap of panel lag polynomial}, \ref{Main Ap of nonpanel} and \ref{Main Ap of highdimensionality}, for any deterministic $K\in\mbN^+$, we have $n^{-1}\hla_k\overset{\mbP}{\longrightarrow}\mbE[\fM_{k,k}(1)]$ for $1\leq k\leq K$.
    \end{pro}
    Readers can refer to \S\ref{sec of mbP dependent} of the supplement for the proof of Proposition \ref{Main Thm of convergence in probability}. 
    \subsubsection*{Outline of the proof of Proposition \ref{Main Thm of convergence in probability}}
    Note that \(\bbX=\bbe\bbU\) by \eqref{Main Eq of nonpanel Xt}, where \(\bbe:=[e_1,\cdots,e_T]\) is the noise matrix and \(\bbU\) is a \(T\times T\) upper triangular matrix with 1 above and on the main diagonal, so the sample correlation matrix \(\hR\) in (\ref{Main Eq of correlation matrix}) can be rewritten as $\hR=\bbM\bbU'\bbe'\diag(\bbe\bbU\bbM\bbU'\bbe')^{-1}\bbe\bbU\bbM$. Consider the singular value decomposition of \(\bbM\bbU'\). Precisely, we have
	\begin{align}
		\bbM\bbU'=\sum_{s=1}^{T-1}\sigma_s\bbw_s\bbv_s'\quad\text{and}\quad\sigma_s:=[2\sin(\pi k/(2T))]^{-1},\label{Main Eq of sigma_k}
	\end{align}
	where \(\bbv_s:=(v_{s,1},\cdots,v_{s,T})'\) and \(\bbw_s:=(w_{s,1},\cdots,w_{s,T})'\) such that $v_{s,t}=\sqrt{\frac{2}{T}}\sin(\pi s(t-1)/T)$ and $w_{s,t}=-\sqrt{\frac{2}{T}}\cos(\pi s(2t-1)/(2T))$ for \(1\leq s\leq T-1\) and $1\leq t\leq T$. When \(s=T\), \(\sigma_T=0,\bbv_T=(1,0,\cdots,0)'\) and \(\bbw_T=\boldsymbol{1}_T/\sqrt{T}\). Next, let \(\hat{F}_k\) be the normalized eigenvector of \(\hla_k\), i.e. \(\hat{\bbR}\hat{F}_k=\hla_k\hat{F}_k\) and $\Vert\hat{F}_k\Vert_2=1$. Since \(\{\bbw_1,\cdots,\bbw_T\}\) forms an orthogonal basis of \(\mathbb{R}^T\), we represent \(\hat{F}_k\) by \(\hat{F}_k:=\sum_{t=1}^T\alpha_{k,t}\bbw_t\), where \(\sum_{k=1}^T\alpha_{1,k}^2=1\). Therefore, we obtain that
    \begin{align}
	    \frac{\hla_k}{n}=\frac{1}{n}\hat{F}_k'\hR\hat{F}_k=\sum_{s,t=1}^T\alpha_{k,s}\alpha_{k,t}\frac{1}{n}\sum_{j=1}^n\widetilde{M}_{j;s,t},\quad\widetilde{M}_{j;s,t}:=\frac{\sigma_s\sigma_t(\bbe_j\bbv_s)(\bbe_j\bbv_t)}{\sum_{l=1}^T\sigma_l^2(\bbe_j\bbv_l)^2},\label{Main Eq of tilde M st}
	\end{align}
    where $\bbe_j$ is the $j$-th {\bf row} of the noise matrix $\bbe$. Then we can prove Proposition \ref{Main Thm of convergence in probability} by showing that $|\alpha_{k,k}|\overset{\mbP}{\longrightarrow}1$ and $\frac{1}{n}\sum_{j=1}^n\widetilde{M}_{j;s,t}\overset{\mbP}{\longrightarrow}\mbE[\fM_{s,t}(1)]$.

    \subsection{CLT for the extreme eigenvalues}
    To further establish the CLT for \(\hla_k\), we need an additional assumption for the cross-sectional matrix \(\bbGa\):
    \begin{Ap}[$m$-dependence\footnote{The $m$-dependence was first introduced by \cite{hoeffding1948central}, where they assumed that $m\in\mbN^+$ is deterministic. In Assumption \ref{Main Ap of m dependent}, we allow $m=m(n)$ to be a function of data dimension $n$ such that $\lim_{n\to\infty}m(n)=\infty$.}]\label{Main Ap of m dependent}
        The \(n\times n\) cross-sectional matrix \(\bbGa\) satisfies that
		$$\{j=1,\cdots,n:\Gamma_{i_1,j}\neq0\}\cap\{j=1,\cdots,n:\Gamma_{i_2,j}\neq0\}=\emptyset$$
		for all \(|i_1-i_2|>m\), where \(m:=m(n)\leq\mro(n^{1/2})\).
    \end{Ap}
	Note that such matrices in Assumption \ref{Main Ap of m dependent} do exist, for example the \(m\) banded toeplitz matrices. Assumption \ref{Main Ap of m dependent} is widely used in estimating high-dimensional covariance matrices, e.g. \cite{bickel2008covariance,cai2011adaptive,chen2013covariance} and \cite{fan2013large}.

    To characterize the asymptotic variance of $\hla_k$, let $\widehat{M}_{j;s,t}(1)$ be a random variable defined as follows:
    \begin{align}
		\widehat{M}_{j;k,l}(x):=\frac{(kl)^{-x}z_{j,k}z_{j,l}}{\sum_{t=1}^{\infty}t^{-2x}z_{j,t}^2},\quad{\rm where\ }x\in[1,+\infty),1\leq k,l\leq K,\label{Main Eq of Mikl}
	\end{align}
    where $\big\{(z_{1,t},\cdots,z_{n,t})'\overset{i.i.d.}{\sim}\mcN(\boldsymbol{0},\tilde{\bbGa}):t\in\mbN^+\big\}$ is a sequence of $n$-dimensional normal vectors with the covariance matrix
    \begin{align}
		\tilde{\bbGa}:=\diag(\bbGa\Psi(1)\Psi(1)'\bbGa')^{-1/2}\bbGa\Psi(1)\Psi(1)'\bbGa'\diag(\bbGa\Psi(1)\Psi(1)'\bbGa')^{-1/2},\label{Main Eq of tibbGa}
	\end{align}
    and $\Psi(1):=\sum_{k=0}^{\infty}\Psi_k$ is defined in Assumption \ref{Main Ap of panel lag polynomial}. Now, the CLT for $\hla_k$ is given as follows:
    \begin{thm}\label{Main Thm of CLT I1}
        Under Assumptions \ref{Main Ap of panel lag polynomial}, \ref{Main Ap of nonpanel}, \ref{Main Ap of highdimensionality} and \ref{Main Ap of m dependent}, suppose \(X_t\) is generated by {\rm (\ref{Main Eq of nonpanel Xt})}, we have
		$$\frac{\sqrt{n}}{\mfm_{k,k}(1)}\left(\frac{\hla_k}{n}-\mbE[\fM_{k,k}(1)]\right)\overset{d}{\longrightarrow}\mcN(0,1),$$
		where
        \begin{align}
            \mfm_{k,k}^2(1):=\Var\left(\frac{1}{\sqrt{n}}\sum_{j=1}^n\widehat{M}_{j;k,k}(1)\right).\label{Main Eq of mfm}
        \end{align}
        Further let \(\bbA_n=[A_{k,l}]_{K\times K}\) be a \(K\times K\) covariance matrix such that  
        $$A_{k,l}:=n^{-1}\Cov\Bigg(\sum_{j=1}^n\widehat{M}_{j;k,k}(1),\sum_{j=1}^n\widehat{M}_{j;l,l}(1)\Bigg)$$
        for \(1\leq k,l\leq K\), where $\widehat{M}_{j;k,k}(1)$ is defined in \eqref{Main Eq of Mikl}. Suppose \(\liminf_{n\to\infty}\lambda_{\min}(\bbA_n)>0\), then
		$$\sqrt{n}\bbA_n^{-1/2}\Bigg(\frac{\hla_1}{n}-\mbE[\fM_{1,1}(1)],\cdots,\frac{\hla_K}{n}-\mbE[\fM_{K,K}(1)]\Bigg)'\overset{d}{\longrightarrow}\mcN(\boldsymbol{0},\bbI_K).$$
    \end{thm}
	Readers can refer to \S\ref{sec of CLT correlation dependent} in the supplement for the detailed proof of Theorem \ref{Main Thm of CLT I1}.
    \begin{remark}\label{Rem of numerical}
        \begin{itemize}
            \item[a.] Similar to Theorem \ref{Main Thm of CLT I1}, we also establish the joint CLT for the first $K$ largest eigenvalues of the {\bf sample covariance matrix} of $\bbX$ generated by high-dimensional random walks, readers can refer to \S\ref{Sec of covariance} in the supplement for details.
            \item[b.] Particularly, if the cross-sectional matrix $\bbGa$ in \eqref{Main Eq of nonpanel Xt} is diagonal, that is, $X_t$ is cross-sectional independent, we can establish the same joint CLT as in Theorem \ref{Main Thm of CLT I1} for more general non-Guassian $\varepsilon_t$ in \eqref{Main Eq of nonpanel Xt}, see \S\ref{Sec of correlation independent} in the supplement for details.
			\item[c.] To apply Theorem \ref{Main Thm of CLT I1} in unit root tests, we need to compute $\mbE[\fM_{k,k}(1)]$ and $\mfm_{k,k}(1)$ in \eqref{Main Eq of mfm}. According to (\ref{Main Eq of fM}), we can numerically compute $\mbE[\fM_{k,k}(1)]$ by the Monte Carlo method, e.g. \(\mbE[\fM_{1,1}(1)]\approx0.4409\). Although the direct estimation of $\mfm_{k,k}(1)$ is generally difficult as the cross-sectional matrix $\bbGa$ is usually unknown, we can use bootstrap method, see Section \ref{Main sec of unit root test}.
        \end{itemize}
    \end{remark}

    \section{Asymptotic spectral properties of the sample correlation matrix of the high-dimensional AR processes}\label{Main Sec of Robusticity}
	In this section, we further investigate the asymptotic spectral properties of the sample correlation matrices of the more general AR($d$) process. Let \(X_t\) be an \(n\)-dimensional AR process generated by
	\begin{align}
		X_t+\sum_{l=1}^da_lX_{t-l}=e_t,\quad e_t=\bbGa\sum_{k=0}^{\infty}\Psi_k\varepsilon_{t-k},\quad\varepsilon_t\overset{i.i.d.}{\sim}\mcN(\boldsymbol{0},\bbI_n),\label{Main Eq of AR process}
	\end{align}
	where \(a_l\in\mbC\) for \(1\leq l\leq d\), and the coefficients $\{\Psi_k:k\in\mbN\}$ of matrix lag polynomial and the cross-sectional matrix \(\bbGa\) satisfy Assumptions \ref{Main Ap of panel lag polynomial} and \ref{Main Ap of nonpanel}, respectively. The characteristic polynomial of (\ref{Main Eq of AR process}) is defined as
	\begin{align}
		f_X(z)=z^d+\sum_{l=1}^da_lz^{d-l}=\prod_{l=1}^d(z-\mfa_l),\label{Main Eq of characteristic polynomial}
	\end{align}
	where \(\mfa_l\in\mbC\) are roots of (\ref{Main Eq of characteristic polynomial}) for \(1\leq l\leq d\). Rewrite (\ref{Main Eq of AR process}) by 
	\begin{align}
		\prod_{l=1}^d(1-\mfa_l L)X_t=e_t,\label{Main Eq of AR process roots}
	\end{align}
	where \(L\) is the time lag operator.

    It is well-known that \(X_t\) is stationary if and only if all \(|\mfa_l|<1\), see Chapter 3.2 in \cite{shumway2000time}. Based on whether \(\mfa_l\) is inside, on or outside the unit circle, we classify all \(\mfa_l\) into three classes. Precisely, we say \(\mfa_l\) is a
	\begin{enumerate}
		\item \emph{stationary root} if \(|\mfa_l|<1\);
		\item \emph{nonstationary root} if \(|\mfa_l|=1\);
		\item \emph{super nonstationary root} if \(|\mfa_l|>1\).
	\end{enumerate}
    Next, we show that the sample correlation matrix of $\bbX$ generated by \eqref{Main Eq of AR process} has different asymptotic spectral properties when its characteristic polynomial \eqref{Main Eq of characteristic polynomial} has different types of roots. Precisely, we have
    \begin{thm}\label{Main Pro of root critiria}
        Under Assumptions \ref{Main Ap of panel lag polynomial}, \ref{Main Ap of nonpanel}, \ref{Main Ap of highdimensionality} and \ref{Main Ap of m dependent}, let $\hR$ be the sample correlation matrix of $\bbX=[X_1,\cdots,X_T]$ generated by an AR($d$) process \eqref{Main Eq of AR process roots}. Let $\mfa_1,\cdots,\mfa_d$ be roots of $X_t$'s characteristic polynomial \eqref{Main Eq of characteristic polynomial} and $\hla_1\geq\cdots\geq\hla_T$ be eigenvalues of $\hR$, then 
        \begin{enumerate}
            \item if all $\mfa_l$ are stationary, we have $n^{-1}\Vert\hR\Vert\overset{\mbP}{\longrightarrow}0$;
            \item if at least one $\mfa_l$ is super nonstationary, we have $\limsup_{n\to\infty}{\rm rank}(\hR)\leq d$ and $\lim_{n\to\infty}\mbP(n^{-1}\Vert\hR\Vert>C)=1$, where $C\in(0,1)$ is a deterministic positive constant. 
        \end{enumerate}
    \end{thm}
    Moreover, if none of $\tau_l$ is super nonstationary and at least one $\mfa_l$ is nonstationary, let's consider a simplified version\footnote{The proof of Theorem \ref{Main Thm of CLT I1} requires the SVD of $\bbM\bbU'$ \eqref{Main Eq of sigma_k}. However, for the toeplitz matrix $\mbU$ in \eqref{Main Eq of mbU}, the explicit expressions of the singular vectors of $\bbM\mbU'$ are unknown, so we consider a simplification \eqref{Main Eq of AR nonstationary}.} of \eqref{Main Eq of AR process roots} as follows:
	\begin{align}
		\prod_{l=1}^d(1-\mfa_l L)X_t=e_t,\quad e_t=\bbGa\varepsilon_t,\quad\varepsilon_t\overset{i.i.d.}{\sim}\mcN(\boldsymbol{0},\bbI_n),\label{Main Eq of AR nonstationary}
	\end{align}
	where \(|\mfa_l|\leq1\) for \(1\leq l\leq d\) and $|\mfa_1|=1$ without loss of generality. Here, we provide that
    \begin{thm}\label{Main Thm of CLT multiple unit roots}
        Under Assumptions \ref{Main Ap of nonpanel}, \ref{Main Ap of highdimensionality} and \ref{Main Ap of m dependent}, for the cross-sectional matrix $\bbGa$ in \eqref{Main Eq of AR nonstationary}, define \(\bbXi=[\Xi_{i_1,i_2}]_{n\times n}:=\diag(\bbGa\bbGa')^{-1/2}\bbGa\bbGa'\diag(\bbGa\bbGa')^{-1/2}\). Further let \((z_{1,t},\cdots,z_{n,t})'\overset{i.i.d.}{\sim}\mcN(\boldsymbol{0},\bbXi):t=1,\cdots,T\). For any \(K\in\mbN^+\) and \(1\leq k\leq K\), let
		$$\mbM_{k,n}=\frac{1}{n}\sum_{i=1}^n\frac{\beta_k^2z_{i,k}^2}{\sum_{t=1}^T\beta_t^2z_{i,t}^2},$$
		where \(\beta_1\geq\cdots\geq\beta_T\) are singular values of \(\bbM\mbU'\) in \eqref{Main Eq of mbU}. Then given any data matrix \(\bbX=[X_1,\cdots,X_T]\) generated by {\rm (\ref{Main Eq of AR nonstationary})}, we have for \(1\leq k\leq K\)
		$$\frac{\sqrt{n}}{\bbm_{k,n}}\left(\frac{\hla_k}{n}-\mbE[\mbM_{k,n}]\right)\overset{d}{\longrightarrow}\mcN(0,1),$$
		where \(\hla_k\) is the first \(k\)-th largest eigenvalue of the sample correlation matrix of \(\bbX\) and \(\bbm_{k,n}^2=n\Var(\mbM_{k,n})\asymp\mrO(1)\).
    \end{thm}
    Readers can refer to \S\ref{Sec of AR process} in the supplement for proofs of Theorems \ref{Main Pro of root critiria} and \ref{Main Thm of CLT multiple unit roots}. 
    \begin{remark}
        As mentioned in Section \ref{Main Sec of introduction}, the sample correlation matrices will have more advantages than the sample covariance matrices for certain statistical inference problems. This could be illustrated in Theorem \ref{Main Thm of CLT multiple unit roots}. In particular, let \(\hSig\) be the sample covariance matrix of \(\bbX=[X_1,\cdots,X_T]\), where \(X_t\) is generated by \eqref{Main Eq of AR process}. If \(X_t\)'s characteristic polynomial \eqref{Main Eq of characteristic polynomial} has one super nonstationary root \(\mfa_1\) ($|\mfa_1|>1$), we can conclude that \(\lim_{n\to\infty}\mbP(\Vert\hSig\Vert\geq\mrO(|\mfa_1|^n))=1\). If all roots of $X_t$'s characteristic polynomial \eqref{Main Eq of characteristic polynomial} are 1, i.e. $(1-L)^dX_t=e_t$ by \eqref{Main Eq of AR process roots}, we can show that $\lim_{n\to\infty}\mbP(\Vert\hSig\Vert\geq\mrO(n^{2d}))=1$. Readers can refer to Remarks \ref{Rem of why not covariance 1} and \ref{Rem of why not covariance 2} in the supplement for detailed estimations of $\Vert\hSig\Vert$. Thus, the divergence rate of \(\Vert\hSig\Vert\) depends on both the roots of (\ref{Main Eq of characteristic polynomial}) and data dimension $n$. By contrast, for the sample correlation matrix \(\hR\), since the absolute values of \(\hR\)'s entries are no more than 1, we have \(n^{-1}\Vert\hR\Vert\in[0,1]\), which suggests that the asymptotic behaviors of \(\Vert\hR\Vert\) is more stable than those of \(\Vert\hSig\Vert\).
    \end{remark}

    \subsection*{Outline of the proofs of Theorems \ref{Main Pro of root critiria} and \ref{Main Thm of CLT multiple unit roots}}
    Recall that \(\bbX=\bbe\bbU\) when \(X_t\) is generated by a random walk, where \(\bbe=[e_1,\cdots,e_T]\) is the noise matrix. Similarly, when \(X_t\) is generated by an AR process (\ref{Main Eq of AR process roots}), we can also represent the data matrix \(\bbX\) by the product of a noise matrix \(\bbe\) and an upper toeplitz matrix. Precisely, for any \(x\in\mbC\), let
	\begin{align}
		\mcT(x)=\left(\begin{array}{ccccc}
			1&-x&0&\cdots&\\
			0&1&-x&0&\cdots\\
			&\ddots&\ddots&\ddots&\\
			&\cdots&0&1&-x\\
			&&\cdots&0&1
		\end{array}\right)\label{Main Eq of mcT(a)}
	\end{align}
	be an upper toeplitz matrix. By Vieta's theorem, we have \(\mcT(x)\mcT(y)=\mcT(y)\mcT(x)\) for any \(x,y\in\mbC\). Given observations \(\bbX=[X_1,\cdots,X_T]\) generated by (\ref{Main Eq of AR process roots}) and \(X_0=\cdots=X_{1-d}=\boldsymbol{0}\), we have \(\bbX\prod_{l=1}^d\mcT(\mfa_l)=\bbe\). Because \(\mcT(x)\) is commutative with respect to multiplications, the matrix \(\prod_{l=1}^d\mcT(\mfa_l)\) is uniquely determined and independent of the order of \(\mcT(\mfa_l)\) in multiplication. Moreover, the eigenvalues of \(\mcT(x)\) are all 1. Therefore, \(\mcT(x)^{-1}\) exists and we have \(\bbX=\bbe\prod_{l=1}^d\mcT(\mfa_l)^{-1}\). For simplicity, we define 
    \begin{align}
        \mbU:=\prod_{l=1}^d\mcT(\mfa_l)^{-1},\label{Main Eq of mbU}
    \end{align}
    and the sample correlation matrix $\hR$ of $\bbX$ generated by an AR process \eqref{Main Eq of AR process roots} can be written as
    \begin{align}
        \hR=\bbM\mbU'\bbe'\diag(\bbe\mbU\bbM\mbU'\bbe')^{-1}\bbe\mbU\bbM.\label{Main Eq of correlation matrix AR process}
    \end{align}
    Similar as the proof of Proposition \ref{Main Thm of convergence in probability}, one essential step is to use the SVD of $\bbM\bbU'$ in \eqref{Main Eq of sigma_k} to represent the extreme eigenvalues of the sample correlation matrices, see \eqref{Main Eq of tilde M st}. Similarly, for $\hR$ in \eqref{Main Eq of correlation matrix AR process}, we need the SVD of $\bbM\mbU'$ to investigate the asymptotic spectral properties. Indeed, different types of roots $\mfa_l$ will lead to different singular structures of $\bbM\mbU'$ and different asymptotic spectral properties of $\hR$. Readers can refer to \S\ref{Sec of AR process} in the supplement for details.

    \section{Applications}\label{Main Sec of applications}
	In this section, we first propose a new high-dimensional unit root test based on the CLT in Theorem \ref{Main Thm of CLT I1}. As an extension of unit root test, we further propose a forward sequential test to determine the number of unit roots in high-dimensional time series data.


	\subsection{Unit roots test}\label{Main sec of unit root test}
	Suppose \(X_t\) is an \(n\)-dimensional time series data generated by
	\begin{align}
		X_t=(\bbI_n-\bbPi)\phi+\bbPi X_{t-1}+e_t,\label{Main Eq of unit root model}
	\end{align}
	where \(\phi\in\mbR^n\) is deterministic, \(\bbPi\) is an \(n\times n\) matrix and $e_t$ is the noise process defined in \eqref{Main Eq of nonpanel Xt}. The unit root test of (\ref{Main Eq of unit root model}) is to test
	\begin{align}
		H_0:\bbPi=\bbI_n\quad{\rm versus}\quad H_1:\Vert\bbPi\Vert<1.\label{Main Eq of unit root test}
	\end{align}
	For the data matrix \(\bbX=[X_1,\cdots,X_T]\) generated by \eqref{Main Eq of unit root model}, let \(\hR\) be the sample correlation matrix of \(\bbX\). Under \(H_0\), we have established the CLT for the largest eigenvalue of \(\hR\) for cross-sectional dependent \(X_t\) in Theorem \ref{Main Thm of CLT I1}, so we construct our test statistic as follows:
	\begin{align}
		\widehat{T}_n(0):=\sqrt{n}\left(\frac{\hla_1}{n}-\mathbb{E}[\fM_{1,1}(1)]\right).\label{Main Eq of test statistic}
	\end{align}
	Consequently, Theorem \ref{Main Thm of CLT I1} implies that
	\begin{align}
		\frac{\widehat{T}_n(0)}{\mfm_{1,1}(1)}\overset{d}{\longrightarrow}\mcN(0,1),\quad\text{where }\mfm_{1,1}(1)\text{ is defined in \eqref{Main Eq of mfm}.}\label{Main Eq of H0 statistic}
	\end{align}
    Under \(H_1\), we have the following results:
    \begin{thm}\label{Main Thm of H1 2}
        Under Assumptions \ref{Main Ap of panel lag polynomial}, \ref{Main Ap of nonpanel} and \ref{Main Ap of highdimensionality}, suppose \(X_t\) is generated by \eqref{Main Eq of unit root model} such that \(\Vert\bbPi\Vert=\tau_0<1\) and \(e_t=\bbGa\sum_{k=0}^{\infty}\Psi_k\varepsilon_{t-k}\), where \(\varepsilon_t\overset{i.i.d.}{\sim}\mcN(\boldsymbol{0},\bbI_n)\). Then the sample correlation matrix \(\hR\) of \(\bbX=[X_1,\cdots,X_T]\) satisfies that
		\begin{align}
			\mbP\big(n^{-1}\Vert\hR\Vert>\mrO(n^{-1/15})\big)\leq\mrO(n^{-1/6}\log^5(n)).\label{Main Eq of H0 convergence rate}
		\end{align}
    \end{thm}
	Readers can refer to \S\ref{ssec of totally stationary alternatives} in the supplement for proofs of Theorem \ref{Main Thm of H1 2}. Now, given a significance level \(\alpha\in(0,1)\), let \(u_{\alpha/2}\) and \(u_{1-\alpha/2}\) be the lower and upper \(\alpha/2\) quantile of \(\widehat{T}_n(0)\) under \(H_0\), then we will reject \(H_0\) if $\widehat{T}_n(0)\notin[u_{\alpha/2},u_{1-\alpha/2}]$. By \eqref{Main Eq of mfm}, we know that $\mfm_{1,1}(1)\asymp\mrO(1)$. Therefore, given a significance level $\alpha\in(0,1)$, we have $|u_{1-\alpha/2}|=|u_{\alpha/2}|\leq\mrO(1)$. On the other hand, by Theorem \ref{Main Thm of H1 2}, our test statistic \(\widehat{T}_n(0)\asymp\mrO_{\mbP}(-\sqrt{n})\) under \(H_1\). Therefore, the asymptotic power of our test is
    $$\lim_{n\to\infty}\mbP\big(\widehat{T}_n(0)\notin[u_{\alpha/2},u_{1-\alpha/2}]\big|H_1\big)=1.$$
    \begin{remark}\label{Rem of unit root test}
        \begin{itemize}
            \item[a.] In practice, for the efficiency of our unit root test, we can check whether $\widehat{T}_n(0)<-\log(n)$ or not. Note that $\lim_{n\to\infty}\mbP(\widehat{T}_n(0)<-\log(n)|H_1)=1$ due to \(\widehat{T}_n(0)\asymp\mrO_{\mbP}(-\sqrt{n})\) under \(H_1\), so we will reject $H_0$ if $\widehat{T}_n(0)<-\log(n)$ and the asymptotic power is still 1.
            \item[b.] As \cite{pesaran2013panel} suggested, it would be more proper to extend the alternative hypothesis \(H_1\) in (\ref{Main Eq of unit root test}) to
			\begin{align}
				\bbPi=\left(\begin{array}{cc}
					\bbI_{n_1}&\boldsymbol{0}_{n_1\times n_2}\\
					\boldsymbol{0}_{n_2\times n_1}&\widetilde{\bbPi}
				\end{array}\right),\label{Main Eq of partially stationary}
			\end{align}
			where \(\Vert\widetilde{\bbPi}\Vert<1\) and \(c_1:=\lim_{n\to\infty}n_1/n\in[0,1)\). Under this generalized alternative hypothesis, we show that the asymptotic power of \(\widehat{T}_n(0)\) is still 1. Readers can refer to \S\ref{ssec of partially stationary alternatives} in the supplement for details of this extension.
        \end{itemize}
    \end{remark}
	Since $\mfm_{1,1}(1)$ in \eqref{Main Eq of mfm} is usually unknown, we provide a bootstrap method to estimate $u_{\alpha/2}$ and $u_{1-\alpha/2}$. By \eqref{Main Eq of Mikl} and \eqref{Main Eq of mfm}, we first estimate $\tilde{\bbGa}$ in \eqref{Main Eq of tibbGa}. 
    Here, we use the hard thresholding method as in \cite{sun2018large} and \cite{zhang2021convergence} to estimate $\tilde{\bbGa}$ as follows:
	\begin{enumerate}
		\item Given \(\bbX=[X_1,\cdots,X_T]\), let \(\Delta X_t=X_t-X_{t-1}\) for \(2\leq t\leq T\), then construct
		$$\bbH:=\frac{1}{T}\sum_{t=2}^T\Delta X_t(\Delta X_t)'+\sum_{l=1}^{[T^{1/2}]}\frac{1}{T-l}\sum_{t=l+1}^T\left(\Delta X_t(\Delta X_{t-l})'+\Delta X_{t-l}(\Delta X_t)'\right),$$
		and \(\tilde{\bbH}:=\diag(\bbH)^{-1/2}\bbH\diag(\bbH)^{-1/2}\).
		\item Choose a threshold \(\nu\asymp\mrO(T^{-1/2}\log(T))\). Let \(T_{\nu}(\tilde{\bbH}):=\big[\tilde{H}_{i,j}1_{|\tilde{H}_{i,j}|>\nu}\big]_{n\times n}\) and denote the SVD of \(T_{\nu}(\tilde{\bbH})\) by $T_{\nu}(\tilde{\bbH})=\sum_{i=1}^n\gamma_i\bbq_i\bbq_i'$, then $\tilde{\bbGa}$ is estimated by \(\bbOme:=\sum_{i=1}^n\max\{\gamma_i,\mu\}^{1/2}\bbq_i\bbq_i'\), where \(\mu>0\) such that \(\mu^2\asymp\mrO\big(n^{-2}\sum_{i,j=1}^n1_{|\tilde{H}_{i,j}|>\nu}\big)\).  
	\end{enumerate}
    Readers can find more technical details for above procedures in \S\ref{sec of estimation variance} of the supplement. 
    \begin{remark}
        For the choice of \(\nu\), readers can refer to \cite{sun2018large}. Moreover, under Assumptions \ref{Main Ap of panel lag polynomial}, \ref{Main Ap of nonpanel}, \ref{Main Ap of highdimensionality} and \ref{Main Ap of m dependent}, Proposition 3.6 in \cite{sun2018large} and Corollary 4.5 in \cite{zhang2021convergence} give that $\mbP\big(n^{-1}\Vert T_{\nu}(\tilde{\bbH})-\tilde{\bbGa}\Vert_F^2\geq m\nu^2\big)\leq\mrO(n^{-1})$, where \(m=m_n\leq\mro(\sqrt{n})\) is defined in Assumption \ref{Main Ap of m dependent}. Since \(T_{\nu}(\tilde{\bbH})\) may not be positive-definite, then we choose another threshold $\mu>0$ and construct $\bbOme=\sum_{i=1}^n\max\{\gamma_i,\mu\}^{1/2}\bbq_i\bbq_i'$ as the square root of $T_{\nu}(\tilde{\bbH})$. Readers can refer to \S2.2 of \cite{chen2013covariance} for details of this technique. 
    \end{remark}
    Now, given a significance level \(\alpha\in(0,1)\), we estimate the lower and upper \(\alpha/2\) quantiles of \(\widehat{T}_n(0)\) as follows:
	\begin{enumerate}
		\item Given \(\bbX=[X_1,\cdots,X_T]\) and the number of bootstraps \(B\), construct \(\Delta X_t=X_t-X_{t-1}\) for \(t=2,\cdots,T\) and \(\bbOme\) by above procedures.
		\item Simulate \(B\) independent standard Gaussian random matrices \(\bbe^{(1)},\cdots,\bbe^{(B)}\in\mbR^{n\times T}\), that is, the entries $e_{i,t}^{(b)}$ of $\bbe^{(b)}$ are i.i.d. standard normal variables. For each $b=1,\cdots,B$, we construct
        $$\bbX^{(b)}:=\bbOme\bbe^{(b)}\diag(1^{-1},\cdots,T^{-1})\quad\text{and}\quad\bbR^{(b)}:=(\bbX^{(b)})'\diag(\bbX^{(b)}(\bbX^{(b)})')^{-1}\bbX^{(b)},$$
        then compute $\widehat{T}_n^{(b)}(0):=\sqrt{n}\big(n^{-1}\Vert\bbR^{(b)}\Vert-\mbE[\fM_{1,1}(1)]\big)$. Finally, \(u_{\alpha/2}\) and \(u_{1-\alpha/2}\) can be estimated by the lower and upper \(\alpha/2\) sample quantiles of \(\{\widehat{T}_n^{(1)}(0),\cdots,\widehat{T}_n^{(B)}(0)\}\).
	\end{enumerate}

    \subsection{Estimating the number of unit roots}\label{Main sec of multiple unit root tests}
    In this subsection, we propose a forward sequential procedure to determine the number of unit roots in an AR($d$) $X_t$. Precisely, suppose \(X_t\) is generated by (\ref{Main Eq of AR process}). For \(1\leq p\leq d\), define a sequence of hypotheses as follows:  
	\begin{center}
		\(\mbH_0^{(p)}\): \(X_t\) has \(p\) unit roots, i.e. \((1-L)^pX_t=e_t\). 
	\end{center}
	For univariate time series data, \cite{dickey1987determining} conducted a backward sequential test for all \(\mbH_0^{(p)}\), that is, starting from testing \(\mbH_0^{(d_0)}\) versus \(\mbH_0^{(d_0-1)}\) for a pre-determined \(d_0\in\mbN^+\), if \(\mbH_0^{(d_0)}\) is rejected, then they continued to test \(\mbH_0^{(d_0-1)}\) versus \(\mbH_0^{(d_0-2)}\); otherwise, they accepted \(\mbH_0^{(d_0)}\). However, determining a proper \(d_0\) is not trivial. For instance, suppose \(X_t\) has \(p_0\in\mbN^+\) unit roots with \(p_0\) being unknown. Choosing a small \(d_0<p_0\) is problematic, while a large \(d_0>p_0\) will make this backward sequential test inefficient. 
	
	Since we have established the unit root test in Section \ref{Main sec of unit root test}, it would be more proper to conduct a forward sequential test to determine the number of unit roots. Specifically, starting from testing \(\mbH_0^{(0)}\) versus \(\mbH_0^{(1)}\), if \(\mbH_0^{(0)}\) is rejected, then we test \(\mbH_0^{(1)}\) versus \(\mbH_0^{(2)}\) and so on in general. Hence, it remains to construct a test statistic for testing \(\mbH_0^{(p)}\) versus \(\mbH_0^{(p+1)}\).
	
	Given the observations \(\bbX=[X_1,\cdots,X_T]\), define the following operator:
	\begin{align}
		\mathscr{T}_p(\bbX):=\left\{\begin{array}{ll}
			\bbX\bbU^{-p}(\bbM\bbU'\bbU\bbM)^{p/2},&p>0,\ p{\rm\ is\ even};\\
			\bbX\bbU^{-p}\bbM\bbU'(\bbU\bbM\bbU')^{(p-1)/2},&p>0,\ p{\rm\ is\ odd};
		\end{array}\right.\label{Main Eq of operator msT}
	\end{align}
	and
	\begin{align}
		\hR(p):=\mathscr{T}_p(\bbX)'\diag(\mathscr{T}_p(\bbX)\mathscr{T}_p(\bbX)')^{-1}\mathscr{T}_p(\bbX),\label{Main Eq of hat Rp}
	\end{align}
	where \(\bbM=\bbI_T-\boldsymbol{1}_{T\times T}\) and \(\bbU\) is defined in \eqref{Main Eq of sigma_k}. Next, we have
    \begin{thm}\label{Main Thm of multiple unit roots test}
        Under Assumptions \ref{Main Ap of panel lag polynomial}, \ref{Main Ap of nonpanel}, \ref{Main Ap of highdimensionality} and \ref{Main Ap of m dependent}, for any $p\in\mbN^+,p\geq2$, suppose \(X_t\) is generated by
		$$(1-L)^pX_t=\bbGa\sum_{k=0}^{\infty}\Psi_k\varepsilon_{t-k},\quad\varepsilon_t\overset{i.i.d.}{\sim}\mcN(\boldsymbol{0},\bbI_n),$$
		let \(\hla_1(p)\) and $\hla_1(p-1)$ be the largest eigenvalue of \(\hR(p)\) and \(\hR(p-1)\) defined in {\rm (\ref{Main Eq of hat Rp})}, respectively, then we have
		\begin{align}
			\left\{\begin{array}{l}
			     \frac{\sqrt{n}}{\mfm_{1,1}(p)}\left(\frac{\hla_1(p)}{n}-\mbE[\fM_{1,1}(p)]\right)\overset{d}{\longrightarrow}\mcN(0,1),\\
			     \frac{\sqrt{n}}{\mfm_{1,1}(p)}\left(\frac{\hla_1(p-1)}{n}-\mbE[\fM_{1,1}(p)]\right)\overset{d}{\longrightarrow}\mcN(0,1), 
			\end{array}\right.\label{Main Eq of multiple unit roots test}
		\end{align}
		where \(\fM_{1,1}(x)\) and \(\mfm_{1,1}(x)\) are defined in {\rm (\ref{Main Eq of fM})} and {\rm (\ref{Main Eq of mfm})}, respectively.
    \end{thm}
	Readers can refer to \S\ref{ssec of test statistics} in the supplement for the proof of Theorem \ref{Main Thm of multiple unit roots test}. Here, we briefly outline the proof of Theorem \ref{Main Thm of multiple unit roots test}. For example, when \(p=2\), we have \(\bbX=\bbe\bbU^2\). By \eqref{Main Eq of operator msT} and \eqref{Main Eq of hat Rp},
    $$\hR(2)=\bbM\bbU'\bbU\bbM\bbe'\diag(\bbe\bbM\bbU'\bbU\bbM\bbU'\bbU\bbM\bbe')^{-1}\bbe\bbM\bbU'\bbU\bbM.$$
    By the SVD of $\bbM\bbU'$ in \eqref{Main Eq of sigma_k}, then \(\bbM\bbU'\bbU\bbM=\bbM\bbU'(\bbM\bbU')'=\sum_{t=1}^{T-1}\sigma_t^2\bbw_t\bbw_t'\). Similar as \eqref{Main Eq of tilde M st}, we can represent the largest eigenvalue $\hla_1(2)$ of $\hR(2)$ as follows:
	$$\frac{\hla_1(2)}{n}=\sum_{s,t=1}^{T-1}\alpha_{1,s}\alpha_{1,t}\frac{1}{n}\sum_{j=1}^n\frac{\sigma_s^2\sigma_t^2(\bbe_j\bbw_s)(\bbe_j\bbw_t)}{\sum_{l=1}^{T-1}\sigma_l^4(\bbe_j\bbw_l)^2},$$
	and we can establish the CLT for $\hla_1(2)$ by the same argument as proving Theorem \ref{Main Thm of CLT I1}. 
    
    To test $\mbH_0^{(p)}$ versus $\mbH_0^{(p+1)}$, construct the test statistic as follows: 
	\begin{align}
		\widehat{T}_n(p):=\sqrt{n}\Bigg(\frac{\hla_1(p)}{n}-\mbE[\fM_{1,1}(p)]\Bigg).\label{Main Eq of hat Tp}
	\end{align}
	Theorem \ref{Main Thm of multiple unit roots test} implies that
	\begin{align*}
		\left\{\begin{array}{ll}
			     \frac{\sqrt{n}}{\mfm_{1,1}(p)}\left(\frac{\hla_1(p)}{n}-\mbE[\fM_{1,1}(p)]\right)\overset{d}{\longrightarrow}\mcN(0,1),&\text{under }\mbH_0^{(p)},\\
			     \frac{\sqrt{n}}{\mfm_{1,1}(p+1)}\left(\frac{\hla_1(p)}{n}-\mbE[\fM_{1,1}(p+1)]\right)\overset{d}{\longrightarrow}\mcN(0,1),&\text{under }\mbH_0^{(p+1)}. 
			\end{array}\right.
	\end{align*}
	Note that \(\fM_{1,1}(x)\) is strictly increasing by (\ref{Main Eq of fM}), and we have \(\mbE[\fM_{1,1}(p+1)]>\mbE[\fM_{1,1}(p)]\). Moreover, we can show that $\mfm_{1,1}(p)\asymp\mrO(1)$ by the same argument as \eqref{Main Eq of mfm}. Hence, we obtain that 
	$$\left\{\begin{array}{ll}
		\widehat{T}_n(p)/\mfm_{1,1}(p)\overset{d}{\longrightarrow}\mcN(0,1),&{\rm under\ }\mbH_0^{(p)},\\
		\widehat{T}_n(p)\asymp\mrO_{\mbP}(\sqrt{n}),&{\rm under\ }\mbH_0^{(p+1)}.
	\end{array}\right.$$
	We reject \(\mbH_0^{(p)}\) if \(\widehat{T}_n(p)>\log(n)\) and the asymptotic power is \(\lim_{n\to\infty}\mbP\big(\widehat{T}_n(p)>\log(n)\big|\mbH_0^{(p+1)}\big)=1\).
    
    Finally, combining with the unit root test in Section \ref{Main sec of unit root test}, we can determine the number of unit roots in high-dimensional time series by the following forward sequential tests:
	\begin{enumerate}
		\item Given \(\bbX\), we first conduct the unit root test in \eqref{Main Eq of unit root test}. Precisely, we construct the test statistic \(\widehat{T}_n(0)\) by \eqref{Main Eq of test statistic}. If \(\widehat{T}_n(0)<-\log(n)\), we reject \(H_1\) and move to step 2; otherwise, we accept \(H_1\) and stop.
		\item Suppose our current test is
		\begin{center}
			\(\mbH_0^{(p)}\): \(X_t\) has \(p\) unit roots.\quad versus\quad\(\mbH_0^{(p+1)}\): \(X_t\) has \(p+1\) unit roots.
		\end{center} 
		where \(p\geq1\). Construct \(\widehat{T}_n(p)\), if \(\widehat{T}_n(p)>\log(n)\), we reject \(\mbH_0^{(p)}\) and move to test \(\mbH_0^{(p+1)}\) versus \(\mbH_0^{(p+2)}\). Otherwise, we reject \(\mbH_0^{(p+1)}\) and accept \(\mbH_0^{(p)}\).
	\end{enumerate}
	\begin{remark}
	    For the backward sequential test for univariate time series in \cite{dickey1987determining}, the acceptance rate, $\mbP(\text{Accept }\mbH_0^{(p)}|\mbH_0^{(p)})$ is much smaller than 1. On the contrary, the numerical experiment in Section \ref{Main sec of Experiment 3} shows that the acceptance rate of our method is always 1, which suggests that our method can report the true number of unit roots in high-dimensional time series data with high accuracy.
	\end{remark}

    \section{Numerical Experiments}\label{Main Sec of numerical}
	In this section, we conduct three numerical experiments to verify Theorems \ref{Main Pro of root critiria}, \ref{Main Thm of CLT multiple unit roots} and demonstrate the performance of two hypothesis tests in Section \ref{Main Sec of applications}.

    \subsection{Experiment 1: asymptotic spectral properties of the sample correlation matrices}
    In this subsection, we verify Theorems \ref{Main Pro of root critiria} and \ref{Main Thm of CLT multiple unit roots}. Construct the following three AR processes:
	\begin{align}
		\left\{\begin{array}{l}
			(1-0.2L)(1+0.5L)X_t^{[1]}=\bbGa \varepsilon_t,\\
			(1-L)(1-e^{{\rm i}\pi/3}L)(1-e^{-{\rm i}\pi/3}L)X_t^{[2]}=\bbGa \varepsilon_t,\\
			(1-0.6L)(1+L)(1+2L)X_t^{[3]}=\bbGa\varepsilon_t,
		\end{array}\right.\label{Main Eq of experiment 1}
	\end{align}
    where \(\bbGa=[(1+|s-t|)^{-1}1_{|s-t|<T^{1/4}}]_{s,t}\in\mbR^{n\times n}\) is the cross-sectional matrix and $\varepsilon_t\overset{i.i.d.}{\sim}\mcN(\boldsymbol{0},\bbI_n)$. Denote \(\hat{\bbR}^{[r]}\) as the sample correlation matrix of \(\bbX^{[r]}=[X_1^{[r]},\cdots,X_T^{[r]}]\) for \(r=1,2,3\), respectively, and let \(\hla_k^{[r]}\) be the \(k\)-th largest eigenvalue of \(\hat{\bbR}^{[r]}\). For each \(r\) and different values of $n$ and $T$, we calculate the sample mean and standard deviation of \(n^{-1}\hla_1^{[r]},n^{-1}\hla_2^{[r]},n^{-1}\hla_3^{[r]}\) based on 200 independent simulations, see Table \ref{Tab of 2}. 
    \begin{table}[!t]
        \caption{Sample mean (S.M.) and standard deviation (S.D.) of $n^{-1}\hla_k^{[r]}$.\label{Tab of 2}}%
        \begin{tabular*}{\columnwidth}{@{\extracolsep\fill}cc|cc|cc|cc@{\extracolsep\fill}}
            \toprule
            &&\multicolumn{2}{c}{$\hla_1^{[r]}$}&\multicolumn{2}{c}{$\hla_2^{[r]}$}&\multicolumn{2}{c}{$\hla_3^{[r]}$}\\\hline
             &$(n,T)$&S.M.&S.D.&S.M.&S.D.&S.M.&S.D.\\\hline
             $X_t^{[1]}$&$(200,400)$&0.0981&0.0042&0.0838&0.0032&0.0767&0.0031\\
             $X_t^{[1]}$&$(500,1000)$&0.0485&0.0015&0.0457&0.0011&0.0414&0.0011\\
             $X_t^{[1]}$&$(1000,2000)$&0.0288&0.0006&0.0276&0.0005&0.0267&0.0004\\\hline
             $X_t^{[2]}$&$(500,1000)$&0.2755&0.0220&0.2206&0.0174&0.1652&0.0158\\
             $X_t^{[2]}$&$(1000,2000)$&0.2671&0.0178&0.2185&0.0137&0.1613&0.0134\\\hline
             $X_t^{[3]}$&$(100,200)$&1&0&0&0&0&0\\
             $X_t^{[3]}$&$(500,1000)$&1&0&0&0&0&0\\\hline
        \end{tabular*}
    \end{table}
    For $X_t^{[1]}$ whose characteristic polynomial having stationary roots only, Theorem \ref{Main Pro of root critiria} claims that $n^{-1}\hla_1^{[1]}\overset{\mbP}{\longrightarrow}0$. According to Table \ref{Tab of 2}, as $n$ and $T$ increase, both the sample mean and standard deviation of $n^{-1}\hla_1^{[1]}$ decrease to 0, which agrees with the first conclusion in Theorem \ref{Main Pro of root critiria}. For $X_t^{[3]}$, there exists one super nonstationary root $-2$, and Table \ref{Tab of 2} shows that ${\rm rank}(\bbR^{[3]})=1$, which agrees with the second conclusion in Theorem \ref{Main Pro of root critiria}. Finally, for $X_t^{[2]}$ with three nonstationary roots, Theorem \ref{Main Thm of CLT multiple unit roots} claims that $\hla_k^{[2]}$ is asymptotically normal. Such asymptotic normality is demonstrated in QQ plots in Figure \ref{Fig of 1}.
    \begin{figure}[!t]
		\centering
		\subfigure[QQ plots of \(n^{-1}\hla_1^{[2]}\), $(n,T)=(500,1000)$.]{\includegraphics[width=0.49\linewidth]{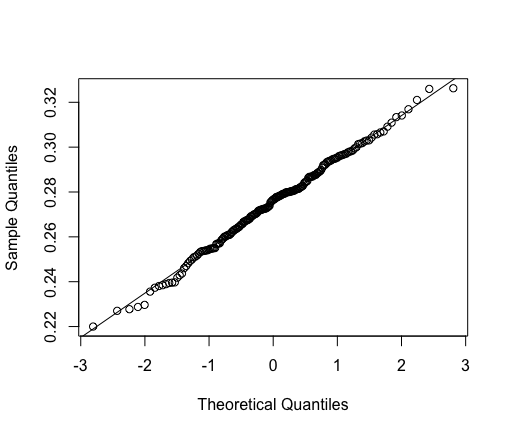}}
		\subfigure[QQ plots of $n^{-1}\hla_2^{[2]}$, $(n,T)=(1000,2000)$.]{\includegraphics[width=0.49\linewidth]{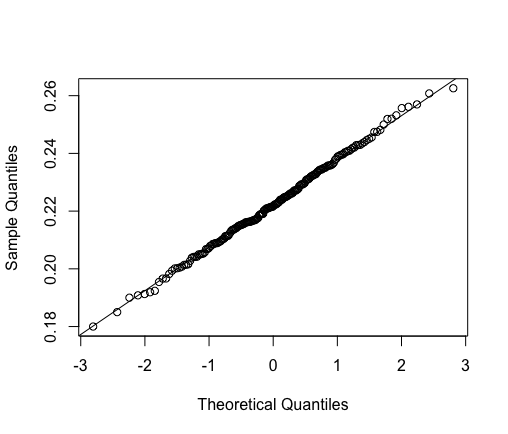}}
		\caption{QQ plots of \(\bbR^{[2]}\)'s first and second largest eigenvalue $\hla_1^{[2]}$ and $\hla_2^{[2]}$ from 200 independent repetitions under different values of data dimension $n$ and sample size $T$.}
		\label{Fig of 1}
	\end{figure}

    \subsection{Experiment 2: unit root test}
	The second experiment is to verify the performance of our unit root test. Suppose $X_t$ is generated by \eqref{Main Eq of unit root model}, where the noise process $e_t$ is generated by $e_t=\bbGa\fe_t$ with the cross-sectional matrix \(\bbGa=[(1+|s-t|)^{-1}1_{|s-t|<T^{1/4}}]_{s,t}\in\mbR^{n\times n}\) and
    \begin{align}
        \fe_t=\diag(P_1,\cdots,P_n)\fe_{t-1}+\varepsilon_t,\quad P_i=0.5+0.2\sin(2\pi i/n),\label{Main Eq of fe_t}
    \end{align}
    where $\varepsilon_t\overset{i.i.d.}{\sim}\mcN(\boldsymbol{0},\bbI_n)$. It is easy to check \(\bbGa\) satisfies Assumptions \ref{Main Ap of nonpanel} and \ref{Main Ap of m dependent}. Consider the following unit root test:
    $$H_0:\bbPi=\bbI_n\quad{\rm versus}\quad H_1:\bbPi=\left(\begin{array}{cc}
		\bbI_{n_1}&\boldsymbol{0}_{n_1\times n_2}\\
		\boldsymbol{0}_{n_2\times n_1}&\widetilde{\bbPi}
	\end{array}\right),$$
	where \(n_1=0.4n\) and $\widetilde{\bbPi}=[0.4^{1+|i-j|}]_{i,j}\in\mbR^{n_2\times n_2}$. It is easy to check that \(\Vert\widetilde{\bbPi}\Vert<1\). We set \(X_0=\phi=\boldsymbol{0}\) and the significance level \(\alpha=0.05\). For different values of data dimension $n$ and sample size $T$, we generate 200 independent \(\bbX=[X_1,\cdots,X_T]\) under both \(H_0\) and \(H_1\). Based on the bootstrap method in Section \ref{Main sec of unit root test}, we compute the empirical size and power and display them in Table \ref{Tab of 1}. 
    \begin{table}[!t]
        \caption{Empirical size/power.\label{Tab of 1}}%
        \begin{tabular*}{\columnwidth}{@{\extracolsep\fill}c|cccc@{\extracolsep\fill}}
            \toprule
            \(n\backslash T\)&60&80&100&120\\\hline
			60&0.095/1&0.090/1&0.060/1&0.055/1\\
			80&0.085/1&0.060/1&0.055/1&0.050/1\\
			100&0.080/1&0.055/1&0.060/1&0.045/1\\
			120&0.070/1&0.060/1&0.050/1&0.045/1\\\hline
        \end{tabular*}
    \end{table}
	All empirical powers are 1, which suggests our unit root can efficiently distinguish \(H_0\) and \(H_1\). When \(n,T\) is relatively small (e.g. \(n=T=60\)), the empirical size is slightly greater than \(0.05\). For larger \(n,T\), all empirical sizes are close to the significance level \(\alpha=0.05\), so our unit root test indeed has good performance for high-dimensional cross-sectional dependent time series data.


    \subsection{Experiment 3: forward sequential test}\label{Main sec of Experiment 3}
	The third experiment is to verify the power of our forward sequential test. The underlying settings are the same as those in Experiment 1 and 2. Precisely, we assume that 
    \begin{align}
        \prod_{l=1}^d(1-\mfa_lL)X_t=\bbGa\fe_t,\label{Main Eq of experiment 3}
    \end{align}
    where the cross-sectional matrix \(\bbGa=[(1+|s-t|)^{-1}1_{|s-t|<T^{1/4}}]_{s,t}\in\mbR^{n\times n}\) and $\fe_t$ is defined as \eqref{Main Eq of fe_t}. For different values of $\mfa_l$ and data dimension $n$ and sample size $T$, we simulate 200 independent samples \(\bbX=[X_1,\cdots,X_T]\), and compute the empirical acceptance rates of our forward sequential test in Section \ref{Main sec of multiple unit root tests} under the following hypotheses:
	\begin{center}
		\(\mbH_0^{(0)}\): \(X_t\) is stationary.\quad\(\mbH_0^{(p)}\): \(X_t\) has \(p\in\mbN^+\) unit roots.
	\end{center} 
    All simulation results about the empirical acceptance rate are summarized in Table \ref{Tab of 3}. For example, roots \((1,1,1)\) means that \(X_t\) is generated by \((1-L)^3X_t=\bbGa\fe_t\). Note that all empirical acceptance rates in Table \ref{Tab of 3} are \(1\), so our forward sequential method is powerful.
    \begin{table}[!t]
        \caption{Empirical acceptance rate.\label{Tab of 3}}%
        \begin{tabular*}{\columnwidth}{@{\extracolsep\fill}c|c|cccc@{\extracolsep\fill}}
            \toprule
            Roots&\(n,T\)&Accept \(\mbH_0^{(0)}\)&Accept \(\mbH_0^{(1)}\)&Accept \(\mbH_0^{(2)}\)&Accept \(\mbH_0^{(3)}\)\\\hline
			\((1,1,1)\)&\((50,100)\)&0&0&0&1\\
			\((1,1,1)\)&\((120,150)\)&0&0&0&1\\
			\((1,1,0.5)\)&\((100,80)\)&0&0&1&0\\
			\((1,0.5{\rm i},-0.5{\rm i})\)&\((100,50)\)&0&1&0&0\\
			\((0.8,0.5,-0.1)\)&\((60,80)\)&1&0&0&0\\\hline
        \end{tabular*}
    \end{table}

     \section{Empirical Study}\label{Main Sec of PPP}
	In this section, we conduct an empirical study about the PPP hypothesis under high-dimensional settings. As a metric for evaluating the relative value of specific goods across different countries, PPP is widely used to compare the absolute purchasing power of countries' currencies. Roughly speaking, if the null hypothesis of PPP holds, real exchange rates of different countries' currencies will be stationary in the long run. Let \(\bbr_t=(r_{1,t},\cdots,r_{n,t})\), where \(r_{i,t}\) is the logarithm of real exchange rates for \(i\)-th country at time \(t\). Suppose
	\begin{align}
		r_{i,t}=\mfc_i+\mfb_i r_{i,t-1}+e_{i,t}\label{Main Eq of exchange rates}
	\end{align} 
	where \(e_{i,t}\) is a stationary error term, \(\mfc_i\) and \(\mfb_i\) are unknown parameters. The null and alternative of PPP hypotheses are
	\begin{align}
		H_0:\mfb_i=1\quad{\rm versus}\quad H_1:\mfb_i<1,\notag
	\end{align}
	Empirical evidence on the stationarity of real exchange rates is abundant but inconclusive. Some earlier works, e.g. \cite{abuaf1990purchasing,jorion1996mean} and \cite{wu1996real}, provided strong evidences to reject \(H_0\) under a ten-country multicurrency system from 1970 to 2000, whereas \cite{cheung1993long,papell1997searching} and \cite{flores1999multivariate} provided some examples to support \(H_0\). In fact, \cite{o1998overvaluation} pointed out some limitations in earlier works indeed making \(H_0\) over-rejected. Later, many researchers aimed to solve this problem by refining the model settings and data structures. For example, \cite{im2003testing} and \cite{hadri2008panel} used the panel unit root tests for the PPP hypothesis to increase the power. \cite{pesaran2012interpretation} suggested that the alternative hypothesis is controversial in many previous literature by assuming \(H_1:{\rm all\ }\mfb_i<1\); instead, a more appropriate alternative should be
	$$H_1:\mfb_i<1,\quad i=1,\cdots,n_1;\ \ \mfb_i=1,\quad i=n_1+1,\cdots,n,\quad\text{for some }n_1<n.$$
	For more recently literature about the PPP hypothesis, readers can refer to \cite{sarno2002purchasing,taylor2002purchasing,pedroni2004panel,banerjee2005testing,moon2005efficient} and \cite{taylor2013real} for in-depth information on the theoretical and empirical aspects.
 
	For our test, we collect the {\bf monthly} (from Jan 2001 to Dec 2021, in total 252 samples) and {\bf quarterly} (from Q1 2001 to Q4 2021, in total 81 samples) period-ending exchange rates (National Currency Per U.S. Dollar) of all countries from the International Monetary Fund’s International Financial Statistics. The sizes of monthly and quarterly data are \(104\) and \(110\), respectively. Let \(\bbr_t\) be the logarithm of real exchange rates of all collected countries at time \(t\), and consider the PPP hypothesis under high-dimensional settings as follows:
	\begin{align}
		H_0:\bbr_t=\bbr_{t-1}+e_t\quad{\rm versus}\quad H_1:\bbr_t=\mfc+\bbPi\bbr_{t-1}+e_t,\label{Main Eq of PPP hypothesis}
	\end{align}
	where \(e_t=\bbGa\sum_{k=0}^{\infty}\Psi_k\varepsilon_{t-k},\varepsilon_t\overset{i.i.d.}{\sim}\mcN(\boldsymbol{0},\bbI_n)\) satisfying Assumptions \ref{Main Ap of panel lag polynomial}, \ref{Main Ap of nonpanel}, \ref{Main Ap of m dependent} and $\bbPi$ is generated as \eqref{Main Eq of partially stationary}
	such that \(\Vert\widetilde{\bbPi}\Vert<1\) and \(c_1:=\lim_{n\to\infty}\frac{n_1}{n}\in[0,1)\). Note that \(\widetilde{\bbPi}\) is not necessarily diagonal. The cross-sectional dependence among \(\bbr_t\) and \(e_t\) is captured by $\bbPi$ and $\bbGa$. We use the \(\widehat{T}_n(0)\) in (\ref{Main Eq of test statistic}) as the test statistic following the established framework in Section \ref{Main sec of unit root test}.
	
	For the monthly data, \(\widehat{T}_n(0)=-14.69<-\log(n)=-4.64\), so we reject \(H_0\). For the quarterly data, \(\widehat{T}_n(0)=-15.11<-\log(n)=-4.70\), so we also reject \(H_0\). In conclusion, our unit root test provides strong evidences that the exchange rates \(\bbr_t\) is not a high-dimensional random walk process, and \(\bbr_t=(r_{1,t},\cdots,r_{n,t})\) indeed contains stationary components as \cite{pesaran2012interpretation} suggested.

\newpage

\newpage
\begin{center}
{\LARGE {\bf Supplementary Materials of the paper ``Testing High-dimensional Nonstationary Time Series''}}
\end{center}
This supplementary document provides all the technical proofs of the results of this paper. It is self-contained without using the results of the main paper.

\appendix
\section{Basic settings}
\setcounter{equation}{0}
\def\theequation{\thesection.\arabic{equation}}
\setcounter{subsection}{0}
In the beginning, let's make some notations here:
\begin{itemize}
	\item \(C_a\) represents a positive constant which depends on some parameter \(a\).
		\item For two real sequence \(\{a_n\}\) and \(\{b_n\}\), $a_n\asymp\mrO(b_n)$ means that for any $n_0\in\mbN^+$, there exists two positive constants $M_1,M_2$ such that
        $$M_1 b_n\leq a_n\leq M_2 b_n,\ \forall n>n_0$$
        Moreover, if \(\{a_n\}\) and \(\{b_n\}\) are sequence of real random variables, then
        \begin{align}
            a_n\asymp\mrO_{\mbP}(b_n)\quad\Longleftrightarrow\quad\lim_{n\to\infty}\mbP(M_1 b_n\leq a_n\leq M_2 b_n)=1\label{Eq of asymp mbP}
        \end{align}
        for two positive constants $M_1=M_1(n_0),M_2=M_2(n_0)$.
		\item Given any integrable random variable/vector \(X\), \(X^{\circ}:=X-\mbE[X]\) denotes its centered version.
		\item The \(L^2\) convergence, the convergence in probability and in distribution are denoted by \(\overset{L^2}{\longrightarrow},\overset{\mbP}{\longrightarrow}\) and \(\overset{d}{\longrightarrow}\), respectively.
        \item Given a matrix \(\bbA=[A_{i,j}]_{n\times n}\), \(\tr(\bbA)=\sum_{i=1}^nA_{i,i}\), \(\bbA'\) denotes the transpose of \(\bbA\), and \(\diag(\bbA)\) is the diagonal matrix made with the main diagonal of \(\bbA\). Moreover, \(\Vert\bbA\Vert\) denotes the spectral norm of \(\bbA\).
\end{itemize}
Next, we present some necessary assumptions and basic settings of our models. Let \(X_t\) be a \(n\)-dimensional time series data, then given the $T$ observations \(X_1,\cdots,X_T\), the high-dimensionality scheme is
\begin{Ap}\label{Ap of highdimensionality}
	As \(n\to\infty\), the number of observations \(T\) also tends to infinity such that \(\lim_{n\to\infty}n/T=c\in(0,\infty)\).
\end{Ap}
Denote $\bbX=[X_1,\cdots,X_T]\in\mbR^{n\times T}$ as the data matrix and \(\bbM:=\bbI_T-\boldsymbol{1}_{T\times T}/T\), where \(\bbI_T\) is the identity matrix with size of \(T\times T\) and \(\boldsymbol{1}_{T\times T}\) is a \(T\times T\) matrix whose entries are all \(1\), then the sample covariance matrix of \(\bbX\) is 
\begin{align}
	\hat{\bbSig}:=\frac{1}{n}\bbM\bbX'\bbX\bbM.\label{Eq of covariance matrix}
\end{align}
Moreover, the sample correlation matrix of $\bbX$ is $\bbD^{-1/2}\bbX\bbM\bbX'\bbD^{-1/2}$, where \(\bbD:=\diag(\bbX\bbM\bbX')\). Since we only focus on the extreme eigenvalues of the sample correlation matrix, and the nonzero eigenvalues of
\begin{align}
	\hat{\bbR}:=\bbM\bbX'\bbD^{-1}\bbX\bbM.\label{Eq of correlation matrix}
\end{align}
and $\bbD^{-1/2}\bbX\bbM\bbX'\bbD^{-1/2}$ are coincide. Without loss of generality, we regard the matrix $\hR$ in \eqref{Eq of correlation matrix} as the sample correlation matrix of $\bbX$ in this article.
Finally, let's briefly introduce the structures of this Supplementary Materials. In \S\ref{Sec of correlation independent} and \S\ref{Sec of correlation dependent}, we establish the joint CLT for the first \(K\) largest eigenvalue of \(\hR\) for cross-sectional independent \(X_t\) and cross-sectional dependent \(X_t\), respectively. In \S\ref{Sec of covariance}, we establish the joint CLT for the first \(K\) largest eigenvalue of \(\hSig\). In \S\ref{Sec of AR process}, we investigate the asymptotic spectral behaviors of the sample correlation matrix generated by more general high-dimensional autoregressive processes. In \S\ref{Sec of applications}, we first investigate the asymptotic behaviors of the sample correlation matrix under the alternative hypothesis of the unit root test in \S\ref{sec of unit root test}, then further construct the statistic for the forward sequential test to determine the number of unit roots in high-dimensional time series in \S\ref{sec of number of unit roots}.

\section{CLT for extreme eigenvalues of the sample correlation matrix of high-dimensional random walks without cross-sectional dependence}\label{Sec of correlation independent}
\setcounter{equation}{0}
\def\theequation{\thesection.\arabic{equation}}
\setcounter{subsection}{0}
Let $X_t$ is a cross-sectional independent random walk generated as follows: 
\begin{align}
	X_t=X_{t-1}+e_t,\quad e_t=\Psi(L)\varepsilon_t=\sum_{k=0}^{\infty}\Psi_k\varepsilon_{t-k},\label{Eq of panel Xt}
\end{align}
where the coefficients \(\{\Psi_k:k\in\mbN\}\) of \(\Psi(L)\) satisfy that
\begin{Ap}\label{Ap of panel lag polynomial}
	All \(\{\Psi_k:=\diag(\varphi_{1,k},\cdots,\varphi_{n,k})\in\mbR^{n\times n}:k\in\mbN^+\}\) are diagonal and there exists two positive constants \(b,B\) such that
	$$\sum_{k=0}^{\infty}(1+k)\Vert\Psi_k\Vert\leq B,\quad{\rm and}\quad\min_{1\leq j\leq n}\inf_{x\in[-\pi,\pi]}\left|\sum_{k=0}^{\infty}\varphi_{j,k}e^{{\rm i}kx}\right|\geq b.$$
\end{Ap}
Moreover, $\varepsilon_t=(\varepsilon_{1,t},\cdots,\varepsilon_{n,t})'$ satisfies that
\begin{Ap}\label{Ap of finite integration}
	All $\varepsilon_{i,t}$ are independent such that $\mbE[\varepsilon_{i,t}]=0$ and $\mbE[\varepsilon_{i,t}^2]=1$ for all $1\leq i\leq n$ and $t\in\mbZ$. Moreover, the densities of all \(\varepsilon_{i,t}\) has a uniform bound for \(i=1,\cdots,n\), i.e. there exists an \(M>0\) such that 
	$${\rm ess}\sup_{x\in\mathbb{R}}p_{\varepsilon_{i,t}}:=\inf\left\{x\in\mbR:\mu(p_{\varepsilon_{i,t}}^{-1}(x,+\infty))=0\right\}\leq M$$
	for all \(i,t\), where \(p_{\varepsilon_{i,t}}(x)\) is the density of \(\varepsilon_{i,t}\) and \(\mu\) is the Borel measure. Besides, we further denote \(\kappa_l:=\sup_{i,t}\mathbb{E}[|\varepsilon_{i,t}|^l]\) for \(l\in\mbN^+\) and assume that all \(\varepsilon_{i,t}\) have uniformly bounded \(8\)-th moment, i.e. \(\kappa_8<\infty\).
\end{Ap}
Given the data matrix \(\bbX=[X_1,\cdots,X_T]\), for any positive integer $K\in\mbN^+$, we will establish the joint CLT for the first $K$ largest eigenvalues of the sample correlation matrix of $\bbX$. Here, let \(\bbU\) be a \(T\times T\) upper triangular matrix with ones above and on the main diagonal, then \eqref{Eq of panel Xt} implies that 
\begin{align}
	\bbX=\bbe\bbU,\label{Eq of bbX}
\end{align}
where \(\bbe=[e_1,\cdots,e_T]\) is the noise matrix. Thus, by \eqref{Eq of correlation matrix}, the sample correlation matrix \(\hR\) of \(\bbX\) in \eqref{Eq of bbX} is
\begin{align}
    \hR=\bbM\bbU'\bbe'\diag(\bbe\bbU\bbM\bbU'\bbe')^{-1}\bbe\bbU\bbM.\label{Eq of correlation matrx random walk 1}
\end{align}
Let \(\hla_1\geq\cdots\geq\hla_K\) be the first \(K\in\mathbb{N}^+\) largest eigenvalues of \(\hat{\bbR}\), we have
\begin{thm}\label{Thm of CLT}
	Under Assumptions {\rm \ref{Ap of highdimensionality}, \ref{Ap of panel lag polynomial}} and {\rm \ref{Ap of finite integration}}, given the observations \(\bbX=[X_1,\cdots,X_T]\) generated by {\rm (\ref{Eq of panel Xt})}, let \(\hla_1\geq\cdots\geq\hla_K\) be the first \(K\in\mbN^+\) largest eigenvalues of the sample correlation matrix \(\hat{\bbR}\) \eqref{Eq of correlation matrx random walk 1} of \(\bbX\), then
	\begin{align}
		\sqrt{n}\left(\frac{\hla_1}{n}-\mbE[\fM_{1,1}],\cdots,\frac{\hla_K}{n}-\mbE[\fM_{K,K}]\right)'\overset{d}{\longrightarrow}\mcN(\bbzeta,\mcS),\label{Eq of correlation CLT}
	\end{align}
	Here, for any $k,l\in\mbN^+$, \(\fM_{k,l}\) is a random variable defined as 
    \begin{align}
        \fM_{k,l}:=\frac{(kl)^{-1}Z_k Z_l}{\sum_{t=1}^{\infty}t^{-2}Z_t^2},\label{Eq of fM}
    \end{align}
    where \(\{Z_t:t\in\mbN^+\}\) is a sequence of i.i.d. \(\mcN(0,1)\). The asymptotic mean \(\bbzeta=(\zeta_1,\cdots,\zeta_K)\)' satisfies that \(|\zeta_k|\leq C_{B,M,\kappa_8,c}\) and \(\mcS\) is the \(K\times K\) covariance matrix such that \(\mcS_{k,l}:=\Cov(\fM_{k,k},\fM_{l,l})\) for \(1\leq k,l \leq K\).
\end{thm}
Next, we present some technical preliminaries for proving Theorem \ref{Thm of CLT}. Denote the singular value decomposition (SVD) of \(\bbM\bbU'\) as
\begin{align}
	\bbM\bbU':=\sum_{k=1}^{T-1}\sigma_k\bbw_k\bbv_k',\quad\quad\sigma_k:=\mu_k^{-1/2}=[2\sin(\pi k/(2T))]^{-1}\label{Eq of SVD of MU}
\end{align}
where \(\bbv_k:=(v_{k,1},\cdots,v_{k,T})'\) and \(\bbw_k:=(w_{k,1},\cdots,w_{k,T})'\) such that
\begin{align}
	v_{k,t}=\sqrt{\frac{2}{T}}\sin(\pi k(t-1)/T),\quad w_{k,t}=-\sqrt{\frac{2}{T}}\cos(\pi k(2t-1)/(2T))\label{Eq of vk}
\end{align}
for \(k=1,\cdots,T-1\). When \(k=T\), \(\sigma_T=0,\bbv_T=(1,0\cdots,0)'\) and \(\bbw_T=\boldsymbol{1}_T/\sqrt{T}\).

Moreover, let \(\hat{F}_1,\cdots,\hat{F}_K\) be the corresponding normalized eigenvectors of $\hla_1,\cdots,\hla_K$. Since \(\{\bbw_1,\cdots,\bbw_T\}\) is a basis of \(\mathbb{R}^T\) and \(\bbM\bbw_T=\boldsymbol{0}\), then $\hR\bbw_T=\boldsymbol{0}$ and $\hat{F}_k'\bbw_T=0$. Hence, we represent \(\hat{F}_k\) by
\begin{align}
	\hat{F}_k:=\sum_{t=1}^{T-1}\alpha_{k,t}\bbw_t.\label{Eq of F1}
\end{align}
where \(\sum_{t=1}^{T-1}\alpha_{k,t}^2=1\). Therefore, we obtain
\begin{align}
	\hla_k=\hat{F}_k'\bbM\bbU'\bbe'\bbD^{-1}\bbe\bbU\bbM\hat{F}_k=\sum_{s,t=1}^{T-1}\alpha_{k,s}\alpha_{k,t}\sigma_s\sigma_t\bbv_s'\bbe'\bbD^{-1}\bbe\bbv_t.\notag
\end{align}
Let \(\bbe_j\) be the \(j\)-th {\bf row} of \(\bbe\) for \(1\leq j\leq n\), then \(D_{j,j}=\sum_{k=1}^{T-1}\sigma_k^2(\bbe_j\bbv_k)^2\) and we define
\begin{align}
	M_{j;k,l}:=\frac{\sigma_k\sigma_l(\bbe_j\bbv_k)(\bbe_j\bbv_l)}{\sum_{t=1}^{T-1}\sigma_t^2(\bbe_j\bbv_t)^2}.\label{Eq of Mjkl}
\end{align}
Thus, we have
\begin{align}
	\frac{\hla_k}{n}=\sum_{s,t=1}^{T-1}\alpha_{k,s}\alpha_{k,t}\frac{1}{n}\sum_{j=1}^nM_{j;k,l}.\label{Eq of la1}
\end{align}
To establish the joint CLT for $(\hla_1,\cdots,\hla_K)'$, we have the following 4 main steps:
\begin{enumerate}
    \item First, we prove that \(\lim_{n\to\infty}\sqrt{n}\mathbb{E}[1-\alpha_{k,k}^2]=0\) for \(1\leq k\leq K\).
		\item Next, we show that
		\begin{align}
			\Bigg(\frac{\hla_1^{\circ}}{\sqrt{n}},\cdots,\frac{\hla_K^{\circ}}{\sqrt{n}}\Bigg)\overset{\mbP}{\longrightarrow}\Bigg(\frac{1}{\sqrt{n}}\sum_{j=1}^n M_{j;1,1}^{\circ},\cdots,\frac{1}{\sqrt{n}}\sum_{j=1}^n M_{j;K,K}^{\circ}\Bigg).\notag
		\end{align}
		\item By the Lindeberg-Feller's CLT, we derive that
		$$\Bigg(\frac{1}{\sqrt{n}}\sum_{j=1}^n M_{j;1,1}^{\circ},\cdots,\frac{1}{\sqrt{n}}\sum_{j=1}^n M_{j;K,K}^{\circ}\Bigg)\overset{d}{\longrightarrow}\mcN(\boldsymbol{0},\mcC).$$
		\item Finally, we can conclude that \(n^{-1/2}\big|\sum_{j=1}^n\mbE[M_{j;k,k}-\fM_{k,k}]\big|\leq C_{B,M,\kappa_8,c}\) by the Lindeberg's principle (see Theorem 1.1 in \cite{chatterjee2006generalization}), so we obtain \eqref{Eq of correlation CLT}.
\end{enumerate}
Finally, as a useful tool for our proof, we cite the following lemma: 
\begin{lem}[Lemma 9 in \cite{onatski2021spurious}]\label{Lem of covariance 1}
    Suppose the noise process \(e_t=\sum_{s=0}^{\infty}\Psi_s\varepsilon_{t-s}\) such that \(\sum_{s=0}^{\infty}(1+s)\Vert\Psi_s\Vert<B\) and \(\varepsilon_s=(\varepsilon_{1,s},\cdots,\varepsilon_{n,s})'\) are independent random vectors with independent entries such that \(\mbE[\varepsilon_{i,s}]=0,\mbE[\varepsilon_{i,s}^2]=1\) and \(\sup_{i,s\in\mbZ}\mbE[\varepsilon_{i,s}^4]<\kappa_4\). Given the noise matrix \(\bbe=[e_1,\cdots,e_T]\), let $\bbe_j$ be the $j$-th {\bf row} of $\bbe$, then for any \(1\leq j\leq n\) and \(1\leq k,l,p,q\leq T\), we have
	\begin{itemize}
        \item Let $f_j(\theta):=\frac{1}{2\pi}\big|\sum_{t=0}^{\infty}\varphi_{j,t}\exp({\rm i}t\theta)\big|^2$ and $\theta_k=2\pi k/T$ for $k=1,\cdots,T$, then
        \begin{align}
	       \bbe_j\bbv_k\overset{d}{\longrightarrow}\mathcal{N}(0,2\pi f_j(\theta_k/2)).\label{Eq of spectral density}
        \end{align}
		\item \(\big|\mathbb{E}[\bbv_l'\bbe_j'\bbe_j\bbv_k]-2\pi f_j(\theta_k/2)\delta_{l,k}\big|\leq CB^2/T\), where \(\delta_{l,k}\) is the Kronecker delta;
		\item \(\big|{\rm Cov}(\bbv_l'\bbe_j'\bbe_j\bbv_k,\bbv_q'\bbe_j'\bbe_j\bbv_p)\big|\leq CB^4(\delta_{l,q}\delta_{k,p}+\delta_{l,p}\delta_{k,q}+(1+\kappa_4)/T)\).
	\end{itemize}
\end{lem}
\subsection{Preliminary lemmas}\label{sec of mean correction}
As we have introduced in the beginning, the first step of establishing the joint CLT in Theorem \ref{Thm of CLT} is to show that $\lim_{n\to\infty}\sqrt{n}\mbE[1-\alpha_{k,k}^2]=0$, where $\alpha_{k,k}$ is defined in \eqref{Eq of F1}. To realized this goal, we need the following lemma: 
\begin{lem}\label{Thm of Lindeberg principle}
	Under Assumptions {\rm \ref{Ap of highdimensionality}, \ref{Ap of panel lag polynomial}} and {\rm \ref{Ap of finite integration}}, for any \(1\leq k,l\leq T-1\), we have
	\begin{align}
	    \big|\mbE[M_{j;k,l}-\fM_{k,l}]\big|\leq C_{B,M,\kappa_8}T^{-1/2}\label{Eq of Lindeberg 1}
	\end{align}
	uniformly in \(1\leq j\leq n\), where \(M_{j;k,l}\) and \(\fM_{k,l}\) are defined in \eqref{Eq of Mjkl} and {\rm (\ref{Eq of fM})}, respectively. Moreover, when $k\neq l$, we have
    \begin{align}
        \big|\mbE[M_{i;k,l}]\big|\leq C_{B,M,\kappa_6}(kl)^{-1}\log^3(T)T^{-1/2}.\label{Eq of Lindeberg 2}
    \end{align}
\end{lem}
Actually, by \eqref{Eq of la1}, to investigate the asymptotic behaviors of $\alpha_{k,k}$, it is essential to have a more comprehensive understanding of the random variables $M_{j;k,l}$ in \eqref{Eq of Mjkl}. The above Lemma \ref{Thm of Lindeberg principle} provide a mean approximation between $M_{j;k,l}$ and $\fM_{k,l}$ in \eqref{Eq of fM}. By \eqref{Eq of fM}, it is relatively easy to handle some basic properties of $\fM_{k,l}$. For example, by symmetry, we know that $\mbE[\fM_{k,l}]=0$ if $k\neq l$; besides, we can use a Monte Carlo method to derive $\mbE[\fM_{k,k}]$ for real applications. Moreover, comparing \eqref{Eq of Lindeberg 1} and \eqref{Eq of Lindeberg 2}, when $k,l$ is relatively small, \eqref{Eq of Lindeberg 1}  indeed provides smaller upper bound of $\big|\mbE[M_{i;k,l}]\big|$ than \eqref{Eq of Lindeberg 2}. However, when $k,l\geq\mrO(\log^2(T))$, \eqref{Eq of Lindeberg 2} will be better than \eqref{Eq of Lindeberg 1}.

Basically, we will prove Lemma \ref{Thm of Lindeberg principle} by the following three steps:
\begin{enumerate}
    \item {\bf Normalization \S\ref{ssec of normalization}:} We extend Theorem 1.1 in \cite{chatterjee2006generalization}, i.e. transform all \(\varepsilon_{i,t}\) in Assumption \ref{Ap of finite integration} into standard normal distributions through the Lindeberg's principle. In this way, all \(\bbe_j\bbv_k\) in \eqref{Eq of Mjkl} will be normal after this transformation, and the difference caused by this transformation will be well controlled by  Lemma \ref{Lem of 1st approximation}.
	\item {\bf Remove the dependence among all \((\bbe_j\bbv_k)^2\) \S\ref{ssec of remove the dependence}:} By Lemma \ref{Lem of covariance 1}, we know that all $(\bbe_j\bbv_t)^2$ are indeed correlated for $t=1,\cdots,T-1$. Here, we will remove dependence among all \((\bbe_j\bbv_k)^2\) while carefully control the error caused by this operation. The key step is to control the total variation distance between high-dimensional Gaussian vectors (Lemma \ref{Thm of cut down TV}).
	\item {\bf Adjust the coefficients \S\ref{ssec of remove the dependence}:} Adjust all coefficients \(\sigma_t\) in \eqref{Eq of Mjkl} to fit the corresponding coefficients in $\fM_{k,l}$ \eqref{Eq of fM}. Moreover, Lemma \ref{Lem of covariance 1} shows that the variance of \(\bbe_j\bbv_t\) are not coincide for $t=1,\cdots,T-1$. Hence, we also unify the variance of \(\bbe_j\bbv_t\) in \eqref{Eq of Mjkl}, see Lemma \ref{Thm of adjust coefficients} for details.
\end{enumerate}
\subsubsection{Some auxiliary results}\label{ssec of preliminary dependent}
In this part, we provide several auxiliary results for proving Lemma \ref{Thm of Lindeberg principle}. For convenience, we simplify the notation \(\operatorname{ess}\sup_{x\in\mathbb{R}}p_{\varepsilon_{i,t}}(x)\) by \(\operatorname{ess}\sup(\varepsilon_{i,t})\) and let
\begin{align}
	x_{j,t}:=\bbe_j\bbv_t.\label{Eq of x}
\end{align}
First, for any \(R\in\mbN^+\), we show that \((\bbe_j\bbv_1,\cdots,\bbe_j\bbv_R)'\) have bounded joint density uniformly for all \(1\leq j\leq n\). Here, we cite the following results:
\begin{lem}[Corollary 2, \cite{bobkov2014bounds}]\label{Lem of bounded density}
	Let \(\{X_k:k=1,\cdots,n\}\) be a sequence of independent \(d\)-dimensional random vectors such that \(\sup_{k=1,\cdots,n}\operatorname{ess}\sup(X_k)<M\), and \(\bba\) be a \(n\)-dimensional unit constant vector, i.e. \(\sum_{k=1}^n a_k^2=1\), then
	$$\operatorname{ess}\sup(S_n)\leq e^{d/2}M,$$
	where \(S_n:=\sum_{k=1}^n a_k X_k\).
\end{lem}
In fact, the condition that \(\Vert\bba\Vert_2=1\) is not essential, we can replace it by any \(\Vert\bba\Vert_2<\infty\) due to the functional ``\(\operatorname{ess}\sup\)'' is homogeneous of degree \(2\), i.e.
$$\operatorname{ess}\sup(\lambda X)^{-2/d}=\lambda^2\operatorname{ess}\sup(X)^{-2/d},\ \lambda\in\mathbb{R}.$$
Hence, when \(\Vert\bba\Vert_2\neq1\), by Lemma \ref{Lem of bounded density}, we have that
$$\Vert\bba\Vert_2^{-2}\operatorname{ess}\sup(S_n)^{-2/d}=\operatorname{ess}\sup(\Vert\bba\Vert_2^{-1}S_n)^{-2/d}\geq e^{-1}M^{-2/d},$$
i.e.
$$\operatorname{ess}\sup(S_n)^{-2/d}\leq\Vert\bba\Vert_2^{-d}e^{d/2}M.$$
Moreover, we also need the following result to deal with the infinite sum of independent random variables:
\begin{lem}[Chapter 4.2, \cite{shiryaev2019probability}]\label{Lem of infinite sum}
	Let \(\{\xi_n:n\in\mathbb{N}^+\}\) be a sequence of independent random variables such that \(\mathbb{E}[\xi_n]=0\) for all \(n\), then if
	$$\sum_{n=1}^{\infty}\mathbb{E}[\xi_n^2]<\infty,$$
	the series \(\sum_{n=1}^{\infty}\xi_n\) converges with probability \(1\).
\end{lem}
Now, we can show that
\begin{lem}\label{Thm of bounded density}
	Under Assumptions {\rm \ref{Ap of panel lag polynomial}} and {\rm \ref{Ap of finite integration}}, for any \(R\in\mathbb{N}^+\), let \(\bby_{j,R}:=(x_{j,1},\cdots,x_{j,R})'\) be a \(R\)-dimensional random vector, where \(x_{j,t}\) is defined in {\rm (\ref{Eq of x})}, then we have
	$$\operatorname{ess}\sup(\bby_{j,R})<C_{R,B,M}$$
	uniformly for all \(1\leq j\leq n\).
\end{lem}
\begin{proof}
	First, for any \(t\in\{1,\cdots,T\}\), define
	\begin{align}
		r_{j,t}^{(0)}:=\sum_{k=0}^{\infty}\varphi_{j,t+k}\varepsilon_{j,-k}=e_{j,t}-\sum_{k=0}^{t-1}\varphi_{j,k}\varepsilon_{j,t-k}:=e_{j,t}-\tilde{e}_{j,t},\label{Eq of rjt}
	\end{align}
	then
	\begin{align}
		x_{j,t}&=\sum_{l=1}^T e_{j,l}v_{t,l}=\sum_{l=1}^T v_{t,l}\tilde{e}_{j,l}+\sum_{l=1}^T v_{t,l}r_{j,l}^{(0)}:=S_{j,t}^1+S_{j,t}^2,\notag
	\end{align}
	where \(v_{t,l}\) has been defined in (\ref{Eq of vk}) and
	$$S_{j,t}^1=\sum_{l=1}^T v_{t,l}\sum_{k=0}^{l-1}\varphi_{j,k}\varepsilon_{j,l-k}=\sum_{l=1}^T\varepsilon_{j,l}\sum_{k=0}^{T-l}\varphi_{j,k}v_{t,k+l}:=\sum_{l=1}^T\varepsilon_{j,l} H_{j,l}^{(t)}.$$
	Next, let's denote \(\bbS_j^1:=(S_{j,1}^1,\cdots,S_{j,R}^1)',\bbS_j^2:=(S_{j,1}^2,\cdots,S_{j,R}^2)'\) and \(\bbH_{j,l}:=(H_{j,l}^{(1)},\cdots,H_{j,l}^{(R)})'\), then
	$$\bby_{j,R}=\bbS_j^1+\bbS_j^2=\sum_{l=1}^T\varepsilon_{j,l}\bbH_{j,l}+\bbS_j^2.$$
	Notice that \(r_{j,t}^{(0)}\) is a infinite sum of random variables, then by Lemma \ref{Lem of infinite sum} and Assumption \ref{Ap of panel lag polynomial}, since
	$$\sum_{k=0}^{\infty}\varphi_{j,t+k}^2\mathbb{E}[\varepsilon_{j,-k}^2]\leq t^{-2}\left(\sum_{k=0}^{\infty}(t+k)|\varphi_{j,t+k}|\right)^2\leq B^2 t^{-2}<\infty,$$
	it implies that \(r_{j,t}^{(0)}\) are well-defined random variables for \(t=1,\cdots,T\), so \(\bbS_j^2\) is a well-defined random vector. Moreover, since all \(r_{j,t}^{(0)}\) depend on \(\{\varepsilon_{j,-k}:k\in\mathbb{N}\}\), then \(\{r_{j,t}^{(0)}:t=1,\cdots,T\}\) and \(\{\tilde{e}_{j,t}:t=1,\cdots,T\}\) are independent, which yields that \(\bbS_j^1\) and \(\bbS_j^2\) are also independent; since the density maximum cannot increase due to convolution multiplication, it is enough to show that \(\bbS_j^1\) has bounded density. Let
	$$\tilde{H}_{j,l}:=\arg\min_{t=1,\cdots,R}|H_{j,l}^{(t)}|\ \ {\rm and\ \ }\tilde{\bbH}_{j,l}:=(H_{j,l}^{(1)}/\tilde{H}_{j,l},\cdots,H_{j,l}^{(R)}/\tilde{H}_{j,l})'$$
	and consider \(\operatorname{ess}\sup(\varepsilon_{j,l}\tilde{\bbH}_{j,l})\). Although \(\varepsilon_{j,l}\tilde{\bbH}_{j,l}\) is a random vector, its density is indeed determined by the univariate random variable \(\varepsilon_{j,l}\), i.e.
	$$\operatorname{ess}\sup(\varepsilon_{j,l}\tilde{\bbH}_{j,l})\leq\operatorname{ess}\sup(\varepsilon_{j,l})\times\max_{t=1,\cdots,R}\big|\tilde{H}_{j,l}/H_{j,l}^{(t)}\big|\leq\operatorname{ess}\sup(\varepsilon_{j,l}).$$
	Moreover, since \(v_{t,r}=\sqrt{2/T}\sin(\pi(r-1)t/T)\), then \(|v_{t,r}|\leq\sqrt{2/T}\) and
	$$\big(H_{j,l}^{(t)}\big)^2\leq\frac{2}{T}\left(\sum_{k=0}^{T-l}|\varphi_{j,k}|\right)^2\leq\frac{2B^2}{T},$$
	where we use Assumption \ref{Ap of panel lag polynomial}. Hence, it gives that
	$$\sum_{l=1}^T\big(\tilde{H}_{j,l}\big)^2\leq 2B^2.$$
	Notice that 
	$$\bbS_j^1=\sum_{l=1}^T\tilde{H}_{j,l}\times\varepsilon_{j,l}\tilde{\bbH}_{j,l},$$
	by Lemma \ref{Lem of bounded density} and Assumption \ref{Ap of finite integration}, it concludes that
	$$\operatorname{ess}\sup(\bbS_j^1)\leq CB^{-R}e^{R/2}M,$$
	which completes our proof.
\end{proof}
Based on Lemma \ref{Thm of bounded density}, \(\bby_{j,R}\) have uniformly bounded densities for $1\leq j\leq n$, then we can further show that:
\begin{lem}\label{Thm of cdf order}
	Under Assumptions {\rm \ref{Ap of panel lag polynomial}} and {\rm \ref{Ap of finite integration}}, for any \(R\in\mathbb{N}^+\), we have
	$$\mathbb{P}\left(\sum_{k=1}^R(\bbe_j\bbv_k)^2\leq x\right)\leq C_{B,M,R}x^{R/2},\quad\forall x\in[0,1].$$
\end{lem}
\begin{proof}
	Since
	\begin{align}
		\mathbb{P}\big(\bby_{j,R}\leq x\big)&=\int_{y_1^2+\cdots+y_R^2\leq x} p_{\bby_{j,R}}(y_1,\cdots,y_R){\rm d}y_1\cdots {\rm d}y_R\notag\\
		&=\int_{\mathcal{D}_x} r^{R-1}\sin^{R-2}\varphi_1\sin^{R-3}\varphi_2\cdots\sin\varphi_{R-2}p_{\bby_{j,R}}(y_1,\cdots,y_R){\rm d}r{\rm d}\varphi_1\cdots{\rm d}\varphi_{R-1},\notag
	\end{align}
	where
	$$\left\{\begin{array}{l}
		y_1:=r\cos\varphi_1\\
		y_2:=r\sin\varphi_1\cos\varphi_2\\
		\vdots\\
		y_{R-1}:=r\sin\varphi_1\sin\varphi_2\cdots\sin\varphi_{R-2}\cos\varphi_{R-1}\\
		y_R:=r\sin\varphi_1\sin\varphi_2\cdots\sin\varphi_{R-2}\sin\varphi_{R-1}
	\end{array}\right.$$
	and \(\mathcal{D}_x:=\{r\in[0,\sqrt{x}];\varphi_1,\cdots,\varphi_{R-2}\in[0,\pi];\varphi_{R-1}\in[0,2\pi]\}\). By Lemma \ref{Thm of bounded density}, we know that \(\bby_{j,R}\) has bounded density, then
	\begin{align}
		\mathbb{P}\big(\bby_{j,R}\leq x\big)&\leq C_{R,B,M}\int_{\mathcal{D}_x}r^{R-1}\sin^{R-2}\varphi_1\sin^{R-3}\varphi_2\cdots\sin\varphi_{R-2}{\rm d}r{\rm d}\varphi_1\cdots{\rm d}\varphi_{R-1}\notag\\
		&\leq C_{R,B,M}\int_0^{\sqrt{x}}r^{R-1}{\rm d}r=C_{R,B,M}x^{R/2},\notag
	\end{align}
	where we use the fact \(\int_0^{\pi}\sin^n\varphi{\rm d}\varphi\leq\int_0^{\pi}\sin\varphi{\rm d}\varphi=1\) in the second inequality.
\end{proof}
Finally, we show that \(x_{j,t}\) has uniformly bounded higher moments, i.e.
\begin{lem}\label{Lem of finite 8th moment}
	Under Assumptions {\rm \ref{Ap of highdimensionality}, \ref{Ap of panel lag polynomial}} and {\rm \ref{Ap of finite integration}}, we have
	$$\mathbb{E}[x_{j,t}^{2l}]\leq C_{B,\kappa_{2l}},\quad1\leq l\leq 4,$$
	where \(x_{j,t}\) is defined in {\rm (\ref{Eq of x})}.
\end{lem}
\begin{proof}
	Without loss of generality, we only prove the case of \(l=4\). For simplicity, we extend the definition of \(r_{j,t}^{(0)}\) in the (\ref{Eq of rjt}) as follows:
	$$r_{j,t}^{(T)}:=\sum_{k=T}^{\infty}\varphi_{j,t+k}\varepsilon_{j,-k},$$
	and \(\tilde{e}_{j,t}:=e_{j,t}-r_{j,t}^{(T)}\), where we abuse the notation \(\tilde{e}_{j,t}\) in the proof of Lemma \ref{Lem of bounded density}. Then we have
	$$\tilde{e}_{j,t}=\sum_{k=0}^{t+T-1} \varphi_{j,k}\varepsilon_{j,t-k}$$
	and
	$$x_{j,t}=\sum_{s=1}^T v_{t,s}e_{j,s}=\sum_{s=1}^T v_{t,s}\tilde{e}_{j,s}+\sum_{s=1}^T v_{t,s}r_{j,s}^{(T)}.$$
	Therefore, by the Hölder's inequality, it gives that
	$$x_{j,t}^8\leq2^7\left[\left(\sum_{s=1}^T v_{t,s}\tilde{e}_{j,s}\right)^8+\left(\sum_{s=1}^T v_{t,s}r_{j,s}^{(T)}\right)^8\right].$$
	Notice that all \(\varepsilon_{j,t}\) are independent with zero mean and unite variance, it implies that
	$$\mathbb{E}\big[(r_{j,s}^{(T)})^8\big]\leq\left(\kappa_8\sum_{k=T}^{\infty}|\varphi_{j,t+k}|^2\right)^4\leq T^{-8}\left(\kappa_8\sum_{k=T}^{\infty}(t+k)|\varphi_{j,t+k}|\right)^8\leq C_{\kappa_8,B}T^{-8},$$
	where we use Assumption \ref{Ap of panel lag polynomial}. Hence, by the Cauchy's inequality and the fact \(\Vert\bbv_t\Vert_2=1\) defined in (\ref{Eq of vk}), it implies that
	$$\mathbb{E}\left[\left(\sum_{s=1}^T v_{t,s}r_{j,s}^{(T)}\right)^8\right]\leq\mathbb{E}\left[\left( \sum_{s=1}^T(r_{j,s}^{(T)})^2\right)^4\right]\leq T^3\sum_{s=1}^T\mathbb{E}\big[(r_{j,s}^{(T)})^8\big]\leq C_{\kappa_8,B}T^{-4}.$$
	Moreover, notice that
	$$\sum_{s=1}^T v_{t,s}\tilde{e}_{j,s}=\sum_{s=1}^T v_{t,s}\sum_{k=0}^{s+T-1}\varphi_{j,k}\varepsilon_{j,s-k}=\sum_{l=1}^T\varepsilon_{j,l}\sum_{k=0}^{T-l}\varphi_{j,k}v_{t,k+l}+\sum_{l=1}^T\varepsilon_{j,1-l}\sum_{k=1}^T\varphi_{j,l+k-1}v_{t,k},$$
	then by Assumption \ref{Ap of panel lag polynomial} again, we have
	\begin{align}
		&\mathbb{E}\left[\left(\sum_{l=1}^T\varepsilon_{j,l}\sum_{k=0}^{T-l}\varphi_{j,k}v_{t,k+l}\right)^8\right]\leq\left[\kappa_8\sum_{l=1}^T\left(\sum_{k=0}^{T-l}\varphi_{j,k}v_{t,k+l}\right)^2\right]^4\notag\\
		&\leq\left[2\kappa_8 T^{-1}\sum_{l=1}^T\left(\sum_{k=0}^{T-l}|\varphi_{j,k}|\right)^2\right]^4\le C_{\kappa_8,B}\notag
	\end{align}
	and
	\begin{align}
		&\mathbb{E}\left[\left(\sum_{l=1}^T\varepsilon_{j,1-l}\sum_{k=1}^T\varphi_{j,l+k-1}v_{t,k}\right)^8\right]\leq\left[\kappa_8\sum_{l=1}^T\left(\sum_{k=1}^T\varphi_{j,l+k-1}v_{t,k}\right)^2\right]^4\notag\\
		&\leq\left[2\kappa_8 T^{-1}\sum_{l=1}^Tl^{-2}\left(\sum_{k=1}^T(l+k)|\varphi_{j,l+k-1}|\right)^2\right]^4\leq C_{\kappa_8,B}T^{-4},\notag
	\end{align}
	then
	\begin{align*}
		&\mathbb{E}\left[\left(\sum_{s=1}^T v_{t,s}\tilde{e}_{j,s}\right)^8\right]\leq2^7\mathbb{E}\left[\left(\sum_{l=1}^T\varepsilon_{j,l}\sum_{k=0}^{T-l}\varphi_{j,k}v_{t,k+l}\right)^8\right]\\
		&+2^7\mathbb{E}\left[\left(\sum_{l=1}^T\varepsilon_{j,1-l}\sum_{k=1}^T\varphi_{j,l+k-1}v_{t,k}\right)^8\right]\leq C_{B,\kappa_8},
	\end{align*}
	which completes our proof.
\end{proof}
\subsubsection{Lindeberg's principle}\label{ssec of normalization}
In this part, we will transform all $\bbe_j\bbv_t$ in $M_{j;k,l}$ \eqref{Eq of Mjkl} by normal random variables. The basic framework follows Lemma 10 in \cite{onatski2021spurious}. Recall the definition of \(\sigma_k\) in (\ref{Eq of SVD of MU}), let
\begin{align}
	\beta_t:=\sigma_t/\sigma_1\quad{\rm and}\quad g_{k,l}(\eta):=\frac{\beta_k\eta_k\beta_l\eta_l}{\sum_{t=1}^{T-1}\beta_t^2\eta_t^2},\label{Eq of beta}
\end{align}
where \(t=1,\cdots,T-1\) and \(\eta=(\eta_1,\cdots,\eta_{T-1})'\). It is easy to see that
\begin{align}
	\frac{2}{\pi t}\leq\beta_t=\frac{\sin(\pi/(2T))}{\sin(\pi t/(2T))}\leq\frac{\pi}{2t}.\label{Eq of beta upper}
\end{align}
Here, we define two \(3T\)-dimensional random vectors \(\vec{x}=(x_1,\cdots,x_{3T})',\vec{y}=(y_1,\cdots,y_{3T})'\) such that
\begin{align}
	x_i:=\left\{\begin{array}{cl}
		\varepsilon_{j,T+1-i}&i=1,\cdots,2T,\\
		\sum_{k=2T+1-i}^{\infty}\varphi_{j,k+2T}\varepsilon_{j,T+1-i-k}&i=2T+1,\cdots,3T,
	\end{array}\right.\label{Eq of xi}
\end{align}
and 
\begin{align}
	y_i:=\left\{\begin{array}{cl}
		{\rm i.i.d.\ }\mathcal{N}(0,1)&i=1,\cdots,2T,\\
		0&i=2T+1,\cdots,3T,
	\end{array}\right.\label{Eq of yi}
\end{align}
where $\varphi_{j,k+2T}$ is defined in Assumption \ref{Ap of panel lag polynomial}, then we have
\begin{align}
	e_{j,t}(\vec{x}):=\sum_{k=0}^{T+t-1}\varphi_{j,k}x_{k+T-t+1}+x_{3T-t+1},\quad t=1,\cdots,T-1,\label{Eq of ejt(x)}
\end{align}
and so does \(e_{j,t}(\vec{y})\). Based on the proof of Theorem 1.1 in \cite{chatterjee2006generalization}, let's show that
\begin{lem}\label{Lem of 1st approximation}
	Under Assumptions {\rm \ref{Ap of highdimensionality}, \ref{Ap of panel lag polynomial}} and {\rm \ref{Ap of finite integration}}, for any integer \(K\in\mbN^+\), let 
	$$h_{k,l}(\vec{x})=\frac{\beta_k(\bbe_j(\vec{x})\bbv_k)\beta_l(\bbe_j(\vec{x})\bbv_l)}{\sum_{t=1}^{T-1}\beta_t^2(\bbe_j(\vec{x})\bbv_t)^2},$$
	where \(\bbe_j(\vec{x})=(e_{j,1}(\vec{x}),\cdots,e_{j,T}(\vec{x}))\), then we have
	\begin{align}
		&\big|\mathbb{E}[h_{k,l}(\vec{x})-h_{k,l}(\vec{y})]\big|\leq C_{B,M,\kappa_8}T^{-1/2},\label{Eq of Lindeberg}\\
        &\big|\mathbb{E}[h_{k,l}(\vec{x})-h_{k,l}(\vec{y})]\big|\leq C_{B,M,\kappa_6}(kl)^{-1}\log^3(T)T^{-1/2}.\label{Eq of Lindeberg principle 1}
	\end{align}
\end{lem}
\begin{proof}
    Actually, the proof of \ref{Eq of Lindeberg} and \eqref{Eq of Lindeberg principle 1} are similar. Before presenting detailed calculations, we make some necessary notations here. First, let \(\vec{z}_i:=(x_1,\cdots,x_i,y_{i+1},\cdots,y_{3T})',\vec{z}_i^0:=(x_1,\cdots,x_{i-1},0,y_{i+1},\cdots,y_{3T})'\) and
	\begin{align}
		A_i:=\mathbb{E}[x_i|x_1,\cdots,x_{i-1}]-\mathbb{E}[y_i]\ \ {\rm and\ \ }B_i:=\mathbb{E}[x_i^2|x_1,\cdots,x_{i-1}]-\mathbb{E}[y_i^2].\notag
	\end{align}
	Since \(\mathbb{E}[h_{k,l}(\vec{x})-h_{k,l}(\vec{y})]=\sum_{i=1}^{3T}\mathbb{E}[h_{k,l}(\vec{z}_i)-h_{k,l}(\vec{z}_{i-1})]\), by third-order Taylor approximation with integral remainder, we have
	$$\left\{\begin{array}{l}
		h_{k,l}(\vec{z}_i)-h_{k,l}(\vec{z}_i^0)=\partial_i h_{k,l}(\vec{z}_i^0)x_i+\frac{1}{2}\partial_i^2 h_{k,l}(\vec{z}_i^0)x_i^2+\frac{1}{2}\int_0^1(1-t)^2\partial_i^3 h_{k,l}(\vec{z}_i^0+t x_i)x_i^3{\rm d}t\\
		h_{k,l}(\vec{z}_{i-1})-h_{k,l}(\vec{z}_i^0)=\partial_i h_{k,l}(\vec{z}_i^0)y_i+\frac{1}{2}\partial_i^2 h_{k,l}(\vec{z}_i^0)y_i^2+\frac{1}{2}\int_0^1(1-t)^2\partial_i^3 h_{k,l}(\vec{z}_i^0+t y_i)y_i^3{\rm d}t
	\end{array}\right.,$$
	where \(\vec{z}_i^0+t x_i=(x_1,\cdots,x_{i-1},t x_i,y_{i+1},\cdots,y_{3T})'\), so does \(\vec{z}_i^0+t y_i\). Then
	\begin{align}
		&\big|\mathbb{E}[h_{k,l}(\vec{z}_i)-h_{k,l}(\vec{z}_{i-1})]\big|\leq\big|\mathbb{E}[\partial_i h_{k,l}(\vec{z}_i^0)(x_i-y_i)]\big|+\frac{1}{2}\big|\mathbb{E}[\partial_i^2 h_{k,l}(\vec{z}_i^0)(x_i^2-y_i^2)]\big|\notag\\
		&+\frac{1}{2}\Bigg|\mathbb{E}\left[\int_0^1(1-t)^2\partial_i^3 h_{k,l}(\vec{z}_i^0+t x_i)x_i^3{\rm d}t\right]\Bigg|+\frac{1}{2}\Bigg|\mathbb{E}\left[\int_0^1(1-t)^2\partial_i^3 h_{k,l}(\vec{z}_i^0+t y_i)y_i^3{\rm d}t\right]\Bigg|.\notag
	\end{align}
    Notice that
	\begin{align}
		&\mathbb{E}[\partial_i h_{k,l}(\vec{z}_i^0)(x_i-y_i)]=\mathbb{E}\big[\mathbb{E}[\partial_i h_{k,l}(\vec{z}_i^0)(x_i-y_i)|\vec{z}_i^0]\big]\notag\\
		&=\mathbb{E}\big[(\mathbb{E}[x_i|x_1,\cdots,x_{i-1}]-\mbE[y_i])\partial_i h_{k,l}(\vec{z}_i^0)\big]=\mathbb{E}[A_i\partial_i h_{k,l}(\vec{z}_i^0)],\notag
	\end{align}
	where we use the fact that all \(y_i\) are independent. Similarly, we have
	$$\mathbb{E}[\partial_i h_{k,l}(\vec{z}_i^0)(x_i^2-y_i^2)]=\mathbb{E}\big[(\mathbb{E}[x_i^2|x_1,\cdots,x_{i-1}]-y_i^2)\partial_i h_{k,l}(\vec{z}_i^0)\big]=\mathbb{E}[B_i\partial_i h_{k,l}(\vec{z}_i^0)].$$
	Hence, it gives that
	\begin{align}
		&\big|\mathbb{E}[h_{k,l}(\vec{z}_i)-h_{k,l}(\vec{z}_{i-1})]\big|\leq\big|\mathbb{E}[A_i\partial_i h_{k,l}(\vec{z}_i^0)]\big|+\frac{1}{2}\big|\mathbb{E}[B_i\partial_i^2 h_{k,l}(\vec{z}_i^0)]\big|\label{Eq of Lindeberg derivatives}\\
		&+\frac{1}{2}\Bigg|\mathbb{E}\left[\int_0^1(1-t)^2\partial_i^3 h_{k,l}(\vec{z}_i^0+t x_i)x_i^3{\rm d}t\right]\Bigg|+\frac{1}{2}\Bigg|\mathbb{E}\left[\int_0^1(1-t)^2\partial_i^3 h_{k,l}(\vec{z}_i^0+t y_i)y_i^3{\rm d}t\right]\Bigg|.\notag
	\end{align}
    It is easy to see that \(A_i=B_i=0\) for \(i=1,\cdots,2T\) according to the (\ref{Eq of xi}) and (\ref{Eq of yi}). For \(i>2T\), since 
	$$|A_i|=\big|\mathbb{E}[x_i|x_1,\cdots,x_{i-1}]\big|\leq\big|\mathbb{E}[x_i^2|x_1,\cdots,x_{i-1}]\big|^{1/2}=B_i^{1/2}$$
	and
	\begin{align}
		&\mathbb{E}[B_i]=\mathbb{E}[x_i^2]=\sum_{k=2T+1-i}^{\infty}\varphi_{j,k+2T}^2\leq T^{-2}\left(\sum_{k=2T+1-i}^{\infty}(k+2T)|\varphi_{j,k+2T}|\right)^2\leq\frac{B^2}{T^2},\label{Eq of EB_i}\\
		&\mathbb{E}[B_i^2]\leq\mathbb{E}[x_i^4]\leq\kappa_4\left(\sum_{k=2T+1-i}^{\infty}\varphi_{j,k+2T}^2\right)^2\leq\frac{\kappa_4 B^4}{T^4},\label{Eq of EB_i2}
	\end{align}
	where we use Assumption \ref{Ap of panel lag polynomial}. By \eqref{Eq of Lindeberg derivatives}, to derive \eqref{Eq of Lindeberg} and \eqref{Eq of Lindeberg principle 1}, it suffices to find the upper bounds of all terms relating the first, second and third (remainders) derivatives of $h_{k,l}(\vec{z}_i^0)$ in \eqref{Eq of Lindeberg derivatives}.
 
    \vspace{5mm}
    \noindent
    {\bf First derivatives.} For \(i>2T\), we have
		$$\big|\mathbb{E}[\partial_i^1 h_{k,l}(\vec{z}_i^0)A_i]\big|\leq\mathbb{E}[|A_i|^2]^{1/2}\mathbb{E}[|\partial_i^1 h_{k,l}(\vec{z}_i^0)|^2]^{1/2}\leq\mathbb{E}[B_i]^{1/2}\mathbb{E}[|\partial_i^1 h_{k,l}(\vec{z}_i^0)|^2]^{1/2}.$$
		By (\ref{Eq of beta}), it implies that
		$$|\partial_i^1 h_{k,l}(\vec{z}_i^0)|\leq\sum_{t=1}^{T-1}|\partial_t^1 g_{k,l}(\vec{z}_i^0)|\cdot|\partial_i^1\bbe_j(\vec{z}_i^0)\bbv_t|,\quad i=1,\cdots,3T,$$
		where
		\begin{align}
			&\partial_i^1\bbe_j(\vec{z}_i^0)\bbv_t=\sum_{s=1}^T v_{t,s}\partial_i^1\Big(\sum_{k=0}^{T+s-1}\varphi_{j,k}x_{k+T-s+1}+x_{3T-s+1}\Big)\notag\\
			&=\left\{\begin{array}{cc}
				\sum_{s=T+1-i}^T v_{t,s}\varphi_{j,i+s-T-1}&i=1,\cdots,2T\\
				v_{t,3T+1-t}&i=2T+1,\cdots,3T
			\end{array}\right.,\notag
		\end{align}
		which implies that 
		\begin{align}
			|\partial_i^1\bbe_j(\vec{z}_i^0)\bbv_t|\leq\sqrt{2/T},\quad i>2T\label{Eq of partial Xt}
		\end{align}
		due to \(|v_{t,3T+1-t}|\leq\sqrt{2/T}\) and Assumption \ref{Ap of panel lag polynomial}. Hence, for \(i>2T\), it gives that
		$$\mathbb{E}[|\partial_i^1 h_{k,l}(\vec{z}_i^0)|^2]\leq\frac{2}{T}\mathbb{E}\left[\left(\sum_{t=1}^{T-1}|\partial_t^1 g_{k,l}(\vec{z}_i^0)|\right)^2\right]=\frac{2}{T}\sum_{t_1,t_2=1}^{T-1}\mathbb{E}\big[|\partial_{t_1}^1 g_{k,l}(\vec{z}_i^0)\partial_{t_2}^1 g_{k,l}(\vec{z}_i^0)|\big],$$
		and
		\begin{align}
			|\partial_t^1 g_{k,l}(\eta)|&\leq\frac{2\beta_t^2\beta_k\beta_l|\eta_k\eta_l\eta_t|}{(\sum_{s=1}^{T-1}\beta_s^2\eta_s^2)^2}+\frac{\beta_k\beta_l(\delta_{k,t}|\eta_l|+\delta_{l,t}|\eta_k|)}{\sum_{s=1}^{T-1}\beta_s^2\eta_s^2}.\label{Eq of 1st derivatives}
		\end{align}
        To derive \eqref{Eq of Lindeberg}, by \eqref{Eq of 1st derivatives}, we have
        \begin{align}
            |\partial_t^1 g_{k,l}(\eta)|\leq\frac{2\beta_t^2|\eta_t|+\beta_k\beta_l(\delta_{k,t}|\eta_l|+\delta_{l,t}|\eta_k|)}{\sum_{s=1}^{T-1}\beta_s^2\eta_s^2}.\notag
        \end{align}
		By Cauchy's inequality, it implies that
        \begin{align}
			&\mathbb{E}\left[|\partial_t^1 g_{k,l}(\vec{z}_i^0)|^2\right]\leq C\mbE[(\beta_t^2|\bbe_j(\vec{z}_i^0)\bbv_t|+\beta_k\beta_l\delta_{k,t}|\bbe_j(\vec{z}_i^0)\bbv_l|+\beta_k\beta_l\delta_{l,t}|\bbe_j(\vec{z}_i^0)\bbv_k|)^4]^{1/2}\notag\\
			&\times\mbE\left[\left(\sum_{s=1}^{T-1}\beta_s^2(\bbe_j(\vec{z}_i^0)\bbv_s)^2\right)^{-4}\right]^{1/2}.\label{Eq of Lindeberg trick}
		\end{align}
        By Lemma \ref{Lem of finite 8th moment} and \eqref{Eq of beta upper}, we know that 
        $$\mbE[(\beta_t^2|\bbe_j(\vec{z}_i^0)\bbv_t|+\beta_k\beta_l\delta_{k,t}|\bbe_j(\vec{z}_i^0)\bbv_l|+\beta_k\beta_l\delta_{l,t}|\bbe_j(\vec{z}_i^0)\bbv_k|)^4]^{1/2}\leq C_{B,\kappa_4}(t^{-4}+(kl)^{-2}(\delta_{k,t}+\delta_{l,t}))$$
        Next, let's show that 
		$$\mbE\left[\left(\sum_{s=1}^{T-1}\beta_s^2(\bbe_j(\vec{z}_i^0)\bbv_s)^2\right)^{-4}\right]\leq C_{B,M},\quad\forall i>2T.$$
		For a pre-specified $R\in\mbN^+$, we know that $(\bbe_j(\vec{z}_i^0)\bbv_1,\cdots,\bbe_j(\vec{z}_i^0)\bbv_R)'$ has bounded density by Lemma \ref{Thm of bounded density}. Furthermore, by Lemma \ref{Thm of cdf order}, it implies that for any $r\in(0,1)$
		\begin{align}
			\mathbb{P}\left(\sum_{k=1}^R(\bbe_j(\vec{z}_i^0)\bbv_k)^2\leq r\right)\leq C_{B,M}r^{R/2}.\label{Eq of 1st bound 1}
		\end{align}
		Therefore, let $R=9$, we have
		\begin{align}
			&\mbE\left[\left(\sum_{s=1}^{T-1}\beta_s^2(\bbe_j(\vec{z}_i^0)\bbv_s)^2\right)^{-4}\right]\leq\beta_9^{-8}\mbE\left[\left(\sum_{s=1}^9(\bbe_j(\vec{z}_i^0)\bbv_s)^2\right)^{-4}\right]\notag\\
			&\leq\beta_9^{-8}+\beta_9^{-8}\mathbb{E}\Bigg[\left(\sum_{s=1}^9(\bbe_j(\vec{z}_i^0)\bbv_s)^2\right)^{-1}\Bigg|\sum_{s=1}^9(\bbe_j(\vec{z}_i^0)\bbv_s)^2\leq1\Bigg]\notag\\
			&\leq\beta_9^{-8}+4\beta_9^{-8}\int_1^{\infty}r^3\mathbb{P}\left(\sum_{k=1}^9(\bbe_j(\vec{z}_i^0)\bbv_k)^2\leq r^{-1}\right){\rm d}r\leq C_{B,M},\label{Eq of 1st bound 2}
		\end{align} 
		and
		\begin{align*}
			&\mbE[|\partial_t^1 g_{k,l}(\vec{z}_i^0)|^2]^{1/2}\leq C_{B,M,\kappa_4}(t^{-2}+(kl)^{-1}(\delta_{k,t}+\delta_{l,t})).
		\end{align*}
		Finally, we can conclude that
		\begin{align}
			&\sum_{i=1}^{3T}\big|\mathbb{E}[\partial_i^1 h_{k,l}(\vec{z}_i^0)A_i]\big|\leq\frac{B}{T}\sum_{i=2T+1}^{3T}\mathbb{E}[|\partial_i^1 h_{k,l}(\vec{z}_i^0)|^2]^{1/2}\label{Eq of 1st bound 3}\\
			&\leq\frac{2B}{T^{3/2}}\sum_{i=2T+1}^{3T}\left(\sum_{t_1,t_2=1}^{T-1}\mathbb{E}\big[|\partial_{t_1}^1 g_{k,l}(\vec{z}_i^0)\partial_{t_2}^1 g_{k,l}(\vec{z}_i^0)|\big]\right)^{1/2}\notag\\
			&\leq\frac{2B}{T^{3/2}}\sum_{i=2T+1}^{3T}\left(\sum_{t_1,t_2=1}^{T-1}\mathbb{E}\big[|\partial_{t_1}^1 g_{k,l}(\vec{z}_i^0)|^2\big]^{1/2}\mathbb{E}\big[|\partial_{t_2}^1 g_{k,l}(\vec{z}_i^0)|^2\big]^{1/2}\right)^{1/2}\notag\\
			&\leq\frac{C_{B,M,\kappa_4}}{T^{1/2}}\left(\sum_{t_1,t_2=1}^{T-1}(t_1^{-2}+\delta_{k,t_1}+\delta_{l,t_1})(t_2^{-2}+\delta_{k,t_2}+\delta_{l,t_2})\right)^{1/2}\leq C_{B,M,\kappa_4}T^{-1/2}.\notag
		\end{align}
        To derive \eqref{Eq of Lindeberg principle 1}, by \eqref{Eq of 1st derivatives}, we have
        \begin{align*}
            |\partial_t^1 g_{k,l}(\eta)|&\leq\frac{2\beta_t\beta_k\beta_l|\eta_k\eta_l|}{(\sum_{s=1}^{T-1}\beta_s^2\eta_s^2)^{3/2}}+\frac{\beta_k\beta_l(\delta_{k,t}|\eta_l|+\delta_{l,t}|\eta_k|)}{\sum_{s=1}^{T-1}\beta_s^2\eta_s^2},
        \end{align*}
        then we can use the same trick as \eqref{Eq of Lindeberg trick} and \eqref{Eq of 1st bound 2} to derive that
        \begin{align*}
            \mbE[|\partial_t^1 g_{k,l}(\vec{z}_i^0)|^2]^{1/2}\leq C_{B,M,\kappa_4}(kl)^{-1}(t^{-1}+(\delta_{k,t}+\delta_{l,t})),
        \end{align*}
        where we use $\beta_k\asymp\mrO(k^{-1})$ by \eqref{Eq of beta upper} (``$\asymp$'' is defined in \eqref{Eq of asymp mbP}). Similar as \eqref{Eq of 1st bound 3}, we can further deduce that
        \begin{align}
            &\sum_{i=1}^{3T}\big|\mathbb{E}[\partial_i^1 h_{k,l}(\vec{z}_i^0)A_i]\big|\leq\frac{B}{T}\sum_{i=2T+1}^{3T}\mathbb{E}[|\partial_i^1 h_{k,l}(\vec{z}_i^0)|^2]^{1/2}\label{Eq of 1st another bound 1}\\
			&\leq\frac{C_{B,M,\kappa_4}(kl)^{-1}}{T^{1/2}}\left(\sum_{t_1,t_2=1}^{T-1}(t_1^{-1}+\delta_{k,t_1}+\delta_{l,t_1})(t_2^{-1}+\delta_{k,t_2}+\delta_{l,t_2})\right)^{1/2}\leq C_{B,M,\kappa_4}(kl)^{-1}\log(T)T^{-1/2}.\notag
        \end{align}
        {\bf Second derivatives.} Similarly, by (\ref{Eq of EB_i2}), for \(i>2T\), we have
		$$\big|\mathbb{E}[\partial_i^2 h_{k,l}(\vec{z}_i^0)B_i]\big|\leq\mathbb{E}\big[|\partial_i^2 h_{k,l}(\vec{z}_i^0)|^2\big]^{1/2}\mathbb{E}[B_i^2]^{1/2}\leq C_{\kappa_4,B}T^{-2}\mathbb{E}\big[|\partial_i^2 h_{k,l}(\vec{z}_i^0)|^2\big]^{1/2},$$
		where
		\begin{align}
			|\partial_i^2 h_{k,l}(\vec{z}_i^0)|\leq\sum_{t_1,t_2=1}^{T-1}|\partial_{t_1 t_2}^2 g_{k,l}(\vec{z}_i^0)|\cdot|\partial_i^1\bbe_j(\vec{z}_i^0)\bbv_{t_1}|\cdot|\partial_i^1\bbe_j(\vec{z}_i^0)\bbv_{t_2}|,\label{Eq of second derivative}
		\end{align}
		and
		\begin{align*}
			&\big|\partial_{t_1 t_2}^2 g_{k,l}(\eta)\big|\leq\frac{8\beta_{t_1}^2\beta_{t_2}^2\beta_k\beta_l|\eta_k\eta_l\eta_{t_1}\eta_{t_2}|}{(\sum_{s=1}^{T-1}\beta_s^2\eta_s^2)^3}+\frac{2\beta_{t_1}^2\beta_k\beta_l(\delta_{t_1,t_2}|\eta_k\eta_l|+\delta_{t_2,k}|\eta_{t_1}\eta_l|+\delta_{t_2,l}|\eta_k\eta_{t_1}|)}{(\sum_{s=1}^{T-1}\beta_s^2\eta_s^2)^2}\\
			&+\frac{2\beta_{t_2}^2\beta_k\beta_l(\delta_{k,t_1}|\eta_l\eta_{t_2}|+\delta_{l,t_1}|\eta_k\eta_{t_2}|)}{(\sum_{s=1}^{T-1}\beta_s^2\eta_s^2)^2}+\frac{\beta_k\beta_l(\delta_{t_1,k}\delta_{t_2,l}+\delta_{t_1,l}\delta_{t_2,k})}{\sum_{s=1}^{T-1}\beta_s^2\eta_s^2}.
		\end{align*}
        To derive \eqref{Eq of Lindeberg}, by the above equation, we have
        $$\big|\partial_{t_1 t_2}^2 g_{k,l}(\eta)\big|\leq\frac{8\beta_{t_1}\beta_{t_2}+4\beta_{t_1}(\delta_{t_1,t_2}\beta_{t_1}+\delta_{t_1,k}+\delta_{t_1,l})+2\beta_{t_2}(\delta_{k,t_1}+\delta_{l,t_1})+\delta_{t_1,k}\delta_{t_2,l}+\delta_{t_1,l}\delta_{t_2,k}}{\sum_{s=1}^{T-1}\beta_s^2\eta_s^2}.$$
		According to (\ref{Eq of second derivative}) and (\ref{Eq of partial Xt}), we have
		\begin{align}
			\mathbb{E}\big[|\partial_i^2 h_{k,l}(\vec{z}_i^0)|^2\big]&\leq4T^{-2}\sum_{t_1,t_2=1}^{T-1}\sum_{t_3,t_4=1}^{T-1}\mathbb{E}\big[|\partial_{t_1 t_2}^2 g_{k,l}(\vec{z}_i^0)\partial_{t_3 t_4}^2 g_{k,l}(\vec{z}_i^0)|\big].\notag
		\end{align}
		It suffices to show that 
		$$\mathbb{E}\big[|\partial_{t_1 t_2}^2 g_{k,l}(\vec{z}_i^0)|^2\big]\leq C_{B,M}(t_1^{-2}t_2^{-2}+t_1^{-2}(\delta_{t_1,t_2}+\delta_{t_1,k}+\delta_{t_1,l})+t_2^{-2}(\delta_{k,t_1}+\delta_{l,t_1})+\delta_{t_1,k}+\delta_{t_1,l}).$$
		The proof is the same as (\ref{Eq of 1st bound 2}), we omit details here. Therefore, we conclude that
		$$\mathbb{E}\big[|\partial_i^2 h_{k,l}(\vec{z}_i^0)|^2\big]\leq C_{B,M}T^{-2}\left(\sum_{t_1,t_2=1}^{T-1}t_1^{-1}t_2^{-1}\right)^2\leq C_{B,M}\log^4 T\times T^{-2}$$
		and
		\begin{align}
			\sum_{i=1}^{3T}\big|\mathbb{E}[\partial_i^2 h_{k,l}(\vec{z}_i^0)B_i]\big|\leq\frac{C_B}{T^2}\sum_{i=2T+1}^{3T}\mathbb{E}\big[|\partial_i^2 h_{k,l}(\vec{z}_i^0)|^2\big]^{1/2}\leq C_{B,M}\frac{\log^2(T)}{T}.\label{Eq of Lindeberg principle 2}
		\end{align}
        Similarly, to derive \eqref{Eq of Lindeberg principle 1}, since
        \begin{align*}
            &\big|\partial_{t_1 t_2}^2 g_{k,l}(\eta)\big|\leq\frac{8\beta_{t_1}\beta_{t_2}\beta_k\beta_l|\eta_k\eta_l|+2\beta_{t_1}^2\beta_k\beta_l\delta_{t_2,t_1}|\eta_l\eta_k|}{(\sum_{s=1}^{T-1}\beta_s^2\eta_s^2)^2}+\frac{2\beta_{t_1}\beta_k\beta_l(\delta_{t_2,k}|\eta_l|+\delta_{t_2,l}|\eta_k|)}{(\sum_{s=1}^{T-1}\beta_s^2\eta_s^2)^{3/2}}\\
			&+\frac{2\beta_{t_2}\beta_k\beta_l(\delta_{k,t_1}|\eta_l|+\delta_{l,t_1}|\eta_k|)}{(\sum_{s=1}^{T-1}\beta_s^2\eta_s^2)^{3/2}}+\frac{\beta_k\beta_l(\delta_{t_1,k}\delta_{t_2,l}+\delta_{t_1,l}\delta_{t_2,k})}{\sum_{s=1}^{T-1}\beta_s^2\eta_s^2},
        \end{align*}
        then we can use the same trick as \eqref{Eq of Lindeberg trick} and \eqref{Eq of 1st bound 2} to conclude  that
        $$\mathbb{E}\big[|\partial_{t_1 t_2}^2 g_{k,l}(\vec{z}_i^0)|^2\big]\leq C_{B,M,\kappa_6}(kl)^{-1}(t_1^{-2}t_2^{-2}+t_1^{-2}(\delta_{t_1,t_2}+\delta_{t_1,k}+\delta_{t_1,l})+t_2^{-2}(\delta_{k,t_1}+\delta_{l,t_1})+\delta_{t_1,k}+\delta_{t_1,l}),$$
        and similar as \eqref{Eq of Lindeberg principle 2}, we can further derive that
        \begin{align}
            \sum_{i=1}^{3T}\big|\mathbb{E}[\partial_i^2 h_{k,l}(\vec{z}_i^0)B_i]\big|\leq\frac{C_B}{T^2}\sum_{i=2T+1}^{3T}\mathbb{E}\big[|\partial_i^2 h_{k,l}(\vec{z}_i^0)|^2\big]^{1/2}\leq C_{B,M,\kappa_6}(kl)^{-1}\log^2(T)T^{-1}.\label{Eq of 1st another bound 2}
        \end{align}
        {\bf Integral remainder.} Finally, we only consider the following case, since the other one is totally the same:
		$$\Bigg|\mathbb{E}\Bigg[\int_0^1(1-t)^2\partial_i^3 h_{k,l}(\vec{z}_i^0+t x_i)x_i^3{\rm d}t\Bigg]\Bigg|\leq\int_0^1(1-t)^2\mathbb{E}\big[|\partial_i^3 h_{k,l}(\vec{z}_i^0+t x_i)||x_i|^3\big]{\rm d}t,$$
		where we use the Fubini's theorem. By the Cauchy's inequality and Lemma \ref{Lem of finite 8th moment}, it gives that
		\begin{align*}
			&\mathbb{E}\big[|\partial_i^3 h_{k,l}(\vec{z}_i^0+t x_i)||x_i|^3\big]\leq\mathbb{E}\big[|\partial_i^3 h_{k,l}(\vec{z}_i^0+t x_i)|^2\big]^{1/2}\mathbb{E}[x_i^6]^{1/2}\\
			&\leq C_{\kappa_6,B}\mathbb{E}\big[|\partial_i^3 h_{k,l}(\vec{z}_i^0+t x_i)|^2\big]^{1/2},
		\end{align*}
		and (\ref{Eq of partial Xt}) further implies that
		\begin{align}
			&|\partial_i^3 h_{k,l}(\vec{z}_i^0+t x_i)|\leq\sum_{t_1,t_2,t_3=1}^{T-1}|\partial_{t_1 t_2 t_3}^3 g_{k,l}(\vec{z}_i^0+t x_i)|\prod_{l=1}^3|\partial_i^1\bbe_j(\vec{z}_i^0+t x_i)\bbv_{t_l}|\notag\\
			&\leq C_B T^{-3/2}\sum_{t_1,t_2,t_3=1}^{T-1}\big|\partial_{t_1 t_2 t_3}^3 g_{k,l}(\vec{z}_i^0+t x_i)\big|.\notag
		\end{align}
		Hence, we obtain that
		\begin{align}
			&\mathbb{E}\big[|\partial_i^3 h_{k,l}(\vec{z}_i^0+t x_i)|^2\big]\leq\label{Eq of Lindeberg principle 5}\\
			&\frac{C_B}{T^3}\sum_{t_1,t_2,t_3=1}^{T-1}\sum_{s_1,s_2,s_3=1}^{T-1}\mathbb{E}\big[\big|\partial_{t_1 t_2 t_3}^3 g_{k,l}(\vec{z}_i^0+t x_i)\big|^2\big]^{1/2}\mathbb{E}\big[\big|\partial_{s_1 s_2 s_3}^3 g_{k,l}(\vec{z}_i^0+t x_i)\big|^2\big]^{1/2},\notag
		\end{align}
		where 
		\begin{align*}
			&\big|\partial_{t_1 t_2 t_3}^3 g_{k,l}(\eta)\big|\leq\frac{48\beta_{t_1}^2\beta_{t_2}^2\beta_{t_3}^2\beta_k\beta_l|\eta_k\eta_l\eta_{t_1}\eta_{t_2}\eta_{t_3}|}{(\sum_{s=1}^{T-1}\beta_s^2\eta_s^2)^4}\\
			&+\frac{8\beta_{t_1}^2\beta_{t_2}^2\beta_k\beta_l(\delta_{t_3,k}|\eta_l\eta_{t_1}\eta_{t_2}|+\delta_{t_3,l}|\eta_k\eta_{t_1}\eta_{t_2}|+\delta_{t_3,t_1}|\eta_k\eta_l\eta_{t_2}|+\delta_{t_3,t_2}|\eta_k\eta_l\eta_{t_1}|)}{(\sum_{s=1}^{T-1}\beta_s^2\eta_s^2)^3}\\
            &+\frac{8\beta_{t_1}^2\beta_{t_3}^2\beta_k\beta_l(\delta_{t_1,t_2}|\eta_k\eta_l\eta_{t_3}|+\delta_{t_2,k}|\eta_{t_1}\eta_l\eta_{t_3}|+\delta_{t_2,l}|\eta_k\eta_{t_1}\eta_{t_3}|)}{(\sum_{s=1}^{T-1}\beta_s^2\eta_s^2)^3}\\
            &+\frac{8\beta_{t_2}^2\beta_{t_3}^2\beta_k\beta_l(\delta_{k,t_1}|\eta_l\eta_{t_2}\eta_{t_3}|+\delta_{l,t_1}|\eta_k\eta_{t_2}\eta_{t_3}|)}{(\sum_{s=1}^{T-1}\beta_s^2\eta_s^2)^3}\\
            &+\frac{2\beta_{t_1}^2\beta_k\beta_l(\delta_{t_1,t_2}(\delta_{t_3,k}|\eta_l|+\delta_{t_3,l}|\eta_k|)+\delta_{t_2,k}(\delta_{t_3,t_1}|\eta_l|+\delta_{t_3,l}|\eta_{t_1}|)+\delta_{t_2,l}(\delta_{t_3,k}|\eta_{t_1}|+\delta_{t_3,t_1}|\eta_k|))}{(\sum_{s=1}^{T-1}\beta_s^2\eta_s^2)^2}\\
			&+\frac{2\beta_{t_2}^2\beta_k\beta_l(\delta_{k,t_1}(\delta_{l,t_3}|\eta_{t_2}|+\delta_{t_2,t_3}|\eta_l|)+\delta_{l,t_1}(\delta_{k,t_3}|\eta_{t_2}|+\delta_{t_2,t_3}|\eta_k|))}{(\sum_{s=1}^{T-1}\beta_s^2\eta_s^2)^2}\\
			&+\frac{2\beta_{t_3}^2\beta_k\beta_l(\delta_{t_1,k}\delta_{t_2,l}|\eta_3|+\delta_{t_1,l}\delta_{t_2,k}|\eta_3|)}{(\sum_{s=1}^{T-1}\beta_s^2\eta_s^2)^2}.
		\end{align*}
		Similar as (\ref{Eq of 1st bound 3}), to prove \eqref{Eq of Lindeberg}, we can use the Hölder's inequality, Lemma \ref{Lem of finite 8th moment} and (\ref{Eq of 1st bound 2}) to derive
		\begin{align}
			&\mbE\big[|\partial_{t_1 t_2 t_3}^3 g_{k,l}(\vec{z}_i^0+tx_i)|^2\big]\leq C_{B,M,\kappa_8}\big(t_1^{-4}t_2^{-4}t_3^{-4}+t_1^{-4}t_2^{-4}(\delta_{t_3,t_1}+\delta_{t_3,t_2}+\delta_{t_3,k}+\delta_{t_3,l})\notag\\
			&+t_1^{-4}t_3^{-4}(\delta_{t_2,t_1}+\delta_{t_2,k}+\delta_{t_2,l})+t_2^{-4}t_3^{-4}(\delta_{t_1,k}+\delta_{t_1,l})\notag\\
            &+t_1^{-4}(\delta_{t_1,t_2}(\delta_{t_3,k}+\delta_{t_3,l})+\delta_{t_2,k}(\delta_{t_3,t_1}+\delta_{t_3,l})+\delta_{t_2,l}(\delta_{t_3,k}+\delta_{t_3,t_1}))\notag\\
            &+t_2^{-4}(\delta_{k,t_1}(\delta_{l,t_3}+\delta_{t_2,t_3})+\delta_{l,t_1}(\delta_{k,t_3}+\delta_{t_2,t_3}))+t_3^{-4}(\delta_{t_1,k}\delta_{t_2,l}+\delta_{t_1,l}\delta_{t_2,k})\big).\label{Eq of Lindeberg principle 3}
		\end{align}
		Since the arguments are totally the same, here we use the following example to explicitly present the calculations. Note that
		\begin{align*}
			&\frac{\beta_{t_1}^2\beta_{t_2}^2\beta_{t_3}^2\beta_k\beta_l|\eta_k\eta_l\eta_{t_1}\eta_{t_2}\eta_{t_3}|}{(\sum_{s=1}^{T-1}\beta_s^2\eta_s^2)^4}\leq\frac{C(t_1t_2t_3)^{-2}|\eta_{t_1}\eta_{t_2}\eta_{t_3}|}{(\sum_{s=1}^{T-1}\beta_s^2\eta_s^2)^3},
		\end{align*}
		by the Hölder's inequality and Lemma \ref{Lem of finite 8th moment}, we have
		\begin{align}
			&\mbE\left[\frac{|(\bbe_j\bbv_{t_1})(\bbe_j\bbv_{t_2})(\bbe_j\bbv_{t_3})|^2}{(\sum_{s=1}^{T-1}\beta_s^2(\bbe_j\bbv_s)^2)^6}\right]\leq\prod_{i=1}^3\mbE\left[(\bbe_j\bbv_{t_i})^8\right]^{1/4}\times\mbE\left[\left(\sum_{s=1}^{T-1}\beta_s^2(\bbe_j\bbv_s)^2\right)^{-24}\right]^{1/4}\notag\\
			&\leq C_{B,M,\kappa_8}\mbE\left[\left(\sum_{s=1}^{T-1}\beta_s^2(\bbe_j\bbv_s)^2\right)^{-24}\right]^{1/4},\label{Eq of Lindeberg trick 2}
		\end{align}
		where we omit \((\vec{z}_i^0+tx_i)\) in \(\bbe_j(\vec{z}_i^0+tx_i)\bbv_t\) to save space. Since we have shown that the joint distribution of \(\big(\bbe_j(\vec{z}_i^0)\bbv_1,\cdots,\bbe_j(\vec{z}_i^0)\bbv_R\big)'\) has the bounded density based on the (\ref{Eq of ejt(x)}) and Lemma \ref{Lem of bounded density}, notice that for \(i\leq2T\), we have
		\begin{align}
			&e_{j,s}(\vec{z}_0^i+tx_i)=\sum_{k=0}^{T+s-1}\varphi_{j,k}(x_{k+T-s+1}1_{k+T-s+1<i}+y_{k+T-s+1}1_{k+T-s+1>i})\notag\\
			&+t(\Psi_{s+i-T-1})_{jj}x_i1_{i\leq 2T,s+i\geq T+2}:=e_{j,t}^1(\vec{z}_0^i+tx_i)+t(\Psi_{s+i-T-1})_{jj}x_i1_{i\leq 2T,s+i\geq T+2},\notag
		\end{align}
		and for \(i>2T\), we have
		\begin{align}
			&e_{j,s}(\vec{z}_0^i+tx_i)=\sum_{k=0}^{T+s-1}\varphi_{j,k}x_{k+T-s+1}+(x_{3T-s+1}1_{i\neq3T-s+1}+tx_i1_{i=3T-s+1})\notag\\
			&:=e_{j,t}^1(\vec{z}_0^i+tx_i)+(x_{3T-s+1}1_{i\neq3T-s+1}+tx_i1_{i=3T-s+1}),\notag
		\end{align}
		by the definitions of \(x_i,y_i\) in (\ref{Eq of xi}) and (\ref{Eq of yi}), it is easy to see that \(e_{j,t}^1(\vec{z}_0^i+tx_i)\) is independent with the rest part, so it is enough to show that \(e_{j,t}^1(\vec{z}_0^i+tx_i)\) has uniformly bounded density, which can be proved by the same argument in Lemma \ref{Lem of bounded density}. Thus, we can conduct Lemma \ref{Thm of bounded density} for \(\big(\bbe_j(\vec{z}_i^0+tx_i)\bbv_1,\cdots,\bbe_j(\vec{z}_i^0+tx_i)\bbv_R\big)'\) and deduce \eqref{Eq of 1st bound 1}. Here, choosing \(R=31\), by the same argument as \eqref{Eq of 1st bound 2}, it gives that
		\begin{align}
			&\mathbb{E}\left[\left(\sum_{s=1}^{T-1}\beta_s^2(\bbe_j\bbv_s)^2\right)^{-24}\right]\leq\beta_{49}^{-48}+\beta_{49}^{-48}\mathbb{E}\left[\left(\sum_{s=1}^{49}\beta_s^2(\bbe_j\bbv_s)^2\right)^{-24}\Bigg|\sum_{s=1}^{49}\beta_s^2(\bbe_j\bbv_s)^2\leq1\right]\notag\\
			&\leq\beta_{49}^{-48}+C_{B,M}\beta_{49}^{-48}\int_1^{\infty}r^{23}\mathbb{P}\left(\sum_{s=1}^{49}\beta_s^2(\bbe_j\bbv_s)^2>r^{-1}\right){\rm d}r\leq C_{B,M}.\label{Eq of Lindeberg trick 3}
		\end{align}
        Now, combining \eqref{Eq of Lindeberg principle 3} and \eqref{Eq of Lindeberg principle 5}, we conclude that
		$$\mathbb{E}\big[|\partial_i^3 h_{k,l}(\vec{z}_i^0+t x_i)|^2\big]\leq\frac{C_{B,M,\kappa_8}}{T^3}\left(\sum_{t_1,t_2,t_3=1}^{T-1}(t_1t_2t_3)^{-2}\right)^2\leq\frac{C_{B,M,\kappa_8}}{T^3},$$
		and
		\begin{align}
			&\sum_{i=1}^{3T}\Bigg|\mathbb{E}\Bigg[\int_0^1(1-t)^2\partial_i^3 h_{k,l}(\vec{z}_i^0+t x_i)x_i^3{\rm d}t\Bigg]\Bigg|\leq C_{\kappa_6,B}\sum_{i=1}^{3T}\int_0^1(1-t)^2\mathbb{E}\big[|\partial_i^3 h_{k,l}(\vec{z}_i^0+t x_i)|^2\big]^{1/2}{\rm d}t\notag\\
			&\leq\frac{C_{B,M,\kappa_8}}{T^{3/2}}\sum_{i=1}^{3T}\int_0^1(1-t)^2{\rm d}t\leq C_{B,M,\kappa_8}T^{-1/2}.\label{Eq of Lindeberg principle 4}
		\end{align}
        To derive \eqref{Eq of Lindeberg principle 1}, similar as \eqref{Eq of Lindeberg principle 3}, we first can prove that
        \begin{align}
			&\mbE\big[|\partial_{t_1 t_2 t_3}^3 g_{k,l}(\vec{z}_i^0+tx_i)|^2\big]\leq C_{B,M,\kappa_6}(kl)^{-1}\big(t_1^{-2}t_2^{-2}t_3^{-2}+t_1^{-2}t_2^{-2}(\delta_{t_3,t_1}+\delta_{t_3,t_2}+\delta_{t_3,k}+\delta_{t_3,l})\notag\\
			&+t_1^{-2}t_3^{-2}(\delta_{t_2,t_1}+\delta_{t_2,k}+\delta_{t_2,l})+t_2^{-2}t_3^{-2}(\delta_{t_1,k}+\delta_{t_1,l})\notag\\
            &+t_1^{-2}(\delta_{t_1,t_2}(\delta_{t_3,k}+\delta_{t_3,l})+\delta_{t_2,k}(\delta_{t_3,t_1}+\delta_{t_3,l})+\delta_{t_2,l}(\delta_{t_3,k}+\delta_{t_3,t_1}))\notag\\
            &+t_2^{-2}(\delta_{k,t_1}(\delta_{l,t_3}+\delta_{t_2,t_3})+\delta_{l,t_1}(\delta_{k,t_3}+\delta_{t_2,t_3}))+t_3^{-2}(\delta_{t_1,k}\delta_{t_2,l}+\delta_{t_1,l}\delta_{t_2,k})\big).\label{Eq of Lindeberg principle 6}
		\end{align}
        For example, note that
        \begin{align*}
			&\frac{\beta_{t_1}^2\beta_{t_2}^2\beta_{t_3}^2\beta_k\beta_l|\eta_k\eta_l\eta_{t_1}\eta_{t_2}\eta_{t_3}|}{(\sum_{s=1}^{T-1}\beta_s^2\eta_s^2)^4}\leq\frac{C(t_1t_2t_3)^{-1}|\eta_k\eta_l|}{(\sum_{s=1}^{T-1}\beta_s^2\eta_s^2)^{5/2}},
		\end{align*}
        then we can use the same trick as \eqref{Eq of Lindeberg trick 2} and \eqref{Eq of Lindeberg trick 3} to show that
        $$\mbE\left[\frac{|\bbe_j\bbv_k\bbe_j\bbv_l|^2}{(\sum_{s=1}^{T-1}\beta_s^2(\bbe_j\bbv_s)^2)^5}\right]\leq C_{B,M,\kappa_6},$$
        which can further concludes \eqref{Eq of Lindeberg principle 6}. Similar as \eqref{Eq of Lindeberg principle 4}, we can obtain that
        \begin{align}
			&\sum_{i=1}^{3T}\Bigg|\mathbb{E}\Bigg[\int_0^1(1-t)^2\partial_i^3 h_{k,l}(\vec{z}_i^0+t x_i)x_i^3{\rm d}t\Bigg]\Bigg|\leq C_{\kappa_6,B}\sum_{i=1}^{3T}\int_0^1(1-t)^2\mathbb{E}\big[|\partial_i^3 h_{k,l}(\vec{z}_i^0+t x_i)|^2\big]^{1/2}{\rm d}t\notag\\
			&\leq\frac{C_{B,M,\kappa_6}(kl)^{-1}\log^3(T)}{T^{3/2}}\sum_{i=1}^{3T}\int_0^1(1-t)^2{\rm d}t\leq C_{B,M,\kappa_6}(kl)^{-1}\log^3(T)T^{-1/2}.\label{Eq of 1st another bound 3}
		\end{align}
		Finally, combining with (\ref{Eq of 1st bound 3}), (\ref{Eq of Lindeberg principle 2}) and (\ref{Eq of Lindeberg principle 4}), we derive \eqref{Eq of Lindeberg}, Combining with \eqref{Eq of 1st another bound 1}, \eqref{Eq of 1st another bound 2} and \eqref{Eq of 1st another bound 3}, we derive \eqref{Eq of Lindeberg principle 1}.
\end{proof}
\subsubsection{Proof of (\ref{Eq of Lindeberg 1})}\label{ssec of remove the dependence}
In this part, we will prove \eqref{Eq of Lindeberg 1} in Lemma \ref{Thm of Lindeberg principle}. By Lemma \ref{Lem of 1st approximation}, we have concluded that
$$\Bigg|\mathbb{E}\Bigg[\frac{\beta_k(\bbe_j(\vec{x})\bbv_k)\beta_l(\bbe_j(\vec{x})\bbv_l)}{\sum_{t=1}^{T-1}\beta_t^2(\bbe_j(\vec{x})\bbv_t)^2}-\frac{\beta_k(\bbe_j(\vec{y})\bbv_k)\beta_l(\bbe_j(\vec{y})\bbv_l)}{\sum_{t=1}^{T-1}\beta_t^2(\bbe_j(\vec{y})\bbv_t)^2}\Bigg]\Bigg|\leq\frac{C_{B,M,\kappa_8}}{\sqrt{T}},\quad1\leq k,l\leq T-1.$$
where all \(\bbe_j(\vec{y})\bbv_t\) are normal for all \(t=1,\cdots,T-1\) due to \(\vec{y}\) \eqref{Eq of yi} is a normal random vector. However, \(|\Cov(\bbe_j(\vec{y})\bbv_k,\bbe_j(\vec{y})\bbv_l)|\leq C_{B,\kappa_4}T^{-1}\) for \(k\neq l\) by Lemma \ref{Lem of covariance 1}. In this part, we will remove these weak dependence among all \(\bbe_j(\vec{y})\bbv_k\). As we have mentioned before, to estimate the error caused by removing weak dependence among all \(\bbe_j(\vec{y})\bbv_k\), we will leverage the total variation distance between high-dimensional Gaussian vectors. For preliminary, we need the following lemma first.
\begin{lem}\label{Lem of cut down}
	Denote \([t]\) to be the largest integer which is no more than \(t\) and let \(\mtm:=[\sqrt{T}]\), then under Assumption {\rm \ref{Ap of panel lag polynomial}}, we have
	$$\left|\mathbb{E}\Bigg[\frac{\beta_k(\bbe_j(\vec{y})\bbv_k)\beta_l(\bbe_j(\vec{y})\bbv_l)}{\sum_{t=1}^{\mtm}\beta_t^2(\bbe_j(\vec{y})\bbv_t)^2}-\frac{\beta_k(\bbe_j(\vec{y})\bbv_k)\beta_l(\bbe_j(\vec{y})\bbv_l)}{\sum_{t=1}^{T-1}\beta_t^2(\bbe_j(\vec{y})\bbv_t)^2}\Bigg]\right|\leq\frac{C_B}{\sqrt{T}}.$$
\end{lem}
\begin{proof}
	Notice that
	\begin{align*}
		&\left|\frac{(\bbe_j(\vec{y})\bbv_k)(\bbe_j(\vec{y})\bbv_l)}{\sum_{t=1}^{\mtm}\beta_t^2(\bbe_j(\vec{y})\bbv_t)^2}-\frac{(\bbe_j(\vec{y})\bbv_k)(\bbe_j(\vec{y})\bbv_l)}{\sum_{t=1}^{T-1}\beta_t^2(\bbe_j(\vec{y})\bbv_t)^2}\right|\leq\frac{|(\bbe_j(\vec{y})\bbv_k)(\bbe_j(\vec{y})\bbv_l)\sum_{t=\mtm+1}^{T-1}\beta_t^2(\bbe_j(\vec{y})\bbv_t)^2}{\big(\sum_{t=1}^{\mfm}\beta_t^2(\bbe_j(\vec{y})\bbv_t)^2\big)^2}
	\end{align*}
	so by the Cauchy's inequality, we have
	\begin{align}
		&\mbE\Bigg[\frac{|(\bbe_j(\vec{y})\bbv_k)(\bbe_j(\vec{y})\bbv_l)\sum_{t=\mtm+1}^{T-1}\beta_t^2(\bbe_j(\vec{y})\bbv_t)^2}{\big(\sum_{t=1}^{\mfm}\beta_t^2(\bbe_j(\vec{y})\bbv_t)^2\big)^2}\Bigg]\leq\mbE\left[(\bbe_j(\vec{y})\bbv_k)^4\right]^{1/4}\mbE\left[(\bbe_j(\vec{y})\bbv_l)^4\right]^{1/4}\notag\\
		&\times\mbE\Bigg[\Bigg(\sum_{t=\mtm+1}^{T-1}\beta_t^2(\bbe_j(\vec{y})\bbv_t)^2\Bigg)^4\Bigg]^{1/4}\times\mbE\Bigg[\Bigg(\sum_{t=\mtm+1}^{T-1}\beta_t^2(\bbe_j(\vec{y})\bbv_t)^2\Bigg)^{-8}\Bigg]^{1/4}.\notag
	\end{align}
	Notice that for any \(k\in\{1,\cdots,T\}\),
	$$\bbe_j(\vec{y})\bbv_k=\sum_{t=1}^Te_{j,t}(\vec{y})v_{k,t}=\sum_{t=1}^T v_{k,t}\sum_{s=0}^{T+t-1}\varphi_{j,s}y_{s+T-t+1}$$
	are normal due to all \(y_i\) are i.i.d. \(\mcN(0,1)\), so we have \(\mbE\left[(\bbe_j(\vec{y})\bbv_l)^8\right]<C_B\) and
	\begin{align}
		&\mbE\Bigg[\Bigg(\sum_{t=\mtm+1}^{T-1}\beta_t^2(\bbe_j(\vec{y})\bbv_t)^2\Bigg)^4\Bigg]\leq C_B\sum_{t_1,t_2,t_3,t_4=\mtm+1}^{T-1}(t_1t_2t_3t_4)^{-2}\leq C_BT^{-2}\notag
	\end{align}
	where we use the (\ref{Eq of beta upper}) and Assumption \ref{Ap of panel lag polynomial}. Similar as (\ref{Eq of 1st bound 1}) and (\ref{Eq of 1st bound 2}), we can also obtain that
	$$\mbE\Bigg[\Bigg(\sum_{t=1}^{\mtm}\beta_t^2(\bbe_j(\vec{y})\bbv_t)^2\Bigg)^{-8}\Bigg]<C_B,$$
	which completes our proof.
\end{proof}
Actually, recall the definition of \(\fM_{k,l}\) in (\ref{Eq of fM}), we define
\begin{align}
	\tilde{\fM}_{k,l}:=\frac{(kl)^{-1}Z_kZ_l}{\sum_{t=1}^{\mtm}t^{-2}Z_t^2},\notag
\end{align}
where \(\mtm=[\sqrt{T}]\). By the same arguments as in Lemma \ref{Lem of cut down}, we can also obtain that 
\begin{align}
    \big|\mbE\big[\tilde{\fM}_{k,l}-\fM_{k,l}\big]\big|\leq CT^{-1/2},\label{Eq of tilde fM}
\end{align}
Now, combining Lemmas \ref{Lem of 1st approximation} and \ref{Lem of cut down}, it gives that for $1\leq k,l\leq K$
\begin{align}
    \big|\mbE\big[M_{j;k,l}-\tilde{g}_{k,l}(\bbz_j^{(1)})\big]\big|\leq C_{B,M,\kappa_8}T^{-1/2},\label{Eq of cut down for remove}
\end{align}
where $\tilde{g}_{k,l}(\vec{z})$ is an $\mfm$-dimensional multivariate function as follows:
\begin{align}
    \tilde{g}_{k,l}(\vec{z})=\frac{x_kx_l}{\sum_{t=1}^{\mfm}x_t^2},\quad\vec{z}=(z_1,\cdots,z_{\mfm})',\label{Eq of tilde g kl}
\end{align}
and $\bbz_j^{(1)}:=\big(\beta_1\bbe_j(\vec{y})\bbv_1,\cdots,\beta_m\bbe_j(\vec{y})\bbv_{\mtm}\big)'$ is an $\mfm$-dimensional normal vector. Hence, to prove \eqref{Eq of Lindeberg 1}, it suffices to show that 
$$\big|\mbE[\tilde{\fM}_{k,l}-\tilde{g}_{k,l}(\bbz_j^{(1)})]\big|\leq C_{B,M,\kappa_8}T^{-1/2}.$$
Note that the entries of normal vector $\bbz_j^{(1)}$ are weakly correlated by Lemma \ref{Lem of covariance 1}. Next, we remove these weak dependence. Let \(\{u_{j,t}:t=1,\cdots,\mtm\}\) be a sequence of {\bf independent} normal variables with zero mean and \(\operatorname{Var}(u_{j,t})=\Var(\bbe_j(\vec{y})\bbv_t)\), further denote
\begin{align}
	\bbz_j^{(1)}:=\big(\beta_1\bbe_j(\vec{y})\bbv_1,\cdots,\beta_m\bbe_j(\vec{y})\bbv_{\mtm}\big)'\ \ {\rm and\ \ }\bbz_j^{(2)}:=\big(\beta_1 u_{j,1},\cdots,\beta_m u_{j,\mtm}\big)'\label{Eq of z1 z2}
\end{align}
to be two $\mfm$-dimensional normal random vectors, whose covariance matrices are \(\Sigma^{(1)},\Sigma^{(2)}\in\mbR^{\mfm\times\mfm}\) respectively. Here, we will show that
$$\big|\mbE\big[\tilde{g}_{k,l}(\bbz_j^{(1)})-\tilde{g}_{k,l}(\bbz_j^{(2)})\big]\big|\leq C_{B,\kappa_4}T^{-1/2}.$$
Note that $|\tilde{g}_{k,l}(\vec{z})|\leq1$, it implies that
\begin{align}
	&\Bigg|\int_{\mathbb{R}^{T-1}}\tilde{g}_{k,l}(\bbz)\big(p_{\bbz^{(1)}}(\bbz)-p_{\bbz^{(2)}}(\bbz)\big){\rm d}\bbz\Bigg|\leq \int_{\mathbb{R}^{T-1}}\big|p_{\bbz^{(0)}}(\bbz)-p_{\bbz^{(1)}}(\bbz)\big|{\rm d}\bbz\leq \operatorname{TV}(\bbz^{(1)},\bbz^{(2)}),\label{Eq of TV}
\end{align}
where \(\operatorname{TV}(\bbz^{(1)},\bbz^{(2)})\) is the total variation distance between \(\bbz^{(1)}\) and \(\bbz^{(2)}\), see \cite{devroye2018total} for the definition of total variation, and \(p_{\bbz^{(i)}}(\bbz)\) is the density function of \(\bbz^{(i)}\), i.e.
$$p_{\bbz^{(i)}}(\bbz)=(2\pi)^{-(T-1)/2}\det(\Sigma^{(i)})^{-1/2}\exp\Big(-\frac{1}{2}\bbz'(\Sigma^{(i)})^{-1}\bbz\Big),\ i=1,2.$$
To control the total variation distance between high-dimensional Gaussian vectors, we cite the following result:
\begin{lem}[Theorem 1.1, \cite{devroye2018total}]\label{Thm of TV}
	Let \(\Sigma_1,\Sigma_1\) be two positive definite \(d\times d\) matrices, then
	$$\operatorname{TV}(\mathcal{N}(\boldsymbol{0},\Sigma_1),\mathcal{N}(\boldsymbol{0},\Sigma_2))\leq\frac{3}{2}\min\Big\{1,\tr(S_d^2)^{1/2}\Big\},$$
	where \(S_d:=\bbI_d-\Sigma_1^{-1/2}\Sigma_2\Sigma_1^{-1/2}\).
\end{lem}
With the help of above results, let's show that
\begin{lem}\label{Thm of cut down TV}
	For two normal vectors \(\bbz^{(1)},\bbz^{(2)}\) defined in the {\rm (\ref{Eq of z1 z2})}, we have \(\operatorname{TV}(\bbz^{(1)},\bbz^{(2)})\leq C_{B,\kappa_4}T^{-1/2}\).
\end{lem}
\begin{proof}
	First, it is clear to see that \(\Sigma^{(2)}\) is positive definite since it is diagonal and all diagonal terms are positive. For \(\Sigma^{(1)}\), since it is symmetric and its diagonal terms are strictly positive, it is enough to show that \(\Sigma^{(1)}\) is diagonal dominated. Given \(t\in\{1,\cdots,m\}\), by Lemma \ref{Lem of covariance 1}, we know that
	\begin{align}
		|\operatorname{Cov}(\bbe_j(\vec{y})\bbv_k,\bbe_j(\vec{y})\bbv_l)-\delta_{k,l}2\pi f_j(\theta_k/2)|\leq C_{B,\kappa_4}T^{-1}.\label{Eq of weak correlation}
	\end{align}
	Recall that \(\beta_t\leq C t^{-1}\) by (\ref{Eq of beta upper}), then for fixed \(l\), it gives that
	$$\sum_{k\neq l}^{\mtm}\beta_k\beta_l|\operatorname{Cov}(\bbe_j(\vec{y})\bbv_k,\bbe_j(\vec{y})\bbv_l)|\leq C_{B,\kappa_4}T^{-1}t^{-1}\sum_{k\neq l}^m k^{-1}\leq C_B T^{-1}l^{-1}\log(T)$$
	and
	$$\operatorname{Var}(\bbe_j(\vec{y})\bbv_l)\geq2\pi f_j(\theta_l/2)-C_{B,\kappa_4}T^{-1}.$$
	Notice that
	$$2\pi|f_j(\theta_t/2)-f_j(0)|\leq\Big|\sum_{k=0}^{\infty}\varphi_{j,k}(\exp({\rm i}k\theta_t/2)-1)\Big|\times\Big(\Big|\sum_{k=0}^{\infty}\varphi_{j,k}\Big|+\Big|\sum_{k=0}^{\infty}\varphi_{j,k}\exp({\rm i}k\theta_t/2)\Big|\Big),$$
	where
	$$\Big|\sum_{k=0}^{\infty}\varphi_{j,k}\exp({\rm i}k\theta_t/2)\Big|\leq\sum_{k=0}^{\infty}|\varphi_{j,k}|\leq B$$
	by Assumption \ref{Ap of panel lag polynomial}. Since \(|\exp({\rm i}x)-1|\leq |x|\) for any \(x\in\mathbb{R}\), then
	$$\Big|\sum_{k=0}^{\infty}\varphi_{j,k}(\exp({\rm i}k\theta_t/2)-1)\Big|\leq\frac{\pi t}{T}\sum_{k=0}^{\infty}k|\varphi_{j,k}|\leq\frac{\pi t B}{T}.$$
	Recall that \(\mtm=[\sqrt{T}]\), so it yields that
	\begin{align}
		2\pi|f_j(\theta_t/2)-f_j(0)|\leq2\pi t B^2T^{-1}\leq C_B T^{-1/2}.\label{Eq of variance difference}
	\end{align}
	By Assumption \ref{Ap of panel lag polynomial}, we have \(2\pi f_j(0)\geq b^2\), so \(2\pi f_j(\theta_t/2)\geq b^2-C_B T^{-1/2}>b^2/2\) and
	$$\beta_t^2\operatorname{Var}(x_{j,t}(y))\geq C_b t^{-2}-C_B T^{-1}t^{-2}.$$
	Since
	$$C_b t^{-2}-C_B T^{-1}t^{-2}>C_B T^{-1}t^{-1}\log T$$
	is equivalent to
	$$\frac{C_{b,B} T-C_B}{\log T}>t,$$
	where \(t\leq\sqrt{T}\), so the left hand side is strictly greater than \(\sqrt{T}\) when \(T\) is sufficiently large, which can imply that
	$$\beta_t^2\operatorname{Var}(x_{j,t}(y))>\sum_{k\neq t}^m\beta_k\beta_t|\operatorname{Cov}(x_{j,t}(y),x_{j,k}(y))|,$$
	i.e. \(\Sigma^{(1)}\) is diagonal dominated, hence positive definite. Finally, by Lemma \ref{Thm of TV}, it gives that
	$$\operatorname{TV}(\bbz^{(1)},\bbz^{(2)})\leq\frac{3}{2}\min\Big\{1,\tr(S^2)^{1/2}\Big\},$$
	where \(S:=\bbI_m-(\Sigma^{(2)})^{-1/2}\Sigma^{(1)}(\Sigma^{(2)})^{-1/2}\). By the (\ref{Eq of weak correlation}) and the definition of \(\bbz^{(2)}\) in (\ref{Eq of z1 z2}), it is easy to see that \(|S_{i,j}|\leq C_{B,\kappa_4}T^{-1}\) for all \(i,j\in\{1,\cdots,m\}\), hence
	$$\tr(S^2)=\sum_{i,j=1}^m|S_{ij}|^2\leq C_B m^2T^{-2}\leq C_B T^{-1},$$
	and \(\operatorname{TV}(\bbz^{(1)},\bbz^{(2)})\leq C_B T^{-1/2}\) as \(T\to\infty\).
\end{proof}
Now, by \eqref{Eq of cut down for remove}, (\ref{Eq of TV}) and Lemma \ref{Thm of cut down TV}, we obtain that for \(1\leq k,l\leq K\)
$$\big|\mbE[M_{j;k,l}-\tilde{g}(\bbz_j^{(2)})]\big|\leq C_{B,M,\kappa_8}T^{-1/2}.$$
Finally, recall that we need to prove $\big|\mbE\big[\tilde{\fM}_{k,l}-\tilde{g}_{k,l}(\bbz_j^{(1)})\big]\big|\leq C_{B,\kappa_4}T^{-1/2}$ to conclude \eqref{Eq of Lindeberg 1}. Now, let \(\{u_{j,t}:t=1,\cdots,m\}\) and \(\{\tilde{u}_{j,t}:t=1,\cdots,m\}\) be two sequences of independent normal variables with zero mean and \(\operatorname{Var}(u_{j,t})=2\pi f_j(\theta_t/2),\operatorname{Var}(\tilde{u}_{j,t})=2\pi f_j(0)\), then define
\begin{align}
		\left\{\begin{array}{l}
			\bbz^{(3)}:=(u_{j,1},\cdots,t^{-1}u_{j,t},\cdots,m^{-1}u_{j,m}),\\
			\bbz^{(4)}:=(\tilde{u}_{j,1},\cdots,t^{-1}\tilde{u}_{j,t},\cdots,m^{-1}\tilde{u}_{j,m}).
		\end{array}\right.\label{Eq of z3 z4}
\end{align}
Note that $\mbE[\tilde{\fM}_{k,l}]=\mbE[\tilde{g}_{k,l}(\bbz_j^{(4)})]$, it suffices to show that $\big|\mbE[\tilde{g}_{k,l}(\bbz_j^{(2)})-\tilde{g}_{k,l}(\bbz_j^{(4)})]\big|\leq C_{B,M,\kappa_8}T^{-1/2}$. Here, we provide that
\begin{lem}\label{Thm of adjust coefficients}
	For Gaussian vectors $\bbz_j^{(2)},\bbz_j^{(3)}$ and $\bbz_j^{(4)}$ defined in \eqref{Eq of z1 z2} and \eqref{Eq of z3 z4}, we have
	$$\big|\mathbb{E}[\tilde{g}_{k,l}(\bbz^{(2)})-\tilde{g}_{k,l}(\bbz^{(3)})]\big|,\big|\mathbb{E}[\tilde{g}_{k,l}(\bbz^{(3)})-\tilde{g}_{k,l}(\bbz^{(4)})]\big|=\mathrm{o}(T^{-1/2}),$$
	where \(\tilde{g}_{k,l}(\vec{z})\) is defined in \eqref{Eq of tilde g kl}.
\end{lem}
\begin{proof}
	Let \(\Sigma^{(3)},\Sigma^{(4)}\) be the covariance matrices of \(\bbz^{(3)},\bbz^{(4)}\) respectively. Similar as the proofs of Lemma \ref{Thm of cut down TV}, we abuse the notation \(S:=\bbI_m-(\Sigma^{(3)})^{-1/2}\Sigma^{(2)}(\Sigma^{(3)})^{-1/2}\) here. Since both \(\Sigma^{(2)}\) and \(\Sigma^{(3)}\) are diagonal, it gives that \(S\) is also diagonal and
	$$S_{tt}=1-t^2\beta_t^2=1-\Big(\frac{t\sin(\pi/2T)}{\sin(\pi t/2T)}\Big)^2=\Big(1-\frac{t\sin(\pi/2T)}{\sin(\pi t/2T)}\Big)\Big(1+\frac{t\sin(\pi/2T)}{\sin(\pi t/2T)}\Big).$$
	Since \(\sin x\geq x-x^3/6\) and \(2x/\pi\leq\sin x\leq x\) for \(x\in[0,\pi/2]\), it implies that
	$$\frac{t\sin(\pi/2T)}{\sin(\pi t/2T)}\geq t\Big(\frac{\pi}{2T}-\frac{1}{6}\Big(\frac{\pi}{2T}\Big)^3\Big)/\Big(\frac{\pi t}{2T}\Big)=1-\frac{1}{6}\Big(\frac{\pi}{2T}\Big)^2$$
	and
	$$\frac{t\sin(\pi/2T)}{\sin(\pi t/2T)}\leq\frac{\pi}{2}.$$
	So
	$$S_{tt}\leq\frac{1}{6}\Big(\frac{\pi}{2T}\Big)^2\Big(1+\frac{\pi}{2}\Big)=\mathrm{O}(T^{-2})$$
	and
	$$\tr(S^2)^{1/2}=\Big(\sum_{t=1}^m S_{tt}^2\Big)^{1/2}=\mathrm{O}(T^{-7/4}),$$
	which implies that \(\operatorname{TV}(\bbz^{(2)},\bbz^{(3)})=\mathrm{o}(T^{-1/2})\), i.e. \(\big|\mathbb{E}[g_{k,l}(\bbz^{(2)})-g_{k,l}(\bbz^{(3)})]\big|=\mathrm{o}(T^{-1/2})\) due to \(g(\bbz)\in[0,1]\). Next, for \(\big|\mathbb{E}[g_{k,l}(\bbz^{(3)})-g_{k,l}(\bbz^{(4)})]\big|\), let's define
	$$\iota_t=\sqrt{\frac{f_j(\theta_t/2)}{f_j(0)}}\ t=1,\cdots,m\quad{\rm and}\quad\tilde{\bbz}^{(4)}:=(\iota_1\tilde{u}_{j,1},\cdots,\iota_t t^{-1}\tilde{u}_{j,t},\cdots,\iota_m m^{-1}\tilde{u}_{j,m}),$$
	where \(f_j(\cdot)\) has been defined in the (\ref{Eq of spectral density}) and \(\theta_t=2\pi t/T\). Due to \(\tilde{\bbz}^{(4)}\) and \(\bbz^{(3)}\) have the same distribution, it is enough to show that 
	$$\big|\mathbb{E}[g_{k,l}(\bbz^{(4)})-g_{k,l}(\tilde{\bbz}^{(4)})]\big|=\mathrm{o}(T^{-1/2}).$$
	Since
	$$g_{k,l}(\bbz^{(4)})-g_{k,l}(\tilde{\bbz}^{(4)})=\frac{\tilde{u}_{j,k}\tilde{u}_{j,l}\sum_{t=1}^m(\iota_t^2-1)t^{-2}\tilde{u}_{j,t}^2}{(\sum_{t=1}^m t^{-2}\tilde{u}_{j,t}^2)(\sum_{t=1}^m \iota_t^2t^{-2}\tilde{u}_{j,t}^2)}$$
	and
	$$\iota_t^2-1=\frac{f_j(\theta_t/2)-f_j(0)}{f_j(0)}.$$
	By Assumption \ref{Ap of panel lag polynomial} and (\ref{Eq of variance difference}), we can conclude that \(|\iota_t^2-1|\leq C_{B,b}t/T\) and
	\begin{align}
		&\big|\mathbb{E}[g_{k,l}(\bbz^{(4)})-g_{k,l}(\tilde{\bbz}^{(4)})]\big|\leq\mathbb{E}\Bigg[\Bigg(\sum_{t=1}^m|\iota_t^2-1|t^{-2}\tilde{u}_{j,k}\tilde{u}_{j,l}\tilde{u}_{j,t}^2\Bigg)^2\Bigg]^{1/2}\mathbb{E}\Bigg[\Bigg(\sum_{t=1}^m t^{-2}\tilde{u}_{j,t}^2\Bigg)^{-4}\Bigg]^{1/2}\notag\\
		&\leq C_{B,b,\kappa_4}T^{-1}\mathbb{E}\Bigg[\Bigg(\sum_{t=1}^mt^{-1}\tilde{u}_{j,k}\tilde{u}_{j,l}\tilde{u}_{j,t}^2\Bigg)^2\Bigg]^{1/2}\mathbb{E}\Bigg[\Bigg(\sum_{t=1}^m t^{-2}\tilde{u}_{j,t}^2\Bigg)^{-4}\Bigg]^{1/2},\notag
	\end{align}
	where we can use the same arguments as those in Lemma \ref{Lem of cut down} to show that
	$$\mathbb{E}\Bigg[\Bigg(\sum_{t=1}^m t^{-2}\tilde{u}_{j,t}^2\Bigg)^{-4}\Bigg]<\infty.$$
	And recall that all \(\tilde{u}_{j,t}\) are i.i.d. normal, then
	$$\mathbb{E}\Bigg[\Bigg(\sum_{t=1}^mt^{-1}\tilde{u}_{j,k}\tilde{u}_{j,l}\tilde{u}_{j,t}^2\Bigg)^2\Bigg]=\sum_{t_1,t_2=1}^m t_1^{-1}t_2^{-1}\mathbb{E}[\tilde{u}_{j,k}^2\tilde{u}_{j,l}^2\tilde{u}_{j,t_1}^2\tilde{u}_{j,t_2}^2]\leq C_B\sum_{t_1,t_2=1}^m t_1^{-1}t_2^{-1}\leq C_B\log^2(T)$$
	hence
	$$\big|\mathbb{E}[g_{k,l}(\bbz^{(4)})-g_{k,l}(\tilde{\bbz}^{(4)})]\big|\leq C_{B,b,\kappa_4}T^{-1}\log T=\mathrm{o}(T^{-1/2}),$$
	which completes our proof.
\end{proof}
Finally, combining \eqref{Eq of tilde fM}, \eqref{Eq of cut down for remove} and Lemma \ref{Thm of adjust coefficients}, we can obtain that 
\begin{align*}
    \big|\mbE\big[M_{j;k,l}-\tilde{\fM}_{k,l}\big]\big|\leq C_{B,M,\kappa_8}T^{-1/2},
\end{align*}
which concludes the (\ref{Eq of Lindeberg 1}) in Lemma \ref{Thm of Lindeberg principle}.
\subsubsection{Proof of (\ref{Eq of Lindeberg 2})}\label{ssec of mean correction for asymmetric}
In this part, we will prove \eqref{Eq of Lindeberg 2} in Lemma \ref{Thm of Lindeberg principle}.
\begin{proof}
    By \eqref{Eq of Lindeberg principle 1} in Lemma \ref{Lem of 1st approximation}, we have concluded that for any \(1\leq k,l\leq T-1\),
    \begin{align}
        \big|\mbE\big[h_{k,l}(\vec{x})-h_{k,l}(\vec{y})\big]\big|\leq C_{B,M,\kappa_6}(kl)^{-1}\log^3(T)T^{-1/2}.\label{Eq of zero mean 1}
    \end{align}
    Combining with Lemma \ref{Lem of cut down}, we further obtain that
    \begin{align}
		\Bigg|\mbE\Bigg[\frac{x_{j,k}(y)x_{j,l}(y)}{\sum_{t=1}^{T-1}\beta_t^2x_{j,t}^2(y)}-\frac{x_{j,k}(y)x_{j,l}(y)}{\sum_{t=1}^{\mtm}\beta_t^2x_{j,t}^2(y)}\Bigg]\Bigg|\leq C_{B,M,\kappa_6}\log^3(T)T^{-1/2},\label{Eq of zero mean 2}
	\end{align}
	where \(\mfm=[\sqrt{T}]\). Without loss of generality, suppose \(k<l\leq\mfm\) and define
	$$\bbz=(x_{j,k}(y),x_{j,1}(y),\cdots,x_{j,k-1}(y),x_{j,k+1}(y),\cdots,x_{j,\mtm}(y))'=(x_{j,k}(y),\bbz_k'),$$
    where $\bbz_k=(x_{j,1}(y),\cdots,x_{j,k-1}(y),x_{j,k+1}(y),\cdots,x_{j,\mtm}(y))'$. The covariance matrix of \(\bbz\) is
	$$\Sigma=\left(\begin{array}{cc}
		\Var(x_{j,k}(y))&\Sigma_k^{12}\\
		\Sigma_k^{21}&\Sigma_k^{22}
	\end{array}\right)\in\mbR^{\mfm\times\mfm}.$$
	By Lemma \ref{Thm of cut down TV}, we know that \((\Sigma_k^{22})^{-1}\) exists, so
	$$x_{j,k}(y)|\bbz_k\sim\mcN(\Sigma_k^{12}(\Sigma_k^{22})^{-1}\bbz_k,\Var(x_{j,k}(y))-\Sigma_k^{12}(\Sigma_k^{22})^{-1}\Sigma_k^{21}).$$
	For simplicity, let \(\mfa_k:=\Sigma_k^{12}(\Sigma_k^{22})^{-1}\in\mbR^{(\mfm-1)\times(\mfm-1)}\), by Lemma \ref{Lem of covariance 1}, we have
	$$\Vert\Sigma_k^{12}\Vert_2\leq C_B T^{-3/4},$$
	and
	$$\min_{\Vert\bba\Vert_2=1}\bba'\Sigma_k^{22}\bba\geq C_b-C_B T^{-1}\Bigg(\sum_{i=1}^{\mtm}|a_i|\Bigg)^2=C_b,$$
	which implies that \(\Vert(\Sigma_k^{22})^{-1}\Vert\leq C_b\) and hence \(\Vert\mfa_k\Vert_2\leq C_{B,b}T^{-3/4}\).
	Consider
	\begin{align}
		&\Bigg|\mbE\Bigg[\frac{x_{j,k}(y)x_{j,l}(y)}{\sum_{t=1}^{\mtm}\beta_t^2x_{j,t}^2(y)}-\frac{x_{j,k}(y)x_{j,l}(y)}{\sum_{t\neq k}^{\mtm}\beta_t^2x_{j,t}^2(y)+\beta_k^2(x_{j,k}(y)-\mfa_k\bbz_k)^2}\Bigg]\Bigg|\notag\\
		&=\Bigg|\mbE\Bigg[\frac{\beta_k^2\mfa_k\bbz_kx_{j,k}(y)x_{j,l}(y)(2x_{j,k}(y)-\mfa_k\bbz_k)}{(\sum_{t=1}^{\mtm}\beta_t^2x_{j,t}^2(y))(\sum_{t\neq k}^{\mtm}\beta_t^2x_{j,t}^2(y)+\beta_k^2(x_{j,k}(y)-\mfa_{k,l}\bbz_k)^2)}\Bigg]\Bigg|.\label{Eq of dominator difference}
	\end{align}
	Since
	\begin{align}
		&\mbE[(\mfa_k\bbz_k)^2]=\sum_{t_1,t_2\neq k}^{\mtm}\mfa_{k,t_1}\mfa_{k,t_2}\mbE[x_{j,t_1}(y)x_{j,t_2}(y)]\notag\\
		&\leq\sum_{t\neq k}^{\mtm}\mfa_{k,t}^2\mbE[x_{j,t}^2(y)]+\sum_{\substack{t_1,t_2\neq k\\t_1\neq t_2}}^{\mtm}|\mfa_{k,t_1}\mfa_{k,t_2}\Cov(x_{j,t_1}(y),x_{j,t_2}(y))|\notag\\
		&\leq C_B\Vert\mfa_k\Vert_2^2+C_{B,\kappa_4}T^{-1}\Bigg(\sum_{t\neq k}^{\mtm}|\mfa_{k,t}|\Bigg)^2\leq C_{B,\kappa_4}\Vert\mfa_k\Vert_2^2\leq C_{B,b,\kappa_4}T^{-3/2},\notag
	\end{align}
	and 
	$${\rm (\ref{Eq of dominator difference})}\leq\mbE[(\mfa_k\bbz_k)^2]^{1/2}\mbE\Bigg[\frac{\beta_k^4x_{j,k}^2(y)x_{j,l}^2(y)(2x_{j,k}(y)-\mfa_k\bbz_k)^2}{(\sum_{t\neq k}^{\mtm}\beta_t^2x_{j,t}^2(y))^4}\Bigg]^{1/2},$$
	note that \(x_{j,t}(y)\) are normal with bounded density, then we can use the same method in Lemmas \ref{Lem of 1st approximation} and \ref{Lem of cut down} to show that  
	$$\mbE\Bigg[\frac{\beta_k^4x_{j,k}^2(y)x_{j,l}^2(y)(2x_{j,k}(y)-\mfa_k\bbz_k)^2}{(\sum_{t\neq k}^{\mtm}\beta_t^2x_{j,t}^2(y))^4}\Bigg]\leq C_B.$$
	Thus, (\ref{Eq of dominator difference}) is bounded by \(C_{B,b,\kappa_4}T^{-3/4}\). Moreover, since \(x_{j,k}(y)-\mfa_k\bbz_k\) is independent with \(\bbz_k\) and $x_{j,k}(y)-\mfa_k\bbz_k$ is a normal variable with zero mean, by symmetry, we have
	\begin{align}
		\mbE\Bigg[\frac{(x_{j,k}(y)-\mfa_k\bbz_k)x_{j,l}(y)}{\sum_{t\neq k}^{\mtm}\beta_t^2x_{j,t}(y)^2+\beta_k^2(x_{j,k}(y)-\mfa_k\bbz_k)^2}\Bigg]=0.\label{Eq of zero mean 3}
	\end{align}
	Similarly, by the Cauchy's inequality and above method, we can also show that
	\begin{align}
		\Bigg|\mbE\Bigg[\frac{\mfa_k\bbz_kx_{j,l}(y)}{\sum_{t\neq k}^{\mtm}\beta_t^2x_{j,t}(y)^2+\beta_k^2(x_{j,k}(y)-\mfa_k\bbz_k)^2}\Bigg]\Bigg|\leq C_B\mbE[(\mfa_k\bbz_k)^2]^{1/2}\leq C_{B,b,\kappa_4}T^{-3/4}.\label{Eq of zero mean 4}
	\end{align}
	Now, combining (\ref{Eq of zero mean 1}), (\ref{Eq of zero mean 2}), (\ref{Eq of dominator difference}), (\ref{Eq of zero mean 3}) and (\ref{Eq of zero mean 4}), we can conclude \eqref{Eq of Lindeberg 2}. For the case when at least one of \(k,l\) is greater than \(m\), we can define (e.g. \(l\leq m<k\))
	$$\bbz:=(x_{j,k}(y),x_{j,1}(y),\cdots,x_{j,m}(y))',$$
	then the above arguments are still valid, so we omit details here to save space.
\end{proof}
\subsection{Asymptotic behaviors of eigenvectors}\label{sec of alpha kt correlation}
Recall that \(\hF_k=\sum_{t=1}^{T-1}\alpha_{k,t}\bbw_t\) is the eigenvector of \(\hla_k\) in (\ref{Eq of F1}). In this section, we will establish the asymptotic behaviors for \(\alpha_{k,t}\) as follows:
\begin{lem}\label{Lem of at convergence}
	Under Assumptions {\rm \ref{Ap of highdimensionality}, \ref{Ap of panel lag polynomial}} and {\rm \ref{Ap of finite integration}}, for any \(K\in\mbN^+\), let \(\hat{F}_k=\sum_{t=1}^{T-1}\alpha_{k,t}\bbw_t\) be the eigenvector corresponding to the \(k\)-th largest eigenvalue of the sample correlation matrix \eqref{Eq of correlation matrx random walk 1} of $\bbX=[X_1,\cdots,X_T]$ generated by \eqref{Eq of panel Xt}, then
	$$\lim_{n\to\infty}\sqrt{n}\mathbb{E}[1-\alpha_{k,k}^2]=0$$
	for \(1\leq k\leq K\).
\end{lem}
As a consequence of above lemma, we can further obtain
\begin{cor}\label{Cor of L2 convergence}
	Under Assumptions {\rm \ref{Ap of highdimensionality}, \ref{Ap of panel lag polynomial}} and {\rm \ref{Ap of finite integration}}, for any \(K\in\mbN^+\), let \(\hat{F}_k=\sum_{t=1}^{T-1}\alpha_{k,t}\bbw_t\) be the eigenvector corresponding to the \(k\)-th largest eigenvalue of the sample correlation matrix \eqref{Eq of correlation matrx random walk 1} of $\bbX=[X_1,\cdots,X_T]$ generated by \eqref{Eq of panel Xt}, then
	$$\sqrt{n}(1-\alpha_{k,k}^2)\overset{\mbP}{\longrightarrow}0$$
	for \(1\leq k\leq K\).
\end{cor}
For preliminary, we first show that \(\lim_{n\to\infty}\sqrt{n}\mathbb{E}[1-\alpha_{1,1}^2]=0\). By the definition of $\hat{F}_1$ in \eqref{Eq of F1}, we know that
$$\alpha_{1,t}\hat{\lambda}_1=\bbw_t'\hat{\bbR}\hat{F}_1=\sigma_t\sum_{k=1}^{T-1}\alpha_{1,k}\sigma_k\bbv_t'\bbe'\bbD^{-1}\bbe\bbv_k=\alpha_{1,t}\sum_{j=1}^nM_{j;t,t}+\sum_{k\neq t}^{T-1}\alpha_{1,k}\sum_{j=1}^n M_{j;k,t},$$
where $M_{j;k,t}$ is defined in \eqref{Eq of Mjkl}. Hence, it gives that
\begin{align}
    \alpha_{1,t}=\frac{\sum_{k\neq t}^{T-1}\alpha_{1,k}n^{-1}\sum_{j=1}^n M_{j;k,t}}{n^{-1}\big(\hla_1-\sum_{j=1}^nM_{j;t,t}\big)}.\label{Eq of alpha 1t example}
\end{align}
Note that $1-\alpha_{1,1}^2=\sum_{t=2}^{T-1}\alpha_{1,t}^2$, to show that $\lim_{n\to\infty}\sqrt{n}\mbE[1-\alpha_{1,1}^2]=0$, we need to investigate the asymptotic behaviors of $\sqrt{n}\mbE[\alpha_{1,t}^2]$ for $2\leq t\leq T-1$. By \eqref{Eq of alpha 1t example}, we first derive the asymptotic behaviors of $\alpha_{1,t}$'s denominators as follows:
\begin{lem}\label{Lem of upper bound}
	Under Assumptions {\rm \ref{Ap of highdimensionality}, \ref{Ap of panel lag polynomial}} and {\rm \ref{Ap of finite integration}}, for any \(K\in\mbN^+\) and \(\hat{F}_1=\sum_{t=1}^{T-1}\alpha_{1,t}\bbw_t\) defined in \eqref{Eq of F1}, let \(\delta_K:=1-\sum_{k=1}^K\mathbb{E}[\fM_{k,k}]\), then there exists a constant \(C_1>0\) such that
	$$\mathbb{P}\Bigg(n^{-1}\Bigg|\hla_1-\sum_{j=1}^nM_{j;1,1}\Bigg|\leq C_1\delta_K^{1/2}\Bigg)\geq1-C_Kn^{-3/5}.$$
\end{lem}
\begin{proof}
    First, let \(N_K\) be a pre-specified integer only depending on \(K\) such that \(N_K\geq K\) and the precise definition for $N_K$ is given in \eqref{Eq of N_K} later, then define
    \begin{align}
        &A_1:=\sum_{k=1}^{N_K}\alpha_{1,k}\sigma_k\bbD^{-1/2}\bbe\bbv_k\quad{\rm and}\quad B_1:=\sum_{k=N_K+1}^{T-1}\alpha_{1,k}\sigma_k\bbD^{-1/2}\bbe\bbv_k,\label{Eq of A1 A2}
    \end{align}
	by (\ref{Eq of la1}), we know that \(\hat{\lambda}_1=\Vert A_1\Vert^2+2\langle A_1,B_1\rangle+\Vert B_1\Vert^2\) and \(\Vert A_1\Vert^2=\sum_{k,l=1}^{N_K}\alpha_{1,k}\alpha_{1,l}\sum_{j=1}^n M_{j;k,l}\), then by the Chebyshev's inequality, we have
	$$\mathbb{P}\Bigg(n^{-1}\Bigg|\sum_{j=1}^n M_{j;k,l}^{\circ}\Bigg|>n^{-1/5}\Bigg)\leq n^{-8/5}\sum_{j=1}^n{\rm Var}(M_{j;k,l})\leq 4n^{-3/5},$$
	where we use the fact that \(|M_{j;k,l}|\leq1\). By (\ref{Eq of Lindeberg 1}), we have \(|\mathbb{E}[M_{j;k,l}-\fM_{k,l}]|\leq C_{B,M,\kappa_8}T^{-1/2}\) uniformly in \(j\) for \(1\leq k,l\leq T-1\), combining with $\mbE[\fM_{k,l}]=0$ for $k\neq l$, it implies that
	\begin{align}
		&n^{-1}\Vert A_1\Vert^2\leq n^{-1}\sum_{k=1}^{N_K}\alpha_{1,k}^2\sum_{j=1}^n\mathbb{E}[M_{j;k,k}]+C_K n^{-1/5}\leq\sum_{k=1}^K\alpha_{1,k}^2\mathbb{E}[\fM_{k,k}]+C_K n^{-1/5}\notag\\
		&\leq\alpha_{1,1}^2\mathbb{E}[\fM_{1,1}]+(1-\alpha_{1,1}^2)\mathbb{E}[\fM_{2,2}]+C_K n^{-1/5}\ {\rm with\ probability}\geq1-C_Kn^{-3/5},\notag
	\end{align}
	where we use the fact that \(\mathbb{E}[\fM_{k,k}]\geq\mathbb{E}[\fM_{l,l}]\) for \(1\leq l\leq k\leq K\) in the last inequality. Similarly, for the \(\Vert B_2\Vert^2\), by the Cauchy's inequality and previous arguments, we have
	\begin{align}
		&n^{-1}\Vert B_1\Vert^2\leq\Bigg(\sum_{k=N_K+1}^{T-1}\alpha_{1,k}^2\Bigg)\Bigg(n^{-1}\sum_{k=N_K+1}^{T-1}\sum_{j=1}^n M_{j;k,k}\Bigg)=\Bigg(\sum_{k=N_K+1}^{T-1}\alpha_{1,k}^2\Bigg)\Bigg(1-n^{-1}\sum_{k=1}^{N_K}\sum_{j=1}^n M_{j;k,k}\Bigg)\notag\\
		&\leq(1-\alpha_{1,1}^2)\Bigg(1-\sum_{k=1}^{N_K}\mathbb{E}[\fM_{k,k}]+C_K n^{-1/5}\Bigg)\ {\rm with\ probability}\geq1-C_Kn^{-3/5}.\notag
	\end{align}
	Here, define 
    \begin{align}
        \delta_K:=1-\sum_{k=1}^{N_K}\mathbb{E}[\fM_{k,k}]\quad{\rm and}\quad\epsilon_K:=\sum_{k=N_K+1}^{T-1}\alpha_{1,k}^2,\label{Eq of delta_K and epsilon_K}
    \end{align}
    since \(n^{-1}|\langle A_1,B_1\rangle|\leq n^{-1}\Vert A_1\Vert\Vert B_1\Vert\leq n^{-1/2}\Vert B_1\Vert\), where we use the fact that \(n^{-1}\Vert A_1\Vert^2\leq1\). Therefore, with probability greater than \(1-C_Kn^{-3/5}\), we obtain that
	\begin{align}
		&n^{-1}\hat{\lambda}_1\leq\alpha_{1,1}^2\mathbb{E}[\fM_{1,1}]+(1-\alpha_{1,1}^2)\mathbb{E}[\fM_{2,2}]+C_K n^{-1/5}\notag\\
		&+(1-\alpha_{1,1}^2)(\delta_K+C_K n^{-1/5})+2(1-\alpha_{1,1}^2)^{1/2}(\delta_K+C_K n^{-1/5})^{1/2}.\label{Eq of upper bound for l1}
	\end{align}
	On the other hand, since \(\hat{\lambda}_1\) is the largest eigenvalue of \(\hat{\bbR}\), then \(\hat{\lambda}_1=\hat{F}_1'\hat{\bbR}\hat{F}_1\geq\bbw_1'\hat{\bbR}\bbw_1\), i.e. \(n^{-1}\sum_{j=1}^n M_{j;1,1}\leq n^{-1}\hat{\lambda}_1\), so we obtain that
	\begin{align}
		\mathbb{E}[\fM_{1,1}]\leq n^{-1}\hat{\lambda}_1+n^{-1/5}\ {\rm with\ probability}\geq1-4n^{-3/5}.\label{Eq of lower bound for l1}
	\end{align}
	Now, combining (\ref{Eq of upper bound for l1}) and (\ref{Eq of lower bound for l1}), we conclude that
	\begin{align}
		1-\alpha_{1,1}^2\leq\frac{2(\delta_K+C_K n^{-1/5})^{1/2}+C_K n^{-1/5}}{\mathbb{E}[\fM_{1,1}-\fM_{2,2}]-\delta_K-C_K n^{-1/5}}\ {\rm with\ probability}\geq1-C_Kn^{-3/5}.\notag
	\end{align}
    For a sufficiently large constant \(N_K\in\mbN^+\), we have \(\delta_K^{1/2}+C_K n^{-1/5}<\mathbb{E}[\fM_{1,1}-\fM_{2,2}]/2\) as \(n\to\infty\), then we obtain
	$$1-\alpha_{1,1}^2\leq\frac{6\delta_K^{1/2}}{\mathbb{E}[\fM_{1,1}-\fM_{2,2}]}\ {\rm with\ probability}\geq1-C_Kn^{-3/5}.$$
	Finally, by (\ref{Eq of A1 A2}), since
	$$n^{-1}\Bigg|\hat{\lambda}_1-\sum_{j=1}^n M_{j;1,1}\Bigg|\leq n^{-1}\Bigg|\Vert A_1\Vert^2-\sum_{j=1}^n M_{j;1,1}\Bigg|+2n^{-1}\Vert A_1\Vert\Vert B_1\Vert+n^{-1}\Vert B_1\Vert^2,$$
	and
	\begin{align}
		n^{-1}\Bigg(\Vert A_1\Vert^2-\sum_{j=1}^n M_{j;1,1}\Bigg)=n^{-1}(\alpha_{1,1}^2-1)\sum_{j=1}^n M_{j;1,1}+n^{-1}\sum_{k\neq1\ {\rm or\ }l\neq1}^{N_K}\alpha_{1,k}\alpha_{1,l}\sum_{j=1}^n M_{j;k,l}.\notag
	\end{align}
	then by Lemma \ref{Lem of upper bound}, with probability greater than \(1-C_K n^{-3/5}\),
	\begin{align}
		&n^{-1}\Bigg|\Vert A_1\Vert^2-\sum_{j=1}^n M_{j;1,1}\Bigg|\notag\\
        &\leq n^{-1}(1-\alpha_{1,1}^2)\sum_{j=1}^n\mathbb{E}[M_{j;1,1}]+n^{-1}\sum_{k=2}^{N_K}\alpha_{1,k}^2\sum_{j=1}^n\mathbb{E}[M_{j;k,k}]+C_K n^{-1/5}\notag\\
		&\leq(1-\alpha_{1,1}^2)\sum_{k=1}^{N_K}\mathbb{E}[\fM_{k,k}]+C_K n^{-1/5}\leq 6\mbE[\fM_{1,1}-\fM_{2,2}]^{-1}\delta_K^{1/2}.\notag
	\end{align}
	Moreover, we have shown that $n^{-1}\Vert B_1\Vert^2\leq(1-\alpha_{1,1}^2)\delta_K$, so it implies that \(2n^{-1}\Vert A_1\Vert\Vert B_1\Vert+n^{-1}\Vert B_1\Vert^2\leq 3\delta_K^{1/2}\) with probability at least \(1-C_K n^{-3/5}\), and
	\begin{align}
		n^{-1}\Bigg|\hat{\lambda}_1-\sum_{j=1}^n M_{j;1,1}\Bigg|\leq(3+6\mbE[\fM_{1,1}-\fM_{2,2}]^{-1})\delta_K^{1/2}\label{Eq of C1}
	\end{align}
    with probability greater than \(1-C_K n^{-3/5}\). And the constant $C_1$ in Lemma \ref{Lem of upper bound} is \(C_1:=3+6\mbE[\fM_{1,1}-\fM_{2,2}]^{-1}\).
\end{proof}
Next, let's prove that \(\sqrt{n}\mathbb{E}[1-\alpha_{1,1}^2]\to0\).
\begin{proof}[Proof of Lemma \ref{Lem of at convergence} for \(\alpha_{1,1}\)]
	Recall the \(N_K\) defined in Lemma \ref{Lem of upper bound}, since 
    \begin{align*}
        \mbE[1-\alpha_{1,1}^2]=\sum_{t=2}^{N_K}\mbE[\alpha_{1,t}^2]+\sum_{t=N_K+1}^{T-1}\mbE[\alpha_{1,t}^2]
    \end{align*}
    let's first show that \(\sqrt{n}\mathbb{E}[\alpha_{1,t}^2]\to0\) for \(2\leq t\leq N_K\). Notice that 
	$$\alpha_{1,t}\hat{\lambda}_1=\bbw_t'\hat{\bbR}\hat{F}_1=\sigma_t\sum_{k=1}^{T-1}\alpha_{1,k}\sigma_k\bbv_t'\bbe'\bbD^{-1}\bbe\bbv_k=\alpha_{1,t}\sum_{j=1}^nM_{j;t,t}+\sum_{k\neq t}^{T-1}\alpha_{1,k}\sum_{j=1}^n M_{j;k,t},$$
	i.e.
	\begin{align}
		\alpha_{1,t}=\frac{\sum_{k\neq t}^{T-1}\alpha_{1,k}n^{-1}\sum_{j=1}^n M_{j;k,t}}{n^{-1}(\hat{\lambda}_1-\sum_{j=1}^nM_{j;t,t})}.\label{Eq of square n at}
	\end{align}
    For \(2\leq t\leq N_K\), define 
    \begin{align}
		\mathcal{F}_t(\epsilon):=\Bigg\{n^{-1}\Bigg|\hat{\lambda}_1-\sum_{j=1}^n M_{j;t,t}\Bigg|>\epsilon\Bigg\},\label{Eq of dominator event}
	\end{align}
    for some \(\epsilon>0\). By the Chebyshev's inequality, we can also show that
	$$\mathbb{P}\Bigg(n^{-1}\Bigg|\sum_{j=1}^n(M_{j;1,1}-M_{j;t,t})^{\circ}\Bigg|>n^{-1/5}\Bigg)\leq 4n^{-3/5},$$
    combining with Lemma \ref{Lem of upper bound}, for a sufficiently large \(N_K\), we have \(C_1\delta_K^{1/2}<\mbE[\fM_{1,1}-\fM_{2,2}]/2\). Thus, let \(\epsilon=\mathbb{E}[\fM_{1,1}-\fM_{2,2}]/2\) in (\ref{Eq of dominator event}), it gives that
	\begin{align}
		&n^{-1}\Bigg|\hat{\lambda}_1-\sum_{j=1}^n M_{j;t,t}\Bigg|\geq n^{-1}\Bigg|\sum_{j=1}^n(M_{j;1,1}-M_{j;t,t})\Bigg|-n^{-1}\Bigg|\hat{\lambda}_1-\sum_{j=1}^n M_{j;1,1}\Bigg|\notag\\
		&\geq\mathbb{E}[\fM_{1,1}-\fM_{2,2}]-n^{-1/5}-C_1\delta_K^{1/2}\notag\\
		&>\mathbb{E}[\fM_{1,1}-\fM_{2,2}]/2\ \ {\rm with\ probability}\geq1-C_Kn^{-3/5},\label{Eq of denominator events 1}
	\end{align}
    i.e. \(\mathbb{P}(\mathcal{F}_t(\epsilon))\geq1- C_Kn^{-3/5}\). Next, by (\ref{Eq of square n at}), we have for \(2\leq t\leq K\)
	\begin{align}
		&\sqrt{n}\mathbb{E}[\alpha_{1,t}^2]=\sqrt{n}\mathbb{E}[\alpha_{1,t}^2|\mathcal{F}_t(\epsilon)]\mathbb{P}(\mathcal{F}_t(\epsilon))+\sqrt{n}\mathbb{E}[\alpha_{1,t}^2|\mathcal{F}_t(\epsilon)^c]\mathbb{P}(\mathcal{F}_t(\epsilon)^c)\notag\\
		&\leq\sqrt{n}\mathbb{P}(\mathcal{F}_t(\epsilon)^c)+\epsilon^{-2}n^{-3/2}\mathbb{E}\Bigg[\Bigg(\sum_{k\neq t}^{T-1}\alpha_{1,k}\sum_{j=1}^n M_{j;k,t}\Bigg)^2\Bigg|\mathcal{F}_t(\epsilon)\Bigg]\mathbb{P}(\mathcal{F}_t(\epsilon))\notag\\
		&\leq C_Kn^{-1/10}+\epsilon^{-2}n^{-3/2}\mathbb{E}\Bigg[\Bigg(\sum_{k\neq2}^{T-1}\alpha_{1,k}\sum_{j=1}^n M_{j;k,2}\Bigg)^2\Bigg].\label{Eq of alpha2}
	\end{align}
    Next, for any \(t>N_K\), by (\ref{Eq of square n at}), the denominator \(n^{-1}(\hla_1-\sum_{j=1}^nM_{j;t,t})\) is greater than \(n^{-1}(\hat{\lambda}_1-\sum_{j=1}^n\sum_{t=N_K+1}^{T-1}M_{j;t,t})\). By the Chebyshev's inequality again, we can obtain that
	$$\frac{1}{n}\sum_{j=1}^n\sum_{t=N_K+1}^{T-1}M_{j;t,t}=1-\frac{1}{n}\sum_{j=1}^n\sum_{t=1}^{N_K} M_{j;t,t}\leq\delta_K+C_Kn^{-1/5}\ {\rm with\ probability}\geq1-C_Kn^{-3/5},$$
	where \(\delta_K\) is defined in (\ref{Eq of delta_K and epsilon_K}). Let's define an event 
    \begin{align}
        \mcE_K(\epsilon)=\Bigg\{n^{-1}\Bigg|\hat{\lambda}_1-\sum_{j=1}^n\sum_{t=N_K+1}^{T-1}M_{j;t,t}\Bigg|>\epsilon\Bigg\},\label{Eq of mcE_K}
    \end{align}
    for \(\epsilon>0\). Here, we choose a sufficiently large \(N_K\) such that \(\delta_K<\mbE[\fM_{1,1}-\fM_{2,2}]/2\), then let \(\epsilon=\mbE[\fM_{1,1}-\fM_{2,2}]/2\) and we can show that \(\mbP(\mcE_K(\epsilon))\geq1- C_Kn^{-3/5}\) by the same method as \eqref{Eq of denominator events 1}. Next, recall that \(\epsilon_K=\sum_{k=K+1}^{T-1}\alpha_{1,k}^2\) defined in (\ref{Eq of delta_K and epsilon_K}), we still have
	$$\sqrt{n}\mathbb{E}[\epsilon_K]=\sqrt{n}\mathbb{E}[\epsilon_K|\mcE_K(\epsilon)]\mathbb{P}(\mcE_K(\epsilon))+\sqrt{n}\mathbb{E}[\epsilon_K|\mcE_K(\epsilon)^c]\mathbb{P}(\mcE_K(\epsilon)^c)\leq C_K n^{-1/10}+\sqrt{n}\mathbb{E}[\epsilon_K|\mcE_K(\epsilon)],$$
    where
    \begin{align}
        &\mbE[\epsilon_K|\mcE_K(\epsilon)]=\sum_{t=N_K+1}^{T-1}\mbE\Bigg[\frac{(\sum_{k\neq t}^{T-1}\alpha_{1,k}n^{-1}\sum_{j=1}^n M_{j;k,t})^2}{[n^{-1}(\hat{\lambda}_1-\sum_{j=1}^nM_{j;t,t})]^2}\Bigg|\mcE_K(\epsilon)\Bigg]\notag\\
        &\leq\epsilon^{-2}\sum_{t=N_K+1}^{T-1}\mbE\Bigg[\Bigg(\sum_{k\neq t}^{T-1}\alpha_{1,k}n^{-1}\sum_{j=1}^n M_{j;k,t}\Bigg)^2\Bigg].\label{Eq of epsilon_K}
    \end{align}
    Now, combining (\ref{Eq of alpha2}) and (\ref{Eq of epsilon_K}), we conclude that
    $$\sqrt{n}\mbE[1-\alpha_{1,1}^2]\leq C_Kn^{-1/10}+\epsilon^{-2}n^{-3/2}\sum_{t=2}^{T-1}\mbE\Bigg[\Bigg(\sum_{k\neq t}^{T-1}\alpha_{1,k}\sum_{j=1}^n M_{j;k,t}\Bigg)^2\Bigg],$$
    where \(\epsilon=\mbE[\fM_{1,1}-\fM_{2,2}]/2\). Hence, it is enough to show that
    $$\lim_{n\to\infty}n^{-3/2}\sum_{t=2}^{T-1}\mbE\Bigg[\Bigg(\sum_{k\neq t}^{T-1}\alpha_{1,k}\sum_{j=1}^n M_{j;k,t}\Bigg)^2\Bigg]=0.$$
    By Cauchy's inequality and $\sum_{k\neq t}^{T-1}\alpha_{1,k}^2\leq1$, we have
	\begin{align*}
		&\mbE\Bigg[\Bigg(\sum_{k\neq t}^{T-1}\alpha_{1,k}\sum_{j=1}^n M_{j;k,t}\Bigg)^2\Bigg]\leq\sum_{k\neq t}^{T-1}\mbE\Bigg[\Bigg(\sum_{j=1}^n M_{j;k,t}\Bigg)^2\Bigg]\\
        &=\sum_{k\neq t}^{T-1}\sum_{j=1}^n\mbE[M_{j;k,t}^2]+\sum_{k\neq t}^{T-1}\sum_{j_1\neq j_2}^n\mbE[M_{j_1;k,t}]\mbE[M_{j_2;k,t}].
	\end{align*}
	Since \(M_{j;k,t}^2=M_{j;k,k}M_{j;t,t}\) and $\sum_{k\neq t}^{T-1}M_{j;k,k}\leq1$ by \eqref{Eq of Mjkl}, then 
    \begin{align}
        &\sum_{t=2}^{T-1}\sum_{k\neq t}^{T-1}\sum_{j=1}^nM_{j;k,t}^2\leq\sum_{j=1}^n\sum_{t=2}^{T-1}M_{j;t,t}\sum_{k\neq t}^{T-1}M_{j;k,k}\leq n.\label{Eq of alpha sum 1}
    \end{align}
    Moreover, by (\ref{Eq of Lindeberg 2}), i.e. \(|\mbE[M_{j;k,t}]|\leq C_{B,M,\kappa_6}(kt)^{-1}\log^3(T)T^{-1/2}\) for \(1\leq k,t\leq T-1\) amd \(k\neq t\), so
	\begin{align}
		&\sum_{t=2}^{T-1}\sum_{k\neq t}^{T-1}\sum_{j_1\neq j_2}^n\mbE[M_{j_1;k,t}]\mbE[M_{j_2;k,t}]\leq C_{B,M,\kappa_6,c}\log^6(T)n.\label{Eq of alpha sum 2}
	\end{align}
	Finally, combining (\ref{Eq of alpha sum 1}) and (\ref{Eq of alpha sum 2}), it gives that
    \begin{align*}
        &n^{-3/2}\sum_{t=2}^{T-1}\mbE\Bigg[\Bigg(\sum_{k\neq t}^{T-1}\alpha_{1,k}\sum_{j=1}^n M_{j;k,t}\Bigg)^2\Bigg]\leq C_{B,M,\kappa_8,c}\log^6(T)T^{-1/2},
    \end{align*}
    which completes our proof.
\end{proof}
Finally, we will inductively prove Lemma \ref{Lem of at convergence}.
\begin{proof}[Proof of Lemma \ref{Lem of at convergence}]
	Currently, we have proven that \(\sqrt{n}\mbE[1-\alpha_{1,1}^2]=0\), let's continue to show that \(\sqrt{n}\mbE[1-\alpha_{2,2}^2]=0\). Similar as Lemma \ref{Lem of upper bound}, we first need to show that \(n^{-1}\big|\hla_2-\sum_{j=1}^nM_{j;2,2}\big|\) is stochastically bounded by \(C_2\delta_K^{1/2}\) with probability at least of \(1-C_K n^{-3/5}\). The key step is to first show \(\alpha_{2,1}\) is stochastically bounded by \(2(1-\alpha_{1,1}^2)\). Since \(\langle\hat{F}_1,\hat{F}_2\rangle=0\), then we have
	$$\langle\hat{F}_1,\hat{F}_2\rangle=\sum_{t=1}^{T-1}\alpha_{2,t}\langle\hat{F}_1,\bbw_t\rangle=\sum_{t=1}^{T-1}\alpha_{2,t}\alpha_{1,t}=0,$$
	i.e. \(\alpha_{2,1}^2=\alpha_{1,1}^{-2}|\sum_{t=2}^{T-1}\alpha_{2,t}\alpha_{1,t}|^2\leq\alpha_{1,1}^{-2}(1-\alpha_{1,1}^2)\leq2(1-\alpha_{1,1}^2)\). By Lemma \ref{Lem of upper bound}, we can conclude that \(\alpha_{2,1}^2\leq12\mbE[\fM_{1,1}-\fM_{2,2}]^{-1}\delta_K^{1/2}\) with probability at least \(1-C_Kn^{-3/5}\) for sufficiently large constant \(N_K\in\mbN^+\). Next, we can repeat the same argument as in Lemma \ref{Lem of upper bound} and obtain that
	\begin{align}
		&n^{-1}\hat{\lambda}_2\leq\alpha_{2,1}^2\mathbb{E}[\fM_{1,1}]+\alpha_{2,2}^2\mathbb{E}[\fM_{2,2}]+(1-\alpha_{2,2}^2)\mathbb{E}[\fM_{3,3}]+3\delta_K^{1/2}\notag\\
		&\leq\alpha_{2,2}^2\mathbb{E}[\fM_{2,2}]+(1-\alpha_{2,2}^2)\mathbb{E}[\fM_{3,3}]+(3+12\mbE[\fM_{1,1}-\fM_{2,2}]^{-1})\delta_K^{1/2}\notag
	\end{align}
    with probability greater than \(1-C_Kn^{-3/5}\). Similarly, since \(n^{-1}\sum_{j=1}^n M_{j;2,2}\leq n^{-1}\hat{\lambda}_2\), we can further imply that 
	$$1-\alpha_{2,2}^2\leq(3+12\mbE[\fM_{1,1}-\fM_{2,2}]^{-1})\mathbb{E}[\fM_{2,2}-\fM_{3,3}]^{-1}\delta_K^{1/2}$$
	with probability greater than \(1-C_Kn^{-3/5}\), so we can consequently obtain \(n^{-1}\big|\hla_2-\sum_{j=1}^nM_{j;2,2}\big|\leq C_2\delta_K^{1/2}\) as (\ref{Eq of C1}), where \(C_2:=3+(3+12\mbE[\fM_{1,1}-\fM_{2,2}]^{-1})\mathbb{E}[\fM_{2,2}-\fM_{3,3}]^{-1}\). By this way, we can inductively construct $C_k$ such that for $1\leq k\leq K$
    $$\mbP\left(n^{-1}\left|\hla_k-\sum_{j=1}^nM_{j;k,k}\right|\leq C_k\delta_K^{1/2}\right)\geq1-C_Kn^{-3/5}.$$
    Based on the above result and
	\begin{align}
		\alpha_{2,t}=\frac{n^{-1}\sum_{k\neq t}^{T-1}\alpha_{2,k}\sum_{j=1}^n M_{j;k,t}}{n^{-1}(\hat{\lambda}_2-\sum_{j=1}^n M_{j;t,t})}\ \ {\rm where\ }t\neq2,\label{Eq of square n general at}
	\end{align}
	we can repeat the same argument as those for \(\lim_{n\to\infty}\sqrt{n}\mbE[1-\alpha_{1,1}^2]=0\) to show that \(\sqrt{n}\mathbb{E}[\alpha_{2,t}^2]\to0\) for \(t\in\{1,\cdots,K\},t\neq2\) and \(\sqrt{n}\sum_{k=N_K+1}^{T-1}\mathbb{E}[\alpha_k^2]\to0\). For instance, we can show that the event \(\mathcal{F}_t(\epsilon):=\{n^{-1}|\hat{\lambda}_2-\sum_{j=1}^n M_{j;t,t}|>\epsilon\}\) has probability greater than \(1-C_Kn^{-3/5}\) for \(\epsilon=\mbE[\fM_{2,2}-\fM_{3,3}]/2\) and sufficiently large \(N_K\), then it gives that
	\begin{align*}
	    &\sqrt{n}\mathbb{E}[\alpha_{2,t}^2]=\sqrt{n}\mathbb{E}[\alpha_{2,t}^2|\mcF_t(\epsilon)]\mathbb{P}(\mcF_t(\epsilon))+\sqrt{n}\mathbb{E}[\alpha_{2,t}^2|\mcF_t(\epsilon)^c]\mathbb{P}(\mcF_t(\epsilon)^c)\\
        &\leq C_K n^{-1/10}+\sqrt{n}\mathbb{E}[\alpha_{2,t}^2|\mcF_t(\epsilon)^c]\mathbb{P}(\mcF_t(\epsilon)^c),
	\end{align*}
	to show the remain term \(\sqrt{n}\mathbb{E}[\alpha_{2,t}^2|\mcF_t(\epsilon)^c]\mathbb{P}(\mcF_t(\epsilon)^c)\) converges to \(0\), it is enough to prove the expectation of the numerator in (\ref{Eq of square n general at}) converges to \(0\), since the proof arguments are totally the same as Lemma \ref{Lem of at convergence}, we omit the details here to save space.
	
	Finally, let's determine the value of \(N_K\) defined in Lemma \ref{Lem of upper bound}. According the above procedures, we need to control the absolute lower bound of the denominator in (\ref{Eq of square n general at}), that is, we hope to find a proper constant \(\epsilon>0\) such that the probabilities of \(n^{-1}|\hat{\lambda}_k-\sum_{j=1}^n M_{j;t,t}|<\epsilon\) (\(t\neq k,1\leq t\leq N_K\)) and \(n^{-1}|\hat{\lambda}_k-\sum_{t=N_K+1}^{T-1}\sum_{j=1}^n M_{j;t,t}|<\epsilon\) have order of \(\mro(n^{-1/2})\). To realize this conclusion, since we show that the probability of \(n^{-1}|\hat{\lambda}_k-\sum_{j=1}^n M_{j;k,k}|<C_k\delta_K^{1/2}\) and \(n^{-1}|\sum_{j=1}^n(M_{j;k,k}-M_{j;t,t})^{\circ}|<n^{-1/5}\) have order of \(\mro(n^{-1/2})\). Similar as \ref{Eq of denominator events 1}, we can obtain that for any $t\neq k$
	\begin{align*}
		&n^{-1}\Bigg|\hat{\lambda}_k-\sum_{j=1}^nM_{j;t,t}\Bigg|\\
        &\geq n^{-1}\Bigg|\sum_{j=1}^n\mbE[M_{j;k,k}-M_{j;t,t}]\Bigg|-n^{-1}\Bigg|\hat{\lambda}_k-\sum_{j=1}^nM_{j;k,k}\Bigg|-n^{-1}\Bigg|\sum_{j=1}^n(M_{j;k,k}-M_{j;t,t})^{\circ}\Bigg|\\
		&\geq\mathbb{E}[\fM_{k,k}-\fM_{t,t}]-C_k\delta_K^{1/2}-\mrO(n^{-1/5})\ {\rm with\ probability}\geq1-\mro(n^{-1/2}).
	\end{align*}
	Therefore, we need \(\delta_K\) is small enough such that for \(1\leq k\leq K\)
    \begin{align*}
        &C_k\delta_K^{1/2}<\min_{1\leq t\leq N_K,t\neq k}\big|\mathbb{E}[\fM_{k,k}-\fM_{t,t}]\big|/2=\mathbb{E}[\fM_{k,k}-\fM_{k+1,k+1}]/2.
    \end{align*}
    Besides, similar as (\ref{Eq of mcE_K}), we also have for \(1\leq k\leq K\)
    \begin{align*}
        &n^{-1}\left|\hla_k-\sum_{j=1}^n\sum_{t=N_K+1}^{T-1}M_{j;t,t}\right|\geq n^{-1}\Bigg|\sum_{j=1}^n\left(\mbE[M_{j;k,k}]-\sum_{t=N_K+1}^{T-1}\mbE[M_{j;t,t}]\right)\Bigg|-n^{-1}\Bigg|\hat{\lambda}_k-\sum_{j=1}^nM_{j;k,k}\Bigg|\\
        &-n^{-1}\Bigg|\sum_{j=1}^n\left(M_{j;k,k}-\sum_{t=N_K+1}^{T-1}M_{j;t,t}\right)^{\circ}\Bigg|\geq\mbE[\fM_{k,k}]-\delta_K-C_k\delta_K^{1/2}-\mrO(n^{-1/5}),
    \end{align*}
    with probability greater than $1-\mro(n^{-1/2})$, so we require that
    $$(1+C_k)\delta_K^{1/2}<\mbE[\fM_{k,k}]/2.$$
    Thus, the \(N_K\) should be chosen such that
    \begin{align}
    	&\delta_K=\sum_{t=N_K+1}^{T-1}\mbE[\fM_{t,t}]\notag\\
        &<\min\big\{\mathbb{E}[\fM_{k,k}-\fM_{k+1,k+1}]^2/(2C_k)^2,\mbE[\fM_{k,k}]^2/(2+2C_k)^2:1\leq k\leq K\big\},\label{Eq of N_K}
    \end{align}
    so we complete the proof of Lemma \ref{Lem of at convergence}.
\end{proof}
\subsection{Joint CLT for the extreme eigenvalues of the sample correlation matrix}\label{sec of CLT bbR}
In this section, we will prove Theorem \ref{Thm of CLT}. Let's briefly outline the proof here. First, we will establish the joint centered CLT for $(\hla_1,\cdots,\hla_K)'$, i.e.
$$n^{-1/2}(\hla_1^{\circ},\cdots,\hla_K^{\circ})'\overset{d}{\longrightarrow}\mcN(\boldsymbol{0},\mcS),$$
where $\hla_k^{\circ}=\hla_k-\mbE[\hla_k]$ and $\mcS\in\mbR^{K\times K}$ is defined in Theorem \ref{Thm of CLT}. Recall that we have shown that $\lim_{n\to\infty}\sqrt{n}\mbE[1-\alpha_{k,k}^2]=0$ in Lemma \ref{Lem of at convergence}. Hence, according to \eqref{Eq of la1}, we can show that
$$\frac{\hla_k^{\circ}}{\sqrt{n}}\overset{\mbP}{\longrightarrow}\frac{1}{\sqrt{n}}\sum_{j=1}^nM_{j;k,k}^{\circ}.$$
Consequently, we obtain that
$$n^{-1/2}(\hla_1^{\circ},\cdots,\hla_K^{\circ})'\overset{\mbP}{\longrightarrow}\Bigg(\frac{1}{\sqrt{n}}\sum_{j=1}^nM_{j;1,1}^{\circ},\cdots,\frac{1}{\sqrt{n}}\sum_{j=1}^nM_{j;K,K}^{\circ}\Bigg)',$$
so it suffices to establish the joint CLT as follows:
$$\Bigg(\frac{1}{\sqrt{n}}\sum_{j=1}^nM_{j;1,1}^{\circ},\cdots,\frac{1}{\sqrt{n}}\sum_{j=1}^nM_{j;K,K}^{\circ}\Bigg)'\overset{d}{\longrightarrow}\mcN(\boldsymbol{0},\mcS).$$
After deriving the joint centered CLT for $n^{-1/2}(\hla_1^{\circ},\cdots,\hla_K^{\circ})'$, combining with Lemma \ref{Thm of Lindeberg principle}, we know that $n^{-1/2}\sum_{j=1}^n\big|\mbE[M_{j;k,k}-\fM_{k,k}]\big|\leq C_{B,M,\kappa_8}$, which finally concludes Theorem \ref{Thm of CLT}.
\subsubsection{Preliminary CLT}\label{ssec of preliminary CLT}
In this part, let's show that
\begin{pro}\label{Thm of pre joint CLT}
	Under Assumptions {\rm \ref{Ap of highdimensionality} \ref{Ap of finite integration}} and {\rm \ref{Ap of panel lag polynomial}}, for any \(K\in\mbN^+\), then
	$$\Bigg(\frac{1}{\sqrt{n}}\sum_{j=1}^nM_{j;1,1}^{\circ},\cdots,\frac{1}{\sqrt{n}}\sum_{j=1}^nM_{j;K,K}^{\circ}\Bigg)'\overset{d}{\longrightarrow}\mcN(\boldsymbol{0},\mcS),$$
	where $M_{j;k,k}$ is defined in \eqref{Eq of Mjkl} and \(M_{j;k,k}^{\circ}:=M_{j;k,k}-\mathbb{E}[M_{j;k,k}]\), \(\mcS\) is a \(K\times K\) covariance matrix defined in Theorem {\rm \ref{Thm of CLT}}.
\end{pro}
For preliminaries, we first prove the following lemma:
\begin{lem}\label{Lem of covariance error}
	For any $K\in\mbN^+$, \(k,l,p,q\in\{1,\cdots,K\}\) and \(1\leq j\leq n\), we have
	$$\lim_{n\to\infty}\big|{\rm Cov}(M_{j;k,l},M_{j;p,q})-{\rm Cov}(\fM_{k,l},\fM_{p,q})\big|=0$$
	uniformly in \(1\leq j\leq n\), where $M_{j;k,l}$ and $\fM_{k,l}$ are defined in \eqref{Eq of Mjkl} and \eqref{Eq of fM}, respectively.
\end{lem}
\begin{proof}
	Note that
	$${\rm Cov}(M_{j;k,l},M_{j;p,q})=\mathbb{E}[M_{j;k,l}M_{j;p,q}]-\mathbb{E}[M_{j;k,l}]\mathbb{E}[M_{j;p,q}],$$
	let's first show that
	$$\lim_{n\to\infty}\big|\mathbb{E}[M_{j;k,l}M_{j;p,q}-\fM_{k,l}\fM_{p,q}]\big|=0$$
	is uniform in \(j\). For any small \(\epsilon>0\), define an event \(\mcG_{1,j}(\epsilon):=\{(\bbe_j\bbv_1)^2\leq\epsilon\}\) and let \(R:=R(\epsilon)\) be a pre-specified integer depending on \(\epsilon\) such that \(R>K\), then denote
	$$\bar{M}_{j;k,l}:=\frac{\beta_k\beta_l(\bbe_j\bbv_k)(\bbe_j\bbv_l)}{\sum_{t=1}^R\beta_t^2(\bbe_j\bbv_t)^2}\quad{\rm and}\quad Z_j(R):=\sum_{t=R+1}^{\infty}\beta_t^2(\bbe_j\bbv_t)^2,$$
	where \(\beta_t\) has been defined in (\ref{Eq of beta}). Since \(|\bar{M}_{j;k,l}|,|M_{j;k,l}|,|\bar{M}_{j;p,q}|,|M_{j;p,q}|\leq1\), we have
	\begin{align}
		&\mathbb{E}\big[|\bar{M}_{j;k,l}\bar{M}_{j;p,q}-M_{j;k,l}M_{j;p,q}|1_{\mcG_{1,j}(\epsilon)}\big]\leq2\mathbb{P}(\mcG_{1,j}(\epsilon)),\notag\\
		&\mathbb{E}\big[|\bar{M}_{j;k,l}\bar{M}_{j;p,q}-M_{j;k,l}M_{j;p,q}|1_{\mcG_{1,j}(\epsilon)^c}\big]\leq\mathbb{E}\big[(|\bar{M}_{j;p,q}(\bar{M}_{j;k,l}-M_{j;k,l})|+|M_{j;k,l}(\bar{M}_{j;p,q}-M_{j;p,q})|)1_{\mcG_{1,j}(\epsilon)^c}\big]\notag\\
		&\leq\mathbb{E}\big[(|\bar{M}_{j;k,l}-M_{j;k,l}|+|\bar{M}_{j;p,q}-M_{j;p,q}|)1_{\mcG_{1,j}(\epsilon)^c}\big]\leq\frac{(\beta_k\beta_l+\beta_p\beta_q)\mathbb{E}[Z_j(R)]}{\beta_1^2\epsilon}.\notag
	\end{align}
	By Lemma \ref{Lem of covariance 1}, we have
	$$\beta_1^{-2}\mathbb{E}[Z_j(R)]\leq\frac{\pi^2}{4}\sum_{t=R+1}^{T-1}t^{-2}(B+CB^2/T)\leq C_BR^{-1}.$$
    and \(\mbP(\mcG_{1,j}(\epsilon))\leq C_B\epsilon^{1/2}\) due to \(\bbe_j\bbv_1\overset{d}{\longrightarrow}\mcN(0,2\pi f_j(0))\) by (\ref{Eq of spectral density}), then we obtain
	$$\mathbb{E}\big[|\bar{M}_{j;k,l}\bar{M}_{j;p,q}-M_{j;k,l}M_{j;p,q}|\times1_{\mcG_{1,j}(\epsilon)^c}\big]\leq C_{K,B}(R\epsilon)^{-1}$$
	and
	$$\mathbb{E}\big[|\bar{M}_{j;k,l}\bar{M}_{j;p,q}-M_{j;k,l}M_{j;p,q}|\big]\leq C_{K,B}(\epsilon^{1/2}+(R\epsilon)^{-1}),$$
    we choose \(R>\epsilon^{-3/2}\), then
	$$\mathbb{E}\big[|\bar{M}_{j;k,l}\bar{M}_{j;p,q}-M_{j;k,l}M_{j;p,q}|\big]\leq C_{K,B}\epsilon^{1/2}.$$
    Moreover, define
	$$\bar{\fM}_{k,l}:=\frac{(kl)^{-1}Z_kZ_l}{\sum_{t=1}^Rt^{-2}Z_t^2}$$
	and we can use the same argument to show that \(\mathbb{E}\big[|\fM_{k,l}\fM_{p,q}-\bar{\fM}_{k,l}\bar{\fM}_{p,q}|\big]\leq C_K\epsilon^{1/2}\). Finally, it remains to bound \(|\mathbb{E}[\bar{\fM}_{k,l}\bar{\fM}_{p,q}-\bar{M}_{j;k,l}\bar{M}_{j;p,q}]|\) uniformly in \(j\). Let \(\tau_{\epsilon}:[0,\infty)\to\mathbb{R}\) be thrice continuously differential function such that
	$$\tau_{\epsilon}(x)=\left\{\begin{array}{ll}
		x&x>\epsilon\\>\epsilon/2&x\in[0,\epsilon]
	\end{array}\right.$$
	and the first three derivatives of \(\tau_{\epsilon}(x)\) bounded for \(x\in[0,\epsilon]\). Furthermore, define
	$$\bar{M}_{j;k,l}^{(\tau)}:=\frac{\beta_k\beta_l(\bbe_j\bbv_k)(\bbe_j\bbv_l)}{\beta_1^2\tau_{\epsilon}((\bbe_j\bbv_1)^2)+\sum_{t=2}^R\beta_t^2(\bbe_j\bbv_t)^2}\quad{\rm and}\quad\bar{\fM}_{k,l}^{(\tau)}:=\frac{(kl)^{-1}Z_kZ_l}{\tau_{\epsilon}(Z_1^2)+\sum_{t=2}^Rt^{-2}Z_t^2},$$
	by Lemma 10 in \cite{onatski2021spurious}, we have \(\big|\mathbb{E}[\bar{M}_{j;k,l}^{(\tau)}\bar{M}_{j;p,q}^{(\tau)}-\bar{\fM}_{k,l}^{(\tau)}\bar{\fM}_{p,q}^{(\tau)}]\big|\leq\epsilon^{1/2}\) uniformly in \(j\) as \(n\to\infty\), since \(|\bar{M}_{j;k,l}^{(\tau)}|,|\tilde{M}_{j,pq}|\leq2\) and
	$$\bar{M}_{j;k,l}\bar{M}_{j;p,q}-\bar{M}_{j;k,l}^{(\tau)}\bar{M}_{j;p,q}^{(\tau)}=\big(\bar{M}_{j;k,l}\bar{M}_{j;p,q}-\bar{M}_{j;k,l}^{(\tau)}\bar{M}_{j;p,q}^{(\tau)}\big)1_{\mcG_{1,j}(\epsilon)},$$
	then we have \(\mathbb{E}\big[|\bar{M}_{j;k,l}\bar{M}_{j;p,q}-\bar{M}_{j;k,l}^{(\tau)}\bar{M}_{j;p,q}^{(\tau)}|\big]\leq4\mbP(\mcG_{1,j}(\epsilon))<C_B\epsilon^{1/2}\). As a result, we conclude that \(\big|\mathbb{E}[M_{j;k,l}M_{j;p,q}-\fM_{k,l}\fM_{p,q}]\big|<C_{B,K}\epsilon^{1/2}\) uniformly in \(j\), combining with Lemma \ref{Thm of Lindeberg principle}, it implies that
	$$\big|{\rm Cov}(M_{j;k,l},M_{j;p,q})-{\rm Cov}(\fM_{k,l},\fM_{p,q})\big|<C_{B,K}\epsilon^{1/2}$$
	uniformly in \(j\) for any small \(\epsilon>0\).
\end{proof}
Now we prove Proposition \ref{Thm of pre joint CLT} as follows:
\begin{proof}[Proof of Proposition \ref{Thm of pre joint CLT}]
	Define \(\mathfrak{M}_{j,K}=(M_{j;1,1},\cdots,M_{j;K,K})'\), it is enough to show that 
	$$\frac{1}{\sqrt{n}}\sum_{j=1}^n\langle\bba,\mathfrak{M}_{j,K}^{\circ}\rangle:=\frac{1}{\sqrt{n}}\sum_{j=1}^n\sum_{k=1}^K a_kM_{j;k,k}^{\circ}$$
	asymptotically weakly converges to a normal distribution for any \(\bba=(a_1,\cdots,a_K)'\in\mathbb{R}^K\). Further let \(\mathfrak{m}_K:=(\fM_{1,1},\cdots,\fM_{K,K})'\) and
	$$s_n^2(\bba):=\sum_{j=1}^n{\rm Var}(\langle\bba,\mathfrak{M}_{j,K}^{\circ}\rangle)\ \ {\rm and\ \ }\fs^2(\bba):={\rm Var}(\langle\bba,\mathfrak{m}_K^{\circ}\rangle)>0.$$
    Here, let's briefly explain why \(\fs^2(\bba)>0\). Otherwise, if there exists \(\bba\) such that \(\fs^2(\bba)=0\), then it implies that \(\langle\bba,\mathfrak{m}_K^{\circ}\rangle=C\) almost surely, where \(C\) is a constant. Hence, we obtain
    $$\frac{\sum_{k=1}^Ka_kk^{-2}z_k^2}{\sum_{t=1}^{\infty}t^{-2}z_t}=C\quad\Rightarrow\quad\sum_{k=1}^K(a_k-C)k^{-2}z_k^2=C\sum_{t=K+1}^{\infty}t^{-2}z_t,$$
    since \(\sum_{t=K+1}^{\infty}t^{-2}z_t\) and \(\sum_{k=1}^K(a_k-C)k^{-2}z_k^2\) are independent, then \(C=0\). Besides, all \(z_k\) are also independent for \(1\leq k\leq K\), so all \(a_k=0\), i.e. \(\bba=\boldsymbol{0}\), which is a contradiction. Thus \(\fs^2(\bba)>0\), then by Lemma \ref{Lem of covariance error}, as \(n,T\to\infty\), we have \({\rm Var}(\langle\bba,\mathfrak{M}_{j,K}^{\circ}\rangle)\to\fs^2(\bba)\) uniformly in \(j\), so it implies that \({\rm Var}(\langle\bba,\mathfrak{M}_{j,K}^{\circ}\rangle)>\fs^2(\bba)/2\) as \(n\to\infty\). Since all \(\langle\bba,\mathfrak{M}_{j,K}^{\circ}\rangle\) are independent and absolutely bounded by \(K\), then for any fixed \(\epsilon>0\)
    $$s_n^2(\bba)=\sum_{j=1}^n{\rm Var}\big(\langle\bba,\mathfrak{M}_{j,K}^{\circ}\rangle\big)>n\fs^2(\bba)/2>K.$$
    as \(n\to\infty\), so the Lindeberg's condition holds:
	$$\frac{1}{s_n^2(\bba)}\sum_{j=1}^n\mathbb{E}\big[|\langle\bba,\mathfrak{M}_{j,K}^{\circ}\rangle|^2 1_{|\langle\bba,\mathfrak{M}_{j,K}^{\circ}\rangle|>\epsilon s_n(\bba)}\big]=0,\ \ {\rm for\ }\forall\epsilon>0.$$
	Therefore, by the Lindeberg-Feller CLT, we obtain
	$$\frac{1}{s_n(\bba)}\sum_{j=1}^n\langle\bba,\mathfrak{M}_{j,K}^{\circ}\rangle\overset{d}{\longrightarrow}\mathcal{N}(0,1),$$
	i.e.
	$$\frac{1}{\sqrt{n}}\sum_{j=1}^n\langle\bba,\mathfrak{M}_{j,K}^{\circ}\rangle\overset{d}{\longrightarrow}\mathcal{N}(0,\fs^2(\bba)),$$
	which completes our proof.
\end{proof}
\subsubsection{Centered CLT}\label{ssec of centered joint CLT}
\begin{pro}\label{Thm of univariate CLT}
	Under Assumptions {\rm \ref{Ap of highdimensionality}, \ref{Ap of panel lag polynomial}} and {\rm \ref{Ap of finite integration}}, for any \(K\in\mathbb{N}^+\), let $\hla_1,\cdots,\hla_K$ be the first $K$ largest eigenvalue of $\hR$ in \eqref{Eq of correlation matrx random walk 1}, then
	$$\Bigg(\frac{\hla_1^{\circ}}{\sqrt{n}},\cdots,\frac{\hla_K^{\circ}}{\sqrt{n}}\Bigg)\overset{d}{\longrightarrow}\mathcal{N}(\boldsymbol{0},\mcS),$$
	where \(\mcS\) is the \(K\times K\) covariance matrix defined in Theorem {\rm \ref{Thm of CLT}}.
\end{pro}
\begin{proof}
	First, we will show that 
	$$n^{-1/2}\Bigg(\hat{\lambda}_k-\sum_{j=1}^n M_{j;k,k}\Bigg)\overset{\mbP}{\longrightarrow}0\quad{\rm for\ }1\leq k\leq K.$$
	Without loss of generality, we only present the proof for \(k=1\), since others are totally the same. In this proof, we will {\bf abuse} the notations \(A_1,B_1\) defined in (\ref{Eq of A1 A2}) as follows:
    $$A_1:=\sum_{k=1}^K\alpha_{1,k}\sigma_k\bbD^{-1/2}\bbe\bbv_k\quad{\rm and}\quad B_1:=\sum_{k=K+1}^{T-1}\alpha_{1,k}\sigma_k\bbD^{-1/2}\bbe\bbv_k,$$
    then we have
	$$\frac{\hat{\lambda}_1^{\circ}}{\sqrt{n}}=\frac{(\Vert A_1\Vert^2)^{\circ}}{\sqrt{n}}+\frac{(\Vert B_1\Vert^2)^{\circ}}{\sqrt{n}}+2\frac{\langle A_1,B_1\rangle^{\circ}}{\sqrt{n}},$$
	Since
	$$\frac{(\Vert A_1\Vert^2)^{\circ}}{\sqrt{n}}=\sum_{k,l=1}^K\frac{1}{\sqrt{n}}\sum_{j=1}^n\big(\alpha_{1,k}\alpha_{1,l}(M_{j;k,l})^{\circ}+\alpha_{1,k}\alpha_{1,l}\mathbb{E}[M_{j;k,l}]-\mathbb{E}[\alpha_{1,k}\alpha_{1,l}M_{j;k,l}]\big),$$
	by Corollary \ref{Cor of L2 convergence} and \(\Var\big(n^{-1/2}\sum_{j=1}^nM_{j;k,l}\big)\leq4\), we have that
	$$\sum_{k\neq1\ {\rm or\ }l\neq1}^{K}\alpha_{1,k}\alpha_{1,l}\frac{1}{\sqrt{n}}\sum_{j=1}^n(M_{j;k,l})^{\circ}\overset{\mbP}{\longrightarrow}0.$$
	By Corollary \ref{Cor of L2 convergence} and (\ref{Eq of Lindeberg 1}) in Lemma \ref{Thm of Lindeberg principle}, we have that
	\begin{align*}
		&\Bigg|\sum_{k\neq1\ {\rm or\ }l\neq1}^K\frac{1}{\sqrt{n}}\sum_{j=1}^n\alpha_{1,k}\alpha_{1,l}\mathbb{E}[M_{j;k,l}]\Bigg|\leq\Bigg|\sum_{k,t=2}^K\sqrt{n}\alpha_{1,k}\alpha_{1,t}n^{-1}\sum_{j=1}^n\mathbb{E}[M_{j;k,t}]\Bigg|\notag\\
		&+|\alpha_{1,1}|\sum_{k=2}^K|\alpha_{1,k}|n^{-1/2}\Bigg|\sum_{j=1}^n\mathbb{E}[M_{j;1,k}]\Bigg|\\
        &\leq\sum_{k,t=2}^K\sqrt{n}|\alpha_{1,k}\alpha_{1,t}|+C_{B,M,\kappa_8,c}\sum_{k=2}^K|\alpha_{1,k}|\overset{\mbP}{\longrightarrow}0.
	\end{align*}
	In addition, since \(|{\rm Cov}(\alpha_{1,k}\alpha_{1,l},n^{-1/2}\sum_{j=1}^nM_{j;k,l})|\leq2\mathbb{E}[\alpha_{1,k}^2\alpha_{1,l}^2]^{1/2}\), where we use the fact that \({\rm Var}(n^{-1/2}\sum_{j=1}^nM_{j;k,l})\leq4\) for \(k\neq l\), we have that
	$$\Bigg|\sum_{k\neq1\ {\rm or\ }l\neq1}^K\frac{1}{\sqrt{n}}\sum_{j=1}^n\mathbb{E}[\alpha_{1,k}\alpha_{1,l}M_{j;kl}]-\mathbb{E}[\alpha_{1,k}\alpha_{1,l}]\mathbb{E}[M_{j;k,l}]\Bigg|\leq2\sum_{k\neq1\ {\rm or\ }l\neq1}^K\mathbb{E}[\alpha_{1,k}^2\alpha_{1,l}^2]^{1/2}\longrightarrow0,$$
	by Lemma \ref{Lem of at convergence} and (\ref{Eq of Lindeberg 1}), we can obtain
	\begin{align*}
	    &\sum_{k\neq1\ {\rm or\ }l\neq1}^K\mathbb{E}[\alpha_{1,k}\alpha_{1,l}]\frac{1}{\sqrt{n}}\sum_{j=1}^n\mathbb{E}[M_{j;k,l}]\leq C_{B,M,\kappa_8,c}\sum_{k\neq1\ {\rm or\ }l\neq1}^K\sqrt{n}\mathbb{E}[\alpha_{1,k}\alpha_{1,l}](\delta_{k,l}+T^{-1/2})\\
        &\leq\sum_{k=2}^K\sqrt{n}\mathbb{E}[\alpha_{1,k}^2]+C_{B,M,\kappa_8,c}T^{-1/2}\sum_{k\neq l;k,l\geq2}^K\sqrt{n}\mbE[\alpha_{1,k}^2]^{1/2}\mbE[\alpha_{1,l}^2]^{1/2}\longrightarrow0.
	\end{align*}
	As a result, we conclude that
	$$\frac{(\Vert A_1\Vert^2)^{\circ}}{\sqrt{n}}=\frac{1}{\sqrt{n}}\sum_{j=1}^n(\alpha_{1,1}^2 M_{j;1,1})^{\circ}+\mro_{\mbP}(1)+\mro(1).$$
	Similarly, by Lemma \ref{Lem of at convergence} and Corollary \ref{Cor of L2 convergence}, we also have
	\begin{align}
		&\frac{1}{\sqrt{n}}\sum_{j=1}^n[M_{j;1,1}^{\circ}-(\alpha_{1,1}^2 M_{j;1,1})^{\circ}]=\frac{1}{\sqrt{n}}\sum_{j=1}^n[(1-\alpha_1^2)M_{j;1,1}]^{\circ}\notag\\
		&=\sqrt{n}(1-\alpha_{1,1}^2)n^{-1}\sum_{j=1}^nM_{j;1,1}-\mathbb{E}\Bigg[\sqrt{n}(1-\alpha_{1,1}^2)n^{-1}\sum_{j=1}^n M_{j;1,1}\Bigg]\overset{\mbP}{\longrightarrow}0,\notag
	\end{align}
	where we use the fact that \(n^{-1}\sum_{j=1}^n M_{j;1,1}\leq1\). Therefore, we conclude that
	$$\Bigg|\frac{(\Vert A_1\Vert^2)^{\circ}}{\sqrt{n}}-\frac{1}{\sqrt{n}}\sum_{j=1}^n M_{j;1,1}^{\circ}\Bigg|\overset{\mbP}{\longrightarrow}0.$$
	Next, by the Cauchy's inequality, we have
	$$\frac{\Vert B_1\Vert^2}{\sqrt{n}}\leq\sqrt{n}\epsilon_K\times\frac{1}{n}\sum_{j=1}^n\sum_{k=K+1}^{T-1}M_{j;k,k},$$
	where \(\epsilon_K=\sum_{k=K+1}^{T-1}\alpha_{1,k}^2\). Since \(n^{-1}\sum_{j=1}^n\sum_{k=K+1}^{T-1}M_{j;k,k}\leq1\) and \(\sqrt{n}\epsilon_K\leq\sqrt{n}(1-\alpha_{1,1}^2)\), we can show that \(n^{-1/2}\mathbb{E}[\Vert B_1\Vert^2]\to0\) and \(n^{-1/2}\Vert B_1\Vert^2\overset{\mbP}{\longrightarrow}0\) by Lemma \ref{Lem of at convergence}, i.e.
	$$\frac{|(\Vert B_1\Vert^2)^{\circ}|}{\sqrt{n}}\leq\frac{\Vert B_1\Vert^2}{\sqrt{n}}+\frac{\mathbb{E}[\Vert B_1\Vert^2]}{\sqrt{n}}\overset{\mbP}{\longrightarrow}0.$$
	Finally, since
	\begin{align}
		\frac{\langle A_1,B_1\rangle}{\sqrt{n}}&=\sum_{k=1}^K\alpha_{1,k}n^{-1/2}\sum_{t=K+1}^{T-1}\alpha_{1,t}\sum_{j=1}^n M_{j;k,t},\notag
	\end{align}
    for each \(k\in\{1,\cdots,K\}\), we have
    \begin{align*}
        &\mbE\Bigg[\Bigg|\sum_{t=K+1}^{T-1}\alpha_{1,t}\sum_{j=1}^n M_{j;k,t}\Bigg|\Bigg]\leq\mbE[1-\alpha_{1,1}^2]^{1/2}\mbE\Bigg[\sum_{t=K+1}^{T-1}\Bigg(\sum_{j=1}^n M_{j;k,t}\Bigg)^2\Bigg]^{1/2},
    \end{align*}
    then by (\ref{Eq of Lindeberg 2}), it yields that
    \begin{align*}
        &\mbE\Bigg[\sum_{t=K+1}^{T-1}\Bigg(\sum_{j=1}^n M_{j;k,t}\Bigg)^2\Bigg]\\
        &=\sum_{j=1}^n \sum_{t=K+1}^{T-1}\mbE[M_{j;k,t}^2]+\sum_{t=K+1}^{T-1}\sum_{j_1\neq j_2}^n\mbE[M_{j_1;k,t}]\mbE[M_{j_2;k,t}]\leq C_{B,M,\kappa_8,c}\log^6(T)T,
    \end{align*}
    where we use the fact $k\neq t$. Hence, by Lemma \ref{Lem of at convergence}, we can conclude that
    \begin{align*}
        &n^{-1/2}\mbE\Bigg[\Bigg|\sum_{t=K+1}^{T-1}\alpha_{1,t}\sum_{j=1}^n M_{j;k,t}\Bigg|\Bigg]\leq C_{B,M,\kappa_8,c}\log^3(T)\mbE[1-\alpha_{1,1}^2]^{1/2}\longrightarrow0.
    \end{align*}
    Since \(K\in\mbN^+\) is a fixed integer
	and \(|\alpha_{1,k}|\leq1\) for \(1\leq k\leq K\), so
    \begin{align*}
        &n^{-1/2}\mbE[|\langle A_1,B_1\rangle|]\leq\sum_{k=1}^Kn^{-1/2}\mbE\Bigg[\Bigg|\sum_{t=K+1}^{T-1}\alpha_{1,t}\sum_{j=1}^n M_{j;k,t}\Bigg|\Bigg]\longrightarrow0,
    \end{align*}
    which implies that \(n^{-1/2}\langle A_1,B_1\rangle^{\circ}\overset{\mbP}{\longrightarrow}0.\) In summary, we conclude that
	$$\Bigg|\frac{\hat{\lambda}_1^{\circ}}{\sqrt{n}}-\frac{1}{\sqrt{n}}\sum_{j=1}^n M_{j;1,1}^{\circ}\Bigg|\overset{\mbP}{\longrightarrow}0.$$
	For \(k\in\{2,\cdots,K\}\), by Lemma \ref{Lem of at convergence}, we can repeat the previous arguments to show that
	$$\Bigg|\frac{\hat{\lambda}_k^{\circ}}{\sqrt{n}}-\frac{1}{\sqrt{n}}\sum_{j=1}^n M_{j;k,k}^{\circ}\Bigg|\overset{\mbP}{\longrightarrow}0.$$
	We omit details here to save space. Now, we have proved that
	$$\Bigg(\frac{\hat{\lambda}_1^{\circ}}{\sqrt{n}},\cdots,\frac{\hat{\lambda}_K^{\circ}}{\sqrt{n}}\Bigg)'=\Bigg(\frac{1}{\sqrt{n}}\sum_{j=1}^nM_{j;1,1}^{\circ},\cdots,\frac{1}{\sqrt{n}}\sum_{j=1}^nM_{j;K,K}^{\circ}\Bigg)'+\mro_{\mbP}(\boldsymbol{1}_K),$$
	where \(\boldsymbol{1}_K\) is a \(K\times 1\) vector with all entries are \(1\). Therefore, by the Slutsky’s lemma, we conclude that 
	$$\Bigg(\frac{\hat{\lambda}_1^{\circ}}{\sqrt{n}},\cdots,\frac{\hat{\lambda}_K^{\circ}}{\sqrt{n}}\Bigg)'\ \ {\rm and\ \ }\Bigg(\frac{1}{\sqrt{n}}\sum_{j=1}^nM_{j;1,1}^{\circ},\cdots,\frac{1}{\sqrt{n}}\sum_{j=1}^nM_{j;K,K}^{\circ}\Bigg)'$$
	have the same asymptotic distribution, then we can complete our proof by Proposition \ref{Thm of pre joint CLT}.
\end{proof}
\subsubsection{Proof of Theorem \ref{Thm of CLT}}\label{ssec of mean corrected CLT}
Until now, we have shown that
\begin{align}
	\big|\mathbb{E}[M_{j;k,k}-\fM_{k,k}]\big|\leq C_{B,M,\kappa_8}T^{-1/2}\notag
\end{align}
in Lemma \ref{Thm of Lindeberg principle}. Next, we will show that
$$\sqrt{n}\Bigg|\frac{\mathbb{E}[\hat{\lambda}_k]}{n}-\mathbb{E}[\fM_{k,k}]\Bigg|\leq C_{B,M,\kappa_8,c},$$
so that we can conclude Theorem \ref{Thm of CLT}.
\begin{proof}[Proof of Theorem \ref{Thm of CLT}]
	Without loss of generality, we only give the detailed proof for \(k=1\). Let's first show that
	$$n^{-1/2}\Bigg|\mbE[\hla_1]-\sum_{j=1}^n\mbE[M_{j;1,1}]\Bigg|=\mro(1).$$
	Note that
	\begin{align*}
	    &n^{-1/2}\Bigg(\mathbb{E}[\hat{\lambda}_1]-\sum_{j=1}^n\mathbb{E}[M_{j;1,1}]\Bigg)=n^{-1/2}\Bigg(\sum_{j=1}^n\mathbb{E}[(\alpha_{1,1}^2-1)M_{j;1,1}]+\sum_{k\neq1\ {\rm or\ }l\neq1}^{T-1}\sum_{j=1}^n\mathbb{E}[\alpha_{1,k}\alpha_{1,l}M_{j;k,l}]\Bigg),
	\end{align*}
	then by Lemma \ref{Lem of at convergence}, we have
    \begin{align*}
        &n^{-1/2}\Bigg|\sum_{j=1}^n\mathbb{E}[(\alpha_{1,1}^2-1)M_{j;1,1}]\Bigg|\leq\sqrt{n}\mathbb{E}[1-\alpha_{1,1}^2]\longrightarrow0.
    \end{align*}
	Next, let's split 
	$$n^{-1/2}\sum_{k\neq1\ {\rm or\ }l\neq1}^{T-1}\sum_{j=1}^n\mathbb{E}[\alpha_{1,k}\alpha_{1,l}M_{j;k,l}]$$
	into the following two parts and claim that
	$$n^{-1/2}\sum_{l=2}^{T-1}\sum_{j=1}^n\mathbb{E}[\alpha_{1,1}\alpha_{1,l}M_{j;1,l}]=\mathrm{o}(1)\ \ {\rm and\ \ }n^{-1/2}\sum_{k,l=2}^{T-1}\sum_{j=1}^n\mathbb{E}[\alpha_{1,k}\alpha_{1,l}M_{j;k,l}]=\mathrm{o}(1).$$
	Since the proof arguments are totally the same as those in Proposition \ref{Thm of univariate CLT}, we only briefly repeat them. For the first term, by the Cauchy's inequality, we have
	$$\mathbb{E}\Bigg[\Bigg|\sum_{l=2}^{T-1}\sum_{j=1}^n\alpha_{1,1}\alpha_{1,l}M_{j;1,l}\Bigg|\Bigg]\leq\mathbb{E}[1-\alpha_{1,1}^2]^{1/2}\mathbb{E}\Bigg[\sum_{t=2}^{T-1}\Bigg(\sum_{j=1}^n M_{j;1,t}\Bigg)^2\Bigg]^{1/2},$$
	by \eqref{Eq of Lindeberg 2}, we have
	\begin{align}
		&\mathbb{E}\Bigg[\sum_{t=2}^{T-1}\Bigg(\sum_{j=1}^n M_{j;1,t}\Bigg)^2\Bigg]=\sum_{t=2}^{T-1}\sum_{j=1}^n\mbE[M_{j;1,t}^2]+\sum_{t=2}^{T-1}\sum_{j_1\neq j_2}^n\mbE[M_{j_1;1,t}]\mbE[M_{j_2;1,t}]\leq C_{B,M,\kappa_8,c}\log^6(T)T,\notag
	\end{align}
	then we obtain that
	$$n^{-1/2}\mathbb{E}\Bigg[\Bigg|\sum_{l=2}^{T-1}\sum_{j=1}^n\alpha_{1,1}\alpha_{1,l}M_{j;1,l}\Bigg|\Bigg]\leq C_{B,M,\kappa_8,c}\log^3(T)\mathbb{E}[1-\alpha_{1,1}^2]^{1/2}\longrightarrow0.$$
	Similarly, for the second part, we have
	\begin{align}
		&n^{-1/2}\mbE\Bigg[\Bigg|\sum_{k,t=2}^{T-1}\alpha_{1,k}\alpha_{1,l}\sum_{j=1}^nM_{j;k,l}\Bigg|\Bigg]\leq n^{-1/2}\mbE[1-\alpha_{1,1}^2]^{1/2}\mbE\Bigg[\sum_{k,t=2}^{T-1}\Bigg(\sum_{j=1}^n M_{j;k,t}\Bigg)^2\Bigg]^{1/2}\notag\\
		&\leq C_{B,M,\kappa_8,c}\log^3(T)\mathbb{E}[1-\alpha_{1,1}^2]^{1/2}\longrightarrow0,\notag
	\end{align}
	Therefore, we show that
	$$n^{-1/2}\Bigg|\mathbb{E}[\hat{\lambda}_1]-\sum_{j=1}^n\mathbb{E}[M_{j;1,1}]\Bigg|=\mathrm{o}(1),$$
	combining with the (\ref{Eq of Lindeberg 1}), it implies that
	\begin{align}
		&\sqrt{n}\Bigg|\frac{\mathbb{E}[\hat{\lambda}_1]}{n}-\mathbb{E}[\fM_{1,1}]\Bigg|\leq n^{-1/2}\Bigg|\mathbb{E}[\hat{\lambda}_1]-\sum_{j=1}^n\mathbb{E}[M_{j;1,1}]\Bigg|+n^{-1/2}\sum_{j=1}^n\big|\mathbb{E}[M_{j;1,1}-\fM_{1,1}]\big|\notag\\
		&\leq\mathrm{o}(1)+C_{B,M,\kappa_8}(n/T)^{1/2}.\notag
	\end{align}
    Now, based on above results and Proposition \ref{Thm of univariate CLT}, we obtain that
    \begin{align*}
        &\sqrt{n}\Bigg(\frac{\hat{\lambda}_1}{n}-\mathbb{E}[\fM_{1,1}],\cdots,\frac{\hat{\lambda}_K}{n}-\mathbb{E}[\fM_{K,K}]\Bigg)'=\Bigg(\frac{\hat{\lambda}_1^{\circ}}{\sqrt{n}},\cdots,\frac{\hat{\lambda}_K^{\circ}}{\sqrt{n}}\Bigg)'\\
        &+\sqrt{n}\Bigg(\frac{\mathbb{E}[\hat{\lambda}_1]}{n}-\mathbb{E}[\fM_{1,1}],\cdots,\frac{\mathbb{E}[\hat{\lambda}_K]}{n}-\mathbb{E}[\fM_{K,K}]\Bigg)\overset{d}{\longrightarrow}\mathcal{N}(\bbzeta,\mcS),
    \end{align*}
	which completes the proof of Theorem \ref{Thm of CLT}.
\end{proof}
\section{CLT for extreme eigenvalues of the sample correlation matrix of high-dimensional random walks with cross-sectional dependence}\label{Sec of correlation dependent}
\setcounter{equation}{0}
\def\theequation{\thesection.\arabic{equation}}
\setcounter{subsection}{0}
Currently, we have established the joint CLT of \((\hla_1,\cdots,\hla_K)'\) for cross-sectional independent \(X_t\) defined in (\ref{Eq of panel Xt}). In this section, we will establish the same CLT for cross-sectional dependent \(X_t\) defined as follows:
\begin{align}
    X_t=X_{t-1}+e_t,\quad e_t=\bbGa\sum_{k=0}^{\infty}\Psi_k\varepsilon_{t-k},\quad\varepsilon_t\overset{i.i.d.}{\sim}\mcN(\boldsymbol{0},\bbI_n),\label{Eq of nonpanel Xt}
\end{align}
where \(\{\Psi_k:k\in\mbN^+\}\) satisfies Assumption \ref{Ap of panel lag polynomial} and a additional assumption 
\begin{align}
    \sum_{k=0}^{\infty}(1+k)^2\Vert\Psi_k\Vert<B.\label{Eq of Ap of panel lag polynomial}
\end{align}
Moreover, \(\bbGa\in\mbR^{n\times n}\) is the cross-sectional matrix such that
\begin{Ap}\label{Ap of nonpanel}
	There exists two positive constants \(m_0,M_0\) such that
	$$\Vert\bbGa\bbGa'\Vert\leq M_0,\quad{\rm and}\quad n^{-1}\tr(\bbGa\bbGa')\geq m_0.$$
\end{Ap}
For simplicity, given the observations \(\bbX=[X_1,\cdots,X_T]\) generated by (\ref{Eq of nonpanel Xt}), we still denote \(\hR\) to be the sample correlation matrix of \(\bbX\), and \(\hla_1\geq\cdots\geq\hla_K\) are the first \(K\in\mbN^+\) largest eigenvalue of \(\hR\). Moreover, let \(\bbe_j\) be the \(j\)-th row of \(\bbe=[e_1,\cdots,e_T]\) defined in (\ref{Eq of nonpanel Xt}). Similar as (\ref{Eq of la1}), for any $K\in\mbN^+$ and $1\leq k\leq K$, we can still represent \(\hla_k\) as follows 
\begin{align}
    \hla_k=\sum_{s,t=1}^{T-1}\alpha_{k,s}\alpha_{k,t}\sum_{i=1}^n\frac{\beta_s\beta_t(\bbe_i\bbv_s)(\bbe_i\bbv_t)}{\sum_{l=1}^{T-1}\tsigma_l^2(\bbe_i\bbv_l)^2},\label{Eq of alpha cross sectional dependent}
\end{align}
where \(\beta_t\) is defined in (\ref{Eq of beta}) and we denote \(\hat{F}_1=\sum_{k=1}^{T-1}\alpha_{1,k}\bbw_k\) (\(\bbw_k\) are defined in (\ref{Eq of vk})) to be the normalized eigenvector of \(\hla_k\). Here, we {\bf abuse} the notation \(M_{j;k,l}\) in (\ref{Eq of fM}) as follows:
\begin{align}
    M_{i;k,l}:=\frac{\beta_k\beta_l(\bbe_i\bbv_k)(\bbe_i\bbv_l)}{\sum_{t=1}^{T-1}\tsigma_t^2(\bbe_i\bbv_t)^2}.\label{Eq of Mikl}
\end{align}
Similar as the proof of Theorem \ref{Thm of CLT}, to establish the CLT for $\hla_k$ for cross-sectional dependent $\bbX$ generated by \eqref{Eq of nonpanel Xt}, the most essential step is to show that
\begin{align*}
	\frac{\hla_k}{\sqrt{n}}\overset{\mbP}{\longrightarrow}\frac{1}{\sqrt{n}}\sum_{i=1}^nM_{i;k,k}=\frac{1}{\sqrt{n}}\sum_{i=1}^n\frac{\beta_k^2(\bbe_i\bbv_k)^2}{\sum_{t=1}^{T-1}\tsigma_t^2(\bbe_i\bbv_t)^2}.
\end{align*}
However, note that $\bbe_{i_1}\bbv_k$ and $\bbe_{i_2}\bbv_k$ are generally not independent for $i_1\neq i_2$ due to the existence of the cross-sectional matrix $\bbGa$ in \eqref{Eq of nonpanel Xt}. Thus, all $\{M_{i;k,k}:1\leq i\leq n\}$ are indeed correlated. Under this situation, proving
$$\Var\left(\frac{1}{\sqrt{n}}\sum_{i=1}^nM_{i;k,k}\right)=\mrO(1)$$
is not trivial as the cross-sectional independent case in Theorem \ref{Thm of CLT}. Precisely, since
$$\Var\left(\frac{1}{\sqrt{n}}\sum_{i=1}^nM_{i;1,1}\right)=\frac{1}{n}\sum_{i_1,i_2=1}^n\Cov(M_{i_1;k,l},M_{i_2;k,l}).$$
we need to find the upper bound of \(\big|\Cov(M_{i_1;k,l},M_{i_2;k,l})\big|\) for \(i_1\neq i_2\). However, there are two the main difficulties of calculating \(\Cov(M_{i_1;k,l},M_{i_2;k,l})\):
\begin{enumerate}
	\item the nonlinearity of \(M_{i;k,l}\);
	\item the correlations among all \(\{\bbe_i\bbv_k:1\leq i\leq n,1\leq k\leq T-1\}\).
\end{enumerate}
To overcome these technical difficulties caused by cross-sectional dependence, we need a stronger version of Assumptions \ref{Ap of panel lag polynomial} and \ref{Ap of finite integration}, that is, all \(\varepsilon_t\overset{i.i.d.}{\sim}\mcN(\boldsymbol{0},\bbI_n)\) and (\ref{Eq of Ap of panel lag polynomial}). Finally, let's briefly introduce the outline of \S\ref{Sec of correlation dependent}. In \S\ref{sec of finite variance}, for any $1\leq k,l\leq T-1$, we show that
$$\Var\left(\frac{1}{\sqrt{n}}\sum_{i=1}^nM_{i;k,l}\right)\leq \mrO((kl)^{-2}).$$
By the above conclusion, in \S\ref{sec of eigenvectors dependent}, we further deduce that $\alpha_{k,k}$ in \eqref{Eq of alpha cross sectional dependent} satisfies that for $1\leq k\leq K$
$$\lim_{n\to\infty}\sqrt{n}\mbE[1-\alpha_{k,k}^2]=0.$$
Thus, in \S\ref{sec of mbP dependent}, we conclude that
$$\frac{\hla_k}{n}\overset{\mbP}{\longrightarrow}\frac{1}{n}\sum_{i=1}^nM_{i;k,k}\overset{\mbP}{\longrightarrow}\mbE[\fM_{k,k}],$$
where $\fM_{k,k}$ is defined in \eqref{Eq of fM}. Finally, in \S\ref{sec of CLT correlation dependent}, we establish the joint CLT for $(\hla_1,\cdots,\hla_K)'$.
\subsection{Preliminary lemmas}\label{sec of finite variance}
In this section, we will show that
\begin{lem}\label{Lem of finite variance}
    Under Assumptions {\rm \ref{Ap of highdimensionality}, \ref{Ap of panel lag polynomial}, \ref{Ap of nonpanel}} and {\rm (\ref{Eq of Ap of panel lag polynomial})}, for any $1\leq k,l\leq T-1$, we have
    \begin{align}
	   \Var\left(\frac{1}{\sqrt{n}}\sum_{i=1}^nM_{i;k,l}\right)\leq C_{B,b,M_0,m_0}(kl)^{-2},\label{Eq of finite variance}
    \end{align}
    where \(M_{i;k,l}\) is defined in {\rm (\ref{Eq of Mikl})}.
\end{lem}
Here, let's first make some notational preliminaries, to distinguish with \(e_t\) in \eqref{Eq of nonpanel Xt}, let
\begin{align}
    \fe_t:=\sum_{k=0}^{\infty}\Psi_k\varepsilon_{t-k}\quad{\rm and}\quad\fe=[\fe_1,\cdots,\fe_T],\label{Eq of fe_j}
\end{align}
and further let \(\vec{\fe}_j\) be the \(j\)-th row of \(\fe=[\fe_1,\cdots,\fe_T]\) for \(1\leq j\leq n\). By (\ref{Eq of nonpanel Xt}) and (\ref{Eq of fe_j}), we have \(\bbe_j=\Gamma_j\fe\), where \(\Gamma_j\) is the \(j\)-th row of \(\bbGa\). Hence, the $M_{i;k,l}$ in \eqref{Eq of Mikl} can be rewritten as 
$$M_{i;k,l}=\frac{\beta_k(\Gamma_i\fe\bbv_k)\beta_l(\Gamma_i\fe\bbv_l)}{\sum_{t=1}^{T-1}\beta_t^2(\Gamma_i\fe\bbv_t)^2}.$$
Note that all $\varepsilon_t\overset{i.i.d.}{\sim}\mcN(\boldsymbol{0},\bbI_n)$ in \eqref{Eq of fe_j}, then \(\vec{\fe}_j\bbv_k\) are all normal. By \eqref{Eq of fe_j}, since $\{\Psi_k:k\in\mbN\}$ in Assumption \ref{Ap of panel lag polynomial} are diagonal, it implies that \(\vec{\fe}_{j_1}\bbv_{k_1},\vec{\fe}_{j_2}\bbv_{k_2}\) are independent for $j_1\neq j_2$, and Lemma \ref{Lem of covariance 1} implies that
$$\left|\Cov(\vec{\fe}_j\bbv_k,\vec{\fe}_j\bbv_l)-2\pi\delta_{k,l}f_j(\pi k/T)\right|\leq C_BT^{-1}.$$
Similar as what we have done in \S\ref{ssec of preliminary dependent}, for each $1\leq j\leq n$ we will replace all $\vec{\fe}_j\bbv_k$ by independent Gaussian variables, that is, remove the weak dependence among all $\{\vec{\fe}_j\bbv_1,\cdots,\vec{\fe}_j\bbv_{T-1}\}$. Importantly, for the cross-sectional dependent $X_t$ in \eqref{Eq of nonpanel Xt}, we need to show that the difference caused by removing these weak dependence has order of \(\mro(T^{-1/2})\). For this reason, we will first refine the upper bound for the weak correlations among all \(\vec{\fe}_j\bbv_k\) in Lemma \ref{Lem of covariance 1}, i.e.
$$\left|\Cov(\vec{\fe}_j\bbv_k,\vec{\fe}_j\bbv_l)\right|\leq C_BT^{-1},\quad k\neq l,$$
see \S\ref{ssec of spectral density approximation} for details. Precisely, we will prove Lemma \ref{Lem of finite variance} by the following three steps:
\begin{enumerate}
    \item First, we will derive a refined error bound of \(|\Cov(\vec{\fe}_j\bbv_k,\vec{\fe}_j\bbv_l)|\) for \(k\neq l\) in Lemma \ref{Lem of covariance adjust}.
    \item Next, let \(z_{i,t}\overset{i.i.d.}{\sim}\mcN(0,1)\) for \(1\leq i\leq n,1\leq t\leq T-1\) and \(\hat{x}_{i,t}=\sum_{j=1}^n\Gamma_{i,j}f_j(0)^{1/2}z_{j,t}\), then we will show that for \(1\leq k,l\leq T-1\)
    $$\mbE\left[\left|\frac{1}{\sqrt{n}}\sum_{i=1}^n\left(\frac{(\Gamma_i\fe\bbv_k)(\Gamma_i\fe\bbv_l)}{\sum_{t=1}^{T-1}\tsigma_t^2(\Gamma_i\fe\bbv_t)^2}-\frac{\hat{x}_{i,k}\hat{x}_{i,l}}{\sum_{t=1}^{T-1}\beta_t^2\hat{x}_{i,t}}\right)\right|^2\right]=\mro(n^{-1/2}).$$
    Precisely, the basic frameworks are similar as those in \S\ref{ssec of remove the dependence}, we first remove the weak dependence among all $\{\vec{\fe}_j\bbv_1,\cdots,\vec{\fe}_j\bbv_{T-1}\}$ for each $1\leq j\leq n$ in \S\ref{ssec of remove the dependence dependent}, then adjust the variance of all $\vec{\fe}_j\bbv_k$ in \S\ref{ssec of adjust the variance dependent}.
    \item Finally, in \S\ref{ssec of positive covariance}, we will show that for \(1\leq k,l\leq T-1\)
    $$\Var\left(\frac{1}{\sqrt{n}}\sum_{i=1}^n\frac{\hat{x}_{i,k}\hat{x}_{i,l}}{\sum_{t=1}^{T-1}\beta_t^2\hat{x}_{i,t}}\right)\leq C_{B,b,M_0,m_0}.$$
\end{enumerate}
Finally, let's make some notations here. Recall that all \(\varepsilon_t\overset{i.i.d.}{\sim}\mcN(\boldsymbol{0},\bbI_n)\), then we have \(\fe_t\sim\mcN(\boldsymbol{0},\sum_{k=0}^{\infty}\Psi_k^2)\) by (\ref{Eq of nonpanel Xt}) and 
\begin{align}
    \vec{\fe}_j=(\fe_{j,1},\cdots,\fe_{j,T})'\sim\mcN(\boldsymbol{0},\mcA^j),\label{Eq of vec fe j}
\end{align}
where \(\mcA^j=[\mcA_{s,t}^j]_{T\times T}\) is a \(T\times T\) Toeplitz matrix such that
\begin{align}
    \mcA_{s,t}^j=\mcA_{t,s}^j=\Cov(\fe_{j,t},\fe_{j,s})=\sum_{k=0}^{\infty}\varphi_{j,k}\varphi_{j,k+|t-s|}:=\phi_{|t-s|}^j.\label{Eq of mcA}
\end{align}
Further let \(\bbV=[\bbv_1,\cdots,\bbv_{T-1}]\) defined in (\ref{Eq of vk}), then
\begin{align}
    (\vec{\fe}_j\bbv_1,\cdots,\vec{\fe}_j\bbv_{T-1})\sim\mcN(0,\bbV'\mcA^j\bbV):=\mcN(0,\mcB^j).\label{Eq of mcB}
\end{align}
\subsubsection{Spectral density approximation}\label{ssec of spectral density approximation}
In this part, we will prove that
\begin{lem}\label{Lem of covariance adjust}
	Under Assumptions {\rm \ref{Ap of highdimensionality}, \ref{Ap of panel lag polynomial}} and {\rm (\ref{Eq of Ap of panel lag polynomial})}, for any sufficiently small \(\delta>0\), we have
	$$\left|\Cov(\vec{\fe}_j\bbv_k,\vec{\fe}_j\bbv_l)\right|\leq C_BT^{3\delta-2},\quad{\rm when}\ \left\{\begin{array}{cc}
		1\leq k\leq T^{\delta},&T^{1-\delta}<l<T-T^{1-\delta},\\
		1\leq l\leq T^{\delta},&T^{1-\delta}<k<T-T^{1-\delta},
	\end{array}\right.$$
    where $\vec{\fe}_j$ and $\bbv_k$ are defined in \eqref{Eq of vec fe j} and \eqref{Eq of vk}, respectively.
\end{lem}
\begin{proof}
    By $\bbv_k$ defined in \eqref{Eq of vk}, we have
    $$\vec{\fe}_j\bbv_k=\sqrt{\frac{2}{T}}\sum_{t=1}^Te_{j,t}\sin(\pi k(t-1)/T)=\sqrt{\frac{2}{T}}\sum_{t=1}^Te_{j,t}\Im(\exp({\rm i}\pi k(t-1)/T)),$$
    then \(\vec{\fe}_j\bbv_k=\sqrt{\pi}{\rm i}\big(e^{{\rm i}\theta_k/2}d_j(\theta_k/2)-e^{-{\rm i}\theta_k/2}d_j(-\theta_k/2)\big)\), where
    \begin{align}
	   d_j(\theta)=\frac{1}{\sqrt{2\pi T}}\sum_{t=1}^Te_{j,t}\exp(-{\rm i}t\theta)\quad{\rm and}\quad\theta_k:=2\pi k/T.\label{Eq of dj}
    \end{align}
    By Theorem 4.4.1 in \cite{brillinger2001time}, the spectral density of \(d_j(\theta)\) is
    \begin{align}
	   f_j(\theta)=\frac{1}{2\pi}\left|\sum_{t=0}^{\infty}\varphi_{j,t}e^{-{\rm i}t\theta}\right|^2.\notag
    \end{align}
	Since
	\begin{align}
		&\Cov(\vec{\fe}_j\bbv_k,\vec{\fe}_j\bbv_l)\label{Eq of cov ejvk 1}\\
		&=\pi\Cov\left(e^{{\rm i}\theta_k/2}d_j(\theta_k/2)-e^{-{\rm i}\theta_k/2}d_j(-\theta_k/2),e^{{\rm i}\theta_l/2}d_j(\theta_l/2)-e^{-{\rm i}\theta_l/2}d_j(-\theta_l/2)\right)\notag\\
		&=\pi e^{{\rm i}(\theta_k-\theta_l)/2}\Cov\left(d_j(\theta_k/2),d_j(\theta_l/2)\right)-\pi e^{{\rm i}(\theta_k+\theta_l)/2}\Cov\left(d_j(\theta_k/2),d_j(-\theta_l/2)\right)\notag\\
		&-\pi e^{-{\rm i}(\theta_k+\theta_l)/2}\Cov\left(d_j(-\theta_k/2),d_j(\theta_l/2)\right)+\pi e^{{\rm i}(\theta_l-\theta_k)/2}\Cov\left(d_j(-\theta_k/2),d_j(-\theta_l/2)\right).\notag
	\end{align}
	Let's first compute
	\begin{align*}
		&\Cov(d_j(\theta_k/2),d_j(\theta_l/2))=\mbE[d_j(\theta_k/2)\cdot d_j(-\theta_l/2)]\\
		&=\frac{1}{2\pi}\int_{-\pi}^{\pi}H_T\left(x-\frac{\pi k}{T}\right)\left(\sum_{t=0}^{\infty}\varphi_{j,t}e^{-{\rm i}xt}\right)\overline{H_T\left(x-\frac{\pi l}{T}\right)\left(\sum_{t=0}^{\infty}\varphi_{j,t}e^{-{\rm i}xt}\right)}dx,
	\end{align*}
	where
	$$H_T(x)=\frac{1}{\sqrt{2\pi T}}\sum_{s=1}^Te^{{\rm i}xs}.$$
	Hence,
	\begin{align*}
		H_T\left(x-\frac{\pi k}{T}\right)\left(\sum_{t=0}^{\infty}\varphi_t e^{-{\rm i}xt}\right)&=\frac{1}{\sqrt{2\pi T}}\sum_{s=1}^Te^{\frac{-\pi{\rm i}ks}{T}}\sum_{t=0}^{\infty}\varphi_{j,t}e^{{\rm i}x(s-t)},\\
		\overline{H_T\left(x-\frac{\pi l}{T}\right)\left(\sum_{t=0}^{\infty}\varphi_{j,t}e^{-{\rm i}xt}\right)}&=\frac{1}{\sqrt{2\pi T}}\sum_{s=1}^Te^{\frac{\pi{\rm i}ls}{T}}\sum_{t=0}^{\infty}\varphi_t e^{-{\rm i}x(s-t)},
	\end{align*}
	which implies that
	\begin{align*}
		&\int_{-\pi}^{\pi}H_T\left(x-\frac{\pi k}{T}\right)\left(\sum_{t=0}^{\infty}\varphi_{j,t}e^{-{\rm i}xt}\right)\overline{H_T\left(x-\frac{\pi l}{T}\right)\left(\sum_{t=0}^{\infty}\varphi_{j,t}e^{-{\rm i}xt}\right)}dx\\
		&=\frac{1}{2\pi T}\sum_{s_1,s_2=1}^Te^{\frac{\pi{\rm i}ls_2}{T}-\frac{\pi{\rm i}ks_1}{T}}\sum_{t_1,t_2=0}^{\infty}\varphi_{j,t_1}\varphi_{j,t_2}\int_{-\pi}^{\pi}e^{{\rm i}x[s_1-s_2-(t_1-t_2)]}dx\\
		&=\frac{1}{2\pi T}\sum_{s_1,s_2=1}^Te^{\frac{\pi{\rm i}ls_2}{T}-\frac{\pi{\rm i}ks_1}{T}}\sum_{\substack{t_1,t_2=0\\t_1-t_2=s_1-s_2}}^{\infty}\varphi_{j,t_1}\varphi_{j,t_2}=\frac{1}{2\pi T}\sum_{s_1,s_2=1}^Te^{\frac{\pi{\rm i}ls_2}{T}-\frac{\pi{\rm i}ks_1}{T}}\phi_{|s_1-s_2|}^j.
	\end{align*}
	First, when \(s_1=s_2\), then \(t_1=t_2\) and
	$$\sum_{s=1}^Te^{\frac{\pi{\rm i}(l-k)s}{T}}=\left\{\begin{array}{cc}
		0&k\equiv l\mod2\\
		\frac{-2}{1-e^{\frac{\pi{\rm i}(k-l)}{T}}}&k\not\equiv l\mod2
	\end{array}\right..$$
	Next, let \(r=|s_1-s_2|\), then we have \(s_1=s_2+r\) or \(s_2=s_1+r\), for the previous case, it gives that
	\begin{align*}
		&\phi_r^j\sum_{s=1}^{T-r}e^{\frac{\pi{\rm i}l(s+r)}{T}-\frac{\pi{\rm i}ks}{T}}=\phi_r^j e^{\frac{\pi{\rm i}lr}{T}}\sum_{s=1}^{T-r}e^{\frac{\pi{\rm i}(l-k)s}{T}}=\phi_r^j\frac{(-1)^{k-l}e^{\frac{\pi{\rm i}kr}{T}}-e^{\frac{\pi{\rm i}lr}{T}}}{1-e^{\frac{\pi{\rm i}(k-l)}{T}}}.
	\end{align*}
	Similarly, for the latter case, we have
	\begin{align*}
		&\phi_r^j\sum_{s=1}^{T-r}e^{\frac{\pi{\rm i}ls}{T}-\frac{\pi{\rm i}k(s+r)}{T}}=\phi_r^j e^{-\frac{\pi{\rm i}kr}{T}}\sum_{s=1}^{T-r}e^{\frac{\pi{\rm i}(l-k)s}{T}}=\phi_r^j\frac{(-1)^{k-l}e^{-\frac{\pi{\rm i}lr}{T}}-e^{-\frac{\pi{\rm i}kr}{T}}}{1-e^{\frac{\pi{\rm i}(k-l)}{T}}}.
	\end{align*}
	Now, let's consider the following two cases:
	\begin{itemize}
		\item \(k\equiv l\mod2\): Notice that \((-1)^{\pm k \pm l}=1\), we have that
		\begin{align*}
			&e^{{\rm i}(\theta_k-\theta_l)/2}\Cov(d_j(\theta_k/2),d_j(\theta_l/2))=\frac{e^{{\rm i}(\theta_k-\theta_l)/2}}{4\pi^2T}\sum_{s_1,s_2=1}^Te^{\frac{\pi{\rm i}ls_2}{T}-\frac{\pi{\rm i}ks_1}{T}}\phi_{|s_1-s_2|}\\
			&=\frac{e^{{\rm i}(\theta_k-\theta_l)/2}}{4\pi^2T\left[1-e^{\frac{\pi{\rm i}(k-l)}{T}}\right]}\sum_{r=1}^{T-1}\phi_r^j\left(e^{\frac{\pi{\rm i}kr}{T}}-e^{-\frac{\pi{\rm i}kr}{T}}-e^{\frac{\pi{\rm i}lr}{T}}+e^{-\frac{\pi{\rm i}lr}{T}}\right)\\
			&=\frac{{\rm i}e^{{\rm i}(\theta_k-\theta_l)/2}}{2\pi^2T\left[1-e^{{\rm i}(\theta_k-\theta_l)/2}\right]}\sum_{r=1}^{T-1}\phi_r^j\left(\sin\left(\frac{\pi kr}{T}\right)-\sin\left(\frac{\pi lr}{T}\right)\right)\\
			&=\left(\frac{1}{1-e^{{\rm i}(\theta_k-\theta_l)/2}}-1\right)\frac{{\rm i}}{2\pi^2T}\sum_{r=1}^{T-1}\phi_r^j\left(\sin\left(r\theta_k/2\right)-\sin\left(r\theta_l/2\right)\right).
		\end{align*}
		Similarly, we have that
		\begin{align*}
			&e^{{\rm i}(\theta_l-\theta_k)/2}\Cov(d_j(-\theta_k/2),d_j(-\theta_l/2))\\
			&=\left(\frac{1}{1-e^{{\rm i}(\theta_l-\theta_k)/2}}-1\right)\frac{{\rm i}}{2\pi^2T}\sum_{r=1}^{T-1}\phi_r^j\left(-\sin\left(r\theta_k/2\right)+\sin\left(r\theta_l/2\right)\right),
		\end{align*}
		and
		\begin{align*}
			&e^{{\rm i}(\theta_k+\theta_l)/2}\Cov(d_j(\theta_k/2),d_j(-\theta_l/2))\\
			&=\left(\frac{1}{1-e^{{\rm i}(\theta_k+\theta_l)/2}}-1\right)\frac{{\rm i}}{2\pi^2T}\sum_{r=1}^{T-1}\phi_r^j\left(\sin\left(r\theta_k/2\right)+\sin\left(r\theta_l/2\right)\right),
		\end{align*}
		and
		\begin{align*}
			&e^{-{\rm i}(\theta_k+\theta_l)/2}\Cov(d_j(-\theta_k/2),d_j(\theta_l/2))\\
			&=\left(\frac{1}{1-e^{-{\rm i}(\theta_k+\theta_l)/2}}-1\right)\frac{{\rm i}}{2\pi^2T}\sum_{r=1}^{T-1}\phi_r^j\left(-\sin\left(r\theta_k/2\right)-\sin\left(r\theta_l/2\right)\right).
		\end{align*}
		Hence, by (\ref{Eq of cov ejvk 1}), we know that
		\begin{align}
			&\Cov\left(\vec{\fe}_j\bbv_k,\vec{\fe}_j\bbv_l\right)=\frac{{\rm i}\pi}{2\pi^2T\left[1-e^{{\rm i}(\theta_k-\theta_l)/2}\right]}\sum_{r=1}^{T-1}\phi_r^j\left(\sin\left(r\theta_k/2\right)-\sin\left(r\theta_l/2\right)\right)\notag\\
			&+\frac{{\rm i}\pi}{2\pi^2T\left[1-e^{{\rm i}(\theta_l-\theta_k)/2}\right]}\sum_{r=1}^{T-1}\phi_r^j\left(\sin\left(r\theta_l/2\right)-\sin\left(r\theta_k/2\right)\right)\notag\\
			&-\frac{{\rm i}\pi}{2\pi^2T\left[1-e^{{\rm i}(\theta_k+\theta_l)/2}\right]}\sum_{r=1}^{T-1}\phi_r^j\left(\sin\left(r\theta_k/2\right)+\sin\left(r\theta_l/2\right)\right)\notag\\
			&-\frac{{\rm i}\pi}{2\pi^2T\left[1-e^{-{\rm i}(\theta_k+\theta_l)/2}\right]}\sum_{r=1}^{T-1}\phi_r^j\left(-\sin\left(r\theta_k/2\right)-\sin\left(r\theta_l/2\right)\right).\label{Eq of cov ejvk 2}
		\end{align}
		When \(k\leq T^{\delta}\), where \(\delta>0\) is a sufficiently small number, consider
		\begin{align*}
			&\sum_{r=1}^{T-1}\phi_r^j\sin\left(r\theta_k/2\right)=\sum_{r=1}^{T-1}\phi_r^j\sin\left(\frac{\pi kr}{T}\right)=\sum_{r=1}^{[T^{1-2\delta}]}+\sum_{r=[T^{1-2\delta}]+1}^{T-1}\phi_r^j\sin\left(\frac{\pi kr}{T}\right).
		\end{align*}
		According to Assumption \ref{Ap of panel lag polynomial}, we know that
		\begin{align*}
			&|\phi_{r}^j|\leq\left(\sum_{k=0}^{\infty}|\varphi_{j,k}|\right)\left(\sum_{k=r}^{\infty}|\varphi_{j,k}|\right)\leq B\sum_{k=r}^{\infty}|\varphi_{j,k}|,
		\end{align*}
		then
		\begin{align*}
			&\sum_{r=1}^{\infty}r|\phi_{r}^j|\leq B\sum_{r=1}^{\infty}r\sum_{k=r}^{\infty}|\varphi_{j,k}|=B\sum_{r=1}^{\infty}r(r+1)|\varphi_{j,r}|/2<B\sum_{r=1}^{\infty}r^2|\varphi_{j,r}|<B^2.
		\end{align*}
		For the previous part, since \(kr\leq T^{1-\delta}\), the for sufficiently large \(T\), we have
		\begin{align*}
			&\left|\sum_{r=1}^{[T^{1-2\delta}]}\phi_r^j\sin\left(\frac{\pi kr}{T}\right)\right|\leq\frac{\pi k}{T}\sum_{r=1}^{[T^{1-2\delta}]}r|\phi_r^j|\leq C_BT^{\delta-1}.
		\end{align*}
		For the later part, we have
		\begin{align*}
			&\left|\sum_{r=[T^{1-2\delta}]+1}^{T-1}\phi_r^j\sin\left(\frac{\pi kr}{T}\right)\right|\leq T^{2\delta-1}\sum_{r=[T^{1-2\delta}]+1}^{T-1}r|\phi_r^j|\leq C_BT^{2\delta-1}.
		\end{align*}
		As a result, we have
		$$\left|\sum_{r=1}^{T-1}\phi_r^j\sin\left(r\theta_k/2\right)\right|\leq C_BT^{2\delta-1}.$$
		Next, consider
		\begin{align*}
			\left|\frac{1}{1-e^{{\rm i}(\theta_k-\theta_l)/2}}-\frac{1}{1-e^{-{\rm i}\theta_l/2}}\right|=\frac{|e^{{\rm i}\theta_k/2}-1|}{|(1-e^{{\rm i}(\theta_k-\theta_l)/2})(1-e^{-{\rm i}\theta_l/2})|},
		\end{align*}
		where \(1\leq k\leq T^{\delta}\) and
		$$|e^{{\rm i}\theta_k/2}-1|\leq|\cos(\theta_k/2)-1|+|\sin(\theta_k/2)|\leq\mrO(T^{\delta-1}).$$
		On the other hand, when \(T^{1-\delta}<l<T-T^{1-\delta}\), we have
		\begin{align*}
			&|1-e^{-{\rm i}\theta_l/2}|\geq|\sin(\theta_l/2)|\geq\sin(\pi T^{-\delta})\geq2T^{-\delta}.
		\end{align*}
		Since \(k\leq T^{\delta}\ll T^{1-\delta}<l\), it implies that
		\begin{align*}
			&|1-e^{{\rm i}(\theta_k-\theta_l)/2}|\geq|\sin(\pi(l-k)/T)|\geq|\sin(\theta_l/4)|\geq T^{-\delta}.
		\end{align*}
		Hence, we can obtain that
		$$\left|\frac{1}{1-e^{{\rm i}(\theta_k-\theta_l)/2}}-\frac{1}{1-e^{-{\rm i}\theta_l/2}}\right|\leq\mrO(T^{3\delta-1}).$$
		In fact, by the same arguments, we can obtain that
		\begin{align}
			\left|\frac{1}{1-e^{{\rm i}(\pm\theta_k\pm\theta_l)/2}}-\frac{1}{1-e^{\pm{\rm i}\theta_l/2}}\right|\leq\mrO(T^{3\delta-1}).\label{Eq of SP density approx 1}
		\end{align}
		Combine with
		$$\sum_{r=1}^{T-1}|\phi_r^j|\left|\pm\sin\left(r\theta_k/2\right)\pm\sin\left(r\theta_l/2\right)\right|\leq2\sum_{r=1}^{T-1}|\phi_r^j|<C_B,$$
		it yields that
		\begin{align}
			&\left|\left(\frac{1}{1-e^{{\rm i}(\pm\theta_k\pm\theta_l)/2}}-\frac{1}{1-e^{\pm{\rm i}\theta_l/2}}\right)\left(\sum_{r=1}^{T-1}\phi_r^j\left(\pm\sin\left(r\theta_k/2\right)\pm\sin\left(r\theta_l/2\right)\right)\right)\right|\leq C_BT^{3\delta-1}.\notag
		\end{align}
		Consequently, it gives that
		\begin{align*}
			&\frac{{\rm i}}{2\pi^2T\left[1-e^{{\rm i}(\theta_k-\theta_l)/2}\right]}\sum_{r=1}^{T-1}\phi_r^j\left(\sin\left(r\theta_k/2\right)-\sin\left(r\theta_l/2\right)\right)\\
			&=\frac{{\rm i}}{2\pi^2T\left[1-e^{-{\rm i}\theta_l/2}\right]}\sum_{r=1}^{T-1}\phi_r^j\left(\sin\left(r\theta_k/2\right)-\sin\left(r\theta_l/2\right)\right)+C_BT^{3\delta-2}.
		\end{align*}
		Since we have shown that \(|1-e^{-{\rm i}\theta_l/2}|\geq|\sin(\theta_l/2)|\geq\sin(\pi T^{-\delta})\geq2T^{-\delta}\) and
		$$\left|\sum_{r=1}^{T-1}\phi_r^j\sin\left(r\theta_k/2\right)\right|\leq C_BT^{2\delta-1},$$
		then
		$$\frac{1}{|1-e^{-{\rm i}\theta_l/2}|}\left|\sum_{r=1}^{T-1}\phi_r^j\sin\left(r\theta_k/2\right)\right|\leq C_BT^{3\delta-1},$$
		and
		\begin{align*}
			&\frac{{\rm i}}{2\pi^2T\left[1-e^{{\rm i}(\theta_k-\theta_l)/2}\right]}\sum_{r=1}^{T-1}\phi_r^j\left(\sin\left(r\theta_k/2\right)-\sin\left(r\theta_l/2\right)\right)\\
			&=\frac{-{\rm i}}{2\pi^2T\left[1-e^{-{\rm i}\theta_l/2}\right]}\sum_{r=1}^{T-1}\phi_r^j\sin\left(r\theta_l/2\right)+C_BT^{3\delta-2}.
		\end{align*}
		Similarly, we can also show that
		\begin{align*}
			&\frac{{\rm i}}{2\pi^2T\left[1-e^{-{\rm i}(\theta_k+\theta_l)/2}\right]}\sum_{r=1}^{T-1}\phi_r^j\left(-\sin\left(r\theta_k/2\right)-\sin\left(r\theta_l/2\right)\right)\\
			&=\frac{-{\rm i}}{2\pi^2T\left[1-e^{-{\rm i}\theta_l/2}\right]}\sum_{r=1}^{T-1}\phi_r^j\sin\left(r\theta_l/2\right)+C_BT^{3\delta-2},
		\end{align*}
		then by (\ref{Eq of cov ejvk 2}), combining the first and last terms, we have
		\begin{align*}
			&\Bigg|\frac{{\rm i}}{2\pi^2T\left[1-e^{{\rm i}(\theta_k-\theta_l)/2}\right]}\sum_{r=1}^{T-1}\phi_r^j\left(\sin\left(r\theta_k/2\right)-\sin\left(r\theta_l/2\right)\right)\\
			&-\frac{{\rm i}}{2\pi^2T\left[1-e^{-{\rm i}(\theta_k+\theta_l)/2}\right]}\sum_{r=1}^{T-1}\phi_r^j\left(-\sin\left(r\theta_k/2\right)-\sin\left(r\theta_l/2\right)\right)\Bigg|\leq C_BT^{3\delta-2}.
		\end{align*}
		Similarly, for the the other two terms, we have the same results, so we omit the details here.
		\item \(k\not\equiv l\mod2\): Notice that \((-1)^{\pm k \pm l}=-1\), then we will obtain that
		\begin{align*}
			&e^{{\rm i}(\theta_k-\theta_l)/2}\Cov(d_j(\theta_k/2),d_j(\theta_l/2))=\frac{1}{4\pi^2T}\sum_{s_1,s_2=1}^Te^{\frac{\pi{\rm i}ls_2}{T}-\frac{\pi{\rm i}ks_1}{T}}\phi_{|s_1-s_2|}\\
			&=\frac{-e^{{\rm i}(\theta_k-\theta_l)/2}}{4\pi^2T\left[1-e^{\frac{\pi{\rm i}(k-l)}{T}}\right]}\left(2\phi_0^j+\sum_{r=1}^{T-1}\phi_r^j\left(e^{\frac{\pi{\rm i}kr}{T}}+e^{-\frac{\pi{\rm i}kr}{T}}+e^{\frac{\pi{\rm i}lr}{T}}+e^{-\frac{\pi{\rm i}lr}{T}}\right)\right)\\
			&=\frac{-e^{{\rm i}(\theta_k-\theta_l)/2}}{2\pi^2T\left[1-e^{{\rm i}(\theta_k-\theta_l)/2}\right]}\left(\phi_0^j+\sum_{r=1}^{T-1}\phi_r^j\left(\cos\left(\frac{\pi kr}{T}\right)+\cos\left(\frac{\pi lr}{T}\right)\right)\right)\\
			&=\left(1-\frac{1}{1-e^{{\rm i}(\theta_k-\theta_l)/2}}\right)\frac{1}{2\pi^2T}\left(\phi_0^j+\sum_{r=1}^{T-1}\phi_r^j\left(\cos\left(r\theta_k/2\right)+\cos\left(r\theta_l/2\right)\right)\right).
		\end{align*}
		Similarly, we have that
		\begin{align*}
			&e^{{\rm i}(\theta_l-\theta_k)/2}\Cov(d_j(-\theta_k/2),d_j(-\theta_l/2))\\
			&=\left(1-\frac{1}{1-e^{{\rm i}(\theta_l-\theta_k)/2}}\right)\frac{1}{2\pi^2T}\left(\phi_0^j+\sum_{r=1}^{T-1}\phi_r^j\left(\cos\left(r\theta_k/2\right)+\cos\left(r\theta_l/2\right)\right)\right),
		\end{align*}
		and
		\begin{align*}
			&e^{{\rm i}(\theta_k+\theta_l)/2}\Cov(d_j(\theta_k/2),d_j(-\theta_l/2))\\
			&=\left(1-\frac{1}{1-e^{{\rm i}(\theta_k+\theta_l)/2}}\right)\frac{1}{2\pi^2T}\left(\phi_0^j+\sum_{r=1}^{T-1}\phi_r^j\left(\cos\left(r\theta_k/2\right)+\cos\left(r\theta_l/2\right)\right)\right),
		\end{align*}
		and
		\begin{align*}
			&e^{-{\rm i}(\theta_k+\theta_l)/2}\Cov(d_j(-\theta_k/2),d_j(\theta_l/2))\\
			&=\left(1-\frac{1}{1-e^{-{\rm i}(\theta_k+\theta_l)/2}}\right)\frac{1}{2\pi^2T}\left(\phi_0^j+\sum_{r=1}^{T-1}\phi_r^j\left(\cos\left(r\theta_k/2\right)+\cos\left(r\theta_l/2\right)\right)\right).
		\end{align*}
		According to (\ref{Eq of SP density approx 1}) and the fact that
		$$\left|\phi_0^j+\sum_{r=1}^{T-1}\phi_r^j\cos\left(r\theta_k/2\right)+\cos\left(r\theta_l/2\right)\right|\leq2\sum_{r=0}^{T-1}|\phi_r^j|<C_B,$$
		we can derive that
		\begin{align}
			&\left|\left(\frac{1}{1-e^{{\rm i}(\pm\theta_k\pm\theta_l)/2}}-\frac{1}{1-e^{\pm{\rm i}\theta_l/2}}\right)\left(\phi_0^j+\sum_{r=1}^{T-1}\phi_r^j\left(\cos\left(r\theta_k/2\right)+\cos\left(r\theta_l/2\right)\right)\right)\right|\leq C_BT^{3\delta-1},\notag
		\end{align}
		i.e.
		\begin{align*}
			&e^{{\rm i}(\pm\theta_k+\theta_l)/2}\Cov(d_j(\pm\theta_k/2),d_j(-\theta_l/2))\\
			&=\left(1-\frac{1}{1-e^{{\rm i}\theta_l/2}}\right)\frac{1}{2\pi^2T}\left(\phi_0^j+\sum_{r=1}^{T-1}\phi_r^j\left(\cos\left(r\theta_k/2\right)+\cos\left(r\theta_l/2\right)\right)\right)+C_BT^{3\delta-2},
		\end{align*}
		which implies that \(\big|\Cov(d_j(\theta_k/2)-d_j(-\theta_k/2),d_j(-\theta_l/2))\big|\leq C_BT^{3\delta-2}\). Similarly, we also have \(\big|\Cov(d_j(\theta_k/2)-d_j(-\theta_k/2),d_j(\theta_l/2))\big|\leq C_BT^{3\delta-2}\).
	\end{itemize}
	which completes our proof.
\end{proof}
Moreover, we give the following results for \(\Var(\vec{\fe}_j\bbv_k)\):
\begin{lem}\label{Lem of variance adjust}
	Under Assumptions {\rm \ref{Ap of panel lag polynomial}} and {\rm \ref{Ap of panel lag polynomial}}, for \(1\leq t\leq T-1\), we have
	$$\left|\Var(\vec{\fe}_j\bbv_t)-2\pi f_j(0)\right|\leq C_Bt/T.$$
\end{lem}
\begin{proof}
	According to Lemma 10 in \cite{onatski2021spurious}, we know that
	$$\left|\Var(\vec{\fe}_j\bbv_t)-2\pi f_j(\pi t/T)\right|\leq C_BT^{-1}.$$
	Notice that
	\begin{align*}
		&2\pi\left|f_j(\pi t/T)-f_j(0)\right|\leq\left|\sum_{k=0}^{\infty}\varphi_{j,k}\left(e^{-{\rm i}\pi kt/T}-1\right)\right|\cdot\left(\left|\sum_{k=0}^{\infty}\varphi_{j,k}e^{-{\rm i}\pi kt/T}\right|+\left|\sum_{k=0}^{\infty}\varphi_{j,k}\right|\right),
	\end{align*}
	by Assumption \ref{Ap of panel lag polynomial} and (\ref{Eq of spectral density}), we know that \(|f_j(\theta)|\leq C_B\) for all \(\theta\in[0,2\pi]\). On the other hand, since
	$$\left|e^{-{\rm i}\pi kt/T}-1\right|\leq\mrO(kt/T),$$
	then
	\begin{align*}
		&\left|\sum_{k=0}^{\infty}\varphi_{j,k}\left(e^{-{\rm i}\pi kt/T}-1\right)\right|\leq\sum_{k=0}^{\infty}|\varphi_{j,k}|\cdot\left|e^{-{\rm i}\pi kt/T}-1\right|\leq\mrO\left(\frac{t}{T}\sum_{k=1}^{\infty}k|\varphi_{j,k}|\right)\leq C_Bt/T,
	\end{align*}
	which completes our proof.
\end{proof}
\subsubsection{Remove the dependence}\label{ssec of remove the dependence dependent}
In this part, we will remove the weak dependence among all \(\vec{\fe}_j\bbv_k\) for $k=1,\cdots,T-1$, where $\fe$ is defined in \eqref{Eq of fe_j}, $\vec{\fe}_j$ is the $j$-th {\bf row} of $\fe$. Recall that $(\vec{\fe}_j\bbv_1,\cdots,\vec{\fe}_j\bbv_{T-1})=\vec{\fe}_j\bbV\sim\mcN(\boldsymbol{0},\mcB^j)$, where \(\mcB^j\) is defined in (\ref{Eq of mcB}), then denote
$$\mcB^j=\diag(\mcB^j)+\Delta^j:=\mcD^j+\Delta^j,$$
and Lemma \ref{Lem of covariance adjust} implies that
$$\Vert\Delta_{k\cdot}^j\Vert_2^2\leq C_B(2T^{1-\delta}\cdot T^{-2}+(T-2T^{1-\delta})T^{6\delta-4})\leq C_BT^{-1-\delta}$$
for \(1\leq k\leq T^{\delta}\). Here, we claim that
\begin{lem}\label{Lem of positive definite Toeplitz}
	Under Assumption {\rm \ref{Ap of panel lag polynomial}}, \(\mcA^j=[\phi_{|s-t|}^j]_{s\times t}\) defined in {\rm (\ref{Eq of mcA})} is a positive definite symmetric Toeplitz matrix.
\end{lem}
\begin{proof}
	According to Lemma 4.1 in \cite{gray2006toeplitz}, the smallest eigenvalue of \(\mcA^j\) is no less than 
	$${\rm ess}\inf_{x\in[0,2\pi]}\left|\sum_{k=0}^{T-1}\phi_k^je^{{\rm i}kx}\right|={\rm ess}\inf_{x\in[0,2\pi]}\left|\sum_{k=0}^{T-1}\varphi_{j,k}e^{{\rm i}kx}\right|^2>b^2$$
	so we can conclude this lemma by Assumption \ref{Ap of panel lag polynomial} when \(T\) is sufficiently large.
\end{proof}
Since \(\mcA^j\) is positive definite, so does \(\mcB^j=\bbV'\mcA^j\bbV\), where \(\bbV=[\bbv_1,\cdots,\bbv_{T-1}]\), then \((\mcB^j)^{1/2}\) exists and we further claim that
\begin{align}
	\tilde{\Delta}^j:=(\mcB^j)^{1/2}-(\mcD^j)^{1/2},\quad\Vert\tilde{\Delta}_{k\cdot}^j\Vert_2^2\leq C_BT^{-1-\delta},\label{Eq of square root error}
\end{align}
for \(1\leq k\leq T^{\delta}\). In fact, let's first find the square root of \((\mcD^j)^{-1/2}\mcB^j(\mcD^j)^{-1/2}=\bbI_{T-1}+(\mcD^j)^{-1/2}\Delta^j(\mcD^j)^{-1/2}:=\bbI_{T-1}+\hat{\Delta}^j\). It is easy to see that \(\Vert\hat{\Delta}_{k\cdot}^j\Vert_2^2\leq C_BT^{-1-\delta}\) for \(1\leq k\leq T^{\delta}\), and since \(\mcD^j\) is invertible according to Assumption \ref{Ap of panel lag polynomial}, the \(\bbI_{T-1}+\hat{\Delta}^j\) is also invertible, which deduces that \(\Vert\hat{\Delta}^j\Vert<1\). Moreover, since
\begin{align*}
	&\left(\bbI_{T-1}+\hat{\Delta}^j\right)^{1/2}=\bbI_{T-1}+\sum_{r=1}^{\infty}\frac{(-1)^r(2r!)}{4^r(1-2r)(r!)^2}(\hat{\Delta}^j)^r,
\end{align*}
we can derive that
\begin{align}
    (\mcB^j)^{1/2}=(\mcD^j)^{1/2}+\sum_{r=1}^{\infty}\frac{(-1)^r(2r!)}{4^r(1-2r)(r!)^2}(\mcD^j)^{1/2}(\hat{\Delta}^j)^r=(\mcD^j)^{1/2}+\tilde{\Delta}^j.\label{Eq of mcB decomposition}
\end{align}
Hence, it implies that
\begin{align*}
	&\tilde{\Delta}_{k\cdot}^j=(\mcD_{k,k}^j)^{1/2}\hat{\Delta}_{k\cdot}^j\sum_{r=1}^{\infty}\frac{(-1)^r(2r!)}{4^r(1-2r)(r!)^2}(\hat{\Delta}^j)^{r-1},
\end{align*}
then
\begin{align*}
	&\Vert\tilde{\Delta}_{k\cdot}^j\Vert_2\leq(\mcD_{k,k}^j)^{1/2}\Vert\hat{\Delta}_{k\cdot}^j\Vert_2\sum_{r=1}^{\infty}\frac{1}{2^r(2r-1)}\Vert\hat{\Delta}^j\Vert^{r-1}\leq\mrO(\Vert\hat{\Delta}_{k\cdot}^j\Vert_2),
\end{align*}
which concludes our claim. 

Now, given \(\vec{z}_j=(z_{j,1},\cdots,z_{j,T-1})'\overset{i.i.d.}{\sim}\mcN(\boldsymbol{0},\bbI_{T-1})\) for \(j=1,\cdots,n\), it gives that $(\vec{\fe}_j\bbV)'$ and $(\mcB^j)^{1/2}\vec{z}_j$ have the same distribution. For simplicity, by \eqref{Eq of mcB decomposition}, we define
$$(\mcB^j)^{1/2}\vec{z}_j=(\mcD^j)^{1/2}\vec{z}_j+\tilde{\Delta}^j\vec{z}_j:=\vec{y}_j+\tilde{\Delta}^j\vec{z}_j,$$
and for $1\leq k\leq T-1$
\begin{align}
    x_{i,k}:=\sum_{j=1}^n\Gamma_{i,j}y_{j,k}+\tilde{\Delta}_{k\cdot}^j\vec{z}_j,\quad\tilde{x}_{i,k}:=\sum_{j=1}^n\Gamma_{i,j}y_{j,k},\label{Eq of X ik cross sectional dependent}
\end{align}
where $\vec{y}_j=(y_{j,1},\cdots,y_{j,T-1})'$, $\tilde{\Delta}_{k\cdot}^j$ is the $k$-th row of $\tilde{\Delta}^j$ and $\bbGa$ is the cross-sectional matrix defined in \eqref{Eq of nonpanel Xt}. Consequently, by \eqref{Eq of Mikl}, we know that $M_{i;k,l}$ and
$$\frac{\beta_k\beta_lx_{i,k}x_{i,l}}{\sum_{t=1}^{T-1}\tsigma_t^2x_{1,t}^2}$$
have the same distribution. Therefore, we will {\bf abuse} the notation $M_{i;k,l}$ as follows:
\begin{align}
    M_{i;k,l}=\frac{\beta_k\beta_lx_{i,k}x_{i,l}}{\sum_{t=1}^{T-1}\tsigma_t^2x_{1,t}^2},\label{Eq of Mikl abuse}
\end{align}
without special clarification, we always assume that $M_{i;k,l}$ is defined by \eqref{Eq of Mikl abuse} instead of \eqref{Eq of Mikl} in the following context. Finally, let
\begin{align}
	\mu_{i,t}:=\sum_{j=1}^n\Gamma_{i,j}\tilde{\Delta}_{t\cdot}^j\vec{z}_j,\label{Eq of mu 1t}
\end{align}
and we will show that 
\begin{lem}\label{Lem of remove dependence}
	Under Assumptions {\rm \ref{Ap of highdimensionality}} and {\rm \ref{Ap of panel lag polynomial}}, for any \(k,l\in\{1,\cdots,T-1\}\), we have
	$$\frac{1}{\sqrt{n}}\sum_{i=1}^n\frac{x_{i,k}x_{i,l}}{\sum_{t=1}^{T-1}\tsigma_t^2x_{i,t}^2}\overset{L^2}{\longrightarrow}\frac{1}{\sqrt{n}}\sum_{i=1}^n\frac{\tilde{x}_{i,k}\tilde{x}_{i,l}}{\sum_{t=1}^{T-1}\tsigma_t^2\tilde{x}_{i,t}^2},$$
    where $x_{i,k}$ and \(\tilde{x}_{i,k}\) are defined in {\rm (\ref{Eq of X ik cross sectional dependent})}, $\beta_t$ is defined in \eqref{Eq of beta}.
\end{lem}
\begin{proof}
	For convenience, we only give the detailed proofs for \(k=l=1\), since the arguments for others are totally the same. First, notice that 
	\begin{align*}
		&\mbE\left[\left(\frac{1}{\sqrt{n}}\sum_{i=1}^n\frac{x_{i,1}^2}{\sum_{t=1}^{T-1}\tsigma_t^2x_{i,t}^2}-\frac{\tilde{x}_{i,1}^2}{\sum_{t=1}^{T-1}\tsigma_t^2\tilde{x}_{i,t}^2}\right)^2\right]\leq\sum_{i=1}^n\mbE\left[\left(\frac{x_{i,1}^2}{\sum_{t=1}^{T-1}\tsigma_t^2x_{i,t}^2}-\frac{\tilde{x}_{i,1}^2}{\sum_{t=1}^{T-1}\tsigma_t^2\tilde{x}_{i,t}^2}\right)^2\right],
	\end{align*}
	and we will show that
	\begin{align}
		\sup_i T\mbE\left[\left(\frac{x_{i,1}^2}{\sum_{t=1}^{T-1}\tsigma_t^2x_{i,t}^2}-\frac{\tilde{x}_{i,1}^2}{\sum_{t=1}^{T-1}\tsigma_t^2\tilde{x}_{i,t}^2}\right)^2\right]\leq C_{B,M_0}T^{-\delta^2},\label{Eq of remove dependence}
	\end{align}
	where \(\delta>0\) is a fixed sufficiently small constant. Without loss of generality, we only present the proofs for \(i=1\) since the constant \(C_{B,M_0}\) in (\ref{Eq of remove dependence}) is independent of \(i\). By (\ref{Eq of mu 1t}), we have
	$$\mu_{1,t}\sim\mcN\left(0,\sum_{j=1}^n\Gamma_{1,j}^2\Vert\tilde{\Delta}_{t\cdot}^j\Vert_2^2\right).$$
	According to (\ref{Eq of square root error}), we have \(\Var(\mu_{1,t})\leq C_{B,M_0}T^{-1-\delta}\) for \(1\leq t<T^{\delta}\), then
	$$\mbP(|\mu_{1,t}|>T^{-1/2-\delta/4})\leq\mrO(\exp(-C_{B,M_0}T^{\delta^2})).$$
	When \(T^{\delta}\leq t<T\), we have \(\Var(\mu_{1,t})\leq C_{B,M_0}T^{-1}\), then
	$$\mbP(|\mu_{1,t}|>t^{\delta}T^{-1/2-\delta^2})\leq\mrO(\exp(-t^{2\delta}T^{-\delta^2}))\leq\mrO(\exp(-C_{B,M_0}T^{\delta^2})).$$
	Now, define
	$$\mcE_1:=\{|\mu_{1,t}|\leq T^{-1/2-\delta/4}:1\leq t<T^{\delta}\}\cup\{|\mu_{1,t}|\leq t^{\delta}T^{-1/2-\delta^2}:T^{\delta}\leq t<T\},$$
	then we can obtain that \(\mbP(\mcE_1)\geq1-\mrO(T\exp(-C_{B,M_0}T^{\delta^2}))\). Hence, conditional on \(\mcE_1\), consider
	\begin{align*}
		&\sum_{t=2}^{T-1}\tsigma_t^2(x_{1,t}^2-\tilde{x}_{1,t}^2)=\sum_{t=2}^{T-1}\tsigma_t^2(2\tilde{x}_{1,t}\mu_{1,t}+\mu_{1,t}^2),
	\end{align*}
	where
	$$\tsigma_t=\frac{\sin(\theta_1)}{\sin(\theta_t)}=\frac{\sin(\pi/2T)}{\sin(\pi t/2T)}\asymp\mrO(t^{-1}),$$
	(``$\asymp$'' is defined in \eqref{Eq of asymp mbP}) and
	\begin{align*}
		&\sum_{t=2}^{T-1}\tsigma_t^2\mu_{1,t}^2\leq\sum_{t=2}^{[T^{\delta}]}\tsigma_t^2\mu_{1,t}^2+\sum_{t=[T^{\delta}]+1}^{T-1}\tsigma_t^{2(1-\delta)}(\tsigma_t^{\delta}\mu_{1,t})^2\\
		&\leq T^{-1-\delta/2}\sum_{t=2}^{[T^{\delta}]}t^{-2}+T^{-1-2\delta^2}\sum_{t=[T^{\delta}]+1}^{T-1}t^{-2(1-\delta)}\leq\mrO(T^{-1-2\delta^2}).
	\end{align*}
	Next, it is easy to see that all \(\tilde{x}_{1,t}\) are independent with each other, and for \(1\leq t<T^{\delta}\), we have
	\begin{align*}
		&\mbP\left(|\tilde{x}_{1,t}\mu_{1,t}|>T^{-1/2-\delta/8}|\mcE_1\right)\leq\mbP\left(|\tilde{x}_{1,t}|>T^{\delta/8}\right)\leq\mrO(\exp(-C_{B,M_0}T^{\delta/8})).
	\end{align*}
	For \(T^{\delta}\leq t<T\), we have
	\begin{align*}
		&\mbP\left(t^{-\delta}|\tilde{x}_{1,t}\mu_{1,t}|>T^{-1/2-\delta^2/2}|\mcE_1\right)\leq\mbP\left(|\tilde{x}_{1,t}|>T^{\delta^2/2}\right)\leq\mrO(\exp(-C_{B,M_0}T^{\delta^2/2})).
	\end{align*}
	Hence, we can derive that
	\begin{align*}
		&\left|\sum_{t=2}^{T-1}\tsigma_t^2\tilde{x}_{1,t}\mu_{1,t}\right|\leq\sum_{t=1}^{[T^{\delta}]}t^{-2}|\tilde{x}_{1,t}\mu_{1,t}|+\sum_{t=[T^{\delta}]+1}^{T-1}t^{-2+\delta}|\tilde{x}_{1,t}t^{-\delta}\mu_{1,t}|\leq\mrO(T^{-1/2-\delta^2/2})
	\end{align*}
	with probability at least of \(1-\mrO(T\exp(-C_{B,M_0}T^{\delta^2/2}))\), i.e.
	$$\mbP\left(\left|\sum_{t=2}^{T-1}\tsigma_t^2(x_{1,t}^2-\tilde{x}_{1,t}^2)\right|>T^{-1/2-\delta^2/2}\Bigg|\mcE_1\right)\leq\mrO(T\exp(-C_{B,M_0}T^{\delta^2/2})).$$
	Consequently, let 
	\begin{align*}
		&\tilde{\mcE}_{1,2}:=\mcE_1\cap\left\{\left|\sum_{t=2}^{T-1}\tsigma_t^2(x_{1,t}^2-\tilde{x}_{1,t}^2)\right|\leq T^{-1/2-\delta^2/2}\right\},
	\end{align*}
	it yields that \(\mbP(\tilde{\mcE}_{1,2})\geq1-\mrO(T\exp(-C_{B,M_0}T^{\delta^2/2}))\). Now, consider
	\begin{align*}
		&T\mbE\left[\left|\frac{x_{1,1}^2}{\sum_{t=1}^{T-1}\tsigma_t^2x_{1,t}^2}-\frac{x_{1,1}^2}{x_{1,1}^2+\sum_{t=2}^{T-1}\tsigma_t^2\tilde{x}_{1,t}^2}\right|^2\right]=T\mbE\left[\left|\frac{x_{1,1}^2\sum_{t=2}^{T-1}\tsigma_t^2(x_{1,t}^2-\tilde{x}_{1,t}^2)}{(\sum_{t=1}^{T-1}\tsigma_t^2x_{1,t}^2)(x_{1,1}^2+\sum_{t=2}^{T-1}\tsigma_t^2\tilde{x}_{1,t}^2)}\right|^2\right]\\
		&=T\mbE\left[\left|\frac{x_{1,1}^2\sum_{t=2}^{T-1}\tsigma_t^2(x_{1,t}^2-\tilde{x}_{1,t}^2)}{(\sum_{t=1}^{T-1}\tsigma_t^2x_{1,t}^2)(x_{1,1}^2+\sum_{t=2}^{T-1}\tsigma_t^2\tilde{x}_{1,t}^2)}\right|^2\Bigg|\tilde{\mcE}_{1,2}\right]\mbP(\tilde{\mcE}_{1,2})\\
		&+T\mbE\left[\left|\frac{x_{1,1}^2\sum_{t=2}^{T-1}\tsigma_t^2(x_{1,t}^2-\tilde{x}_{1,t}^2)}{(\sum_{t=1}^{T-1}\tsigma_t^2x_{1,t}^2)(x_{1,1}^2+\sum_{t=2}^{T-1}\tsigma_t^2\tilde{x}_{1,t}^2)}\right|^2\Bigg|\tilde{\mcE}_{1,2}^c\right]\mbP(\tilde{\mcE}_{1,2}^c)\\
		&\leq T^{-\delta^2}\mbE\left[\frac{x_{1,1}^4}{\big(\sum_{t=1}^{T-1}\tsigma_t^2x_{1,t}^2\big)^2\big(x_{1,1}^2+\sum_{t=2}^{T-1}\tsigma_t^2\tilde{x}_{1,t}^2\big)^2}\right]+4T^3\mbP(\tilde{\mcE}_{1,2}^c)\leq C_{B,M_0}T^{-\delta^2},
	\end{align*}
	where we claim that
	$$\mbE\left[\frac{x_{1,1}^4}{\big(\sum_{t=1}^{T-1}\tsigma_t^2x_{1,t}^2\big)^2\big(x_{1,1}^2+\sum_{t=2}^{T-1}\tsigma_t^2\tilde{x}_{1,t}^2\big)^2}\right]<C_{B,M_0}.$$
	In fact, by the Holder's inequality, we have that
	\begin{align*}
		&\mbE\left[\frac{x_{1,1}^4}{\big(\sum_{t=1}^{T-1}\tsigma_t^2x_{1,t}^2\big)^2\big(x_{1,1}^2+\sum_{t=2}^{T-1}\tsigma_t^2\tilde{x}_{1,t}^2\big)^2}\right]\\
		&\leq\mbE\left[x_{1,1}^{12}\right]^{1/3}\mbE\left[\left(\sum_{t=1}^{T-1}\tsigma_t^2x_{1,t}^2\right)^{-6}\right]^{1/3}\mbE\left[\left(x_{1,1}^2+\sum_{t=2}^{T-1}\tsigma_t^2\tilde{x}_{1,t}^2\right)^{-6}\right]^{1/3}.
	\end{align*}
	Since \(x_{1,t}\sim\mcN\big(0,2\pi\sum_{j=1}^n\Gamma_{1,j}^2f_j(\theta_t/2)\big)\), then \(\mbE\left[x_{1,1}^{12}\right]<C_{B,M_0}\) by Assumptions \ref{Ap of panel lag polynomial} and \ref{Ap of nonpanel}. Moreover, according to Lemma \ref{Thm of cdf order}, we know that \(\mbP\left(\sum_{t=1}^{14}x_{1,t}^2\leq x\right)\leq C_{B,M_0}x^7\) for \(x\in[0,1]\), then
	\begin{align*}
		&\mbE\left[\left(\sum_{t=1}^{T-1}\tsigma_t^2x_{1,t}^2\right)^{-6}\right]\leq\tsigma_{14}^{-12}\mbE\left[\left(\sum_{t=1}^{14}x_{1,t}^2\right)^{-6}\right]\leq\tsigma_{14}^{-12}\int_1^{\infty}r^5\mbP\left(\sum_{t=1}^{14}x_{1,t}^2\leq r^{-1}\right)dr+\tsigma_{14}^{-12}\leq C_{B,M_0},
	\end{align*}
	so does \(\mbE\big[\big(x_{1,1}^2+\sum_{t=2}^{T-1}\tsigma_t^2\tilde{x}_{1,t}^2\big)^{-6}\big]\).
	Finally, notice that
	\begin{align*}
		&\left|\frac{\tilde{x}_{1,1}^2}{\sum_{t=1}^{T-1}\tsigma_t^2\tilde{x}_{1,t}^2}-\frac{x_{1,1}^2}{x_{1,1}^2+\sum_{t=2}^{T-1}\tsigma_t^2\tilde{x}_{1,t}^2}\right|\leq\frac{|\tilde{x}_{1,1}^2-x_{1,1}^2|}{\sum_{t=2}^{T-1}\tsigma_t^2\tilde{x}_{1,t}^2},
	\end{align*}
	by the previous argument, we can also derive that
	$$\mbP\left(\left|x_{1,1}^2-\tilde{x}_{1,1}^2\right|>T^{-1/2-\delta^2/2}|\mcE_1\right)\leq\mrO(T\exp(-C_{B,M_0}T^{\delta^2/2})),$$
	then define
	$$\tilde{\mcE}_{1,1}:=\mcE_1\cap\left\{\left|x_{1,1}^2-\tilde{x}_{1,1}^2\right|\leq T^{-1/2-\delta^2/2}\right\},$$
	we can conclude that \(\mbP(\tilde{\mcE}_{1,1})\geq1-\mrO(T\exp(-C_{B,M_0}T^{\delta^2/2}))\) and
	\begin{align*}
		&T\mbE\left[\left|\frac{\tilde{x}_{1,1}^2}{\sum_{t=1}^{T-1}\tsigma_t^2\tilde{x}_{1,t}^2}-\frac{x_{1,1}^2}{x_{1,1}^2+\sum_{t=2}^{T-1}\tsigma_t^2\tilde{x}_{1,t}^2}\right|^2\right]\leq T\mbE\left[\frac{|\tilde{x}_{1,1}^2-x_{1,1}^2|^2}{\big|\sum_{t=2}^{T-1}\tsigma_t^2\tilde{x}_{1,t}^2\big|^2}\Bigg|\tilde{\mcE}_{1,1}\right]\mbP(\tilde{\mcE}_{1,1})\\
		&+T\mbE\left[\frac{|\tilde{x}_{1,1}^2-x_{1,1}^2|^2}{\big|\sum_{t=2}^{T-1}\tsigma_t^2\tilde{x}_{1,t}^2\big|^2}\Bigg|\tilde{\mcE}_{1,1}^c\right]\mbP(\tilde{\mcE}_{1,1}^c)\leq C_{B,M_0}T^{-\delta^2}.
	\end{align*}
	Now, combine the previous two results, we prove that
	$$T\mbE\left[\left|\frac{\tilde{x}_{1,1}^2}{\sum_{t=1}^{T-1}\tsigma_t^2\tilde{x}_{1,t}^2}-\frac{x_{1,1}^2}{\sum_{t=1}^{T-1}\tsigma_t^2x_{1,t}^2}\right|^2\right]\leq C_{B,M_0}T^{-\delta^2},$$
	which completes our proof.
\end{proof}
\subsubsection{Adjust the variance}\label{ssec of adjust the variance dependent}
After removing the weak dependence among \(x_{i,t}\) by Lemma \ref{Lem of remove dependence}, all \(\tilde{x}_{i,t}\) defined in \eqref{Eq of X ik cross sectional dependent} are indeed independent. However, since
$$\tilde{x}_{i,t}=\sum_{j=1}^n\Gamma_{i,j}(\mcD_{t,t}^j)^{1/2}z_{j,t},$$
it is easy to see
$$\Var(\tilde{x}_{i,t})=\sum_{j=1}^n\Gamma_{i,j}^2\mcD_{t,t}^j,$$
i.e. the variance of all \(\tilde{x}_{i,1},\cdots,\tilde{x}_{i,T-1}\) could be different. Next, we will unify the variance of \(\tilde{x}_{i,t}\) as follows:
\begin{lem}\label{Lem of unify variance}
	Under Assumptions {\rm \ref{Ap of highdimensionality}, \ref{Ap of panel lag polynomial}} and {\rm \ref{Ap of nonpanel}}, let 
	$$\hat{x}_{i,t}=\sum_{j=1}^n\Gamma_{i,j}(2\pi f_j(0))^{1/2}z_{j,t},$$
    where \(z_{j,t}\overset{i.i.d.}{\sim}\mcN(0,1)\), then for \(1\leq k,l\leq T-1\), we have
	$$\frac{1}{\sqrt{n}}\sum_{i=1}^n\frac{\tilde{x}_{i,k}\tilde{x}_{i,l}}{\sum_{t=1}^{T-1}\tsigma_t^2\tilde{x}_{i,t}^2}\overset{L^2}{\longrightarrow}\frac{1}{\sqrt{n}}\sum_{i=1}^n\frac{\hat{x}_{i,k}\hat{x}_{i,l}}{\sum_{t=1}^{T-1}\tsigma_t^2\hat{x}_{i,t}^2},$$
    where $\tilde{x}_{i,k}$ is defined in \eqref{Eq of X ik cross sectional dependent}.
\end{lem}
\begin{proof}
	For convenience, we only give the detailed proofs for \(k=l=1\), since the arguments for others are totally the same. Similar as what we have done in Lemma \ref{Lem of remove dependence}, we will show that 
	\begin{align}
		\sup_i T\mbE\left[\left|\frac{\tilde{x}_{i,1}^2}{\sum_{t=1}^{T-1}\tsigma_t^2\tilde{x}_{i,1}^2}-\frac{\hat{x}_{i,1}^2}{\sum_{t=1}^{T-1}\tsigma_t^2\hat{x}_{i,t}^2}\right|^2\right]\leq C_{B,M_0}\log(T)T^{-1/3}.\label{Eq of unify variance}
	\end{align}
	Notice that
	$$\tilde{x}_{i,t}-\hat{x}_{i,t}=\sum_{j=1}^n\Gamma_{i,j}\left((\mcD_{t,t}^j)^{1/2}-(2\pi f_j(0))^{1/2}\right)z_{j,t}=\sum_{j=1}^n\Gamma_{i,j}\frac{\mcD_{t,t}^j-2\pi f_j(0)}{(\mcD_{t,t}^j)^{1/2}+(2\pi f_j(0))^{1/2}}z_{j,t},$$
	according to Assumption \ref{Ap of panel lag polynomial}, we know that \((\mcD_{t,t}^j)^{1/2}+(2\pi f_j(0))^{1/2}>C_b\), then by Lemma \ref{Lem of variance adjust}, we have
	\begin{align*}
		&\Var(\tilde{x}_{i,t}-\hat{x}_{i,t})\leq C_b\sum_{j=1}^n\Gamma_{i,j}^2\left|\mcD_{t,t}^j-2\pi f_j(0)\right|^2\leq C_{B,M_0}t^2T^{-2},
	\end{align*}
	for \(t=1,\cdots,T-1\), i.e.
	\begin{align*}
		&\Var(t^{-1}(\tilde{x}_{i,t}-\hat{x}_{i,t}))\leq C_{B,M_0}T^{-2},
	\end{align*}
	so we obtain that
	$$\mbP\left(t^{-1}\left|\tilde{x}_{i,t}-\hat{x}_{i,t}\right|>T^{-5/6}\right)\leq\mrO\left(\exp(-C_{B,M_0}T^{1/3})\right).$$
	Here, we define an event
	$$\fE_i:=\left\{t^{-1}\left|\tilde{x}_{i,t}-\hat{x}_{i,t}\right|\leq T^{-5/6}:t=1,\cdots,T-1\right\},$$
	then \(\mbP(\fE_i)\geq1-\mrO\left(T\exp(-C_{B,M_0}T^{1/3})\right)\). Next, conditional on \(\fE_i\), since
	$$t^{-2}\left(\tilde{x}_{i,t}^2-\hat{x}_{i,t}^2\right)=2t^{-2}\hat{x}_{i,t}\left(\tilde{x}_{i,t}-\hat{x}_{i,t}\right)+\left(\tilde{x}_{i,t}-\hat{x}_{i,t}\right)^2,$$
	where \(t^{-2}\left(\tilde{x}_{i,t}-\hat{x}_{i,t}\right)^2\leq T^{-5/3}\) and 
	\begin{align*}
		&\mbP\left(t^{-1}\left|\tilde{x}_{i,t}-\hat{x}_{i,t}\right|\cdot|\hat{x}_{i,t}|>T^{-2/3}|\fE_i\right)\leq\mrO(\exp(-C_{B,M_0}T^{1/3})),
	\end{align*}
	so we can deduce that
	$$t^{-2}\left|\tilde{x}_{i,t}^2-\hat{x}_{i,t}^2\right|\leq2(T^{-5/3}+t^{-1}T^{-2/3})$$
	with probability at least of \(1-\mrO(\exp(-C_{B,M_0}T^{1/3}))\), which further implies that
	$$\sum_{t=2}^{T-1}\tsigma_t^2\left|\tilde{x}_{i,t}^2-\hat{x}_{i,t}^2\right|\Big|\fE_i\leq2(T^{-2/3}+\log(T)T^{-2/3})\leq4\log(T)T^{-2/3}.$$
	Now, define
	$$\hat{\fE}_i:=\fE_i\cap\left\{\sum_{t=2}^{T-1}\tsigma_t^2\left|\tilde{x}_{i,t}^2-\hat{x}_{i,t}^2\right|\leq4\log(T)T^{-2/3}\right\},$$
	then \(\mbP(\hat{\fE}_i)\geq1-\mrO(T\exp(-C_{B,M_0}T^{1/3}))\). Hence,
	\begin{align*}
		&T\mbE\left[\left|\frac{\tilde{x}_{i,1}^2}{\sum_{t=1}^{T-1}\tsigma_t^2\tilde{x}_{i,t}^2}-\frac{\tilde{x}_{i,1}^2}{\tilde{x}_{i,1}^2+\sum_{t=2}^{T-1}\tsigma_t^2\hat{x}_{i,t}^2}\right|^2\right]\leq T\mbE\left[\frac{\big(\sum_{t=2}^{T-1}\tsigma_t^2\left|\tilde{x}_{i,t}^2-\hat{x}_{i,t}^2\right|\big)^2}{\big(\sum_{t=1}^{T-1}\tsigma_t^2\tilde{x}_{i,t}^2\big)^2}\Bigg|\hat{\fE}_i\right]\mbP(\hat{\fE}_i)\\
		&+T\mbE\left[\frac{\big(\sum_{t=2}^{T-1}\tsigma_t^2\left|\tilde{x}_{i,t}^2-\hat{x}_{i,t}^2\right|\big)^2}{\big(\sum_{t=1}^{T-1}\tsigma_t^2\tilde{x}_{i,t}^2\big)^2}\Bigg|\hat{\fE}_i^c\right]\mbP(\hat{\fE}_i^c)\leq C_{B,M_0}\log(T)T^{-1/3}.
	\end{align*}
	Similarly, we can also show that
	$$T\mbE\left[\left|\frac{\hat{x}_{i,1}^2}{\sum_{t=1}^{T-1}\tsigma_t^2\hat{x}_{i,t}^2}-\frac{\tilde{x}_{i,1}^2}{\tilde{x}_{i,1}^2+\sum_{t=2}^{T-1}\tsigma_t^2\hat{x}_{i,t}^2}\right|^2\right]\leq C_{B,M_0}\log(T)T^{-1/3},$$
	which completes our proof.
\end{proof}
\subsubsection{Covariance estimation}\label{ssec of positive covariance}
Now, by Lemmas \ref{Lem of remove dependence} and \ref{Lem of unify variance}, we have shown that 
$$\frac{1}{\sqrt{n}}\sum_{i=1}^n\frac{x_{i,k}x_{i,l}}{\sum_{t=1}^{T-1}\tsigma_t^2x_{i,1}^2}\overset{L^2}{\longrightarrow}\frac{1}{\sqrt{n}}\sum_{i=1}^n\frac{\hat{x}_{i,k}\hat{x}_{i,l}}{\sum_{t=1}^{T-1}\tsigma_t^2\hat{x}_{i,t}^2},$$
where $x_{i,k}$ and $\hat{x}_{i,k}$ are defined in \eqref{Eq of X ik cross sectional dependent} and Lemma \ref{Lem of unify variance}, then it implies that
$$\frac{1}{n}\Var\left(\sum_{i=1}^n\frac{x_{i,k}x_{i,l}}{\sum_{t=1}^{T-1}\tsigma_t^2x_{i,1}^2}\right)\to\frac{1}{n}\Var\left(\sum_{i=1}^n\frac{\hat{x}_{i,k}\hat{x}_{i,l}}{\sum_{t=1}^{T-1}\tsigma_t^2\hat{x}_{i,t}^2}\right).$$
Note that
$$\frac{1}{n}\Var\left(\sum_{i=1}^n\frac{\hat{x}_{i,k}\hat{x}_{i,l}}{\sum_{t=1}^{T-1}\tsigma_t^2\hat{x}_{i,t}^2}\right)=\frac{1}{n}\sum_{i_1,i_2=1}^n\Cov\left(\frac{\hat{x}_{i_1,k}\hat{x}_{i_1,l}}{\sum_{t=1}^{T-1}\tsigma_t^2\hat{x}_{i_1,t}},\frac{\hat{x}_{i_2,k}\hat{x}_{i_2,l}}{\sum_{t=1}^{T-1}\tsigma_t^2\hat{x}_{i_2,t}^2}\right),$$
we will show that
\begin{lem}\label{Lem of positive covariance}
	For any \(i_1,i_2\in\{1,\cdots,n\}\), we have
	\begin{align}
		\Cov\left(\frac{\hat{x}_{i_1,k}^2}{\sum_{t=1}^{T-1}\tsigma_t^2\hat{x}_{i_1,1}^2},\frac{\hat{x}_{i_2,k}^2}{\sum_{t=1}^{T-1}\tsigma_t^2\hat{x}_{i_2,1}^2}\right)\geq0,\label{Eq of positive covariance 1}
	\end{align}
	and
	\begin{align}
		\left|\Cov\left(\frac{\hat{x}_{i_1,k}\hat{x}_{i_1,l}}{\sum_{t=1}^{T-1}\tsigma_t^2\hat{x}_{i_1,1}^2},\frac{\hat{x}_{i_2,k}\hat{x}_{i_2,l}}{\sum_{t=1}^{T-1}\tsigma_t^2\hat{x}_{i_2,1}^2}\right)\right|\leq\mrO\left(\rho_{i_1,i_2}^2\right),\label{Eq of positive covariance 2}
	\end{align}
	where \(k,l\in\{1,\cdots,T-1\}\), $\hat{x}_{i,k}$ is defined in Lemma \ref{Lem of unify variance} and
    $$\rho_{i_1,i_2}:=\frac{\sum_{j=1}^n\Gamma_{i_1,j}\Gamma_{i_2,j}f_j(0)}{(\sum_{j=1}^n\Gamma_{i_1,j}^2f_j(0))^{1/2}(\sum_{j=1}^n\Gamma_{i_2,j}^2f_j(0))^{1/2}}.$$
\end{lem}
\begin{proof}
	By the definition of $\hat{x}_{i,k}$ in Lemma \ref{Lem of unify variance}, we have
	$$\hat{x}_i=(\hat{x}_{i,1},\cdots,\hat{x}_{i,T-1})\sim\mcN\left(\boldsymbol{0},2\pi\sum_{j=1}^n\Gamma_{i,j}^2f_j(0)\bbI_{T-1}\right),$$
	since \(M_{i_1;1,1}\) is a ratio of quadratic forms, then we assume \((\hat{x}_{i,1},\cdots,\hat{x}_{i,T-1})\sim\mcN(\boldsymbol{0},\bbI_{T-1})\) without loss of generality, then
	$$\Cov(\hat{x}_{i_1,t_1},\hat{x}_{i_2,t_2})=\frac{\delta_{t_1,t_2}\sum_{j=1}^n\Gamma_{i_1,j}\Gamma_{i_2,j}f_j(0)}{(\sum_{j=1}^n\Gamma_{i_1,j}^2f_j(0))^{1/2}(\sum_{j=1}^n\Gamma_{i_2,j}^2f_j(0))^{1/2}}.$$
	For simplicity, denote
	$$\rho_{i_1,i_2}:=\frac{\sum_{j=1}^n\Gamma_{i_1,j}\Gamma_{i_2,j}f_j(0)}{(\sum_{j=1}^n\Gamma_{i_1,j}^2f_j(0))^{1/2}(\sum_{j=1}^n\Gamma_{i_2,j}^2f_j(0))^{1/2}}:=\arcsin\tau_{i_1,i_2},$$
	where \(\tau_{i_1,i_2}\in[-\pi/2,\pi/2]\), then conditional \(\hat{x}_{i_2,t}\), we have
	$$\hat{x}_{i_1,t}|\hat{x}_{i_2,t}\sim\mcN\left(\hat{x}_{i_2,t}\sin\tau_{i_1,i_2},\cos^2\tau_{i_1,i_2}\right),$$
	which further implies
	\begin{align*}
		&\mbE\left[\frac{\hat{x}_{i_1,k}\hat{x}_{i_1,l}}{\sum_{t=1}^{T-1}\tsigma_t^2\hat{x}_{i_1,1}^2}\frac{\hat{x}_{i_2,k}\hat{x}_{i_2,l}}{\sum_{t=1}^{T-1}\tsigma_t^2\hat{x}_{i_2,1}^2}\right]=\mbE\left[\frac{\hat{x}_{i_2,k}\hat{x}_{i_2,l}}{\sum_{t=1}^{T-1}\tsigma_t^2\hat{x}_{i_2,1}^2}\mbE\left[\frac{\hat{x}_{i_1,k}\hat{x}_{i_1,l}}{\sum_{t=1}^{T-1}\tsigma_t^2\hat{x}_{i_1,1}^2}\Bigg|\hat{x}_{i_2}\right]\right]\\
		&=\mbE\left[\frac{(\hat{x}_{i_2,k}\sin\tau_{i_1,i_2}+w_k\cos\tau_{i_1,i_2})(\hat{x}_{i_2,l}\sin\tau_{i_1,i_2}+w_l\cos\tau_{i_1,i_2})}{\sum_{t=1}^{T-1}\tsigma_t^2(\hat{x}_{i_2,t}\sin\tau_{i_1,i_2}+w_t\cos\tau_{i_1,i_2})^2}\frac{\hat{x}_{i_2,1}^2}{\sum_{t=1}^{T-1}\tsigma_t^2\hat{x}_{i_2,1}^2}\right],
	\end{align*}
	where \(w=(w_1,\cdots,w_{T-1})\sim\mcN(\boldsymbol{0},\bbI_{T-1})\) is independent with \(\hat{x}_{i_2}\). For convenience, we simplify the above equation by the following forms:
	$$H_{k,l}^{i_1,i_2}(\tau):=\mbE\left[\frac{(z_k\sin\tau+w_k\cos\tau)(z_l\sin\tau+w_l\cos\tau)}{\sum_{t=1}^{T-1}\tsigma_t^2(z_t\sin\tau+w_t\cos\tau)^2}\frac{z_kz_l}{\sum_{t=1}^{T-1}\tsigma_t^2z_t^2}\right],$$
	where \(z=(z_1,\cdots,z_{T-1})\sim\mcN(\boldsymbol{0},\bbI_{T-1})\) is independent with \(w\) and 
	$$G_{k,l}^{i_1,i_2}(\tau):=\mbE\left[\frac{(z_k\sin\tau+w_k\cos\tau)(z_l\sin\tau+w_l\cos\tau)}{\sum_{t=1}^{T-1}\tsigma_t^2(z_t\sin\tau+w_t\cos\tau)^2}\right].$$
	Hence,
	\begin{align*}
		&\mcC_{k,l}^{i_1,i_2}:=\Cov\left(\frac{\hat{x}_{i_1,k}\hat{x}_{i_1,l}}{\sum_{t=1}^{T-1}\tsigma_t^2\hat{x}_{i_1,1}^2},\frac{\hat{x}_{i_2,k}\hat{x}_{i_2,l}}{\sum_{t=1}^{T-1}\tsigma_t^2\hat{x}_{i_2,1}^2}\right)=H_{k,l}^{i_1,i_2}(\tau)-G_{k,l}^{i_1,i_2}(\tau)\mbE\left[\frac{\hat{x}_{i_2,k}\hat{x}_{i_2,l}}{\sum_{t=1}^{T-1}\tsigma_t^2\hat{x}_{i_2,1}^2}\right]
	\end{align*}
	is a smooth function of \(\tau\). In fact, since all \(h_t\overset{i.i.d.}{\sim}\mcN(0,1)\), it implies that \(\frac{d}{d\tau}G_{i_1,i_2}^{k,l}(\tau)=0\). Hence, it yields that \(\frac{d}{d\tau}\mcC_{k,l}^{i_1,i_2}(\tau)=\frac{d}{d\tau}H_{k,l}^{i_1,i_2}(\tau)\).
	\begin{itemize}
		\item Proof of (\ref{Eq of positive covariance 1}): Without loss of generality, we assume \(k=1\) and abbreviate \(H_{i_1,i_2}^{1,1}(\tau)\) by \(H(\tau)\). Define \(h_t:=z_t\sin\tau+w_t\cos\tau\), then
		\begin{align*}
			\frac{d}{d\tau}H(\tau)&=2\mbE\left[\frac{z_1^2(z_1\sin\tau+w_1\cos\tau)(z_1\cos\tau-w_1\sin\tau)\sum_{t=2}^{T-1}\tsigma_t^2h_t^2}{\big(\sum_{t=1}^{T-1}\tsigma_t^2z_t^2\big)\big(\sum_{t=1}^{T-1}\tsigma_t^2h_t^2\big)^2}\right]\\
			&+2\sum_{s=2}^{T-1}\tsigma_s^2\mbE\left[\frac{z_1^2h_1^2(z_s\sin\tau+w_s\cos\tau)(-z_s\cos\tau+w_s\sin\tau)}{\big(\sum_{t=1}^{T-1}\tsigma_t^2h_t^2\big)^2\big(\sum_{t=1}^{T-1}\tsigma_t^2z_t^2\big)}\right].
		\end{align*}
		Next, let \(z_s:=r_s\cos\beta_s\) and \(w_s:=r_s\sin\beta_s\), where \(\beta_s\in[0,2\pi]\), then we have
		\begin{align*}
			&\mbE\left[\frac{z_1^2h_1^2(z_s\sin\tau+w_s\cos\tau)(-z_s\cos\tau+w_s\sin\tau)}{\big(\sum_{t=1}^{T-1}\tsigma_t^2h_t^2\big)^2\big(\sum_{t=1}^{T-1}\tsigma_t^2z_t^2\big)}\right]\\
			&=\mbE\left[\frac{-z_1^2h_1^2r_s^2\sin(2(\tau+\beta_s))}{\big(\tsigma_s^2r_s^2\cos^2\beta_s+\sum_{t=1,t\neq s}^{T-1}\tsigma_t^2z_t\big)\big(\tsigma_s^2r_s^2\sin^2(\tau+\beta_s)+\sum_{t=1,t\neq s}^{\infty}\tsigma_t^2h_t^2\big)^2}\right]\\
			&=\mbE\left[\frac{-z_1^2h_1^2r_s^2\sin(2(\gamma_s))}{\big(\tsigma_s^2r_s^2\cos^2(\gamma_s-\tau)+\sum_{t=1,t\neq s}^{T-1}\tsigma_t^2z_t\big)\big(\tsigma_s^2r_s^2\sin^2\gamma_s+\sum_{t=1,t\neq s}^{\infty}\tsigma_t^2h_t^2\big)^2}\right],
		\end{align*}
		where \(\gamma_s:=\beta_s+\tau\). It is easy to see that
		$$g_s(\gamma_s):=\frac{-z_1^2h_1^2r_s^2\sin(2\gamma_s)}{\big(\tsigma_s^2r_s^2\sin^2\gamma_s+\sum_{t=1,t\neq s}^{\infty}\tsigma_t^2h_t^2\big)^2}$$
		is periodic function of \(\gamma_s\) with the period of \(\pi\), and \(g_s(\pi/2+\delta)=-g_s(\pi/2-\delta)\geq0\) for \(\delta\in[0,\pi/2]\). On the other hand, since \(\delta,\tau\in[0,\pi/2]\), we have
		\begin{align*}
			&\cos^2(\pi/2-\delta-\tau)=\sin^2(\delta+\tau)\geq\sin^2(\tau-\delta)=\cos^2(\pi/2+\delta-\tau),
		\end{align*}
		so
		\begin{align*}
			0&\leq\frac{-g_s(\pi/2-\delta)}{\tsigma_s^2r_s^2\cos^2(\pi/2-\delta-\tau)+\sum_{t=1,t\neq s}^{T-1}\tsigma_t^2z_t^2}\leq\frac{g_s(\pi/2+\delta)}{\tsigma_s^2r_s^2\cos^2(\pi/2+\delta-\tau)+\sum_{t=1,t\neq s}^{T-1}\tsigma_t^2z_t^2},
		\end{align*}
		and
		$$\mbE\left[\frac{z_1^2h_1^2(z_s\sin\tau+w_s\cos\tau)(-z_s\cos\tau+w_s\sin\tau)}{\big(\sum_{t=1}^{T-1}\tsigma_t^2h_t^2\big)^2\big(\sum_{t=1}^{T-1}\tsigma_t^2z_t^2\big)}\right]\geq0,\quad{\rm for\ }\tau\in[0,\pi/2].$$
		Similarly, when \(\tau\in[-\pi/2,0]\) and \(\delta\in[0,\pi/2]\), we have
		\begin{align*}
			&\cos^2(\pi/2-\delta-\tau)=\sin^2(\delta+\tau)\leq\sin^2(\delta-\tau)=\cos^2(\pi/2+\delta-\tau),
		\end{align*}
		i.e.
		\begin{align*}
			\frac{-g_s(\pi/2-\delta)}{\tsigma_s^2r_s^2\cos^2(\pi/2-\delta-\tau)+\sum_{t=1,t\neq s}^{T-1}\tsigma_t^2z_t^2}\geq\frac{g_s(\pi/2+\delta)}{\tsigma_s^2r_s^2\cos^2(\pi/2+\delta-\tau)+\sum_{t=1,t\neq s}^{T-1}\tsigma_t^2z_t^2}\geq0,
		\end{align*}
		and
		$$\mbE\left[\frac{z_1^2h_1^2(z_s\sin\tau+w_s\cos\tau)(-z_s\cos\tau+w_s\sin\tau)}{\big(\sum_{t=1}^{T-1}\tsigma_t^2h_t^2\big)^2\big(\sum_{t=1}^{T-1}\tsigma_t^2z_t^2\big)}\right]\leq0,\quad{\rm for\ }\tau\in[-\pi/2,0].$$
		For \(s=1\), notice that
		\begin{align*}
			&\mbE\left[\frac{z_1^2(z_1\sin\tau+w_1\cos\tau)(z_1\cos\tau-w_1\sin\tau)\sum_{t=2}^{T-1}\tsigma_t^2h_t^2}{\big(\sum_{t=1}^{T-1}\tsigma_t^2h_t^2\big)^2\big(\sum_{t=1}^{T-1}\tsigma_t^2z_t^2\big)}\right]\\
			&=\mbE\left[\frac{r_1^2\cos^2(\gamma_1-\tau)}{r_1^2\cos^2(\gamma_1-\tau)+\sum_{t=2}^{T-1}\tsigma_t^2z_t^2}\frac{r_1^2\sin(2\gamma_1)\sum_{t=2}^{T-1}\tsigma_t^{-2}h_t^2}{\big(r_1^2\sin^2\gamma_1+\sum_{t=2}^{T-1}\tsigma_t^{-2}h_t^2\big)^2}\right],
		\end{align*}
		where \(\gamma_1:=\beta_1+\tau\). Let
		$$g_1(\gamma_1):=\frac{r_1^2\sin(2\gamma_1)\sum_{t=2}^{\infty}t^{-2}h_t^2}{\big(r_1^2\sin^2\gamma_1+\sum_{t=2}^{\infty}t^{-2}h_t^2\big)^2},$$
		which is a periodic function of \(\gamma_1\) with the period of \(\pi\), and \(g_1(\pi/2-\delta)=-g_1(\pi/2+\delta)\geq0\) for \(\delta\in[0,\pi/2]\). By the same arguments for \(s\geq2\), when \(\delta,\tau\in[0,\pi/2]\), we have
		\begin{align*}
			\frac{r_1^2\cos^2(\pi/2-\delta-\tau)g_1(\pi/2-\delta)}{r_1^2\cos^2(\pi/2-\delta-\tau)+\sum_{t=2}^{T-1}\tsigma_t^2z_t^2}&\geq\frac{-r_1^2\cos^2(\pi/2+\delta-\tau)g_1(\pi/2+\delta)}{r_1^2\cos^2(\pi/2+\delta-\tau)+\sum_{t=2}^{T-1}\tsigma_t^2z_t^2}\geq0,
		\end{align*}
		and
		$$\mbE\left[\frac{z_1^2(z_1\sin\tau+w_1\cos\tau)(z_1\cos\tau-w_1\sin\tau)\sum_{t=2}^{T-1}\tsigma_t^2h_t^2}{\big(\sum_{t=1}^{T-1}\tsigma_t^2h_t^2\big)^2\big(\sum_{t=1}^{T-1}\tsigma_t^2z_t^2\big)}\right]\geq0,\quad{\rm for\ }\tau\in[0,\pi/2].$$
		Similarly, when \(\tau\in[-\pi/2,0]\) and \(\delta\in[0,\pi/2]\), we have
		\begin{align*}
			0&\leq\frac{r_1^2\cos^2(\pi/2-\delta-\tau)g_1(\pi/2-\delta)}{r_1^2\cos^2(\pi/2-\delta-\tau)+\sum_{t=2}^{T-1}\tsigma_t^2z_t^2}\leq\frac{-r_1^2\cos^2(\pi/2+\delta-\tau)g_1(\pi/2+\delta)}{r_1^2\cos^2(\pi/2+\delta-\tau)+\sum_{t=2}^{T-1}\tsigma_t^2z_t^2},
		\end{align*}
		and
		$$\mbE\left[\frac{z_1^2(z_1\sin\tau+w_1\cos\tau)(z_1\cos\tau-w_1\sin\tau)\sum_{t=2}^{T-1}\tsigma_t^2h_t^2}{\big(\sum_{t=1}^{T-1}\tsigma_t^2h_t^2\big)^2\big(\sum_{t=1}^{T-1}\tsigma_t^2z_t^2\big)}\right]\leq0,\quad{\rm for\ }\tau\in[-\pi/2,0].$$
		Consequently, we conclude that  
		$$\tau\frac{d}{d\tau}H(\tau)\geq0,\quad{\rm for\ }\tau_1\in[-\pi/2,\pi/2].$$
		When \(\tau=0\), i.e. \(\hat{x}_{i_1,t}\) and \(\hat{x}_{i_2,t}\) are independent for \(t=1,\cdots,T-1\), then we have
		$$\Cov\left(\frac{\hat{x}_{i_1,1}^2}{\sum_{t=1}^{T-1}\tsigma_t^2\hat{x}_{i_1,1}^2},\frac{\hat{x}_{i_2,1}^2}{\sum_{t=1}^{T-1}\tsigma_t^2\hat{x}_{i_2,1}^2}\right)=0.$$
		Therefore, given \(\tau\in[-\pi/2,\pi/2]\), we can derive that
		\begin{align*}
			&\Cov\left(\frac{\hat{x}_{i_1,1}^2}{\sum_{t=1}^{T-1}\tsigma_t^2\hat{x}_{i_1,1}^2},\frac{\hat{x}_{i_2,1}^2}{\sum_{t=1}^{T-1}\tsigma_t^2\hat{x}_{i_2,1}^2}\right)=\int_0^{\tau}\frac{d}{d\tau}H(\tau)d\tau\geq0.
		\end{align*}
		\item Proof of (\ref{Eq of positive covariance 2}): For any \(k,l\in\{1,\cdots,T-1\}\), we abbreviate \(H_{k,l}^{i_1,i_2}(\tau)\) by \(H_{k,l}(\tau)\), since
		\begin{align*}
			&\frac{d}{d\tau}H_{k,l}(\tau)=\mbE\left[\frac{z_kz_l(h_k(z_l\cos\tau-w_l\sin\tau)+h_l(z_k\cos\tau-w_k\sin\tau))}{\big(\sum_{t=1}^{T-1}\tsigma_t^2z_t^2\big)\big(\sum_{t=1}^{T-1}\tsigma_t^2h_t^2\big)}\right]\\
			&+2\sum_{s=1}^{T-1}\tsigma_s^2\mbE\left[\frac{z_kz_lh_kh_lh_s(-z_s\cos\tau+w_s\sin\tau)}{\big(\sum_{t=1}^{T-1}\tsigma_t^2z_t^2\big)\big(\sum_{t=1}^{T-1}\tsigma_t^2h_t^2\big)^2}\right],
		\end{align*}
		and
		\begin{align*}
			&\frac{d}{d\tau}H_{k,l}(0)=\mbE\left[\frac{z_kz_l^2w_k+z_k^2z_lw_l}{\big(\sum_{t=1}^{T-1}\tsigma_t^2z_t^2\big)\big(\sum_{t=1}^{T-1}\tsigma_t^2w_t^2\big)}\right]-2\sum_{s=2}^{T-1}\tsigma_s^2\mbE\left[\frac{z_kz_lz_sw_kw_lw_s}{\big(\sum_{t=1}^{T-1}\tsigma_t^2z_t^2\big)\big(\sum_{t=1}^{T-1}\tsigma_t^2w_t^2\big)^2}\right],
		\end{align*}
		where
		\begin{align*}
			&\mbE\left[\frac{z_kz_l^2w_k}{\big(\sum_{t=1}^{T-1}\tsigma_t^2z_t^2\big)\big(\sum_{t=1}^{T-1}\tsigma_t^2w_t^2\big)}\right]=\mbE\left[\frac{z_kz_l^2}{\sum_{t=1}^{T-1}\tsigma_t^2z_t^2}\right]\mbE\left[\frac{w_k}{\sum_{t=1}^{T-1}\tsigma_t^2w_t^2}\right]=0,
		\end{align*}
		and
		\begin{align*}
			&\mbE\left[\frac{z_kz_lz_sw_kw_lw_s}{\big(\sum_{t=1}^{T-1}\tsigma_t^2z_t^2\big)\big(\sum_{t=1}^{T-1}\tsigma_t^2w_t^2\big)^2}\right]=\mbE\left[\frac{z_kz_lz_s}{\sum_{t=1}^{T-1}\tsigma_t^2z_t^2}\right]\mbE\left[\frac{w_kw_lw_s}{\big(\sum_{t=1}^{T-1}\tsigma_t^2w_t^2\big)^2}\right]=0.
		\end{align*}
		Hence, \(\frac{d}{d\tau}H_{k,l}(0)=0\). Next, notice that
		\begin{align*}
			&\frac{d^2}{d\tau^2}H_{k,l}(\tau)=\mbE\left[\frac{2z_kz_l(h_k'h_l'-h_kh_l)}{\big(\sum_{t=1}^{T-1}\tsigma_t^2z_t^2\big)\big(\sum_{t=1}^{T-1}\tsigma_t^2h_t^2\big)}\right]-\sum_{s=1}^{T-1}\tsigma_s^2\mbE\left[\frac{2z_kz_l(h_kh_l'+h_k'h_l)h_sh_s'}{\big(\sum_{t=1}^{T-1}\tsigma_t^2z_t^2\big)\big(\sum_{t=1}^{T-1}\tsigma_t^2h_t^2\big)^2}\right]\\
			&-\sum_{s=1}^{T-1}\tsigma_s^2\mbE\left[\frac{2z_kz_l\big(h_kh_lh_sh_s'\big)'}{\big(\sum_{t=1}^{T-1}\tsigma_t^2z_t^2\big)\big(\sum_{t=1}^{T-1}\tsigma_t^2h_t^2\big)^2}\right]+\sum_{s,r=1}^{T-1}\tsigma_s^2\tsigma_r^2\mbE\left[\frac{8z_kz_l\big(h_kh_lh_sh_s'h_rh_r'\big)}{\big(\sum_{t=1}^{T-1}\tsigma_t^2z_t^2\big)\big(\sum_{t=1}^{T-1}\tsigma_t^2h_t^2\big)^3}\right],
		\end{align*}
		where \(h_t'=\frac{d}{d\tau}h_t=z_t\cos\tau-w_t\sin\tau\sim\mcN(0,1)\) and \((h_t')'=-h_t\). For the first term in the above equation, by the Cauchy's inequality, we have
		\begin{align}
			&\left|\mbE\left[\frac{z_kz_l(h_k'h_l'-h_kh_l)}{\big(\sum_{t=1}^{T-1}\tsigma_t^2z_t^2\big)\big(\sum_{t=1}^{T-1}\tsigma_t^2h_t^2\big)}\right]\right|\leq\mbE\left[\left(\sum_{t=1}^{T-1}\tsigma_t^2h_t^2\right)^{-4}\right]^{1/4}\times\mbE\left[\left(\sum_{t=1}^{T-1}\tsigma_t^2z_t^2\right)^{-4}\right]^{1/4}\notag\\
			&\times\mbE[z_k^4z_l^4]^{1/4}\mbE[(h_k'h_l'-h_kh_l)^4]^{1/4}\leq C,\label{Eq of finite integration of positive covariance}
		\end{align}
		where we use the fact that \(\sum_{t=1}^{T-1}\tsigma_t^2h_t^2\geq\tsigma_{10}^2\sum_{t=1}^{10}h_t^2\sim\tsigma_{10}^2\chi^2(10)\) and the inverse chi square distribution with degree of freedom greater than \(10\) has the finite 4th moment; and all \(z_k,z_l,h_k,h_l,h_k',h_l'\) are standard normal. It is easy to see this constant \(C\) is independent of \(\tau\). Similarly, we can also show that all other three terms are bounded by some constants independent of \(\tau\), so we omit details here for convenience. In a word, we show that
		$$\left|\frac{d^2}{d\tau^2}H_{k,l}(\tau)\right|<C,\quad{\rm for\ }\tau\in[-\pi/2,\pi/2],$$
		then
		$$\left|\frac{d}{d\tau}H_{k,l}(\tau)\right|\leq\frac{d}{d\tau}H_{k,l}(0)+\left|\int_0^{\tau}\frac{d^2}{d\tau^2}H_{k,l}(\tau)d\tau\right|\leq\int_0^{\tau}\left|\frac{d^2}{d\tau^2}H_{k,l}(\tau)\right|d\tau<C|\tau|,$$
		and
		$$|\mcC_{k,l}(\tau)|=\left|\int_0^{\tau}\frac{d}{d\tau}H_{k,l}(\tau)d\tau\right|\leq\int_0^{\tau}\left|\frac{d}{d\tau}H_{k,l}(\tau)\right|d\tau\leq C\tau^2\leq C\rho^2.$$
		where we use \(\rho=\sin\tau\geq\frac{2}{\pi}\tau\) in the last inequality.
	\end{itemize}
	Now we complete our proof.
\end{proof}
Finally, let's prove Lemma \ref{Lem of finite variance} as follows: 
\begin{proof}[Proof of Lemma \ref{Lem of finite variance}]
    Recall the definition of $M_{i;k,l}$ in \eqref{Eq of Mikl abuse}, by Lemmas \ref{Lem of remove dependence} and \ref{Lem of unify variance}, we have concluded that
$$\frac{1}{n}\mbE\left[\left|\sum_{i=1}^n\left(M_{i;k,l}-\frac{\beta_k\beta_l\hat{x}_{i,k}\hat{x}_{i,l}}{\sum_{t=1}^{T-1}\tsigma_t^2\hat{x}_{i,t}^2}\right)\right|^2\right]\leq\frac{C_{B,M_0,c}}{(kl)^2T^{\delta^2}},$$
where we use $\beta_k\leq \mrO(k^{-1})$ by \eqref{Eq of beta upper}. In other words, it gives that 
$$\Var\left(\frac{1}{\sqrt{n}}\sum_{i=1}^nM_{i;k,l}\right)\to\frac{1}{n}\sum_{i_1,i_2=1}^n\Cov\left(\frac{\tsigma_k\tsigma_l\hat{x}_{i_1,k}\hat{x}_{i_1,l}}{\sum_{t=1}^{T-1}\tsigma_t^2\hat{x}_{i_1,t}^2},\frac{\tsigma_k\tsigma_l\hat{x}_{i_2,k}\hat{x}_{i_2,l}}{\sum_{t=1}^{T-1}\tsigma_t^2\hat{x}_{i_2,t}^2}\right).$$
By Lemma \ref{Lem of positive covariance}, it gives that
$$\left|\Cov\left(\frac{\tsigma_k\tsigma_l\hat{x}_{i_1,k}\hat{x}_{i_1,l}}{\sum_{t=1}^{T-1}\tsigma_t^2\hat{x}_{i_1,t}^2},\frac{\tsigma_k\tsigma_l\hat{x}_{i_2,k}\hat{x}_{i_2,l}}{\sum_{t=1}^{T-1}\tsigma_t^2\hat{x}_{i_2,t}^2}\right)\right|\leq\mrO\left((kl)^{-2}\rho_{i_1,i_2}^2\right),$$
where
$$\rho_{i_1,i_2}^2=\frac{(\sum_{j=1}^n\Gamma_{i_1,j}\Gamma_{i_2,j}f_j(0))^2}{(\sum_{j=1}^n\Gamma_{i_1,j}^2f_j(0))(\sum_{j=1}^n\Gamma_{i_2,j}^2f_j(0))}.$$
Hence, define
\begin{align}
	\bbF:=\diag(f_1(0),\cdots,f_n(0)),\quad\tilde{\bbGa}=\diag(\bbGa\bbF\bbGa')^{-1/2}\bbGa\bbF\bbGa'\diag(\bbGa\bbF\bbGa')^{-1/2},\label{Eq of bbF}
\end{align}
where $f_i(0)$ is defined in \eqref{Eq of spectral density}. Note that \(\rho_{i_1,i_2}\) is the \((i_1,i_2)\)-th entry of \(\tilde{\bbGa}\), and
$$\Var\left(\frac{1}{\sqrt{n}}\sum_{i=1}^nM_{i;k,l}\right)\leq\frac{1}{n}\mrO\left((kl)^{-2}\Vert\tilde{\bbGa}\Vert_F^2\right).$$
According to Assumptions \ref{Ap of panel lag polynomial} and \ref{Ap of nonpanel}, since
$$\sum_{j=1}^n\Gamma_{i,j}^2f_j(0)\geq C_{b,m_0},\quad{\rm for\ }i=1,\cdots,n,$$
then
$$\frac{1}{n}\Vert\tilde{\bbGa}\Vert_F^2\leq\frac{C_{b,m_0}}{n}\Vert\bbGa\bbF\bbGa'\Vert_F^2\leq C_{b,m_0}\Vert\bbGa\bbF\bbGa'\Vert^2\leq C_{B,b,M_0,m_0}.$$
which concludes (\ref{Eq of finite variance}).
\end{proof}

\subsection{Asymptotic behaviors of eigenvectors}\label{sec of eigenvectors dependent}
\begin{lem}\label{Lem of alpha}
	Under Assumptions {\rm \ref{Ap of highdimensionality}, \ref{Ap of panel lag polynomial}} and {\rm \ref{Ap of nonpanel}}, for any $K\in\mbN^+$ and \(k\in\{1,\cdots,K\}\), let $\hat{F}_k=\sum_{t=1}^{T-1}\alpha_{k,t}\bbw_t$ be the eigenvector corresponding to the $k$-th largest eigenvalue of the sample correlation matrix of $\bbX=[X_1,\cdots,X_T]$ generated by \eqref{Eq of nonpanel Xt}, then
    $$\lim_{n\to\infty}\sqrt{n}\mbE[1-\alpha_{k,k}^2]=0.$$
\end{lem}
The proof of above Lemma is the same as Lemma \ref{Lem of at convergence}. Before proving Lemma \ref{Lem of alpha}, we first need to show that:
\begin{align}
	\left|\frac{1}{\sqrt{n}}\sum_{i=1}^n\mbE\left[M_{i;k,l}-\fM_{k,l}\right]\right|\leq\frac{C_{B,M_0,c}}{(kl)T^{\delta^2/2}}.\label{Eq of mean difference}
\end{align}
where 
$$\fM_{k,l}=\frac{(kl)^{-1}Z_kZ_l}{\sum_{t=1}^{\infty}t^{-2}Z_t^2},$$
is defined in \eqref{Eq of fM} and \(\{Z_t:t\in\mbN^+,Z_t\overset{i.i.d.}{\sim}\mcN(0,1)\}\). By Lemmas \ref{Lem of remove dependence} and \ref{Lem of unify variance}, combining with \eqref{Eq of beta upper}, we know that for any sufficiently small $\delta>0$
$$\left|\mbE\left[\frac{1}{n}\sum_{i=1}^nM_{i;k,l}-\frac{\beta_k\beta_lZ_kZ_l}{\sum_{t=1}^{T-1}\beta_t^2Z_t^2}\right]\right|\leq C_{B,M_0,c}(kl)^{-1}T^{-1/2-\delta^2/2}.$$
By the same tricks as those in Lemma \ref{Lem of cut down}, we can conclude that 
$$\mbE\left|\left[\frac{Z_kZ_l}{\sum_{t=1}^{\infty}t^{-2}Z_t^2}-\frac{Z_kZ_l}{\sum_{t=1}^{T-1}t^{-2}Z_t^2}\right]\right|\leq\mrO(T^{-1})$$
Hence, to prove \eqref{Eq of mean difference}, it suffices to show that
$$\left|\mbE\left[\frac{Z_kZ_l}{\sum_{t=1}^{T-1}\beta_t^2Z_t^2}-\frac{Z_kZ_l}{\sum_{t=1}^{T-1}t^{-2}Z_t^2}\right]\right|<\mro(T^{-1}).$$
By \eqref{Eq of beta upper}, notice that
$$|t^{-2}-\tsigma_t^2|=t^{-2}\left(1-\frac{t\sin(\pi/2T)}{\sin(\pi t/2T)}\right)\left(1+\frac{t\sin(\pi/2T)}{\sin(\pi t/2T)}\right)\leq\frac{t^{-2}\pi^2}{8T^2},$$
where we use the fact that \(\sin x\geq x-x^3/6\) and \(2x/\pi\leq\sin x\leq x\) for \(x\in[0,\pi/2]\), then
\begin{align*}
	&\left|\mbE\left[\frac{Z_kZ_l}{\sum_{t=1}^{T-1}t^{-2}Z_t^2}-\frac{Z_kZ_l}{\sum_{t=1}^{T-1}\tsigma_t^2Z_t^2}\right]\right|\leq\sum_{t=1}^{T-1}|t^{-2}-\tsigma_t^2|\mbE\left[\frac{Z_kZ_lZ_t^2}{(\sum_{t=1}^{T-1}t^{-2}Z_t^2)(\sum_{t=1}^{T-1}\tsigma_t^2Z_t^2)}\right]\\
	&\leq\mrO\left(\sum_{t=1}^{T-1}|t^{-2}-\tsigma_t^2|\right)=\mrO(T^{-2}),
\end{align*}
which concludes that (\ref{Eq of mean difference}).
\begin{proof}[Proof of Lemma \ref{Lem of alpha}]
	Since the whole proofs of this lemma is nearly the same as those for Lemma \ref{Lem of at convergence}, then we only consider the case when \(k=1\). Similar as (\ref{Eq of A1 A2}), denote
	$$A_1:=\sum_{k=1}^{N_K}\alpha_{1,k}\sigma_k\bbD^{-1/2}\bbGa\bbe\bbv_k,\quad{\rm and}\quad B_1:=\sum_{k=N_K+1}^{T-1}\alpha_{1,k}\sigma_k\bbD^{-1/2}\bbGa\bbe\bbv_k,$$
	where \(N_K\) is a pre-specified integer only depending on \(K\), then
	$$\hla_1=\Vert A_1\Vert_2^2+\Vert B_1\Vert_2^2+2\langle A_1,B_1\rangle.$$
	To prove this lemma, it is enough to show that \(\lim_{n\to\infty}\sqrt{n}\mbE[\alpha_{1,t}^2]=0\) for \(t\geq2\). Notice that
	$$\alpha_{1,t}=\frac{\sum_{k\neq t}^{T-1}\alpha_{1,k}n^{-1}\sum_{i=1}^nM_{i;k,t}}{n^{-1}\hla_1-n^{-1}\sum_{i=1}^nM_{i;t,t}}.$$
	\begin{itemize}
		\item Let's first show that for \(t=1,\cdots,N_K\),
		\begin{align}
			\mbP\left(\frac{1}{n}\left|\hla_1-\sum_{i=1}^nM_{i;t,t}\right|>\mbE[\fM_{1,1}-\fM_{2,2}]/2\right)\geq1-C_{B,b,M_0,m_0,c}n^{-3/5},\label{Eq of denominators 1}
		\end{align}
		and
		\begin{align}
			\mbP\left(\frac{1}{n}\left|\hla_1-\sum_{t=N_K+1}^{T-1}\sum_{i=1}^nM_{i;t,t}\right|>\mbE[\fM_{1,1}-\fM_{2,2}]/2\right)\geq1-C_{B,b,M_0,m_0,c}n^{-3/5}.\label{Eq of denominators 2}
		\end{align}
		First, according to (\ref{Eq of mean difference}), (\ref{Eq of finite variance}) and the Chebyshev's inequality, it yields that
		\begin{align}
			\mbP\left(\left|n^{-1}\sum_{i=1}^nM_{i;k,l}-\mbE[\fM_{k,l}]\right|>n^{-1/5}\right)\leq C_{B,b,M_0,m_0,c}n^{-3/5}.\label{Eq of concentration M kl}
		\end{align}
		By the Cauchy's inequality, we have
		\begin{align*}
			&\frac{1}{n}\Vert B_1\Vert_2^2\leq\left(1-\sum_{k=1}^{N_K}\alpha_{1,k}^2\right)\sum_{t=N_K+1}^{T-1}\frac{1}{n}\sum_{i=1}^nM_{i;t,t}\\
			&\leq(1-\alpha_{1,1}^2)\left(1-\sum_{k=1}^{N_K}\mbE[\fM_{k,k}]+\mrO(n^{-1/5})\right)=(1-\alpha_{1,1}^2)\left(\mfb_K+\mrO(n^{-1/5})\right)
		\end{align*}
		with probability at least of \(1-C_{B,b,M_0,m_0,c}n^{-3/5}\), where 
        $$\mfb_k:=1-\sum_{k=1}^{N_K}\mbE[\fM_{k,k}].$$
        Similarly, we can obtain that
		\begin{align*}
			&\frac{1}{n}\Vert A_1\Vert_2^2\leq\frac{1}{n}\sum_{k=1}^{N_K}\alpha_{1,k}^2\mbE[\fM_{k,k}]+\mrO(n^{-1/5})\leq\alpha_{1,1}^2\mbE[\fM_{1,1}]+(1-\alpha_{1,1}^2)\mbE[\fM_{2,2}]+\mrO(n^{-1/5})
		\end{align*}
		with probability at least of \(1-C_{B,b,M_0,m_0,c}n^{-3/5}\). Consequently, it yields that
		\begin{align*}
			&\frac{\hla_1}{n}\leq\frac{1}{n}\Vert A_1\Vert_2^2+\frac{1}{n}\Vert B_1\Vert_2^2+\frac{2}{n}\Vert A_1\Vert_2\Vert B_1\Vert_2\\
			&\leq\alpha_{1,1}^2\mbE[\fM_{1,1}]+(1-\alpha_{1,1}^2)(\mbE[\fM_{2,2}]+\mfb_K)+2(1-\alpha_{1,1}^2)^{1/2}\mfb_K^{1/2}+\mrO(n^{-1/5})
		\end{align*}
		with probability at least of \(1-C_{B,b,M_0,m_0,c}n^{-3/5}\). On the other hand,
		$$\frac{\hla_1}{n}\geq\frac{1}{n}\sum_{i=1}^nM_{i;1,1}\geq\mbE[\fM_{1,1}]-n^{-1/5}$$
		with probability at least of \(1-\mrO(n^{-3/5})\). Combining the above two results, we obtain that
		\begin{align*}
			&(1-\alpha_{1,1}^2)\mbE[\fM_{1,1}]\leq(1-\alpha_{1,1}^2)(\mbE[\fM_{2,2}]+\mfb_K)+2(1-\alpha_{1,1}^2)^{1/2}\mfb_K^{1/2}+\mrO(n^{-1/5})\\
			&\Longrightarrow\quad1-\alpha_{1,1}^2\leq\frac{2\mfb_K^{1/2}+\mrO(n^{-1/5})}{\mbE[\fM_{1,1}-\fM_{2,2}]-\mfb_K}\leq\frac{6\mfb_K^{1/2}}{\mbE[\fM_{1,1}-\fM_{2,2}]},
		\end{align*}
		with probability at least of \(1-C_{B,b,M_0,m_0,c}n^{-3/5}\), where we choose a sufficiently large \(K>0\) such that \(\mfb_K=1-\sum_{k=1}^{N_K}\mbE[\fM_{k,k}]\ll\mbE[\fM_{1,1}-\fM_{2,2}]^2/4\). Moreover, it further implies that
		\begin{align*}
			&\left|\frac{\hla_1}{n}-\frac{1}{n}\sum_{i=1}^nM_{i;1,1}\right|\\
			&\leq(1-\alpha_{1,1}^2)(\mbE[\fM_{2,2}+\fM_{1,1}]+\mfb_K)+2(1-\alpha_{1,1}^2)^{1/2}\mfb_K^{1/2}+\mrO(n^{-1/5})\leq\mrO(\mfb_K^{1/2})
		\end{align*}
		with probability at least of \(1-C_{B,b,M_0,m_0,c}n^{-3/5}\). Therefore, for \(t\geq2\), it gives that
		\begin{align*}
			&\left|\frac{\hla_1}{n}-\frac{1}{n}\sum_{i=1}^nM_{i;t,t}\right|\geq\frac{1}{n}\left|\sum_{i=1}^n(M_{i;1,1}-M_{i;t,t})\right|-\mrO(\mfb_K^{1/2})\\
			&\geq\mbE[\fM_{1,1}-\fM_{t,t}]-\mrO(\mfb_K^{1/2})>\mbE[\fM_{1,1}-\fM_{2,2}]/2
		\end{align*}
		with probability at least of \(1-\mrO(n^{-3/5})\), so we conclude (\ref{Eq of denominators 1}). Similarly,
		\begin{align*}
			&\left|\frac{\hla_1}{n}-\frac{1}{n}\sum_{t=N_K+1}^{T-1}\sum_{i=1}^nM_{i;t,t}\right|\geq\frac{1}{n}\left|\sum_{i=1}^n\left(M_{i;1,1}-\sum_{t=N_K+1}^{T-1}M_{i;t,t}\right)\right|-\mrO(\mfb_K^{1/2})\\
			&\geq\mbE[\fM_{1,1}]-\mrO(\mfb_K^{1/2})>\mbE[\fM_{1,1}-\fM_{2,2}]/2
		\end{align*}
		with probability at least of \(1-C_{B,b,M_0,m_0,c}n^{-3/5}\), so we conclude (\ref{Eq of denominators 2}).
		\item Define
		$$\mcF_t:=\left\{\frac{1}{n}\left|\hla_1-\sum_{i=1}^nM_{i;t,t}\right|>\mbE[\fM_{1,1}-\fM_{2,2}]/2\right\},\quad t=1,\cdots,N_K,$$
		and
		$$\mcF_{N_K+1}:=\left\{\frac{1}{n}\left|\hla_1-\sum_{t=N_K+1}^{T-1}\sum_{i=1}^nM_{i;t,t}\right|>\mbE[\fM_{1,1}-\fM_{2,2}]/2\right\},$$
		then for \(2\leq t\leq N_K\), we have
		\begin{align*}
			&\sqrt{n}\mbE[\alpha_{1,t}^2]=\sqrt{n}\mbE[\alpha_{1,t}^2|\mcF_t]\mbP(\mcF_t)+\sqrt{n}\mbE[\alpha_{1,t}^2|\mcF_t^c]\mbP(\mcF_t^c)\\
			&\leq\sqrt{n}\mbE[\alpha_{1,t}^2|\mcF_t]+\sqrt{n}\mbP(\mcF_t^c)\\
			&\leq C_{B,b,M_0,m_0,c}n^{-1/10}+\frac{4n^{-3/2}}{\mbE[\fM_{1,1}-\fM_{2,2}]^2}\mbE\left[\left(\sum_{k\neq t}^{T-1}\alpha_{1,k}\sum_{i=1}^nM_{i;k,t}\right)^2\right],
		\end{align*}
		by the Cauchy's inequality, we have
		\begin{align*}
			&n^{-3/2}\mbE\left[\left(\sum_{k\neq t}^{T-1}\alpha_{1,k}\sum_{i=1}^nM_{i;k,t}\right)^2\right]\leq n^{-3/2}\sum_{k\neq t}^{T-1}\mbE\left[\left(\sum_{i=1}^nM_{i;k,t}\right)^2\right]\\
			&\leq n^{-3/2}\sum_{k\neq t}^{T-1}\left(\Var\left(\sum_{i=1}^nM_{i;k,t}\right)+\left(\sum_{i=1}^n\mbE[M_{i;k,t}]\right)^2\right).
		\end{align*}
		By (\ref{Eq of mean difference}), when \(k\neq t\), since $\mbE[\fM_{k,t}]=0$, we have 
		$$\left|\frac{1}{\sqrt{n}}\sum_{i=1}^n\mbE[M_{i;k,t}]\right|\leq\frac{C_{B,M_0,c}}{(kt)n^{\delta^2/2}}$$
		so
		$$n^{-3/2}\sum_{k\neq t}^{T-1}\left(\sum_{i=1}^n\mbE[M_{i;k,t}]\right)^2\leq\frac{C_{B,M_0,c}}{t^2n^{1/2+\delta^2/2}}.$$
		Moreover, by (\ref{Eq of finite variance}), we have
		\begin{align*}
			&n^{-1}\Var\left(\sum_{i=1}^nM_{i;k,t}\right)\leq C_{B,b,M_0,m_0}(kt)^{-2},
		\end{align*}
		then
		$$n^{-3/2}\sum_{k\neq t}^{T-1}\Var\left(\sum_{i=1}^nM_{i;k,t}\right)\leq C_{B,b,M_0,m_0,c}t^{-2}n^{-1/2},$$
		which implies that 
		\begin{align}
			n^{-3/2}\mbE\left[\left(\sum_{k\neq t}^{T-1}\alpha_{1,k}\sum_{i=1}^nM_{i;k,t}\right)^2\right]\leq C_{B,b,M_0,m_0,c}t^{-2}n^{-1/2},\label{Eq of alpha 1 t}
		\end{align}
		and \(\lim_{n\to\infty}\sqrt{n}\mbE[\alpha_{1,t}^2]=0\) for \(2\leq t\leq N_K\). Finally, let \(\mfc_K:=\sum_{t=N_K+1}^{T-1}\alpha_{1,t}^2\), similar as the case of \(t=2\), we have
		\begin{align*}
			&\sqrt{n}\mbE[\mfc_K]=\sqrt{n}\mbE[\mfc_K|\mcF_{N_K+1}]\mbP(\mcF_{N_K+1})+\sqrt{n}\mbE[\mfc_K|\mcF_{N_K+1}^c]\mbP(\mcF_{N_K+1}^c)\\
			&\leq C_{B,b,M_0,m_0,c}n^{-1/10}+\frac{4n^{-3/2}}{\mbE[\fM_{1,1}-\fM_{2,2}]^2}\sum_{t=K+1}^{T-1}\mbE\left[\left(\sum_{k\neq t}^{T-1}\alpha_{1,k}\sum_{i=1}^nM_{i;k,t}\right)^2\right]\leq C_{B,b,M_0,m_0,c}n^{-1/10},
		\end{align*}
		where we use (\ref{Eq of alpha 1 t}) in the last inequality.
	\end{itemize}
	Since \(\sqrt{n}\mbE[1-\alpha_{1,1}^2]=\sum_{t=2}^{N_K}\sqrt{n}\mbE[\alpha_{1,t}^2]+\sqrt{n}\mbE[\mfc_K]\leq C_{B,b,M_0,m_0,c}n^{-1/10}\), it completes our proof for \(k=1\). Finally, for general \(k\geq2\), we can inductively prove that \(\lim_{n\to\infty}\sqrt{n}\mbE[1-\alpha_{k,k}^2]=0\) based on \(\lim_{n\to\infty}\sqrt{n}\mbE[1-\alpha_{k-1,k-1}^2]=0\), just as what we have done in \S\ref{sec of alpha kt correlation}. For the choice of \(N_K\), we can still use (\ref{Eq of N_K}), since the arguments are totally the same, we omit the details here to save space.
\end{proof}
\subsection{Limit of the convergence in probability}\label{sec of mbP dependent}
In this subsection, we will first find the limit of the convergence in probability for $n^{-1}\hla_k$, i.e.
\begin{pro}\label{Thm of convergence in probability}
	Under Assumptions {\rm \ref{Ap of highdimensionality}, \ref{Ap of panel lag polynomial}, \ref{Ap of nonpanel}} and {\rm (\ref{Eq of Ap of panel lag polynomial})}, for any $K\in\mbN^+$ \(1\leq k\leq K\), let $\hla_k$ be the $k$-th largest eigenvalue of the sample correlation matrix of $\bbX=[X_1,\cdots,X_T]$ generated by \eqref{Eq of nonpanel Xt}, then
	\begin{align*}
		\frac{\hla_k}{n}\overset{\mbP}{\longrightarrow}\mbE[\fM_{k,k}],
	\end{align*}
	where $\fM_{k,k}$ is defined in \eqref{Eq of fM}.
\end{pro}
\begin{proof}
	For convenience, we only present the proof for \(k=1\), since others are the same. Recall that
	$$\frac{\hla_1}{\sqrt{n}}=\sum_{k,l=1}^{T-1}\alpha_{1,k}\alpha_{1,l}\frac{1}{\sqrt{n}}\sum_{i=1}^nM_{i;k,l}=\frac{1}{\sqrt{n}}\left(\Vert A_1\Vert^2+\Vert B_1\Vert^2+2\langle A_1,B_1\rangle\right),$$
    where \(A_1,B_1\) are defined in Lemma \ref{Lem of alpha}. We will show that
	$$\frac{1}{\sqrt{n}}(\Vert A_1\Vert_2^2)^{\circ}\overset{\mbP}{\longrightarrow}\frac{1}{\sqrt{n}}\sum_{i=1}^nM_{i;1,1}^{\circ},\quad\frac{1}{\sqrt{n}}(\Vert B_1\Vert_2^2)^{\circ},\frac{1}{\sqrt{n}}\langle A_1,B_1\rangle^{\circ}\overset{\mbP}{\longrightarrow}0.$$
	Since
	\begin{align*}
		&\frac{1}{\sqrt{n}}(\Vert A_1\Vert_2^2)^{\circ}=\sum_{k,l=1}^{N_K}\frac{1}{\sqrt{n}}\sum_{i=1}^n\left(\alpha_{1,k}\alpha_{1,l}(M_{i;k,l})^{\circ}+\alpha_{1,k}\alpha_{1,l}\mbE[M_{i;k,l}]-\mbE[\alpha_{1,k}\alpha_{1,l}M_{i;k,l}]\right),
	\end{align*}
	by (\ref{Eq of finite variance}) and Lemma \ref{Lem of alpha}, we can use the Chebyshev's inequality to imply that
	\begin{align*}
		&\sum_{k\neq1,l\neq1}^{N_K}\alpha_{1,k}\alpha_{1,l}\frac{1}{\sqrt{n}}\sum_{i=1}^nM_{i;k,l}^{\circ}\overset{\mbP}{\longrightarrow}0.
	\end{align*}
	Moreover, by Lemma \ref{Lem of alpha} and (\ref{Eq of mean difference}), it gives that
	\begin{align*}
		&\left|\sum_{k\neq1,l\neq1}^{N_K}\alpha_{1,k}\alpha_{1,l}\frac{1}{\sqrt{n}}\sum_{i=1}^n\mbE[M_{i;k,l}]\right|\leq\sum_{k\neq1,l\neq1}^{N_K}(kl)^{-1}\left(\delta_{k,l}\sqrt{n}+n^{-\delta^2/8}\right)\left|\alpha_{1,k}\alpha_{1,l}\right|\overset{\mbP}{\longrightarrow}0.
	\end{align*}
	By (\ref{Eq of finite variance}), we know that
	$$\sum_{k\neq1,l\neq1}^{N_K}\frac{1}{\sqrt{n}}\sum_{i=1}^n\Cov(\alpha_{1,k}\alpha_{1,l},M_{i;k,l})\leq\sum_{k\neq1,l\neq1}^{N_K}\mrO((kl)^{-1}\mbE[\alpha_{1,k}^2\alpha_{1,l}^2]^{1/2})\to0,$$
	and
	\begin{align*}
		&\left|\sum_{k\neq1,l\neq1}^{N_K}\mbE[\alpha_{1,k}\alpha_{1,l}]\frac{1}{\sqrt{n}}\sum_{i=1}^n\mbE[M_{i;k,l}]\right|\leq\sum_{k\neq1,l\neq1}^{N_K}(kl)^{-1}\left(\delta_{k,l}\sqrt{n}+n^{-\delta^2/8}\right)\mbE\left[\left|\alpha_{1,k}\alpha_{1,l}\right|\right]\to0,
	\end{align*}
	which implies that
	$$\sum_{k\neq1,l\neq1}^{N_K}\frac{1}{\sqrt{n}}\sum_{i=1}^n\mbE[\alpha_{1,k}\alpha_{1,l}M_{i;k,l}]\to0.$$
	Hence, we obtain that
	$$\frac{1}{\sqrt{n}}(\Vert A_1\Vert_2^2)^{\circ}\overset{\mbP}{\longrightarrow}\frac{1}{\sqrt{n}}\sum_{i=1}^nM_{i;1,1}^{\circ}+\frac{1}{\sqrt{n}}\sum_{i=1}^n\left((\alpha_{1,1}^2-1)M_{i;1,1}\right)^{\circ}\overset{\mbP}{\longrightarrow}\frac{1}{\sqrt{n}}\sum_{i=1}^nM_{i;1,1}^{\circ},$$
	where we use Lemma \ref{Lem of alpha} in the last step. Next, let's show that
	$$\frac{1}{\sqrt{n}}(\Vert B_1\Vert_2^2)^{\circ}\overset{\mbP}{\longrightarrow}0.$$
	By the Cauchy's inequality, we have
	\begin{align*}
		&\frac{1}{\sqrt{n}}\Vert B_1\Vert_2^2\leq\sqrt{n}(1-\alpha_{1,1}^2)\times\sum_{t=N_K+1}^{T-1}\frac{1}{n}\sum_{i=1}^nM_{i;t,t}\\
		&=\sqrt{n}(1-\alpha_{1,1}^2)\times\left(1-\sum_{t=1}^{N_K}\frac{1}{n}\sum_{i=1}^nM_{i;t,t}\right)\overset{\mbP}{\longrightarrow}0,
	\end{align*}
	where we use (\ref{Eq of concentration M kl}) in the last step, so combine with Lemma \ref{Lem of alpha}, it implies that \(n^{-1/2}(\Vert B_1\Vert_2^2)^{\circ}\overset{\mbP}{\longrightarrow}0\). Finally, since
	$$\frac{1}{\sqrt{n}}\left|\langle A_1,B_1\rangle\right|\leq\frac{1}{\sqrt{n}}\Vert A_1\Vert_2\times\Vert B_1\Vert_2\overset{\mbP}{\longrightarrow}0,$$
	and by the Cauchy's inequality, we have
	\begin{align*}
		&\frac{1}{n}\mbE\left[\langle A_1,B_1\rangle\right]^2=\frac{1}{n}\left(\sum_{k=1}^{N_K}\sum_{t=N_K+1}^{T-1}\sum_{i=1}^n\mbE\left[\alpha_{1,k}\alpha_{1,t}M_{i;k,t}\right]\right)^2\\
		&\leq\mbE[1-\alpha_{1,1}^2]\times\sum_{k=1}^{N_K}\sum_{t=N_K+1}^{T-1}\frac{1}{n}\mbE\left[\left(\sum_{i=1}^nM_{i;k,t}\right)^2\right]\\
		&\leq\mbE[1-\alpha_{1,1}^2]\times\sum_{k=1}^{N_K}\sum_{t=N_K+1}^{T-1}\frac{1}{n}\left(\Var\left(\sum_{i=1}^nM_{i;k,t}\right)+\left(\sum_{i=1}^n\mbE\left[M_{i;k,t}\right]\right)^2\right)\\
		&\leq\mbE[1-\alpha_{1,1}^2]\times\sum_{k=1}^{N_K}\sum_{t=N_K+1}^{T-1}C_{B,b,M_0,m_0,c}(kt)^{-2}\leq\sqrt{n}\mbE[1-\alpha_{1,1}^2]\times\log(n) n^{-1/2}\to0,
	\end{align*}
	where we use the (\ref{Eq of mean difference}), (\ref{Eq of finite variance}) and \ref{Lem of alpha} in the last line of above equations. Now, we conclude that
	$$\frac{1}{\sqrt{n}}\langle A_1,B_1\rangle^{\circ}\overset{\mbP}{\longrightarrow}0,$$
	which completes our proof.
\end{proof}
\subsection{Joint CLT for the extreme eigenvalues of the sample correlation matrix}\label{sec of CLT correlation dependent}
In Proposition \ref{Thm of convergence in probability}, we have shown that
$$\frac{\hla_k}{n}\overset{\mbP}{\longrightarrow}\mbE[\fM_{k,k}],$$
so in this subsection, we will further establish the joint CLT for $(\hla_1,\cdots,\hla_K)$. Here, let's make some necessary notations first. Let $\{\vec{z}_t=(z_{1,t},\cdots,z_{n,t})'\overset{i.i.d.}{\sim}\mcN(\boldsymbol{0},\tilde{\bbGa}):t\in\mbN^+\}$ be a sequence of i.i.d. normal vectors, where $\tilde{\bbGa}\in\mbR^{n\times n}$ is defined in \eqref{Eq of bbF}, then define
\begin{align}
    \hat{M}_{i;k,l}=\frac{(kl)^{-1}z_{i,l}z_{i,l}}{\sum_{t=1}^{\infty}t^{-2}z_{i,t}^2}.\label{Eq of hat Mikl}
\end{align}
By Lemmas \ref{Lem of remove dependence}, \ref{Lem of unify variance} and Proposition \ref{Thm of convergence in probability}, we can conclude that
$$\frac{\hla_k^{\circ}}{\sqrt{n}}\overset{\mbP}{\longrightarrow}\frac{1}{\sqrt{n}}\sum_{i=1}^n\hat{M}_{i;k,k}^{\circ},$$
where $\hla_k^{\circ}=\hla_k-\mbE[\hla_k]$. Therefore, to establish the CLT for \(\hla_k\), it suffices to establish the CLT for $n^{-1/2}\sum_{i=1}^n\hat{M}_{i;k,k}^{\circ}$ . However, note that all \(\hat{M}_{i;k,k}\) in \eqref{Eq of hat Mikl} are not independent. Hence, to establis the CLT for correlated random variables, we propose the following addition assumption:
\begin{Ap}\label{Ap of m dependent}
	Given the \(n\times n\) cross-sectional matrix \(\bbGa\), denote \(\Gamma_i\) to be the \(i\)-th row of \(\bbGa\) and we require that
	$$\{j=1,\cdots,n:\Gamma_{i_1,j}\neq0\}\cap\{j=1,\cdots,n:\Gamma_{i_2,j}\neq0\}=\emptyset$$
	for all \(|i_1-i_2|>m\), where \(m:=m(n)=\mro(n^{1/2})\).
\end{Ap}
\begin{remark}
	The above assumption suggests the sets of nonzero entries in the \(i_1\)-th row and \(i_2\)-th row of \(\bbGa\) do not have overlap if \(|i_1-i_2|>m\). And such matrix indeed exists, for example the \(m\) banded toeplitz matrices.
\end{remark}
\begin{remark}
	Moreover, the above assumption suggests \(\{\hat{M}_{1;k,k},\cdots,\hat{M}_{n;k,k}\}\) is a \(m\)-dependent sequence, see Definition 1 in \cite{hoeffding1948central}. Since
	$$\hat{M}_{i;k,k}=\frac{(kl)^{-1}z_{i,l}z_{i,l}}{\sum_{t=1}^{\infty}t^{-2}z_{i,t}^2},$$
	where \(\hat{x}_i=(\hat{x}_{i,1},\cdots,\hat{x}_{i,T-1})\sim\mcN\big(\boldsymbol{0},\bbI_{T-1}\big)\) and \(\Cov(\hat{x}_{i_1},\hat{x}_{i_2})=\rho_{i_1,i_2}\bbI_{T-1}\), where \(\rho_{i_1,i_2}\) is the \((i_1,i_2)\)-th entry of \(\tilde{\bbGa}\) defined in (\ref{Eq of bbF}), since 
	$$\rho_{i_1,i_2}=\Gamma_{i_1}\bbF\Gamma_{i_2}'=\sum_{j=1}^n\Gamma_{i_1,j}\Gamma_{i_2,j}f_j(0),$$
	and at least one of \(\Gamma_{i_1,j},\Gamma_{i_2,j}\) is zero for all \(j=1,\cdots,n\) when \(|i_1-i_2|>m\), it implies that \(\rho_{i_1,i_2}=0\) if \(|i_1-i_2|>m\), i.e. \(\hat{x}_{i_1},\hat{x}_{i_2}\) are independent, so that \(\hat{M}_{i_1;k,k},\hat{M}_{i_2;k,k}\) will be independent.
\end{remark}
Now, let's show that
\begin{lem}\label{Lem of preliminary CLT}
	Under Assumptions {\rm \ref{Ap of highdimensionality}, \ref{Ap of panel lag polynomial}, \ref{Ap of nonpanel}, \ref{Ap of m dependent}} and {\rm (\ref{Eq of Ap of panel lag polynomial})}, for any $K\in\mbN^+$ and $1\leq k\leq K$, we have
	$$\frac{1}{\sqrt{n}\mfm_{k,k}}\sum_{i=1}^n\hat{M}_{i;k,k}^{\circ}\overset{d}{\longrightarrow}\mcN(0,1),$$
    where $\hat{M}_{i;k,k}$ is defined in \eqref{Eq of hat Mikl} and
    \begin{align}
	    \mfm_{k,k}^2:=\Var\left(\frac{1}{\sqrt{n}}\sum_{i=1}^n\hat{M}_{i;k,k}\right)\asymp\mrO(1).\label{Eq of mfm kk}
	\end{align}
    Notation ``$\asymp$'' is defined in \eqref{Eq of asymp mbP}.
\end{lem}
\begin{proof}
	According to Lemma \ref{Lem of positive covariance}, it implies that $\Cov\left(\hat{M}_{i_1;k,k},\hat{M}_{i_2;k,k}\right)\geq0$, so
	$$\mfm_{k,k}^2=\frac{1}{n}\sum_{i=1}^n\Var\left(\hat{M}_{i;k,k}\right)+\frac{1}{n}\sum_{i_1\neq i_2}^n\Cov\left(\hat{M}_{i_1;k,k},\hat{M}_{i_2;k,k}\right)\geq\frac{1}{n}\sum_{i=1}^n\Var\left(\hat{M}_{i;k,k}\right)=\Var(\fM_{k,k}),$$
	where all \(\Var\big(\hat{M}_{i;k,k}\big)=\Var(\fM_{k,k})>0\) are equal. Combining with Lemma \ref{Lem of finite variance}, it implies that \(\mfm_{k,k}\asymp\mrO(1)\). By Theorem 1.4 in \cite{janson2021central}, it is enough to check that \(\{\hat{M}_{1;k,k}^{\circ},\cdots,\hat{M}_{n;k,k}^{\circ}\}\) satisfies the following version of the Lindeberg's condition:
	$$\lim_{n\to\infty}\frac{m}{n\mfm_{k,k}^2}\sum_{i=1}^n\mbE\left[(\hat{M}_{i;k,k}^{\circ})^21_{|\hat{M}_{i;k,k}^{\circ}|>\epsilon\sqrt{n}\mfm_{k,k}/m}\right]=0,\quad\forall\epsilon>0.$$
	By Assumption \ref{Ap of m dependent}, it is easy to see that \(\lim_{n\to\infty}\sqrt{n}\mfm_{k,k}/m=\infty\); on the other hand, \(|\hat{M}_{i;k,k}^{\circ}|\leq2\), then \(1_{|\hat{M}_{i;k,k}|>\epsilon\sqrt{n}\mfm_{k,k}/m}=0\) for sufficiently large \(n\), which concludes this lemma.
\end{proof}
Finally, we can conclude that
\begin{thm}\label{Thm of CLT I1}
	Under Assumptions {\rm \ref{Ap of highdimensionality}, \ref{Ap of panel lag polynomial}, \ref{Ap of nonpanel}, \ref{Ap of m dependent}} and {\rm (\ref{Eq of Ap of panel lag polynomial})}, for any $K\in\mbN^+$ and $1\leq k\leq K$, let \(\hla_k\) be the first \(k\)-th largest eigenvalue of the sample correlation matrix of \(\bbX=[X_1,\cdots,X_T]\) generated by {\rm (\ref{Eq of nonpanel Xt})}, then
	$$\frac{\sqrt{n}}{\mfm_{k,k}}\left(\frac{\hla_k}{n}-\mbE[\fM_{k,k}]\right)\overset{d}{\longrightarrow}\mcN(0,1),$$
	where $\mfm_{k,k}$ is defined in \eqref{Eq of mfm kk}. Moreover, let $\bbA_n=[A_{k,l}]$ be a $K\times K$ covariance matrix such that
	$$A_{k,l}:=\frac{1}{n}\Cov\left(\sum_{i=1}^n\hat{M}_{i;k,k},\sum_{i=1}^n\hat{M}_{i;l,l}\right),$$
    where $\hat{M}_{i;k,k}$ is defined in \eqref{Eq of hat Mikl}. If $\liminf_{n\to\infty}\sigma_{\min}(\bbA_n)>0$, where $\sigma_{\min}(\bbA_n)$ is the smallest singular value of $\bbA_n$, then
    $$\sqrt{n}\bbA_n^{-1/2}\left(\frac{\hla_1}{n}-\mbE[\fM_{1,1}],\cdots,\frac{\hla_K}{n}-\mbE[\fM_{K,K}]\right)'\overset{d}{\longrightarrow}\mcN(\boldsymbol{0},\bbI_K).$$
\end{thm}
\begin{proof}
	For simplicity, we only present the details for $k=1$, since the proofs for others are the same. According to Proposition \ref{Thm of convergence in probability} and Lemmas \ref{Lem of remove dependence}, \ref{Lem of unify variance}, it gives that
	$$\frac{\hla_k^{\circ}}{\sqrt{n}}\overset{\mbP}{\longrightarrow}\frac{1}{\sqrt{n}}\sum_{i=1}^nM_{i;k,k}^{\circ}.$$
	Combine with the above Lemma \ref{Lem of preliminary CLT}, we know that
	$$\frac{\hla_k^{\circ}}{\sqrt{n}\mfm_{k,k}}\overset{d}{\longrightarrow}\mcN(0,1),$$
	where we use the fact that \(\mfm_{k,k}\asymp\mrO(1)\) by \eqref{Eq of mfm kk}. Next, will show that 
	$$\lim_{n\to\infty}\frac{1}{\sqrt{n}}\left|\mbE\left[\hla_1-\sum_{i=1}^nM_{i;1,1}\right]\right|=0.$$
	Since
	\begin{align*}
		&\frac{1}{\sqrt{n}}\left|\mbE\left[\hla_1-\sum_{i=1}^nM_{i;1,1}\right]\right|\leq\frac{1}{\sqrt{n}}\sum_{i=1}^n\mbE\left[(1-\alpha_{1,1}^2)M_{i;1,1}\right]+\frac{1}{\sqrt{n}}\mbE\left[\left|\sum_{k\ {\rm or}\ l\neq1}^{T-1}\sum_{i=1}^n\alpha_{1,k}\alpha_{1,l}M_{i;k,l}\right|\right],
	\end{align*}
	by Lemma \ref{Lem of alpha}, we have
	$$\frac{1}{\sqrt{n}}\sum_{i=1}^n\mbE\left[(1-\alpha_{1,1}^2)M_{i;1,1}\right]\leq\sqrt{n}\mbE\left[(1-\alpha_{1,1}^2)\right]\longrightarrow0,$$
	and
	\begin{align*}
		&\frac{1}{\sqrt{n}}\mbE\left[\left|\sum_{k,l>1}^{T-1}\sum_{i=1}^n\alpha_{1,k}\alpha_{1,l}M_{i;k,l}\right|\right]\leq\mbE\left[\left(\sum_{k,l>1}^{T-1}\alpha_{1,k}^2\alpha_{1,l}^2\right)^{1/2}\times\left(\sum_{k,l>1}^{T-1}\frac{1}{n}\left(\sum_{i=1}^nM_{i;k,l}\right)^2\right)^{1/2}\right]\\
		&\leq\mbE\left[\left(1-\alpha_{1,1}^2\right)\times\left(\sum_{i=1}^n\sum_{k,l>1}^{T-1}M_{i;k,l}^2\right)^{1/2}\right]\leq\sqrt{n}\mbE\left[1-\alpha_{1,1}^2\right]\longrightarrow0.
	\end{align*}
	Moreover, by (\ref{Eq of mean difference}) and (\ref{Eq of finite variance}), we can further obtain that
	\begin{align*}
		&\frac{1}{\sqrt{n}}\mbE\left[\left|\sum_{l>1}^{T-1}\sum_{i=1}^n\alpha_{1,1}\alpha_{1,l}M_{i;1,l}\right|\right]\leq\mbE[1-\alpha_{1,1}^2]^{1/2}\times\mbE\left[\frac{1}{n}\sum_{l>1}^{T-1}\left(\sum_{i=1}^nM_{i;1,l}\right)^2\right]^{1/2}\\
		&\leq\mbE[1-\alpha_{1,1}^2]^{1/2}\times\left(\sum_{l>1}^{T-1}\frac{1}{n}\Var\left(\sum_{i=1}^nM_{i;1,l}\right)+\frac{1}{n^2}\mbE\left[\sum_{i=1}^nM_{i;1,l}\right]^2\right)^{1/2}\\
		&\leq\mrO\left(\mbE[1-\alpha_{1,1}^2]^{1/2}\right)\longrightarrow0.
	\end{align*}
	Now, we obtain that
    $$\frac{\sqrt{n}}{\mfm_{1,1}}\left(\frac{\hla_1}{n}-\mbE[\fM_{1,1}]\right)\overset{d}{\longrightarrow}\mcN(0,1).$$
    Finally, by the above arguments, we show that for $1\leq k\leq K$
    $$\sqrt{n}\left(\frac{\hla_k}{n}-\mbE[\fM_{k,k}]\right)\overset{\mbP}{\longrightarrow}\frac{1}{\sqrt{n}}\sum_{i=1}^n(\hat{M}_{i;k,k}-\mbE[\fM_{k,k}]).$$
    To establish the joint CLT for $(\hla_1,\cdots,\hla_K)'$, it suffices to establish the joint CLT for
    $$\frac{1}{\sqrt{n}}\left(\sum_{i=1}^n(\hat{M}_{i;1,1}-\mbE[\fM_{1,1}]),\cdots,\sum_{i=1}^n(\hat{M}_{i;K,K}-\mbE[\fM_{K,K}])\right)',$$
    whose covariance matrix is $\bbA_n$ defined in Theorem \ref{Thm of CLT I1}. Since $\liminf_{n\to\infty}\sigma_{\min}(\bbA_n)>0$, where $\sigma_{\min}(\bbA_n)$ is the smallest singular value of $\bbA_n$, then $\bbA_n$ is positive semi-definite and $\bbA_n^{-1/2}$ always exists. Consequently, for any unit $K$-dimensional vector $\bba=(a_1,\cdots,a_K)'$, we can use the same Lindeberg's condition for $m$-dependent random variables as in Lemma \ref{Lem of preliminary CLT} to show that
    $$\frac{1}{\sqrt{n\cdot \bba'\bbA_n\bba}}\sum_{k=1}^Ka_k\sum_{i=1}^n(\hat{M}_{i;k,k}-\mbE[\fM_{k,k}])\overset{d}{\longrightarrow}\mcN(0,1),$$
    we omit details here to save space, so we conclude the joint CLT in Theorem \ref{Thm of CLT I1}.
\end{proof}
\section{CLT for extreme eigenvalues of the sample covariance matrix of high-dimensional random walks}\label{Sec of covariance}
\setcounter{equation}{0}
\def\theequation{\thesection.\arabic{equation}}
\setcounter{subsection}{0}
Let's consider a $n$-dimensional random walk $X_t$ generated by 
\begin{align}
    X_t=X_{t-1}+e_t,\quad e_t=\sum_{k=0}^{\infty}\Psi_k\varepsilon_{t-k},\label{Eq of random walk covariance}
\end{align}
where $\varepsilon_t=(\varepsilon_{1,t},\cdots,\varepsilon_{n,t})'$ and $\{\Psi_k:k\in\mbN\}$ satisfy that
\begin{Ap}\label{Ap of noise}
	\(\varepsilon_{i,t}\) are independent for all \(1\leq i\leq n\) and \(t\in\mathbb{Z}\) such that \(\mathbb{E}[\varepsilon_{i,t}]=0,\mathbb{E}[\varepsilon_{i,t}^2]=1\) and \(\kappa_6:=\sup_{i,t}\mathbb{E}[\varepsilon_{i,t}^6]<\infty\).
\end{Ap}
\begin{Ap}\label{Ap of effective rank}
	There exists \(B>0\) such that $\sum_{k=0}^{\infty}(1+k)\Vert\Psi_k\Vert\leq B$. Further denote
	\begin{align}
		\Psi(1):=\sum_{k=0}^{\infty}\Psi_k\quad{\rm and}\quad\bbW:=\Psi(1)'\Psi(1),\label{Eq of bbW}
	\end{align}
	there exists a positive constant \(b\) such that \(n^{-1}\tr(\bbW)\geq b\), and the effective rank \(\tr(\bbW)/\Vert\bbW\Vert=\mrO(n)\).
\end{Ap}
In this section, for any $K\in\mbN^+$, we establish the joint CLT for the first $K$ largest eigenvalues of the sample covariance of \(\hat{\bbSig}\) in (\ref{Eq of covariance matrix}) of $\bbX=[X_1,\cdots,X_T]$ generated by \eqref{Eq of random walk covariance}. 
\begin{thm}\label{Thm of original CLT}
	Under Assumptions {\rm \ref{Ap of highdimensionality}, \ref{Ap of noise}} and {\rm \ref{Ap of effective rank}}, for any $K\in\mbN^+$ and $1\leq k\leq K$, let \(\tichi_k\) be the \(k\)-th largest eigenvalue of the sample covariance matrix \(\hat{\bbSig}\) in {\rm (\ref{Eq of covariance matrix})} of $\bbX=[X_1,\cdots,X_T]$ generated by \eqref{Eq of random walk covariance}, then
	$$\frac{\tr(\bbW)}{\sqrt{\tr(\bbW^2)}}\left(\frac{\hat{\chi}_1}{n\tr(\bbW)}-\frac{1}{(c\pi)^2},\cdots,\frac{\hat{\chi}_K}{n\tr(\bbW)}-\frac{1}{(c\pi K)^2}\right)'\overset{d}{\longrightarrow}\mcN(\boldsymbol{0},\mcC),$$
    where $\mcC=\diag(\mcC_{1,1},\cdots,\mcC_{K,K})$ is a $K\times K$ diagonal covariance matrix such that $\mcC_{k,k}=2(c\pi k)^{-4}$.
\end{thm}
Here, let's briefly outline the proof of Theorem \ref{Thm of original CLT}. The key idea is to leverage the joint CLT for extreme eigenvalues of the sample covariance matrix of the I(1) terms in Beveridge-Nelson decomposition of \(X_t\). Precisely, the Beveridge-Nelson decomposition of \(X_t\) is given as
\begin{align}
	\left\{\begin{array}{l}
		X_t=\Psi(1)\xi_t+\Psi^*(L)\varepsilon_t,\\
		\bbX=\Psi(1)\bbve\bbU+\Psi^*(L)\bbve,
	\end{array}\right.\label{Eq of BN decomposition}
\end{align}
where \(\Psi^*(L):=\sum_{k=0}^{\infty}\Psi_k^* L^k\) with \(\Psi_k^*:=-\sum_{i=k+1}^{\infty}\Psi_i\) and \(\xi_t:=\sum_{j=1}^t\varepsilon_t\), \(\bbve:=[\varepsilon_1,\cdots,\varepsilon_T]\). Then the sample covariance of the \(\mathrm{I}(1)\) term in the Beveridge-Nelson decomposition of \(\bbX\bbM\) is defined as
\begin{align}
	\tilde{\bbSig}:=\frac{1}{n}\bbM\bbU'\bbve'\Psi(1)'\Psi(1)\bbve\bbU\bbM=\frac{1}{n}\bbM\bbU'\bbve'\bbW\bbve\bbU\bbM.\label{Eq of BN covariance matrix}
\end{align}
Denote \(\hat{\chi}_1\geq\cdots\geq\hat{\chi}_T\) and \(\tilde{\chi}_1\geq\cdots\geq\tilde{\chi}_T\) to be the eigenvalues of \(\hat{\bbSig}\) and \(\tilde{\bbSig}\) respectively, with the corresponding normalized eigenvectors \(\hat{H}_1,\cdots,\hat{H}_T\) and \(\tilde{H}_1,\cdots,\tilde{H}_T\). Notice that \(\{\bbw_1,\cdots,\bbw_T\}\) in (\ref{Eq of vk}) forms an orthogonal basis of \(\mbR^T\), similar as (\ref{Eq of la1}), we have
\begin{align}
    \tichi_k=\tilde{H}_k'\tilde{\bbSig}\tilde{H}_k=\sum_{s,t=1}^{T-1}\sigma_s\sigma_t\alpha_{k,s}\alpha_{k,t}\bbv_s'\bbve'\bbW\bbve\bbv_t,\quad\tilde{H}_k:=\sum_{t=1}^{T-1}\alpha_{k,t}\bbw_t,\label{Eq of H1}
\end{align}
where \(\sum_{k=1}^{T-1}\alpha_{k,t}^2=1\) for \(1\leq k\leq T-1\). 

Now, based on $\tilde{\bbSig}$ in \eqref{Eq of BN covariance matrix}, we can prove Theorem \ref{Thm of original CLT} by the following two steps:
\begin{enumerate}
    \item Establish the joint CLT for $(\tichi_1,\cdots,\tichi_K)'$;
    \item Show that $n^{-3/2}|\tichi_k-\chi_k|\overset{\mbP}{\longrightarrow}0$.
\end{enumerate}
In the end, as a useful tool for our proof, we cite the following result:
\begin{lem}[Lemma 2 in \cite{onatski2021spurious}]\label{Lem of covariance}
	Under Assumptions {\rm \ref{Ap of highdimensionality}} and {\rm \ref{Ap of noise}}, let \(\bbve=[\varepsilon_1,\cdots,\varepsilon_T]\) to be the noise matrix such that \(\bbve_j\) be the \(j\)-th {\bf row} of \(\bbve\), further let \(\vec{a}_i\) and \(\bbA\) be any deterministic \(T\)-dimensional unit vectors for \(i=1,\cdots,4\) and \(n\times n\) matrix, respectively, then
	\begin{align*}
		&\mbE[\vec{a}_1'\bbve'\bbA\bbve\vec{a}_2]=\langle\vec{a}_1,\vec{a}_2\rangle\tr(\bbA),\\
		&\left|\Cov\big(\vec{a}_1'\bbve'\bbA\bbve\vec{a}_2,\vec{a}_3'\bbve'\bbA\bbve\vec{a}_4\big)-\Vert\bbA\Vert_F^2\left(\langle\vec{a}_1,\vec{a}_3\rangle\langle\vec{a}_2,\vec{a}_4\rangle+\langle\vec{a}_1,\vec{a}_4\rangle\langle\vec{a}_2,\vec{a}_3\rangle\right)\right|\\
		&\leq2\kappa_4\tr(\bbA\circ\bbA)\sum_{t=1}^T|\vec{a}_{1,t}\vec{a}_{2,t}\vec{a}_{3,t}\vec{a}_{4,t}|.
	\end{align*}
\end{lem}

\subsection{Joint CLT for extreme eigenvalues for the sample covariance matrix in the BN decomposition}\label{sec of CLT tisig}
\begin{pro}\label{Thm of pre CLT}
	Under Assumptions {\rm \ref{Ap of highdimensionality}, \ref{Ap of noise}, \ref{Ap of effective rank}} and {\rm \ref{Ap of effective rank}}, for any \(K\in\mbN^+\) and $1\leq k\leq K$, let $\tichi_k$ be the $k$-th largest eigenvalue of $\tilde{\bbSig}$ defined in \eqref{Eq of BN covariance matrix}, then
	\begin{align}
		\frac{\tr(\bbW)}{\sqrt{\tr(\bbW^2)}}\left(\frac{\tichi_1}{n\tr(\bbW)}-\frac{1}{(c\pi)^2},\cdots,\frac{\tichi_K}{n\tr(\bbW)}-\frac{1}{(c\pi K)^2}\right)'\overset{d}{\longrightarrow}\mcN(\boldsymbol{0},\mcC),\label{Eq of tilde joint CLT}
	\end{align}
	where \(\tichi_k\) is defined in {\rm (\ref{Eq of H1})} and \(\mcC=\diag(\mcC_{1,1},\cdots,\mcC_{K,K})\) is a \(K\times K\) diagonal matrix such that \(\mcC_{k,k}=2(c\pi k)^{-4}\).
\end{pro}
To prove Proposition \ref{Thm of pre CLT}, we will use the same frameworks as those for Theorems \ref{Thm of CLT} and \ref{Thm of CLT I1}. Basically, recall the $\alpha_{k,t}$ defined in \eqref{Eq of H1}, we will first show that
$$\lim_{n\to\infty}\sqrt{n}\mbE[1-\alpha_{k,k}^2]=0.$$
Therefore, we can conclude that
$$\sqrt{n}\left|\frac{\tichi_k}{n\tr(\bbW)}-\frac{1}{n}\bbv_k'\bbve'\tilde{\bbW}\bbve\bbv_k\right|\overset{\mbP}{\longrightarrow}0,$$
so it suffices to establish the joint CLT for $n^{-1}(\bbv_1'\bbve'\tilde{\bbW}\bbve\bbv_1,\cdots,\bbv_K'\bbve'\tilde{\bbW}\bbve\bbv_K)'$ to prove Proposition \ref{Thm of pre CLT}. In the following two parts, we will prove Proposition \ref{Thm of pre CLT} based on the above outline.
\subsubsection{Preliminary CLT}\label{sec of preliminary CLT covariance}
Here, we denote some necessary notations. Let
\begin{align}
	\tilde{\bbW}:=\bbW/\tr(\bbW)\label{Eq of tiW}
\end{align} 
by Assumption \ref{Ap of effective rank}, it yields that 
$$\tr(\tilde{\bbW})=1\quad{\rm and}\quad\Vert\tilde{\bbW}\Vert=\mathrm{O}(n^{-1}).$$
Moreover, by Lemma \ref{Lem of covariance}, for any \(k_1,k_2,t_1,t_2\in\{1,\cdots,T\}\), we have
\begin{itemize}
	\item \(\mathbb{E}[y_{k_1,t_1}]=\delta_{k_1,t_1}\), where \(\delta_{k_1,t_1}\) is the Kronecker delta;
	\item \(\big|\Cov(y_{k_1,t_1},y_{k_2,t_2})-\tr(\tilde{\bbW}^2)(\delta_{k_1,k_2}\delta_{t_1,t_2}+\delta_{k_1,t_2}\delta_{k_2,t_1})\big|\leq8\kappa_4 T^{-1}\tr(\tilde{\bbW}^2)\).
\end{itemize}
In this part, we first establish the joint CLT for $n^{-1}(\bbv_1'\bbve'\tilde{\bbW}\bbve\bbv_1,\cdots,\bbv_K'\bbve'\tilde{\bbW}\bbve\bbv_K)'$ as follows:
\begin{lem}\label{Lem of pre CLT}
	Under Assumptions {\rm \ref{Ap of highdimensionality}, \ref{Ap of noise}} and {\rm \ref{Ap of effective rank}}, for any \(K\in\mathbb{N}^+\), let
    \begin{align}
	   x_{j,k}:=\bbve_j\bbv_k\quad{\rm and}\quad y_{k,t}:=\bbv_k'\bbve'\tilde{\bbW}\bbve\bbv_t,\label{Eq of x y}
    \end{align}
    then
	\begin{align*}
		\frac{\tr(\bbW)}{\sqrt{\tr(\bbW^2)}}\big(y_{1,1}-1,\cdots,y_{K,K}-1\big)'\overset{d}{\longrightarrow}\mcN(\boldsymbol{0},2\bbI_K).
	\end{align*}
\end{lem}
\begin{proof}
	Let \(\bba=(a_1,\cdots,a_K)'\) be any \(K\)-dimensional unit vector, it is enough to show that
	$$\frac{\tr(\bbW)}{\sqrt{\tr(\bbW^2)}}\sum_{k=1}^Ka_k(y_{k,k}-1)\overset{d}{\longrightarrow}\mcN(0,2).$$
	Define a sequence of sigma fields \(\mcF_{l,K}:=\sigma\{x_{i,k}:1\leq i\leq l,1\leq k\leq K\}\) for \(l=0,1,\cdots,n-1\) and \(\mcF_{0,K}=\emptyset\), then
	\begin{align*}
		&\sqrt{n}\sum_{k=1}^Ka_k(y_{k,k}-1)=\sqrt{n}\sum_{i,j=1}^n\tilde{W}_{i,j}\sum_{k=1}^Ka_k(x_{i,k}x_{j,k}-\delta_{i,j})\\
		&=\sum_{l=1}^n\sqrt{n}\sum_{k=1}^Ka_k\left(\tilde{W}_{l,l}(x_{l,k}^2-1)+2\sum_{r=1}^{l-1}\tilde{W}_{l,r}x_{l,k}x_{r,k}\right):=\sum_{l=1}^nH_{l,K},
	\end{align*}
	where \(\{(H_{l,K},\mcF_{l,K}):l=1,\cdots,n\}\) is a sequence of martingale differences due to \(\mbE[H_{l,K}|\mcF_{l-1,K}]=0\). By the CLT of martingale differences (see \cite{alj2014conditions}), we first need to compute
	\begin{align*}
		&\sum_{l=1}^n\mbE[H_{l,K}^2|\mcF_{l-1,K}]=\sum_{l=1}^n\mbE\left[\left(\sum_{k=1}^Ka_k\left(\tilde{W}_{l,l}(x_{l,k}^2-1)+2\sum_{r=1}^{l-1}\tilde{W}_{l,r}x_{l,k}x_{r,k}\right)\right)^2\right]\\
		&=\sum_{l=1}^n\sum_{k=1}^Ka_k^2\mbE\left[\left(\tilde{W}_{l,l}(x_{l,k}^2-1)+2\sum_{r=1}^{l-1}\tilde{W}_{l,r}x_{l,k}x_{r,k}\right)^2\right]\\
		&=\sum_{l=1}^n\sum_{k=1}^K\left(a_k^2\tilde{W}_{l,l}^2\mbE\big[(x_{l,k}^2-1)^2\big]+4\sum_{r=1}^{l-1}\tilde{W}_{l,r}^2\right)\\
		&=2n\tr(\tiW^2)\sum_{k=1}^Ka_k^2+n\sum_{l=1}^n\sum_{k=1}^Ka_k^2\tilde{W}_{l,l}^2\left(\mbE\big[x_{l,k}^4\big]-3\right)\longrightarrow2n\tr(\tiW^2),
	\end{align*}
	where we use \(\mbE\big[x_{i_1,k_1}x_{i_2,k_2}\big]=\delta_{i_1,i_2}\delta_{k_1,k_2}\) and \(n\sum_{l=1}^n\tilde{W}_{l,l}^2\leq\mrO(1)\) and
	\begin{align*}
		&\mbE\big[x_{l,k}^4\big]=\mbE\left[\left(\sum_{t=1}^T\varepsilon_{l,t}v_{t,k}\right)^4\right]=3\sum_{t_1,t_2=1}^Tv_{t_1,k}^2v_{t_2,k}^2\mbE\big[\varepsilon_{t_1,k}^2\varepsilon_{t_2,k}^2\big]+\sum_{t=1}^Tv_{t,k}^4\big(\mbE[\varepsilon_{t,k}^4]-3\big)\\
		&=3+C_{\kappa_4}T^{-1}.
	\end{align*}
	Next, let's verify that all \(H_{l,K}\) satisfy the Lyapounov’s condition, i.e. \(\max_l\mbE[|\sqrt{n}H_{l,K}|^3]\leq\mrO(1)\). By the Hölder's inequality, we have
	\begin{align*}
		&\mbE\big[|H_{l,K}|^3\big]=n^3\mbE\left[\left|\sum_{k=1}^Ka_k\left(\tilde{W}_{l,l}(x_{l,k}^2-1)+2\sum_{r=1}^{l-1}\tilde{W}_{l,r}x_{l,k}x_{r,k}\right)\right|^3\right]\\
		&\leq n^3\left(\sum_{k=1}^K|a_k|^{3/2}\right)^2\sum_{k=1}^K\mbE\left[\left|\tilde{W}_{l,l}(x_{l,k}^2-1)+2\sum_{r=1}^{l-1}\tilde{W}_{l,r}x_{l,k}x_{r,k}\right|^3\right]\\
		&\leq4K^2n^3\tilde{W}_{l,l}^3\sum_{k=1}^K\mbE\big[|x_{l,k}^2-1|^3\big]+32K^2n^3\sum_{k=1}^K\mbE\big[|x_{l,k}|^3\big]\mbE\left[\left|\sum_{r=1}^{l-1}\tilde{W}_{l,r}x_{r,k}\right|^3\right].
	\end{align*}
	Since
	\begin{align}
		&\mathbb{E}\left[x_{j,k}^6\right]=\mathbb{E}\left[\left(\sum_{s=1}^T\varepsilon_{j,s}v_{s,k}\right)^6\right]=15\sum_{s_1,s_2,s_3=1}^T |v_{s_1,k}v_{s_2,k}v_{s_3,k}|^2\mathbb{E}\big[|\varepsilon_{j,s_1}\varepsilon_{j,s_2}\varepsilon_{j,s_3}|^2\big]\notag\\
		&\leq15\kappa_6\sum_{s_1,s_2,s_3=1}^T |v_{s_1,k}v_{s_2,k}v_{s_3,k}|^2\leq120\kappa_6,\notag
	\end{align}
	where we use the fact that \(\mathbb{E}\left[\prod_{l=1}^8\varepsilon_{j,s_l}\right]=0\) if there is one \(s_l\) different with others, and
	\begin{align*}
		&\mbE\left[\left|\sum_{r=1}^{l-1}\tilde{W}_{l,r}x_{r,k}\right|^6\right]=15\sum_{r_1,r_2,r_3=1}^{l-1}\big|\tilde{W}_{l,r_1}\tilde{W}_{l,r_2}\tilde{W}_{l,r_3}\big|^2\mbE\big[\big|x_{r_1,k}x_{r_2,k}x_{r_3,k}\big|^2\big]\leq C_{\kappa_6}\Vert\tiW\Vert^6,
	\end{align*}
	further combining with Assumption \ref{Ap of effective rank}, we can obtain that	
	\begin{align*}
		&\mbE\big[|H_{l,K}|^3\big]\leq C_{K,\kappa_6}+C_{K,\kappa_6}n^3\sum_{k=1}^K\mbE\left[\left|\sum_{r=1}^{l-1}\tilde{W}_{l,r}x_{r,k}\right|^6\right]^{1/2}\\
		&\leq C_{K,\kappa_6}+C_{K,\kappa_6}n^3\Vert\tiW\Vert^3\leq C_{K,\kappa_6}, 
	\end{align*}
	then we fulfill the two conditions in Theorem 1.3 of \cite{alj2014conditions}, so we conclude that
	\begin{align*}
		&\frac{\sqrt{n}\sum_{k=1}^Ka_k(y_{k,k}-1)}{\sqrt{2n\tr(\tiW^2)}}\overset{d}{\longrightarrow}\mcN(0,1),
	\end{align*}
	which completes our proof.
\end{proof}
\subsubsection{Asymptotic behavior of eigenvectors}\label{sec of alpha kt}
In this part, we will establish the asymptotic behaviors of \(\alpha_{t,k}\) defined in (\ref{Eq of H1}),
\begin{lem}\label{Lem of alpha convergence rate}
	Under Assumptions {\rm \ref{Ap of highdimensionality}, \ref{Ap of noise}} and {\rm \ref{Ap of effective rank}}, recall \(\sigma_k\) defined in {\rm (\ref{Eq of SVD of MU})}, define 
	\begin{align}
		\tisig_k=\frac{\sigma_k}{n}=\frac{1}{2n\sin(\pi k/(2T))},\label{Eq of tilde sigma}
	\end{align}	 
	then for any \(K\in\mbN^+\) and \(1\leq k\leq K\), we have 
	$$\lim_{n\to\infty}\sqrt{n}\mbE[1-\alpha_{k,k}^2]=0,$$
	where \(\alpha_{k,k}\) is defined in {\rm (\ref{Eq of H1})}.
\end{lem}
\begin{proof}
    First, let's show that \(\sqrt{n}\mbE[1-\alpha_{1,1}]=0\) through the following two steps 
    \begin{itemize}
        \item Similar as (\ref{Eq of A1 A2}), we have \(\tichi_1/(n\tr(\bbW))=n^{-2}(\vec{\gamma}_1'\tilde{\bbW}\vec{\gamma}_1+\vec{\gamma}_2'\tilde{\bbW}\vec{\gamma}_2+2\vec{\gamma}_1'\tilde{\bbW}\vec{\gamma}_2)\) by (\ref{Eq of H1}), where
    \begin{align}
	   \vec{\gamma}:=\sum_{k=1}^{T-1}\alpha_{1,k}\sigma_k\bbve\bbv_k=\sum_{k=1}^{N_K}\cdots+\sum_{k=N_K+1}^{T-1}\cdots:=\vec{\gamma}_1+\vec{\gamma}_2,\label{Eq of gamma}
    \end{align}
    and \(N_K\) is a pre-specified integer only depending on \(K\), and we will show that there exists a constant \(C_1>0\) such that
    \begin{align}
        &\mathbb{P}\big(|\tichi_1/(n\tr(\bbW))-\tisig_1^2y_{1,1}|\leq C_1N_K^{-1/2}\big)\geq1-\mathrm{O}(n^{-3/5}).\label{Eq of tichi1 upper bound}
    \end{align}
    Since
    \begin{align*}
        &|\tichi_1/(n\tr(\bbW))-\tisig_1^2y_{1,1}|\leq|n^{-2}\vec{\gamma}_1'\tilde{\bbW}\vec{\gamma}_1-\tisig_1^2y_{1,1}|+n^{-2}|\vec{\gamma}_2'\tilde{\bbW}\vec{\gamma}_2|+2n^{-2}|\vec{\gamma}_2'\tilde{\bbW}\vec{\gamma}_1|,
    \end{align*}
    let's first show that 
	\begin{align}
		\mathbb{P}\big(n^{-2}\vec{\gamma}_2\tilde{\bbW}\vec{\gamma}_2>C_cN_K^{-1}+\epsilon\big)<\epsilon^{-2}\mathrm{O}(n^{-1}).\label{Eq of Chebyshev 1}
	\end{align}
	Notice that 
	\begin{align}
		&n^{-2}\big|\vec{\gamma}_2'\tiW\vec{\gamma}_2\big|\leq\Vert\tiW\Vert\cdot n^{-2}\Vert\vec{\gamma}_2\Vert_2^2\leq\Vert\tiW\Vert\cdot\sum_{j=1}^n\left(\sum_{k=N_K+1}^{T-1}\alpha_{1,k}\tisig_k x_{j,k}\right)^2\notag\\
		&\leq\Vert\tiW\Vert\cdot(1-\alpha_{1,1}^2)\sum_{j=1}^n\sum_{k=N_K+1}^{T-1}\tisig_k^2x_{j,k}^2\leq(1-\alpha_{1,1}^2)\frac{C}{n}\sum_{k=N_K+1}^{T-1}\tisig_k^2\sum_{j=1}^nx_{j,k}^2,\label{Eq of gamma2 quadratic}
	\end{align}
	where we use the Cauchy's inequality and \(\Vert\tiW\Vert\leq\mrO(n^{-1})\). Since \(\tisig_k\leq(T/2n)k^{-1}\), then by Lemma \ref{Lem of covariance} and (\ref{Eq of x y}), it concludes that
    \begin{align}
		&\frac{1}{n}\sum_{k=N_K+1}^{T-1}\tisig_k^2\sum_{j=1}^n\mbE[x_{j,k}^2]=\sum_{k=N_K+1}^{T-1}\tisig_k^2\leq C_cN_K^{-1}.\label{Eq of gamma2 quadratic mean}
	\end{align}
    Similarly, we have
	\begin{align}
		&n^{-2}\Var\left(\sum_{k=N_K+1}^{T-1}\tisig_k^2\sum_{j=1}^nx_{j,k}^2\right)=n^{-2}\sum_{k_1,k_2=N_K+1}^{T-1}\sum_{j_1,j_2=1}^n\tisig_{k_1}^2\tisig_{k_2}^2\Cov\big(x_{j_1,k_1}^2,x_{j_2,k_2}^2\big)\notag\\
		&\leq C_{\kappa_4}n^{-2}\sum_{k_1,k_2=N_K+1}^{T-1}\sum_{j_1,j_2=1}^n\tisig_{k_1}^2\tisig_{k_2}^2\delta_{j_1,j_2}\langle\bbv_{k_1},\bbv_{k_2}\rangle^2\leq C_{\kappa_4}n^{-1}N_K^{-2}.\label{Eq of gamma2 quadratic variance}
	\end{align}
	Then we can conclude (\ref{Eq of Chebyshev 1}) by the Chebyshev's inequality. Moreover, by (\ref{Eq of x y}), we have that
	$$n^{-2}\vec{\gamma}_1'\tilde{\bbW}\vec{\gamma}_1=\sum_{t_1,t_2=1}^{N_K}\alpha_{1,t_1}\alpha_{1,t_2}\tisig_{t_1}\tisig_{t_2}y_{t_1,t_2},$$
	by Lemma \ref{Lem of covariance}, Chebyshev's inequality and Assumption \ref{Ap of effective rank}, we can conclude that
	$$\mbP\left(\big|y_{t_1,t_2}-\delta_{t_1,t_2}\big|>\epsilon\right)<\frac{3\tr(\bbW^2)}{\epsilon^2\tr(\bbW)^2}\leq\frac{3n\Vert\bbW\Vert^2}{\epsilon^2\tr(\bbW)^2}=\epsilon^{-2}\mrO(n^{-1}).$$
	Thus, it gives that
	\begin{align}
		&\mbP\left(n^{-2}\vec{\gamma}_1'\tilde{\bbW}\vec{\gamma}_1\leq\sum_{t=1}^{N_K}\alpha_{1,t}^2\tisig_t^2+C_K\epsilon\right)\geq1-C_K\epsilon^{-2}\mathrm{O}(n^{-1}).\label{Eq of Chebyshev 2}
	\end{align}
    Since \(\tiW\) is positive semi-definite, by the Cauchy's inequality, (\ref{Eq of Chebyshev 1}) and (\ref{Eq of Chebyshev 2}), we have
	$$n^{-2}|\vec{\gamma}_1'\tilde{\bbW}\vec{\gamma}_2|\leq|n^{-2}\vec{\gamma}_1'\tilde{\bbW}\vec{\gamma}_1|^{1/2}|n^{-2}\vec{\gamma}_2'\tilde{\bbW}\vec{\gamma}_2|^{1/2}\leq C_c N_K^{-1/2},$$
	where we use the fact that \(\sum_{t=1}^{N_K}\alpha_{1,t}^2\tisig_t^2<C_c\) by (\ref{Eq of tilde sigma}), so let \(\epsilon=n^{-1/5}\) in (\ref{Eq of Chebyshev 1}) and (\ref{Eq of Chebyshev 2}), it gives that
	\begin{align}
		&\frac{\tichi_1}{n\tr(\bbW)}=n^{-2}\vec{\gamma}_1'\tilde{\bbW}\vec{\gamma}_1+n^{-2}\vec{\gamma}_2'\tilde{\bbW}\vec{\gamma}_2+2n^{-2}\vec{\gamma}_1'\tilde{\bbW}\vec{\gamma}_2\leq \sum_{t=1}^{N_K}\alpha_{1,t}^2\tisig_t^2+C_cN_K^{-1}\notag\\
		&+C_cN_K^{-1/2}+C_Kn^{-1/5}\leq\sum_{t=1}^{N_K}\alpha_{1,t}^2\tisig_t^2+C_c N_K^{-1/2},\label{Eq of upper bound}
	\end{align}
	with probability at least of \(1-\mathrm{O}(n^{-3/5})\). On the other hand, since
	\begin{align}
		\frac{\tichi_1}{n\tr(\bbW)}\geq\tisig_1^2y_{1,1}\geq \tisig_1^2-n^{-1/5},\label{Eq of lower bound}
	\end{align}
	with probability at least of \(1-\mathrm{O}(n^{-3/5})\), then combine with (\ref{Eq of upper bound}) and (\ref{Eq of lower bound}), it gives that
	$$(1-\alpha_{1,1}^2)\tisig_1^2\leq\sum_{t=2}^{N_K}\alpha_{1,t}^2\tisig_t^2+C_cN_K^{-1/2}\leq(1-\alpha_{1,1}^2)\tisig_2^2+C_cK^{-1/2}$$
	with probability at least of \(1-\mathrm{O}(n^{-3/5})\), i.e.
	$$1-\alpha_{1,1}^2\leq C_c(\tisig_1^2-\tisig_2^2)^{-1}N_K^{-1/2}$$
	with probability at least of \(1-\mathrm{O}(n^{-3/5})\). Finally, since
	$$n^{-2}\vec{\gamma}_1'\tilde{\bbW}\vec{\gamma}_1-\tisig_1^2y_{1,1}=(1-\alpha_{1,1}^2)\tisig_1^2y_{1,1}+\sum_{k\neq1\ {\rm or\ }t\neq1}^{N_K}\alpha_{1,k}\alpha_{1,t}\tisig_k\tisig_ty_{k,t},$$
	we can conclude that
	\begin{align*}
		&\big|n^{-2}\vec{\gamma}_1'\tilde{\bbW}\vec{\gamma}_1-\tisig_1^2y_{1,1}\big|\leq(1-\alpha_{1,1}^2)\tisig_1^2+\sum_{k=2}^{N_K}\alpha_{1,k}^2\tisig_k^2+C_K n^{-1/5}\\
		&\leq (1-\alpha_{1,1}^2)\sum_{k=1}^{N_K}\tisig_k^2\leq C_c(\tisig_1^2-\tisig_2^2)^{-1}N_K^{-1/2}
	\end{align*}
	with probability at least of \(1-\mathrm{O}(n^{-3/5})\). Combining with (\ref{Eq of Chebyshev 1}), it yields that
	\begin{align}
		&\big|\tichi_1/(n\tr(\bbW))-\tisig_1^2y_{1,1}\big|\notag\\
		&\leq\big|n^{-2}\vec{\gamma}_1'\tilde{\bbW}\vec{\gamma}_1-\tisig_1^2y_{1,1}\big|+\big|n^{-2}\vec{\gamma}_2'\tilde{\bbW}\vec{\gamma}_2\big|+2\big|n^{-2}\vec{\gamma}_2'\tilde{\bbW}\vec{\gamma}_2\big|^{1/2}\big|n^{-2}\vec{\gamma}_1'\tilde{\bbW}\vec{\gamma}_1\big|^{1/2}\notag\\
		&\leq C_c(\tisig_1^2-\tisig_2^2)^{-1}N_K^{-1/2}+C_cN_K^{-1/2}:=C_1N_K^{-1/2}\notag
	\end{align}
	with probability greater than \(1-\mathrm{O}(n^{-3/5})\).
	\item According to (\ref{Eq of H1}) and (\ref{Eq of x y}), we have
	$$\frac{\alpha_{1,k}\tichi_1}{\tr(\bbW)}=\frac{1}{n}\bbw_k'\bbM\bbU'\bbve'\tilde{\bbW}\bbve\bbU\bbM\tiF_1=\frac{1}{n}\sigma_k\sum_{t=1}^{T-1}\sigma_t\alpha_{1,t}\bbv_k'\bbve'\tilde{\bbW}\bbve\bbv_t,$$
	i.e.
	\begin{align}
		\sqrt{n}\alpha_{1,k}^2=\frac{\tisig_k^2\sqrt{n}(\sum_{t\neq k}^{T-1}\tisig_t\alpha_{1,t}y_{k,t})^2}{(\tichi_1/(n\tr(\bbW))-\tisig_k^2y_{k,k})^2},\label{Eq of alpha}
	\end{align}
	where \(\tisig_k=n^{-1}\sigma_k\) and \(y_{k,t}=\bbv_k'\bbve'\tilde{\bbW}\bbve\bbv_t\). First, for the denominator of (\ref{Eq of alpha}), by (\ref{Eq of tichi1 upper bound}) and \(\mathbb{P}\big(|y_{k,k}-1|>\epsilon\big)\leq \epsilon^{-2}\mathrm{O}(n^{-1})\), we have for a sufficiently large \(N_K\)
	\begin{align*}
		&|\tichi_1/(n\tr(\bbW))-\tisig_k^2y_{k,k}|\geq|\tisig_1^2y_{1,1}-\tisig_k^2y_{k,k}|-|\tichi_1/(n\tr(\bbW))-\tisig_1^2y_{1,1}|\\
		&\geq\tisig_1^2-\tisig_k^2-C_1 N_K^{-1/2}>(\tisig_1^2-\tisig_2^2)/2:=\epsilon_{1,2}
	\end{align*}
	with probability at least of \(1-\mathrm{O}(n^{-3/5})\). Let's define an event \(\mathcal{E}_{1,k}:=\{|\tichi_1/(n\tr(\bbW))-\tisig_k^2y_{k,k}|>\epsilon_{1,2}\}\), then \(\mathbb{P}(\mathcal{E}_{1,k}^c)\leq\mathrm{O}(n^{-3/5})\) for \(k=2,\cdots,N_K\) and
	\begin{align}
		&\sqrt{n}\mathbb{E}[\alpha_{1,k}^2]=\sqrt{n}\mathbb{E}[\alpha_{1,k}^2|\mathcal{E}_{1,k}]\mathbb{P}(\mcE_{1,k})+\sqrt{n}\mathcal{E}[\alpha_{1,k}^2|\mathcal{E}_{1,k}^c]\mathbb{P}(\mathcal{E}_{1,k}^c)\notag\\
		&\leq\mathrm{O}(n^{-1/10})+\epsilon_{1,2}^{-2}\tisig_k^2\sqrt{n}\mathbb{E}\left[\left(\sum_{t\neq k}^{T-1}\tisig_t\alpha_{1,t}y_{k,t}\right)^2\right].\label{Eq of sum numerator 1}
	\end{align}
	Next, for \(k>N_K\), let's define
	\begin{align}
		A_K:=\sum_{s=N_K+1}^{T-1}\tisig_s^2y_{s,s}\quad{\rm and}\quad\mathcal{E}_{1,N_K+1}:=\big\{\big|\tichi_1/(n\tr(\bbW))-A_K\big|>\epsilon_{1,2}\big\}.\label{Eq of AK}
	\end{align}
	Since \(\mathbb{E}[A_K]=\sum_{s=N_K+1}^{T-1}\tisig_s^2\leq CN_K^{-1}\) and
	\begin{align}
		&\operatorname{Var}(A_K)=\sum_{s,r=N_K+1}^{T-1}\tisig_s^2\tisig_r^2\operatorname{Cov}(y_{s,s},y_{r,r})\leq C_{\kappa_4}\sum_{s,r=N_K+1}^{T-1}\tisig_s^2\tisig_r^2\delta_{s,r}\tr(\tiW^2)\notag\\
		&\leq C_{\kappa_4}\tr(\tilde{\bbW}^2)\sum_{s=N_K+1}^{T-1}\tisig_s^4\leq CK^{-2}\tr(\tilde{\bbW}^2)=\mathrm{O}(n^{-1})\notag
	\end{align}
	by Lemma \ref{Lem of covariance} and Assumption \ref{Ap of effective rank}, we can let \(N_K\) be sufficiently large such that \(\epsilon_{1,2}<\tisig_1^2-\mathbb{E}[A_K]\) and induce that
	$$\big|\tichi_1/(n\tr(\bbW))-A_K\big|\geq\tisig_1^2-\mathbb{E}[A_K]-C_{\kappa_4}N_K^{-1/2}>\epsilon_{1,2}$$
	with probability at least of \(1-\mathrm{O}(n^{-3/5})\). Then we obtain that
	\begin{align}
		&\sqrt{n}\sum_{k=N_K+1}^{T-1}\mathbb{E}[\alpha_{1,k}^2]\notag\\
		&=\sqrt{n}\sum_{k=N_K+1}^{T-1}\mathbb{E}[\alpha_{1,k}^2|\mathcal{E}_{1,N_K+1}]\mathbb{P}(\mathcal{E}_{1,N_K+1})+\sqrt{n}\sum_{k=N_K+1}^{T-1}\mathbb{E}[\alpha_{1,k}^2|\mathcal{E}_{1,N_K+1}^c]\mathbb{P}(\mathcal{E}_{1,N_K+1}^c)\notag\\
		&\leq\mathrm{O}(n^{-1/10})+\epsilon_{1,2}^{-2}\sum_{k=K+1}^{T-1}\tisig_k^2\sqrt{n}\mathbb{E}\left[\left(\sum_{t\neq k}^{T-1}\tisig_t\alpha_{1,t}y_{k,t}\right)^2\right].\label{Eq of sum numerator 2}
	\end{align}
	Hence, combining with (\ref{Eq of sum numerator 1}) and (\ref{Eq of sum numerator 2}), it is enough to show that the sum of the expectation of numerator in (\ref{Eq of alpha}) over all \(k>1\) converges to zero, i.e.
	$$\lim_{n\to\infty}\sum_{k=2}^{T-1}\tisig_k^2\sqrt{n}\mathbb{E}\left[\left(\sum_{t\neq k}^{T-1}\tisig_t\alpha_{1,t}y_{k,t}\right)^2\right]=0.$$ 
	By the Cauchy's inequality and Lemma \ref{Lem of covariance}, we have that
	\begin{align}
		&\sqrt{n}\mathbb{E}\left[\left(\sum_{t\neq k}^{T-1}\tisig_t\alpha_{1,t}y_{k,t}\right)^2\right]\leq\sqrt{n}\mathbb{E}\left[\left(\sum_{t\neq k}^{T-1}\alpha_{1,t}^2\right)\left(\sum_{t\neq k}^{T-1}\tisig_t^2y_{k,t}^2\right)\right]\leq\sqrt{n}\sum_{t\neq k}^{T-1}\tisig_t^2\mathbb{E}[y_{k,t}^2]\notag\\
		&\leq\sqrt{n}C_{\kappa_4}\tr(\tilde{\bbW}^2)\sum_{t\neq k}^{T-1}\tisig_t^2=\mathrm{O}(n^{-1/2}),\notag
	\end{align}
	where we use the fact that \(\tisig_t\asymp\mathrm{O}(t^{-1})\) (``$\asymp$'' is defined in \eqref{Eq of asymp mbP}) and \(\Vert\tiW\Vert\leq\mrO(n^{-1})\) in the last inequality above. Now, we conclude that
	$$\sum_{k=2}^{T-1}\tisig_k^2\sqrt{n}\mathbb{E}\left[\left(\sum_{t\neq k}^{T-1}\tisig_t\alpha_{1,t}y_{k,t}\right)^2\right]\leq\mathrm{O}(n^{-1/2})\sum_{k=2}^{T-1}\tisig_k^2=\mathrm{O}(n^{-1/2}).$$
	i.e.
	\begin{align*}
		&\sqrt{n}\mathbb{E}[1-\alpha_{1,1}^2]=\sqrt{n}\sum_{k=2}^{T-1}\mathbb{E}[\alpha_{1,k}^2]\\
		&\leq\mathrm{O}(n^{-1/10})+\epsilon_{1,2}^{-2}\sum_{k=2}^{T-1}\tisig_k^2\sqrt{n}\mathbb{E}\left[\left(\sum_{t\neq k}^{T-1}\tisig_t\alpha_{1,t}y_{k,t}\right)^2\right]=\mathrm{O}(n^{-1/10}),
	\end{align*}
    \end{itemize}
	Until now, we obtain that \(\lim_{n\to\infty}\sqrt{n}\mbE[1-\alpha_{1,1}^2]=0\). For general cases, we can inductively prove that \(\lim_{n\to\infty}\sqrt{n}\mbE[1-\alpha_{k,k}^2]=0\) based on \(\lim_{n\to\infty}\sqrt{n}\mbE[1-\alpha_{k-1,k-1}^2]=0\), just as what we have done in \S\ref{sec of alpha kt correlation}. For the choice of \(N_K\), we can still use (\ref{Eq of N_K}), since the arguments are totally the same, we omit the details here to save space.
\end{proof}
\begin{cor}\label{Cor of alpha convergence rate}
	As a consequence of Lemma {\rm \ref{Lem of alpha convergence rate}}, it gives that \(\sqrt{n}(1-\alpha_{k,k}^2)\overset{\mbP}{\longrightarrow}0\) for $1\leq k\leq K$.
\end{cor}
\subsubsection{Proof of Proposition \ref{Thm of pre CLT}}
\begin{proof}
	Let's first show that 
	\begin{align}
		\frac{\tichi_k^{\circ}}{\sqrt{n}\tr(\bbW)}\overset{\mbP}{\longrightarrow}\sqrt{n}\tisig_k^2(y_{k,k}-1).\label{Eq of mbP tichi_k}
	\end{align}
	Without loss of generality, we only provide the detailed proof of \(k=1\), others are totally the same. By (\ref{Eq of gamma}), we have
	$$\frac{\tichi_1^{\circ}}{\sqrt{n}\tr(\bbW)}=n^{-3/2}(\vec{\gamma}_1'\tilde{\bbW}\vec{\gamma}_1+\vec{\gamma}_2'\tilde{\bbW}\vec{\gamma}_2+2\vec{\gamma}_1'\tilde{\bbW}\vec{\gamma}_2)^{\circ}.$$
	Let's first prove that
	\begin{align*}
		n^{-3/2}(\vec{\gamma}_1'\tilde{\bbW}\vec{\gamma}_1)^{\circ}-\sqrt{n}\tisig_1^2y_{1,1}^{\circ}\overset{\mbP}{\longrightarrow}0,
	\end{align*}
	where
	$$n^{-3/2}(\vec{\gamma}_1'\tilde{\bbW}\vec{\gamma}_1)^{\circ}=\sqrt{n}\sum_{k,t=1}^{N_K}\tisig_k\tisig_t\big(\alpha_{1,k}\alpha_{1,t}y_{k,t}^{\circ}+\alpha_{1,k}\alpha_{1,t}\delta_{kt}-\mathbb{E}[\alpha_{1,k}\alpha_{1,t}y_{k,t}]\big).$$
	By Corollary \ref{Cor of alpha convergence rate} and Lemma \ref{Lem of covariance}, we have that \(\sqrt{n}\alpha_{1,k}^2\overset{\mbP}{\longrightarrow}0\) for \(k\geq2\) and \(n\Var(y_{k,t})\leq C_{\kappa_4}n\tr(\tiW^2)\leq\mrO(1)\), it implies that
	$$\sqrt{n}\sum_{k\neq1\ {\rm or\ }t\neq1}^{N_K}\tisig_k\tisig_t\alpha_{1,k}\alpha_{1,t}(y_{k,t}-\delta_{k,t})\quad{\rm and}\quad\sqrt{n}\sum_{k\neq1\ {\rm or\ }t\neq1}^{N_K}\tisig_k\tisig_t\alpha_{1,k}\alpha_{1,t}\delta_{kt}\overset{\mbP}{\longrightarrow}0.$$
	Furthermore, if \(k\neq1\) or \(t\neq1\), since
	\begin{align}
		&\big|\operatorname{Cov}(\alpha_{1,k}\alpha_{1,t},\sqrt{n}y_{k,t})\big|\leq\mathbb{E}[\alpha_{1,k}^2\alpha_{1,t}^2]^{1/2}\operatorname{Var}(\sqrt{n}y_{k,t})^{1/2}\notag\\
		&\leq C_{\kappa_4}\sqrt{n}\tr(\tilde{\bbW}^2)^{1/2}\mathbb{E}[\alpha_{1,k}^2\alpha_{1,t}^2]^{1/2}\longrightarrow0,\label{Eq of covariance of alpha and y}
	\end{align}
	it implies that
	$$\sqrt{n}\sum_{k\neq1\ {\rm or\ }t\neq1}^{N_K}\big|\mathbb{E}[\alpha_{1,k}\alpha_{1,t}y_{k,t}]-\mathbb{E}[\alpha_{1,k}\alpha_{1,t}]\mathbb{E}[y_{k,t}]\big|\to0.$$
	Therefore, we conclude that \(n^{-3/2}(\vec{\gamma}_1'\tilde{\bbW}\vec{\gamma}_1)^{\circ}-\sqrt{n}y_{1,1}^{\circ}\overset{\mbP}{\longrightarrow}0\). Next, we will show that \(n^{-3/2}(\vec{\gamma}_2'\tilde{\bbW}\vec{\gamma}_2)^{\circ}\overset{\mbP}{\longrightarrow}0\). By (\ref{Eq of gamma2 quadratic}), it has that
	\begin{align*}
		n^{-3/2}\vec{\gamma}_2'\tilde{\bbW}\vec{\gamma}_2\leq C\sqrt{n}(1-\alpha_{1,1}^2)\times\frac{1}{n}\sum_{j=1}^n\sum_{k=N_K+1}^{T-1}\tisig_k^2x_{j,k}^2,
	\end{align*}
	where \(\sqrt{n}(1-\alpha_{1,1}^2)\overset{\mbP}{\longrightarrow}0\) by Corollary \ref{Cor of alpha convergence rate}. Combining with (\ref{Eq of gamma2 quadratic mean}) and (\ref{Eq of gamma2 quadratic variance}), we can conclude that
	$$n^{-3/2}\vec{\gamma}_2'\tilde{\bbW}\vec{\gamma}_2\overset{\mbP}{\longrightarrow}0.$$
	Next, by (\ref{Eq of covariance of alpha and y}), it implies that
	\begin{align}
		&\sqrt{n}\mathbb{E}[\vec{\gamma}_2'\tilde{\bbW}\vec{\gamma}_2]\leq\sum_{k,t=N_K+1}^{T-1}\tisig_k\tisig_t\sqrt{n}\big|\mathbb{E}[\alpha_{1,k}\alpha_{1,t}y_{k,t}]\big|\notag\\
		&\leq\sum_{k=N_K+1}^{T-1}\tisig_k^2\sqrt{n}\mathbb{E}[\alpha_{1,k}^2]+2\sqrt{n}\tr(\tilde{\bbW}^2)^{1/2}\sum_{k,t=N_K+1}^{T-1}\tisig_k\tisig_t\mathbb{E}[\alpha_{1,k}^2\alpha_{1,t}^2]^{1/2}\notag\\
		&\leq K^{-1}\sqrt{n}\mathbb{E}[1-\alpha_{1,1}^2]+\mathbb{E}[1-\alpha_{1,1}^2]^{1/2}\sum_{k,t=N_K+1}^{T-1}\tisig_k\tisig_t\to0,\notag
	\end{align}
	where we use the fact that \(\tr(\tilde{\bbW}^2)=\mathrm{O}(n^{-1})\) and \(\sum_{k,t=N_K+1}^{T-1}\tisig_k\tisig_t\leq\big(\sum_{t=1}^{T-1}t^{-1}\big)^2\leq\log^2 T\) and \(\mathbb{E}[1-\alpha_{1,1}^2]\leq\mro(\sqrt{n})\) by Corollary \ref{Cor of alpha convergence rate}. Hence, we obtain that \(\sqrt{n}(\vec{\gamma}_2'\tilde{\bbW}\vec{\gamma}_2)^{\circ}\overset{\mbP}{\longrightarrow}0\). Finally, let's show that \(\sqrt{n}(\vec{\gamma}_1'\tilde{\bbW}\vec{\gamma}_2)^{\circ}\overset{\mbP}{\longrightarrow}0\). Notice that
	\begin{align}
		&\sqrt{n}\vec{\gamma}_1'\tilde{\bbW}\vec{\gamma}_2=\sqrt{n}\sum_{k=1}^{N_K}\sum_{t=N_K+1}^{T-1}\tisig_k\tisig_t\alpha_{1,k}\alpha_{1,t}y_{k,t}\notag
	\end{align}
	By the Cauchy's inequality, we have that
	\begin{align}
		&\sqrt{n}\mbE\big[\big|\vec{\gamma}_1'\tilde{\bbW}\vec{\gamma}_2\big|\big]\leq\sqrt{n}\mbE\left[\sum_{k=1}^{N_K}\sum_{t=N_K+1}^{T-1}\alpha_{1,k}^2\alpha_{1,t}^2\right]^{1/2}\mbE\left[\sum_{k=2}^{N_K}\sum_{t=N_K+1}^{T-1}\tisig_k^2\tisig_t^2y_{k,t}^2\right]^{1/2}\notag\\
		&\leq\big(\sqrt{n}\mbE[1-\alpha_{1,1}^2]\big)^{1/2}\left(\sqrt{n}\sum_{k=1}^{N_K}\sum_{t=N_K+1}^{T-1}\tisig_k^2\tisig_t^2\mbE[y_{k,t}^2]\right)^{1/2}.\notag
	\end{align}
	Since \(\mathbb{E}[y_{k,t}]=0\) for \(k\neq t\) and \(\Var(y_{1,t})\leq C_{\kappa_4}\tr(\tiW^2)\) by Lemma \ref{Lem of covariance}, then \(\mbE[y_{k,t}^2]\leq C_{\kappa_4}\tr(\tiW^2)\) and
	$$\sqrt{n}\sum_{k=1}^{N_K}\sum_{t=N_K+1}^{T-1}\tisig_k^2\tisig_t^2\mbE[y_{k,t}^2]\leq C_{\kappa_4}\sum_{k=1}^{N_K}\sum_{t=N_K+1}^{T-1}\tisig_k^2\tisig_t^2\sqrt{n}\tr(\tiW^2)\leq C_{\kappa_4}n^{-1/2}.$$
	Combining Lemma \ref{Lem of alpha convergence rate} and Markov's inequality, we can conclude that
	\begin{align*}
		&\sqrt{n}\mbE\big[\big|\vec{\gamma}_1'\tilde{\bbW}\vec{\gamma}_2\big|\big]\longrightarrow0\quad{\rm and}\quad\sqrt{n}\vec{\gamma}_1'\tilde{\bbW}\vec{\gamma}_2\overset{\mbP}{\longrightarrow}0,
	\end{align*}
	so (\ref{Eq of mbP tichi_k}) is valid. Next, let's show that
	$$\lim_{n\to\infty}\sqrt{n}\Big|\frac{\mathbb{E}[\tichi_k]}{n\tr(\bbW)}-\tisig_k^2\Big|=0.$$
	Here, we only present the proof for \(k=1\) as well, since others are the same. Note that
	$$\frac{\mathbb{E}[\tichi_1]}{n\tr(\bbW)}-\tisig_1^2=\tisig_1^2\mathbb{E}[\alpha_{1,1}^2y_{1,1}-1]+\sum_{k\neq1\ {\rm or\ }t\neq1}^{T-1}\tisig_k\tisig_t\mathbb{E}[\alpha_{1,k}\alpha_{1,t}y_{k,t}],$$
	then by Lemma \ref{Lem of covariance}, \ref{Lem of alpha convergence rate}, Assumption \ref{Ap of effective rank} and (\ref{Eq of covariance of alpha and y}), we have
	\begin{align*}
		&\sqrt{n}|\mathbb{E}[\alpha_{1,1}^2y_{1,1}-1]|=\sqrt{n}|\mathbb{E}[(1-\alpha_{1,1}^2)y_{1,1}]|\\
		&\leq\sqrt{n}\mathbb{E}[1-\alpha_{1,1}^2]+\sqrt{n}\mathbb{E}[(1-\alpha_{1,1}^2)^2]^{1/2}\operatorname{Var}(y_{1,1})^{1/2}\longrightarrow0,
	\end{align*}
	Similarly, for \(k\neq1\) or \(t\neq1\), by previous arguments, we can also conclude that
	\begin{align}
		&\sqrt{n}\left|\sum_{k\neq1\ {\rm or\ }t\neq1}^{T-1}\tisig_k\tisig_t\mathbb{E}[\alpha_{1,k}\alpha_{1,t}y_{k,t}]\right|\notag\\
		&\leq\sqrt{n}\sum_{k=2}^{T-1}\tisig_k^2\mathbb{E}[\alpha_{1,k}^2]+C\sqrt{n}\mathbb{E}[1-\alpha_{1,1}^2]^{1/2}\tr(\tiW^2)^{1/2}\sum_{k\neq1\ {\rm or\ }t\neq1}^{T-1}\tisig_k\tisig_t\notag\\
		&\leq C\sqrt{n}\mathbb{E}[1-\alpha_{1,1}^2]+C\mathbb{E}[\sqrt{n}(1-\alpha_{1,1}^2)]^{1/2}n^{-1/4}\log^2 T\to0.\notag
	\end{align}
	Hence, we obtain that
	$$\sqrt{n}\Bigg|\frac{\mathbb{E}[\tichi_1]}{n\tr(\bbW)}-\tisig_1^2\Bigg|\leq\sqrt{n}|\mathbb{E}[\alpha_{1,1}^2y_{1,1}-1]|+\sqrt{n}\Bigg|\sum_{k\neq1\ {\rm or\ }t\neq1}^{T-1}\tisig_k\tisig_t\mathbb{E}[\alpha_{1,k}\alpha_{1,t}y_{k,t}]\Bigg|\to0.$$
	Consequently, it yields that for \(1\leq k\leq K\)
	$$\sqrt{n}\left(\frac{\tichi_k}{n\tr(\bbW)}-\tisig_k^2\right)\overset{\mbP}{\longrightarrow}\sqrt{n}\tisig_k^2(y_{k,k}-1).$$
	By Assumption \ref{Ap of highdimensionality} and (\ref{Eq of tilde sigma}), we know that
	\begin{align*}
		\lim_{n\to\infty}\tisig_k=\lim_{n\to\infty}\frac{1}{2n\sin(\pi k/2T)}=(c\pi k)^{-1},
	\end{align*}
	combining with Lemma \ref{Lem of pre CLT} and Slutsky’s lemma, we conclude that (\ref{Eq of tilde joint CLT}).
\end{proof}
\subsection{Joint CLT for extreme eigenvalues of the sample covariance matrix}\label{sec of CLT sig}
Until now, we have establish the joint CLT for \(\tichi_1,\cdots,\tichi_K\) in Proposition \ref{Thm of pre CLT}. In this section, we will finally prove Theorem \ref{Thm of original CLT}. As we have mentioned before, let \(\hat{\chi}_k\) and \(\tichi_k\) be the \(k\)-th largest eigenvalue of \(\hat{\bbSig}\) and \(\tilde{\bbSig}\) defined in {\rm (\ref{Eq of covariance matrix})} and {\rm (\ref{Eq of BN covariance matrix})}, respectively, the key step is to show that
$$\frac{|\hat{\chi}_k-\tichi_k|}{\sqrt{n}\tr(\bbW)}\overset{\mbP}{\longrightarrow}0,$$
For preliminaries, let's first cite the following result:
\begin{lem}[Chapter 9.12.5 of \cite{bai2010spectral}]\label{Thm of Extreme eigenvalue}
	Let \(\bbX=[X_{ij}]_{p\times n}\) be a random matrix with size of \(p\times n\), whose entries \(\{X_{ij}\}\) are independent complex random variables with mean zero, variance one, and finite fourth moments, and \(|X_{ij}|\leq \eta_n n^{1/2}\), where \(\eta_n\to0\) and \(\eta_n n^{1/4}\to\infty\). If \(\frac{p}{n}\to y>0\), then for any \(x>(1+\sqrt{y})^2\) and \(l>0\), the spectral norm of \(\bbS_n=n^{-1}\bbX\bbX^*\) satisfies that
	$$\mathbb{P}\left(\Vert\bbS_n\Vert>x\right)\leq \mathrm{o}(n^{-l}).$$
\end{lem}
Now, we first prove the following lemma:
\begin{lem}\label{Lem of bound spectral norm}
	Under Assumptions {\rm \ref{Ap of highdimensionality}, \ref{Ap of noise}} and {\rm \ref{Ap of effective rank}}, then
	$$\mathbb{P}\big(n^{-1/2}\Vert\Psi^*(L)\bbve\bbM\Vert>C_{B,c}\big)=\mathrm{O}(n^{-1/4}),$$
    where $\Psi^*(L)$ is defined in \eqref{Eq of BN decomposition} and $\bbve=[\varepsilon_1,\cdots,\varepsilon_T]$ defined in \eqref{Eq of random walk covariance}.
\end{lem}
\begin{proof}
	Since \(\Vert\bbM\Vert=1\), then by the triangle inequality, it yields that
	\begin{align}
		\Vert\Psi^*(L)\bbve\bbM\Vert\leq\Vert\Psi^*(L)\bbve\Vert\leq\sum_{k=0}^T\Vert\Psi_k^*\Vert\Vert\bbve_{-k}\Vert+\Vert r_T\Vert,\label{Eq of bound e+ rT}
	\end{align}
	where \(\bbve_{-k}:=[\varepsilon_{1-k},\cdots,\varepsilon_{T-k}]\in\mbR^{n\times T}\) and \(r_T:=\sum_{k=T+1}^{\infty}\Psi_k^*\bbve_{-k}\). it is easy to see that \(\Vert\bbve_{-k}\Vert\leq\Vert\bbve_+\Vert\), where \(\bbve_+:=[\varepsilon_{1-T},\cdots,\varepsilon_T]\in\mbR^{n\times 2T}\), and
	$$\sum_{k=0}^T\Vert\Psi_k^*\Vert\Vert\bbve_{-k}\Vert\leq\Vert\bbve_+\Vert\sum_{k=0}^T\Vert\Psi_k^*\Vert\leq\Vert\bbve_+\Vert\sum_{k=0}^Tk\Vert\Psi_k\Vert\leq B\Vert\bbve_+\Vert,$$
	where we use the definition \(\Psi_k^*:=-\sum_{i=k+1}^{\infty}\Psi_i\) and Assumption \ref{Ap of effective rank}. Next, define an event
	$$\mathcal{E}:=\{|\varepsilon_{it}|\leq n^{3/8}:i=1,\cdots,n;t=1-T,\cdots,T\},$$
	then
	$$\mathbb{P}(\mathcal{E}^c)\leq\sum_{i=1}^n\sum_{t=1-T}^T\mathbb{P}(|\epsilon_{it}|>n^{3/8})\leq\kappa_6\sum_{i=1}^n\sum_{t=1-T}^T n^{-9/4}=\mathrm{O}(n^{-1/4}),$$
	where we use Assumption \ref{Ap of noise} and following Chebyshev's inequality:
	$$\mathbb{P}(|\epsilon_{it}|>n^{3/8})\leq n^{-9/4}\mathbb{E}[|\varepsilon_{it}|^6]\leq\kappa_6n^{-9/4}.$$
	Then for any \(x>(1+\sqrt{c/2})^2\), we have that
	$$\mathbb{P}((2T)^{-1}\Vert\bbve_+\bbve_+'\Vert>x)\leq\mathbb{P}((2T)^{-1}\Vert\bbve_+\bbve_+'\Vert>x|\mathcal{E})+\mathbb{P}(\mathcal{E}^c)=\mathrm{O}(n^{-1/4}),$$
	where we use the fact \(\varepsilon_{it}\) are independent with zero mean, unit variance and finite fourth moment, then apply Lemma \ref{Thm of Extreme eigenvalue} for \(l=1/4\), it further implies that
	\begin{align}
		\mathbb{P}(n^{-1/2}\Vert\bbve_+\Vert>x)=\mathrm{O}(n^{-1/4})\label{Eq of bound e+}
	\end{align}
	for any \(x>1+\sqrt{2/c}\). It remains to bound \(\Vert r_T\Vert\), since \(\Psi_k^*\) are all diagonal, then given a unit vector \(y=(y_1,\cdots,y_T)'\in\mathbb{R}^T\), we have that
	$$\mathbb{E}[\Vert r_T y\Vert_2^2]=\sum_{i=1}^n\sum_{t_1,t_2=1}^T y_{t_1}y_{t_2}\mathbb{E}[(r_T)_{it_1}(r_T)_{it_2}]\leq\sum_{i=1}^n\sum_{t_1,t_2=1}^T |y_{t_1}y_{t_2}|\mathbb{E}[(r_T)_{it_1}^2+(r_T)_{it_2}^2]/2,$$
	by Assumption \ref{Ap of noise} and \ref{Ap of effective rank}, we can obtain that
	$$\mathbb{E}[(r_T)_{it}^2]=\mathbb{E}\left[\left(\sum_{k=T+1}^{\infty}(\Psi_k^*)_{ii}\varepsilon_{i,t-k}\right)^2\right]=\sum_{k=T+1}^{\infty}(\Psi_k^*)_{ii}^2\leq\left(\sum_{k=T+1}^{\infty}|(\Psi_k^*)_{ii}|\right)^2$$
	and
	$$\sum_{k=T+1}^{\infty}(\Psi_k^*)_{ii}\leq T^{-1}\sum_{k=T+1}^{\infty}k|(\Psi_k^*)_{ii}|\leq BT^{-1}.$$
	Therefore,
	$$\sum_{i=1}^n\sum_{t_1,t_2=1}^T |y_{t_1}y_{t_2}|\mathbb{E}[(r_T)_{it_1}^2+(r_T)_{it_2}^2]/2\leq C_{B,c}T^{-1}\left(\sum_{t=1}^T|y_t|\right)^2=\mathrm{O}(1)$$
	and
	$$\mathbb{E}[\Vert r_T\Vert^2]=\sup_{\Vert y\Vert_2=1}\mathbb{E}[\Vert r_T y\Vert_2^2]=\mathrm{O}(1),$$
	by the Chebyshev's inequality, it yields that
	\begin{align}
		\mathbb{P}(\Vert r_T\Vert>\sqrt{n})=\mathrm{O}(n^{-1}).\label{Eq of bound rT}
	\end{align}
	As a result, for any \(x>1+\sqrt{2/c}\), combine with (\ref{Eq of bound e+ rT}), (\ref{Eq of bound e+}) and (\ref{Eq of bound rT}), we can conclude that
	$$\mathbb{P}\big(n^{-1/2}\Vert\Psi^*(L)\bbve\bbM\Vert>Bx+1\big)\leq\mathbb{P}(\Vert r_T\Vert>\sqrt{n})+\mathbb{P}(n^{-1/2}\Vert\bbve_+\Vert>x)=\mathrm{O}(n^{-1/4}),$$
	which completes our proof.
\end{proof}
Finally, let's give the proof for Theorem \ref{Thm of original CLT}.
\begin{proof}[Proof of Theorem \ref{Thm of original CLT}]
	By (\ref{Eq of BN decomposition}), it gives that
	\begin{align*}
		&\big|\hat{\chi}_k^{1/2}-\tichi_k^{1/2}\big|\leq n^{-1/2}\Vert\bbX\bbM-\Psi(1)\bbve\bbU\bbM\Vert=n^{-1/2}\Vert\Psi^*(L)\bbve\bbM\Vert,
	\end{align*}
	then
	$$\frac{|\hat{\chi}_k-\tichi_k|}{\sqrt{n}\tr(\bbW)}\leq\frac{\Vert\Psi^*(L)\bbve\bbM\Vert|\hat{\chi}_k^{1/2}+\tichi_k^{1/2}|}{n\tr(\bbW)}.$$
	According to Lemma \ref{Lem of bound spectral norm},
	$$\mathbb{P}\big(n^{-1}\Vert\Psi^*(L)\bbve\bbM\Vert^2\leq C_{B,c}\big)\geq1-\mathrm{O}_{\mathbb{P}}(n^{-1/4}),$$
	so
	$$\frac{|\hat{\chi}_k-\tichi_k|}{\sqrt{n}\tr(\bbW)}\leq C_{B,c}\frac{|\hat{\chi}_k^{1/2}+\tichi_k^{1/2}|}{\sqrt{n}\tr(\bbW)}$$
	with probability at least of \(1-\mathrm{O}_{\mathbb{P}}(n^{-1/4})\). By (\ref{Eq of BN covariance matrix}) and Lemma \ref{Thm of Extreme eigenvalue}, we know that
	$$|\tichi_k|^{1/2}\leq n^{-1/2}\Vert\tilde{\bbSig}\Vert^{1/2}\leq C_{B,c}n$$
	with probability at least of \(1-\mathrm{O}_{\mathbb{P}}(n^{-1/4})\). Similarly, by (\ref{Eq of BN decomposition}), we have
	$$|\hat{\chi}_k|^{1/2}\leq n^{-1/2}\Vert\hat{\bbSig}\Vert^{1/2}\leq n^{-1/2}\Vert\tilde{\bbSig}\Vert^{1/2}+n^{-1/2}\Vert\Psi^*(L)\bbve\bbM\Vert\leq C_{B,c}n$$
	with probability at least of \(1-\mathrm{O}_{\mathbb{P}}(n^{-1/4})\). Finally, since \(\tr(\bbW)\geq nb\) by Assumption \ref{Ap of effective rank}, then
	$$\frac{|\hat{\chi}_k-\tichi_k|}{\sqrt{n}\tr(\bbW)}\leq C_{B,c}\frac{|\hat{\chi}_k^{1/2}+\tichi_k^{1/2}|}{\sqrt{n}\tr(\bbW)}\leq C_{B,b,c}n^{-1/2}$$
	with probability at least of \(1-\mathrm{O}_{\mathbb{P}}(n^{-1/4})\). Finally, by Assumption \ref{Ap of effective rank}, we can conclude that
    $$\tr(\bbW)/\tr(\bbW^2)^{1/2}\asymp\mrO(\sqrt{n}),$$
    which completes our proof.
\end{proof}
\section{Asymptotic spectral properties of sample correlation matrices of high-dimensional autoregressive processes with cross-sectional dependence}\label{Sec of AR process}
\setcounter{equation}{0}
\def\theequation{\thesection.\arabic{equation}}
\setcounter{subsection}{0}
In this section, we will further investigate the asymptotic spectral properties of the sample correlation matrix of high-dimensional autoregressive (AR) processes. First, let's make some notations here. Let $X_t$ be $n$-dimensional stochastic process generated by a \(d\)-th (\(d\in\mbN^+\) is a fixed integer) order AR process as follows:
\begin{align}
	X_t+\sum_{l=1}^da_lX_{t-l}=e_t,\quad e_t=\bbGa\fe_t=\bbGa\sum_{k=0}^{\infty}\Psi_k\varepsilon_{t-k},\quad\varepsilon_t\overset{i.i.d.}{\sim}\mcN(\boldsymbol{0},\bbI_n),\label{Eq of AR process}
\end{align}
where \(\{\Psi_k:k\in\mbN\}\) and \(\bbGa\) satisfy Assumptions \ref{Ap of panel lag polynomial} and \ref{Ap of nonpanel}, respectively. Next, the characteristic polynomial of (\ref{Eq of AR process}) is defined as
\begin{align}
	f_X(z)=z^d+\sum_{l=1}^da_lz^{d-l}.\label{Eq of characteristic polynomial}
\end{align}
and let \(\mfa_1,\cdots,\mfa_d\) be the roots of \(f_X(z)=0\), i.e. \(f_X(z)=\prod_{l=1}^d(z-\mfa_l)\). Then, we can rewrite (\ref{Eq of AR process}) as follows:
\begin{align}
    &\prod_{l=1}^d(1-\mfa_lL)X_t=e_t,\label{Eq of AR roots}
\end{align}
where \(L\) is the time lag operator. For simplicity, we still denote \(\hR\) \eqref{Eq of correlation matrix} to be the sample correlation matrix of \(\bbX=[X_1,\cdots,X_T]\) generated by (\ref{Eq of AR process}). Based on whether \(\mfa_l\) is inside, onside or outside the unit circle, we classify all \(\mfa_l\) into three classes. Precisely, we call \(\mfa_l\) is a
\begin{enumerate}
	\item {\bf stationary roots} if \(|\mfa_l|<1\);
	\item {\bf nonstationary roots} if \(|\mfa_l|=1\);
	\item {\bf super nonstationary roots} if \(|\mfa_l|>1\).
\end{enumerate} 
In this section, we will show that 
\begin{thm}\label{Thm of AR spectral properties}
    Under Assumptions {\rm \ref{Ap of highdimensionality}, \ref{Ap of panel lag polynomial}, \ref{Ap of nonpanel}} and \eqref{Eq of Ap of panel lag polynomial}, let $\hR$ be the sample correlation matrix \eqref{Eq of correlation matrix} of $\bbX=[X_1,\cdots,X_T]$ generated by a $d$-th order AR process \eqref{Eq of AR process}, let $\mfa_1,\cdots,\mfa_d$ be roots of $X_t$'s characteristic polynomial \eqref{Eq of characteristic polynomial} and $\hla_1\geq\cdots\geq\hla_T$ be eigenvalues of $\hR$, then 
        \begin{enumerate}
            \item if all $\mfa_l$ are stationary, we have $n^{-1}\Vert\hR\Vert\overset{\mbP}{\longrightarrow}0$;
            \item if at least one $\mfa_l$ is super nonstationary, we have $\limsup_{n\to\infty}{\rm rank}(\hR)\leq d$, and there exists a positive constant $C=C_{d,B,b,M_0,m_0}\in(0,1)$ such that $\lim_{n\to\infty}\mbP(n^{-1}\Vert\hR\Vert>C)=1.$
        \end{enumerate}
        Moreover, suppose the noise process $e_t$ in \eqref{Eq of AR process} satisfies that $e_t=\bbGa\varepsilon_t$, where $\varepsilon_t\overset{i.i.d.}{\sim}\mcN(\boldsymbol{0},\bbI_n)$, and $\bbGa$ satisfies Assumption {\rm \ref{Ap of m dependent}}, then for any $K\in\mbN^+$, 
        \begin{itemize}
            \item[3.] if all $\tau_l$ are {\bf not} super nonstationary and at least one $\mfa_l$ is nonstationary, we have $\bbm_{k,n}^{-1}(\hla_k/n-\mbE[\mbM_{k,n}])\overset{d}{\longrightarrow}\mcN(0,1)$ for $1\leq k\leq K$, where the precise definitions of $\bbm_{k,n}$ and $\mbM_{k,n}$ are postponed to Proposition \ref{Thm of CLT multiple unit roots} later.
        \end{itemize}
\end{thm}
Before proving the above Theorem, we present some necessary notations here. Recall that \(\bbX=\bbe\bbU\) in (\ref{Eq of bbX}) when \(X_t\) is generated by a random walk (\ref{Eq of nonpanel Xt}). Similarly, when \(X_t\) is an AR process generated by (\ref{Eq of AR process}), we can still represent \(\bbX\) by \(\bbe\) times a upper toeplitz matrix. Precisely, for any \(x\in\mbC\), denote
\begin{align}
	\mcT(x)=\left(\begin{array}{ccccc}
		1&-x&0&\cdots&\\
		0&1&-x&0&\cdots\\
		&\ddots&\ddots&\ddots&\\
		&\cdots&0&1&-x\\
		&&\cdots&0&1
	\end{array}\right)\label{Eq of mcT(a)}
\end{align}
to be a upper toeplitz matrix with main diagonals are 1 and the fist sup-diagonals are \(-x\), while others are 0. It is easy to see that \(\mcT(x)\mcT(y)=\mcT(y)\mcT(x)\) for any \(x,y\in\mbC\), that is, the matrix \(\mcT(x)\) is commutative with any other \(\mcT(y)\) in multiplication. Moreover, by the Vieta's theorem, we can further show that
\begin{align}
    \prod_{l=1}^d\mcT(\mfa_l)=\left(\begin{array}{ccccccc}
	1&a_1&\cdots&a_d&0&\cdots&0\\
	0&1&a_1&\cdots&a_d&0&\\
	&\ddots&\ddots&\ddots&&\ddots&\\
	0&\cdots&0&1&a_1&\cdots&a_d\\
	&&&\ddots&\ddots&\ddots&\\
	0&&&\cdots&0&1&a_1\\
	0&&&&\cdots&0&1
\end{array}\right),\label{Eq of bbT}
\end{align}
where \(\prod_{l=1}^d\mcT(\mfa_l)\) is uniquely determined and independent of the sequence of \(\mcT(\mfa_l)\) in multiplication. Now, let \(X_0=\cdots=X_{1-d}=\boldsymbol{0}\), 
we have
$$\bbX\prod_{l=1}^{d-1}\mcT(\mfa_l)=\bbe\quad\Longrightarrow\quad\bbX=\bbe\prod_{l=1}^{d-1}\mcT(\mfa_l)^{-1},$$
where \(\bbX=[X_1,\cdots,X_t]\) and \(\bbe=[e_1,\cdots,e_T]\). Here, \(\mcT(\mfa_l)^{-1}\) exists due to all eigenvalues of \(\mcT(\mfa_l)\) are 1. For simplicity, denote 
\begin{align}
    \mbT:=\prod_{l=1}^d\mcT(\mfa_l),\quad\mbU:=\mbT^{-1}.\label{Eq of mbT and mbU}
\end{align}
Thus, according to \eqref{Eq of correlation matrix}, the sample correlation matrix of $\bbX$ will be
\begin{align}
    \hR=\bbM\mbU'\bbX'\diag(\bbX\mbU\bbM\mbU'\bbX')^{-1}\bbX\mbU\bbM.\label{Eq of correlation matrix AR process}
\end{align}
Similar as proofs of Theorems \ref{Thm of CLT} and \ref{Thm of CLT I1}, the key idea of proving Theorem \ref{Thm of AR spectral properties} is to leverage the SVD of $\bbM\mbU'$ to represent the eigenvalues of $\hR$ in \eqref{Eq of correlation matrix AR process}. In the following three subsections, we will prove the three situations in Theorem \ref{Thm of AR spectral properties}, respectively.
\subsection{Autoregressive processes with stationary roots}\label{sec of stationary roots}
In this subsection, let's assume that all $\mfa_l$ in \eqref{Eq of AR roots} are stationary, i.e. $|\mfa_l|<1$, then we will prove that $n^{-1}\Vert\hR\Vert\overset{\mbP}{\longrightarrow}0$. First, we need the following lemma:

\begin{lem}\label{Thm of stationary roots}
	Under Assumptions {\rm \ref{Ap of highdimensionality}, \ref{Ap of panel lag polynomial}} and \eqref{Eq of Ap of panel lag polynomial}, let \(\fe=[\fe_1,\cdots,\fe_T]\) be generated by \(\fe_t=\sum_{k=0}^{\infty}\Psi_k\varepsilon_{t-k}\) and \(\varepsilon_t\overset{i.i.d.}{\sim}\mcN(\boldsymbol{0},\bbI_n)\). For any given integer \(d>1\) and \(\mfa_1,\cdots,\mfa_d\) such that \(|\mfa_l|<1\) for all \(1\leq l\leq d\), define 
	$$\tilde{\fe}:=\fe\prod_{l=1}^d\mcT(\mfa_l)^{-1}=[\tilde{\fe}_1,\cdots,\tilde{\fe}_T],$$
	then there exists \(\{\tilde{\Psi}_k:k\in\mbN\}\) satisfying Assumption {\rm \ref{Ap of panel lag polynomial}} and \eqref{Eq of Ap of panel lag polynomial} such that \(\tilde{\fe}_t:=\sum_{k=0}^{\infty}\tilde{\Psi}_k\varepsilon_{t-k}\).
\end{lem}
\begin{proof}
	We will inductively prove this claim. First, consider \(\tilde{\fe}=\fe\mcT(\mfa_1)^{-1}\), i.e. \(\mcT(\mfa_1)\tilde{\fe}=\fe\), then
	$$\tilde{\fe}_t=\mfa_1\tilde{\fe}_{t-1}+\fe_t=\sum_{k=0}^{\infty}\mfa_1^k\fe_{t-k}=\sum_{k=0}^{\infty}\mfa_1^k\sum_{l=0}^{\infty}\Psi_l\varepsilon_{t-k-l}=\sum_{k=0}^{\infty}\left(\sum_{l=0}^k\mfa_1^{k-l}\Psi_l\right)\varepsilon_{t-k}.$$
	Hence, \(\tilde{\Psi}_k=\sum_{l=0}^k\mfa_1^{k-l}\Psi_l\) and it is easy to see that 
	\begin{align*}
		&\sum_{k=0}^{\infty}(1+k)^2\Vert\tilde{\Psi}_k\Vert\leq\sum_{k=0}^{\infty}(1+k)^2\sum_{l=0}^k|\mfa_1|^{k-l}\Vert\Psi_l\Vert=\sum_{k=0}^{\infty}\Vert\Psi_k\Vert\sum_{l=k}^{\infty}(1+l)^2|\mfa_1|^l\\
		&<C_{\mfa_1}\sum_{k=0}^{\infty}(1+k)^2\Vert\Psi_k\Vert<C_{\mfa_1}B.
	\end{align*}
	On the other hand, denote \(\tilde{\Psi}_k=\diag(\tilde{\varphi}_{1,k},\cdots,\tilde{\varphi}_{n,k})\), then
	$$\min_j\inf_{x\in[0,2\pi]}\left|\sum_{k=0}^{\infty}\tilde{\varphi}_{j,k}e^{{\rm i}tx}\right|=|1-\mfa_1|^{-1}\min_j\inf_{x\in[0,2\pi]}\left|\sum_{k=0}^{\infty}\varphi_{j,k}e^{{\rm i}tx}\right|>0,$$
	so \(\{\tilde{\Psi}_k:k\in\mbN\}\) satisfies Assumption \ref{Ap of panel lag polynomial}. Now, for any fixed general \(d\), we can repeat the following procedure and inductively prove \(\{\tilde{\Psi}_k:k\in\mbN\}\) satisfies Assumption \ref{Ap of panel lag polynomial}.
\end{proof}
\begin{remark}
    Importantly, the above Lemma shows that stationary roots will not change the asymptotic behaviors of extreme eigenvalues of the sample correlation matrix $\hR$ \eqref{Eq of correlation matrix AR process} of \(\bbX\). Suppose the roots of \(X_t\)'s characteristic polynomials are \(\mfa_1,\cdots,\mfa_{d-1},1\), where \(|\mfa_l|<1\) for \(l<d\). Since we know that \(\bbX=\bbGa\fe\prod_{l=1}^d\mcT(\mfa_l)^{-1}\), where the product \(\prod_{l=1}^d\mcT(\mfa_l)^{-1}\) is independent of the order \(\mcT(\mfa_l)^{-1}\), then we rewrite \(\bbX\) as follows:
	$$\bbX=\bbGa\left(\fe\prod_{l=1}^{d-1}\mcT(\mfa_l)^{-1}\right)\mcT(1)^{-1}=\bbGa\tilde{\fe}\bbU,$$
	where \(\bbU=[1_{s\leq t}]_{s,t}\in\mbR^{T\times T}\) and \(\tilde{\fe}\) satisfies Assumption \ref{Ap of panel lag polynomial}. Therefore, given other Assumptions \ref{Ap of highdimensionality}, \ref{Ap of nonpanel} and \ref{Ap of m dependent}, we can conclude Theorem \ref{Thm of CLT I1} again under this situation.
\end{remark}
Finally, by \eqref{Eq of AR process}, since $\bbX=\bbGa\fe\mbU$, where $\fe=[\fe_1,\cdots,\fe_T]$, by Lemma \ref{Thm of stationary roots}, we know that $\tilde{\fe}_t$ is a stationary noise process satisfying Assumption \ref{Ap of panel lag polynomial} and \eqref{Eq of Ap of panel lag polynomial}, where $\tilde{\fe}=\fe\mbU=[\tilde{\fe}_1,\cdots,\tilde{\fe}_T]$. Hence, $\bbX=\bbGa\tilde{\fe}$, it gives that $X_t=\bbGa\tilde{\fe}_t$, then we can apply Theorem \ref{Thm of H1 norm} (set $\bbPi=\boldsymbol{0}$) later to conclude the first claim in Theorem \ref{Thm of AR spectral properties}.
\subsection{Autoregressive processes with super nonstationary roots}\label{sec of super nonstationary roots}
In this subsection, we assume that at least one of $\mfa_1,\cdots,\mfa_d$ in \eqref{Eq of AR roots} is super nonstationary, that is, there exists at least one $|\mfa_l|>1$. Next, we will show that
\begin{pro}\label{Thm of totally nonstationary roots}
	Under Assumptions {\rm \ref{Ap of highdimensionality}, \ref{Ap of panel lag polynomial}} and {\rm \ref{Ap of nonpanel}}, suppose \(X_t\) is an \(n\)-dimensional AR process generated by {\rm (\ref{Eq of AR process})} and its characteristic polynomial {\rm (\ref{Eq of characteristic polynomial})} has \(d_0\) super nonstationary roots, where \(1\leq d_0\leq d\). Without loss of generality, assume \(\mfa_1=\max_{1\leq l\leq d}|\mfa_l|\). Let \(\hla_k\) be the \(k\)-th largest eigenvalues of the sample correlation matrix \(\hR\) \eqref{Eq of correlation matrix AR process} of \(\bbX=[X_1,\cdots,X_T]\), then
	\begin{align*}
	    &\mbP(\hla_k>\mrO(|\mfa_1|^{-T/16}))\leq\mrO(|\mfa_1|^{-T/16}),\quad \forall k>d_0,\quad\text{and}\quad\lim_{n\to\infty}\mbP(n^{-1}\hla_1>C_{d_0,B,b,M_0,m_0})=1.
	\end{align*}
\end{pro}
Before presenting the proof of Proposition \ref{Thm of totally nonstationary roots}, we make the following notations. For a matrix $A\in\mbR^{n\times m}$ such that $n\leq m$, we denote $\sigma_t(A)$ to be the $t$-th largest singular value of $A$, where $1\leq t\leq n$. Moreover, if $A\in\mbR^{n\times n}$ is symmetric, we denote $\lambda_t(A)$ to be $t$-th largest eigenvalue of $A$.
\begin{proof}[Proof of Proposition \ref{Thm of totally nonstationary roots}]
	Without loss of generality, let \(|\mfa_1|\geq\cdots\geq|\mfa_{d_0}|>1\), where \(1\leq d_0\leq d\), and \(|\mfa_l|\leq1\) for \(d_0<l\leq d\). Define 
    $$\mbT_1:=\prod_{l=1}^{d_0}\mcT(\mfa_l),\quad\mbU_1:=\mbT_1^{-1},\quad\mbT_2:=\prod_{l=d_0+1}^d\mcT(\mfa_l),\quad\mbU_2:=\mbT_2^{-1}.$$
    Let's first investigate the singular values of \(\mbU_1\bbM\). Recall that $\mfa_1,\cdots,\mfa_{d_0}$ are all super nonstationary, then denote 
	$$\prod_{l=1}^{d_0}(x-\mfa_l)=x^{d_0}+\sum_{l=1}^{d_0}b_lx^{d_0-l},$$
	and \(\mfb_0:=1+\sum_{l=1}^{d_0}b_l^2\) and \(\mfb_k:=b_k+\sum_{l=1}^{d_0-k}b_lb_{l+k}\) for \(1\leq k\leq d_0\). Here, we define a \(T\times T\) symmetric banded toeplitz matrix as follows:
	$$\bbT={\rm Toeplitz}(\mfb_0,\cdots,\mfb_{d_0},0,\cdots,0),$$
	where the main diagonals of $\bbT$ are \(\mfb_0\) and \(k\)-th sub and sup-diagonals are \(\mfb_k\), while others are 0. Furthermore, we define the characteristic function of $\bbT$ as follows:
	$$g(x)=\mfb_0+\sum_{k=1}^{d_0}\mfb_k(e^{{\rm i}kx}+e^{-{\rm i}kx})=\left|e^{-{\rm i}d_0x}+\sum_{l=1}^{d_0}b_le^{-{\rm i}kx}\right|^2=\prod_{l=1}^{d_0}\left|e^{-{\rm i}x}-\mfa_l\right|^2,$$
	it is easy to see that \(g(x)\geq\prod_{l=1}^{d_0}\left|1-\mfa_l\right|^2>0\), then by Lemma 4.1 in \cite{gray2006toeplitz}, we know that 
	\begin{align}
	    \prod_{l=1}^{d_0}\left|1-|\mfa_l|\right|^2\leq\lambda_T(\bbT)\leq\lambda_1(\bbT)\leq\prod_{l=1}^{d_0}\left|1+|\mfa_l|\right|^2,\label{Eq of bbT singular value bound}
	\end{align}
	where \(\lambda_t(\bbT)\) represents the \(t\)-th largest eigenvalue of \(\bbT\). Moreover, notice that
	\begin{align*}
		\mbT_1'\mbT_1+\sum_{k=1}^{d_0}\mfv_k\mfv_k'=\bbT,
	\end{align*}
	where \(\mfv_k=(b_k,\cdots,b_1,0,\cdots,0)'\in\mbR^T\) for \(k=1,\cdots,d_0\). According to Theorem 1.1 in \cite{tyrtyshnikov1996unifying}, we have
	$$\lambda_T(\bbT-\mfv_1\mfv_1')\leq\lambda_T(\bbT)\leq\lambda_{T-1}(\bbT-\mfv_1\mfv_1')\leq\lambda_{T-1}(\bbT)\leq\cdots\leq\lambda_1(\bbT-\mfv_1\mfv_1')\leq\lambda_1(\bbT),$$
	Inductively, we can obtain that
	$$\lambda_T(\mbT_1'\mbT_1)\leq\lambda_{T-1}(\mbT_1'\mbT_1)\leq\cdots\leq\lambda_{T-d_0+1}(\mbT_1'\mbT_1)\leq\lambda_T(\bbT)\leq\cdots,$$
	i.e.
	$$\lambda_1(\mbU_1\mbU_1')\geq\lambda_2(\mbU_1\mbU_1')\geq\cdots\geq\lambda_{d_0}(\mbU_1\mbU_1')\geq\lambda_T(\bbT)^{-1}\geq\cdots.$$
	Notice that \(\mbU_1\bbM\mbU_1'=\mbU_1\mbU_1'-\mbU_1\boldsymbol{1}_T\boldsymbol{1}_T'\mbU_1'/T\), then by Theorem 1.1 in \cite{tyrtyshnikov1996unifying}, we have
	$$\lambda_t(\mbU_1\mbU_1')\geq\lambda_t(\mbU_1\bbM\mbU_1'),\quad t=1,\cdots,T.$$
	Hence, let \(\sigma_t(\mbU_1\bbM)\) be the \(t\)-th largest singular value of \(\mbU_1\bbM\), then the number of \(\sigma_t(\mbU_1\bbM)\) greater than \(\lambda_T(\bbT)^{-1/2}\) is at most \(d_0\), i.e.
	\begin{align}
	    \sigma_1(\mbU_1\bbM)\geq\sigma_2(\mbU_1\bbM)\geq\cdots\sigma_{d_0}(\mbU_1\bbM)\geq\lambda_T(\bbT)^{-1/2}\geq\sigma_{d_0+1}(\mbU_1\bbM)\geq\cdots.\label{Eq of singular value inequality 1}
	\end{align}
	On the other hand, notice that 
	$$\mcT(\mfa_l)^{-1}=[1_{s\leq t}\mfa_l^{t-s}]_{s,t}\in\mbR^{T\times T}$$
	is an upper toeplitz matrix with main diagonals are \(1\) and the \(k\)-th sup-diagonals are \(\mfa_l^k\). Hence, it is easy to see that
	\begin{align*}
		\sigma_1(\mcT(\mfa_l)^{-1}\bbM)^2\geq\sum_{t=0}^{T-1}|\mfa_l|^{2t}/2=\frac{\mfa_l^{2T}-1}{2(\mfa_l^2-1)}\quad\Longrightarrow\quad\sigma_1(\mcT(\mfa_l)^{-1}\bbM)\geq\mrO(|\mfa_l|^{T-1}).
	\end{align*}
	Moreover, by Lemma 4.1 in \cite{gray2006toeplitz}, we have 
    $$|\mfa_l|-1<\sigma_T(\mcT(\mfa_l))\leq\sigma_1(\mcT(\mfa_l))<1+|\mfa_l|$$
    then it gives that $\sigma_T(\mcT(\mfa_l)^{-1})\geq\sigma_1(\mcT(\mfa_l))^{-1}>(1+|\mfa_l|)^{-1}$. So we can further conclude that
	$$\sigma_1(\mbU_1\bbM)\geq\sigma_1(\mcT(\mfa_1)^{-1}\bbM)\prod_{l=2}^{d_0}\sigma_T(\mcT(\mfa_l)^{-1})\geq C|\mfa_1|^{T-1}\prod_{l=2}^{d_0}(1+|\mfa_l|)^{-1}.$$
	Finally, for \(\mbT_2=\prod_{l=d_0+1}^d\mcT(\mfa_l)\), where \(|\mfa_l|\leq 1\) and \(l=d_0+1,\cdots,d\), since \(-\mfa_l^{-1}\mcT(\mfa_l)\) is a Jordan matrix, by Theorem 1 in \cite{erxiong1994bounds} and Lemma 4.1 in \cite{gray2006toeplitz}, we can imply that
	$$T^{-1}\leq\sigma_T(\mcT(\mfa_l))\leq\sigma_1(\mcT(\mfa_l))\leq2,$$
	then 
	$$\sigma_1(\mbU_2)\leq T^{d-d_0},\quad\sigma_T(\mbU_2)\geq2^{d_0-d}.$$
	Since $\mbU_1\mbU_2=\mbU_2\mbU_1$, we have
	\begin{align*}
		&\sigma_t(\mbU_1\mbU_2\bbM)\leq\sigma_1(\mbU_2)\sigma_t(\mbU_1\bbM)\leq T^{d-d_0}\sigma_t(\mbU_1\bbM),\\
		&\sigma_1(\mbU_1\mbU_2\bbM)\geq\sigma_1(\mbU_1\bbM)\sigma_T(\mbU_1\mbU_2)\geq2^{d_0-d}\sigma_1(\mcG^{-1}\bbM).
	\end{align*}
	Now, define
	$$\tsigma_t(\mbU_1\mbU_2\bbM)=\frac{\sigma_t(\mbU_1\mbU_2\bbM)}{\sigma_1(\mbU_1\mbU_2\bbM)},$$
	by \eqref{Eq of bbT singular value bound} and \eqref{Eq of singular value inequality 1}, we have \(\sigma_t(\mbU_1\bbM)\leq\prod_{l=1}^{d_0}|1-|\mfa_l||^{-1}\) for \(t>d_0\), then for sufficiently large \(T\), we have
	$$\tsigma_t(\mbU_1\mbU_2\bbM)\leq C\frac{(T/2)^{d-d_0}\prod_{l=1}^{d_0}|1-|\mfa_l||^{-1}}{|\mfa_1|^{T-1}\prod_{l=2}^{d_0}(1+|\mfa_l|)^{-1}}\leq C_{\mfa}|\mfa_1|^{-T/2},\quad t>d_0,$$
	where we use the fact that \(|\mfa_1|>1\). Next, let the SVD of \(\mbU_1\mbU_2\bbM\) be 
	$$\mbU_1\mbU_2\bbM=\sum_{t=1}^T\sigma_t(\mbU_1\mbU_2\bbM)\vec{\mfv}_t\vec{\mfw}_t',$$
	then by \eqref{Eq of correlation matrix AR process}, it yields that
	\begin{align*}
		\hR&=\bbM\bbX'\diag(\bbX\bbM\bbX')^{-1}\bbX\bbM\\
		&=\bbM\mbU_2'\mbU_1'\fe'\bbGa'\diag(\bbGa\fe\mbU_1\mbU_2\bbM\mbU_2'\mbU_1'\fe'\bbGa')^{-1}\bbGa\fe\mbU_1\mbU_2\bbM\\
		&=\sum_{k,l=1}^T\vec{\mfw}_k\vec{\mfw}_l'\sum_{i=1}^n\frac{\tsigma_k(\mbU_1\mbU_2\bbM)\tsigma_l(\mbU_1\mbU_2\bbM)(\bbGa_i\fe\vec{\mfv}_k)(\bbGa_i\fe\vec{\mfv}_l)}{\sum_{t=1}^T\tsigma_t(\mbU_1\mbU_2\bbM)^2(\bbGa_i\fe\vec{\mfv}_t)^2},
	\end{align*}
	where \(\fe=[\fe_1,\cdots,\fe_T]\) is defined in \eqref{Eq of AR process} and $\bbGa_i$ is the $i$-th row of $\bbGa$. Next, we will show that
	$$\sum_{k>d_0\ {\rm or}\ l>d_0}^T\vec{\mfw}_k\vec{\mfw}_l'\sum_{i=1}^n\frac{\tsigma_k(\mbU_1\mbU_2\bbM)\tsigma_l(\mbU_1\mbU_2\bbM)(\bbGa_i\fe\vec{\mfv}_k)(\bbGa_i\fe\vec{\mfv}_l)}{\sum_{t=1}^T\tsigma_t(\mbU_1\mbU_2\bbM)^2(\bbGa_i\fe\vec{\mfv}_t)^2}\overset{\mbP}{\longrightarrow}\boldsymbol{0}_{T\times T}.$$
	Here, let's define an event
	$$\mcE(1):=\left\{\big|\bbGa_i\fe\vec{\mfv}_1\big|>|\mfa_1|^{-T/8}:i=1,\cdots,n\right\},$$
	it is easy to see that
	$$\mbP\big(\mcE(1)^c\big)\leq\sum_{i=1}^n\mbP\left(\big|\bbGa_i\fe\vec{\mfv}_1\big|\leq|\mfa_1|^{-T/8}\right)\leq C_{B,M_0}n|\mfa_1|^{-T/8}.$$
	Moreover, since
	$$\mbE\left[\frac{|\bbGa_i\fe\vec{\mfv}_k\bbGa_i\fe\vec{\mfv}_l|^2}{\big(\sum_{t=1}^T\tsigma_t(\mbU_1\mbU_2\bbM)^2(\bbGa_i\fe\vec{\mfv}_t)^2\big)^2}\Bigg|\mcE(1)\right]\leq C_{B,M_0}|\mfa_1|^{T/2}\mbE\left[|\bbGa_i\fe\vec{\mfv}_k\bbGa_i\fe\vec{\mfv}_l|^2\right]\leq C_{B,M_0}|\mfa_1|^{T/2},$$
	then
	\begin{align*}
		&\mbE\left[\left\Vert\sum_{k>d_0\ {\rm or}\ l>d_0}^T\vec{\mfw}_k\vec{\mfw}_l'\sum_{i=1}^n\frac{\tsigma_k(\mbU_1\mbU_2\bbM)\tsigma_l(\mbU_1\mbU_2\bbM)(\bbGa_i\fe\vec{\mfv}_k)(\bbGa_i\fe\vec{\mfv}_l)}{\sum_{t=1}^T\tsigma_t(\mbU_1\mbU_2\bbM)^2(\bbGa_i\fe\vec{\mfv}_t)^2}\right\Vert_F^2\Bigg|\mcE(1)\right]\\
		&\leq\sum_{k>d_0\ {\rm or}\ l>d_0}^T\tsigma_k(\mbU_1\mbU_2\bbM)^2\tsigma_l(\mbU_1\mbU_2\bbM)^2\sum_{i=1}^n\mbE\left[\frac{|\bbGa_i\fe\vec{\mfv}_k\bbGa_i\fe\vec{\mfv}_l|^2}{\big(\sum_{t=1}^T\tsigma_t(\mbU_1\mbU_2\bbM)^2(\bbGa_i\fe\vec{\mfv}_t)^2\big)^2}\Bigg|\mcE(1)\right]\\
		&\leq C_{B,M_0}T^2n|\mfa_1|^{-T/2},
	\end{align*}
	and
	\begin{align*}
		&\mbE\left[\left\Vert\sum_{k>d_0\ {\rm or}\ l>d_0}^T\vec{\mfw}_k\vec{\mfw}_l'\sum_{i=1}^n\frac{\tsigma_k(\mbU_1\mbU_2\bbM)\tsigma_l(\mbU_1\mbU_2\bbM)(\bbGa_i\fe\vec{\mfv}_k)(\bbGa_i\fe\vec{\mfv}_l)}{\sum_{t=1}^T\tsigma_t(\mbU_1\mbU_2\bbM)^2(\bbGa_i\fe\vec{\mfv}_t)^2}\right\Vert_F^2\Bigg|\mcE(1)^c\right]\mbP(\mcE(1)^c)\\
		&\leq T^4n^2\mbP(\mcE(1)^c)\leq C_{B,M_0}T^4n^3|\mfa_1|^{-T/8},
	\end{align*}
	which implies that
	\begin{align}
	    \mbE\left[\left\Vert\sum_{k>d_0\ {\rm or}\ l>d_0}^T\vec{\mfw}_k\vec{\mfw}_l'\sum_{i=1}^n\frac{\tsigma_k(\mbU_1\mbU_2\bbM)\tsigma_l(\mbU_1\mbU_2\bbM)(\bbGa_i\fe\vec{\mfv}_k)(\bbGa_i\fe\vec{\mfv}_l)}{\sum_{t=1}^T\tsigma_t(\mbU_1\mbU_2\bbM)^2(\bbGa_i\fe\vec{\mfv}_t)^2}\right\Vert_F^2\right]\leq\mrO(|\tau_1|^{-T/8}).\label{Eq of finte rank super nonstationary}
	\end{align}
	Finally, define
	$$\hat{\mcR}:=\sum_{k,l=1}^{d_0}\vec{\mfw}_k\vec{\mfw}_l'\sum_{i=1}^n\frac{\tsigma_k(\mbU_1\mbU_2\bbM)\tsigma_l(\mbU_1\mbU_2\bbM)(\bbGa_i\fe\vec{\mfv}_k)(\bbGa_i\fe\vec{\mfv}_l)}{\sum_{t=1}^T\tsigma_t(\mbU_1\mbU_2\bbM)^2(\bbGa_i\fe\vec{\mfv}_t)^2},$$
	it is easy to see than \({\rm rank}(\hat{\mcR})\leq d_0\), by the Wielandt-Hoffman inequality (Theorem 2.5 in \cite{gray2006toeplitz}), it gives that
	$$\sum_{i=1}^{d_0}\big(\lambda_i(\hR)-\lambda_i(\hat{\mcR})\big)^2+\sum_{i=d_0+1}^T\lambda_i(\hR)^2\leq\Vert\hR-\hat{\mcR}\Vert_F^2,$$
	which implies that
	$$\lambda_i(\hR)\overset{L^2}{\longrightarrow}\lambda_i(\hat{\mcR}),\quad i=1,\cdots,d_0,\quad\lambda_i(\hR)\overset{L^2}{\longrightarrow}0,\quad i=d_0+1,\cdots,T,$$
    By \eqref{Eq of finte rank super nonstationary} and Markov's inequality, we can show that $\mbP(\hla_k>\mrO(|\mfa_1|^{-T/16}))\leq\mrO(|\mfa_1|^{-T/16})$. Finally, note that
    \begin{align}
        \hla_1=\Vert\hR\Vert\geq\vec{\mfw}_1'\hR\vec{\mfw}_1=\sum_{i=1}^n\frac{\tsigma_1(\mbU_1\mbU_2\bbM)^2(\bbGa_i\fe\vec{\mfv}_1)^2}{\sum_{t=1}^T\tsigma_t(\mbU_1\mbU_2\bbM)^2(\bbGa_i\fe\vec{\mfv}_t)^2},\label{Eq of super nonstationary inequality 1}
    \end{align}
    and
    \begin{align*}
        \frac{\tsigma_1(\mbU_1\mbU_2\bbM)^2(\bbGa_i\fe\vec{\mfv}_1)^2}{\sum_{t=1}^T\tsigma_t(\mbU_1\mbU_2\bbM)^2(\bbGa_i\fe\vec{\mfv}_t)^2}\geq\frac{\tsigma_1(\mbU_1\mbU_2\bbM)^2(\bbGa_i\fe\vec{\mfv}_1)^2}{\tsigma_1(\mbU_1\mbU_2\bbM)^2\sum_{t=1}^{d_0}(\bbGa_i\fe\vec{\mfv}_t)^2+\sum_{t=d_0+1}^T\tsigma_t(\mbU_1\mbU_2\bbM)^2(\bbGa_i\fe\vec{\mfv}_t)^2},
    \end{align*}
    since we have shown that $\tsigma_t(\mbU_1\mbU_2\bbM)\leq\mrO(|\tau_1|^{-T/2})$ for all $t>d_0$, it implies that
    \begin{align*}
        &\mbP\left(\sum_{t=d_0+1}^T\tsigma_t(\mbU_1\mbU_2\bbM)^2(\bbGa_i\fe\vec{\mfv}_t)^2>\mrO(T^2|\tau_1|^{-T})\right)\leq\mrO(|\tau_1|^{-T}),
    \end{align*}
    and
    \begin{align}
        &\mbP\left(\frac{1}{n}\sum_{i=1}^n\frac{\tsigma_1(\mbU_1\mbU_2\bbM)^2(\bbGa_i\fe\vec{\mfv}_1)^2}{\sum_{t=1}^T\tsigma_t(\mbU_1\mbU_2\bbM)^2(\bbGa_i\fe\vec{\mfv}_t)^2}\geq\frac{1}{n}\sum_{i=1}^n\frac{(\bbGa_i\fe\vec{\mfv}_1)^2}{\sum_{t=1}^{d_0}(\bbGa_i\fe\vec{\mfv}_t)^2}\right)\leq1-\mrO(|\tau_1|^{-T}).\label{Eq of super nonstationary inequality 2}
    \end{align}
    Recall $\fe=[\fe_1,\cdots,\fe_T]$ defined in \eqref{Eq of AR process}, since all \(\varepsilon_t\overset{i.i.d.}{\sim}\mcN(\boldsymbol{0},\bbI_n)\) and \(\Psi_k\) are diagonal, then 
    $$\fe\vec{\mfv}_t\sim\mcN(\boldsymbol{0},\diag(\vec{\mfv}_t'\mcA^1\vec{\mfv}_t,\cdots,\vec{\mfv}_t'\mcA^n\vec{\mfv}_t)),$$
    where \(\mcA^j\) is defined in (\ref{Eq of mcA}). By Assumption \ref{Ap of panel lag polynomial} and Lemma \ref{Lem of positive definite Toeplitz}, \(b^2\leq\vec{\mfv}_t'\mcA^j\vec{\mfv}_t\leq B\) for all \(1\leq j\leq n\). Hence, we know that $\bbGa_i\fe\vec{\mfv}_t\sim\mcN(0,\mfn_{i,t}^2)$, where $\mfn_{i,t}^2=\sum_{j=1}^n\Gamma_{i,j}^2\vec{\mfv}_t'\mcA^j\vec{\mfv}_t$. By Assumption \ref{Ap of nonpanel}, we know that $C_{b,m_0}\leq \mfn_{i,t}^2\leq C_{B,M_0}$. Hence, it gives that
    \begin{align*}
        \mbE\left[\frac{(\bbGa_i\fe\vec{\mfv}_1)^2}{\sum_{t=1}^{d_0}(\bbGa_i\fe\vec{\mfv}_t)^2}\right]\geq\frac{\mbE[|\bbGa_i\fe\vec{\mfv}_1|]^2}{\sum_{t=1}^{d_0}\mbE\big[(\bbGa_i\fe\vec{\mfv}_t)^2\big]}=\frac{2\mfn_{i,1}^2}{\pi\sum_{t=1}^{d_0}\mfn_{i,t}^2}\geq C_{d_0,B,b,M_0,m_0},
    \end{align*}
    and
    $$\frac{1}{n}\sum_{i=1}^n\mbE\left[\frac{(\bbGa_i\fe\vec{\mfv}_1)^2}{\sum_{t=1}^{d_0}(\bbGa_i\fe\vec{\mfv}_t)^2}\right]\geq C_{d_0,B,b,M_0,m_0}.$$
    Moreover, we can use the same method as Lemma \ref{Lem of finite variance} to show that
    $$\frac{1}{n}\Var\left(\sum_{i=1}^n\frac{(\bbGa_i\fe\vec{\mfv}_1)^2}{\sum_{t=1}^{d_0}(\bbGa_i\fe\vec{\mfv}_t)^2}\right)\leq\mrO(1),$$
    we omit details here to save space. Therefore, by Chebyshev's inequality, we can conclude that
    $$\lim_{n\to\infty}\mbP\left(\frac{1}{n}\sum_{i=1}^n\frac{(\bbGa_i\fe\vec{\mfv}_1)^2}{\sum_{t=1}^{d_0}(\bbGa_i\fe\vec{\mfv}_t)^2}\geq C_{d_0,B,b,M_0,m_0}\right)=1,$$
    combining \eqref{Eq of super nonstationary inequality 1} and \eqref{Eq of super nonstationary inequality 2} with the above equation, we can derive that $\lim_{n\to\infty}\mbP(n^{-1}\hla_1>C_{d_0,B,b,M_0,m_0})=1$.
\end{proof}
As a consequence of Proposition \ref{Thm of totally nonstationary roots}, once \(X_t\)'s characteristic polynomial has one totally nonstationary root, the rank of its sample correlation matrix will be asymptotically finite. Finally, let's further investigate how the super nonstationary roots effect the asymptotic spectral properties of the sample covariance maitrx of $\bbX=[X_1,\cdots,X_T]$ generated by \eqref{Eq of AR process}.
\begin{remark}\label{Rem of why not covariance 1}
	Given the observations \(\bbX\) defined in Proposition \ref{Thm of totally nonstationary roots}, denote \(\hSig\) to be the sample covariance matrix of \(\bbX\), i.e. 
	$$\hSig=\frac{1}{n}\bbM\bbX'\bbX\bbM,$$
	then we claim that \(\Vert\hSig\Vert\geq\mrO(|\mfa_1|^{CT})\) for some $C>0$. Recall that \(\bbX=\bbe\mbU_1\mbU_2\) in the previous proof, and the SVD of \(\mbU_1\mbU_2\bbM=\sum_{t=1}^T\sigma_t(\mbU_1\mbU_2\bbM)\vec{\mfv}_t\vec{\mfw}_t'\), then we have
	\begin{align*}
		&\Vert\hSig\Vert\geq\vec{\mfw}_1'\hSig\vec{\mfw}_1=\frac{\sigma_1(\mbU_1\mbU_2\bbM)^2}{n}\vec{\mfv}_t'\bbe'\bbe\vec{\mfv}_t=\frac{\sigma_1(\mbU_1\mbU_2\bbM)^2}{n}\vec{\mfv}_t'\fe'\bbGa'\bbGa\fe\vec{\mfv}_t,
	\end{align*}
	where \(\fe=[\fe_1,\cdots,\fe_T]\) and \(\fe_t=\sum_{k=0}^{\infty}\Psi_k\varepsilon_{t-k}\). Since all \(\varepsilon_t\overset{i.i.d.}{\sim}\mcN(\boldsymbol{0},\bbI_n)\) and \(\Psi_k\) are diagonal, then \(\fe\vec{\mfv}_1\sim\mcN(\boldsymbol{0},\diag(\vec{\mfv}_1'\mcA^1\vec{\mfv}_1,\cdots,\vec{\mfv}_1'\mcA^n\vec{\mfv}_1))\), where \(\mcA^j\) is defined in (\ref{Eq of mcA}). By Assumption \ref{Ap of panel lag polynomial} and Lemma \ref{Lem of positive definite Toeplitz}, \(b^2\leq\vec{\mfv}_1'\mcA^j\vec{\mfv}_1\leq B\) for all \(1\leq j\leq n\), so \(\diag(\vec{\mfv}_1'\mcA^1\vec{\mfv}_1,\cdots,\vec{\mfv}_1'\mcA^n\vec{\mfv}_1)^{-1/2}\fe\vec{\mfv}_1\sim\mcN(\boldsymbol{0},\bbI_n)\). Thus, by the Gaussian concentration inequality, we have
	\begin{align*}
		&\Vert\hSig\Vert\geq C_{B,m_0}\sigma_1(\mbU_1\mbU_2\bbM)^2=\mrO(|\mfa_1|^{CT})
	\end{align*}
	with probability of \(1-\mrO(e^{-Cn})\), i.e. \(\Vert\hSig\Vert\geq\mrO(|\mfa_1|^{CT})\) with probability of \(1\) by the previous proof.
\end{remark}
\subsection{Autoregressive processes with nonstationary roots}\label{sec of nonstationary roots}
As we have shown in Lemma \ref{Thm of stationary roots}, when the characteristic polynomial \eqref{Eq of characteristic polynomial} has both stationary roots and (super) nonstationary roots, the asymptotic spectral properties of $\hR$ in \eqref{Eq of correlation matrix AR process} only depend on the (super) nonstationary roots. Moreover, we have shown that super nonstationary roots makes ${\rm rank}(\hR)$ will be asymptotically finite in Proposition \ref{Thm of totally nonstationary roots}. In this subsection, we will investigate how the nonstationary roots alone effect the asymptotic spectral properties of $\hR$ in \eqref{Eq of correlation matrix AR process}, we assume that \(X_t\)’s characteristic polynomial (\ref{Eq of characteristic polynomial}) only has the nonstationary roots without loss of generality, i.e. all $|\mfa_l|=1$ \eqref{Eq of AR roots} for $1\leq l\leq d$. To overcome technical difficulties, we will simplify \(e_t\) by \(e_t\overset{i.i.d.}{\sim}\mcN(\boldsymbol{0},\bbGa\bbGa')\), i.e. \(X_t\) is generated by
$$\prod_{l=1}^d(1-\mfa_l L)X_t=e_t=\bbGa\fe_t,\quad\fe_t\overset{i.i.d.}{\sim}\mcN(\boldsymbol{0},\bbI_n),$$
where \(|\mfa_l|=1\) for \(1\leq l\leq d\). Thus, given the observations \(\bbX\) and noise matrix \(\bbe\), we have \(\bbX=\bbe\prod_{l=1}^d\mcT(\mfa_l)\). The reason is that the SVD structures of general toeplitz matrices (e.g. \(\prod_{l=1}^d\mcT(\mfa_l)\)) are unknown for us, so it is indeed hard to obtain the paralleling results of Lemma 10 in \cite{onatski2021spurious} as follows:
$$\left|\Cov(\fe_j\bbv_k,\fe_j\bbv_l)-2\pi\delta_{k,l}f_j(\theta_k/2)\right|<C_BT^{-1}.$$
On the other hand, if \(\fe=[\fe_1,\cdots,\fe_T]\) is just a standard Gaussian matrix, then
$$\Cov(\fe_j\bbv_k,\fe_j\bbv_l)=\delta_{k,l}.$$
For simplicity, we rewrite \(\prod_{l=1}^d(1-\mfa_l L)X_t=e_t\) in \eqref{Eq of AR roots} as follows:
\begin{align}
    \prod_{l=1}^{m_d}(1-\mfa_l L)^{d_l}X_t=\bbGa\fe_t,\quad\fe_t\overset{i.i.d.}{\sim}\mcN(\boldsymbol{0},\bbI_n),\label{Eq of AR nonstationary}
\end{align}
where \(\sum_{l=1}^{m_d}d_l=d\) and \(\mfa_l\) are different such that \(|\mfa_l|=1\) for \(1\leq l \leq m_d\), then we will show that
\begin{pro}\label{Thm of CLT multiple unit roots}
    Under Assumptions {\rm \ref{Ap of highdimensionality}, \ref{Ap of nonpanel}} and {\rm \ref{Ap of m dependent}}, for the cross-sectional matrix $\bbGa$ in \eqref{Eq of AR nonstationary}, define
    $$\bbXi=[\Xi_{i_1,i_2}]_{n\times n}:=\diag(\bbGa\bbGa')^{-1/2}\bbGa\bbGa'\diag(\bbGa\bbGa')^{-1/2}.$$
    Moreover, let $\{\vec{z}_t:=(z_{1,t},\cdots,z_{n,t})'\overset{i.i.d.}{\sim}\mcN(\boldsymbol{0},\bbXi):t=1,\cdots,T\}$ be a sequence of i.i.d. random normal vectors, then for any $K\in\mbN^+$ and \(1\leq k\leq K\), define
    $$\mbM_{k,n}=\frac{1}{n}\sum_{i=1}^n\frac{\sigma_k^2z_{i,k}^2}{\sum_{t=1}^T\sigma_t^2z_{i,t}^2},$$
    where \(\sigma_1\geq\cdots\geq\sigma_T\) are singular values of \(\mbT^{-1}\bbM\). Then given \(T\) observations \(\bbX=[X_1,\cdots,X_T]\) generated by \eqref{Eq of AR nonstationary}, we have
	$$\frac{\sqrt{n}}{\bbm_{k,n}}\left(\frac{\hla_k}{n}-\mbE\big[\mbM_{k,n}\big]\right)\overset{d}{\longrightarrow}\mcN(0,1),$$
	where \(\hla_k\) is the first \(k\)-th largest eigenvalue of the sample correlation matrix of \(\bbX\) and \(\bbm_{k,n}^2=n\Var(\mbM_{k,n})\).
\end{pro}
The key step of proving the above theorem is to show that (e.g \(k=1\))
$$\frac{\hla_1-\mbE[\hla_1]}{\sqrt{n}}\overset{\mbP}{\longrightarrow}\frac{1}{\sqrt{n}}\sum_{i=1}^n\frac{\tsigma_1^2(\bbGa_i\fe\bbv_1)^2}{\sum_{t=1}^T\tsigma_t^2(\bbGa_i\fe\bbv_t)^2}-\mbE\left[\frac{\tsigma_1^2(\bbGa_i\fe\bbv_1)^2}{\sum_{t=1}^T\tsigma_t^2(\bbGa_i\fe\bbv_t)^2}\right],$$
where $\bbGa_i$ is the $i$-th row of $\bbGa$ in \eqref{Eq of AR nonstationary}, \(\beta_k=\sigma_k/\sigma_1\) and 
\begin{align}
	\mbT^{-1}\bbM(\mbT^{-1})'=\sum_{k=1}^T\sigma_k^2\bbv_k\bbv_k'.\label{Eq of multiple unit roots SVD}
\end{align}
Here, we abuse the notations $\sigma_k$ in above definition \eqref{Eq of multiple unit roots SVD}. In general, for a matrix $A\in\mbR^{n\times m}$ such that $n\leq m$, we denote $\sigma_t(A)$ to be the $t$-th largest singular value of $A$, where $1\leq t\leq n$. Moreover, if $A\in\mbR^{n\times n}$ is symmetric, we denote $\lambda_t(A)$ to be $t$-th largest eigenvalue of $A$. 

Although we can repeat the proofs of Lemmas \ref{Lem of positive covariance}, \ref{Lem of alpha} and Theorem \ref{Thm of convergence in probability} to Theorem \ref{Thm of CLT multiple unit roots}, the essential question is to figure out the values of \(\tsigma_t,t=1,\cdots,T\). Unfortunately, it is indeed hard to obtain the explicit forms of \(\sigma_t\), but we can still provide some asymptotic properties of \(\sigma_t\) as follows:
\begin{lem}\label{Lem of singular values}
	Denote \(\sigma_t:=\sigma_t(\mbT^{-1}\bbM)\) to be the \(t\)-th largest singular value of \(\mbT^{-1}\bbM\), without loss of generality, assume that \(d_1=\max_{1\leq l\leq m_d}d_l\), then we have
	\begin{itemize}
		\item \(\sigma_1\asymp\mrO(T^{d_1})\) (``$\asymp$'' is defined in \eqref{Eq of asymp mbP});
		\item \(\sum_{t=1}^T\tsigma_t^2\asymp\mrO(1)\), where \(\tsigma_t=\sigma_t/\sigma_1\);
		\item \(\tsigma_t\asymp\mrO(l^{-d_1})\) for \(1\leq t\leq T^{\delta}\), where \(\delta>0\) is a sufficiently small fixed constant.
	\end{itemize}
\end{lem}
\begin{proof}
    Let's first focus on the singular values of $\mbT$. Recall that \(\prod_{l=1}^{m_d}(1-\mfa_l L)^{d_l}X_t=e_t\) in (\ref{Eq of AR nonstationary}), where \(d_l\in\mbN^+\) and \(\sum_{l=1}^{m_d}d_l=d\). For simplicity, we denote
    $$\prod_{l=1}^{m_d}(1-\mfa_l L)^{d_l}X_t=X_t+\sum_{l=1}^da_lX_{t-l}.$$
    Note that \(\mbT=\prod_{l=1}^{m_d}\mcT(\mfa_l)^{d_l}\) in (\ref{Eq of bbT}) is an upper toeplitz matrix, so \(\mbT^{-1}\) is also an upper toeplitz matrix, i.e. 
        $$\mbT^{-1}=\left(\begin{array}{ccccc}
	       1&b_1&\cdots&b_{T-2}&b_{T-1}\\
	       0&1&b_1&\cdots&b_{T-2}\\
	       &\ddots&\ddots&\ddots&\vdots\\
	       0&\cdots&0&1&b_1\\
	       0&&\cdots&0&1
        \end{array}\right),$$
        where \(b_k\) be the \(k\)-th sup-diagonal terms of \(\mbT^{-1}\) for \(1\leq k\leq T-1\). It is easy to see that
        $$b_k+\sum_{l=1}^da_lb_{k-l}=0,\quad {\rm for\ }d+1\leq l\leq T-1,$$
        and \(b_1,\cdots,b_d\) can be solved by
        $$\left(\begin{array}{cccc}
	       1&b_1&\cdots&b_d\\
	       
	       &\ddots&\ddots&\ddots\\
	       &\ddots&1&b_1\\
	       &&0&1
        \end{array}\right)=\left(\begin{array}{cccc}
	       1&a_1&\cdots&a_d\\
	       
	       &\ddots&\ddots&\ddots\\
	       &\ddots&1&a_1\\
	       &&0&1
        \end{array}\right)^{-1}.$$
        Hence, we can derive \(b_k=\sum_{l=1}^{m_d}\mfa_l^k\sum_{r=1}^{d_l}\mfc_{l,r}k^{r-1}\), where \(\mfc_{l,r}\) can be solved by \(b_1,\cdots,b_d\). Consequently, \(|b_k|\leq C_{d,m_d,\mfa}k^{d_1-1}\) and
        \begin{align*}
            &\Vert\mbT^{-1}\bbM\Vert\leq\Vert\mbT^{-1}\Vert_F\leq C_{d,m_d,\mfa}\left(\sum_{k=1}^T(T+1-k)k^{2(d_0-1)}\right)^{1/2}\leq C_{d,m_d,\mfa}T^{d_1}.
        \end{align*}
        On the other hand, due to all \(|\mfa_l|=1\), then \(\Vert\mcT(\mfa_l)\Vert\leq2\), i.e. \(\sigma_T(\mcT(\mfa_l)^{-1})\geq1/2\). Thus,
        $$\Vert\mbT^{-1}\bbM\Vert\geq\sigma_1(\mcT(\mfa_1)^{-d_1}\bbM)\prod_{l=2}^{m_d}\sigma_T(\mcT(\mfa_l)^{-d_l})>2^{d_1-d}\sigma_1(\mcT(\mfa_1)^{-d_1}).$$
        Moreover, by direct calculations, \(\mcT(\mfa_1)^{-d_1}\) is a \(T\times T\) upper toeplitz matrix with main diagonals are 1 and \(k\)-th sup-diagonals are \(\binom{k+d_1-1}{d_1-1}\mfa_1^k\). Now, if $\mfa_1\neq1$, let \(\vec{a}=T^{-1/2}(1,\mfa_1,\cdots,\mfa_1^{T-1})\) be a \(T\)-dimensional unit vector, then
        \begin{align*}
            &\Vert\mcT(\mfa_1)^{-d_1}\vec{a}\Vert_2^2=T^{-1}\sum_{t=1}^T\Bigg|\mfa_1^{t-1}\sum_{k=0}^{T-t}\binom{k+d_1-1}{d_1-1}\Bigg|^2=T^{-1}\sum_{t=1}^T\left(\sum_{k=0}^{t-1}\binom{k+d_1-1}{d_1-1}\right)^2,\\
            &T^{-2}\Vert\mcT(\mfa_1)^{-d_1}\boldsymbol{1}_T\boldsymbol{1}_T'\vec{a}\Vert_2^2=T^{-3}\left|\sum_{t=1}^T\mfa_1^{t-1}\right|^2\sum_{t=1}^T\left|\sum_{k=0}^{T-t}\binom{k+d_1-1}{d_1-1}\mfa_1^k\right|^2.
        \end{align*}
        so it gives that
        \begin{align*}
            &\Vert\mcT(\mfa_1)^{-d_1}\bbM\vec{a}\Vert_2\geq\Vert\mcT(\mfa_1)^{-d_1}\vec{a}\Vert_2-T^{-2}\Vert\mcT(\mfa_1)^{-d_1}\boldsymbol{1}_T\boldsymbol{1}_T'\vec{a}\Vert_2\\
            &=(T^{-1/2}-T^{-3/2}|1-\mfa_1|^{-1})\left(\sum_{t=1}^T\left(\sum_{k=0}^{t-1}\binom{k+d_1-1}{d_1-1}\right)^2\right)^{1/2}\\
            &>\frac{T^{-1/2}}{2}\left(\sum_{t=1}^T\left(\sum_{k=0}^{t-1}(k/d_1)^{d_1-1}\right)^2\right)^{1/2}>C_{d_1}T^{-1/2}\left(\sum_{t=1}^Tt^{2d_1}\right)^{1/2}>C_{d_1}T^{d_1},
        \end{align*}
        where we use the fact \(|\mfa_1|=1\). Otherwise, if $\mfa_1=1$, let $\vec{a}=T^{-1/2}(1,-1,1,-1,\cdots)\in\mbR^T$, then 
        \begin{align*}
            &\Vert\mcT(\mfa_1)^{-d_1}\bbM\vec{a}\Vert_2\geq\Vert\mcT(\mfa_1)^{-d_1}\vec{a}\Vert_2-T^{-2}\Vert\mcT(\mfa_1)^{-d_1}\boldsymbol{1}_T\boldsymbol{1}_T'\vec{a}\Vert_2\\
            &=T^{-1/2}\left(\sum_{t=1}^T\left(\sum_{k=0}^{t-1}(-1)^k\binom{k+d_1-1}{d_1-1}\right)^2\right)^{1/2}-T^{-3/2}\left(\sum_{t=1}^T\left(\sum_{k=0}^{t-1}\binom{k+d_1-1}{d_1-1}\right)^2\right)^{1/2}>C_{d_1}T^{d_1},
        \end{align*}
        so we conclude that \(\sigma_1(\mcT(\mfa_1)^{-d_1}\bbM)=\Vert\mcT(\mfa_1)^{-d_1}\bbM\Vert>C_{d_1}T^{d_1}\).
        As a result, we obtain that \(\Vert\mbT^{-1}\Vert>2^{d_1-d}\sigma_1(\mcT(\mfa_1)^{-d_1})>C_{d_1,d}T^{d_1}\), i.e. \(\Vert\mbT^{-1}\bbM\Vert\asymp\mrO(T^{d_1})\). 

    Next, we construct a \((T+d)\times(T+d)\) circulant matrix as follows:
		$$\mfC(T+d):=\left(\begin{array}{ccccccc}
			1&a_1&\cdots&a_d&0&\cdots&0\\
			0&1&a_1&\cdots&a_d&\cdots&\vdots\\
			\vdots&\vdots&\ddots&\ddots&&\ddots&\\
			&&&1&a_1&\cdots&a_d\\
			a_d&0&\cdots&0&1&\cdots&a_{d-1}\\
			\vdots&\ddots&&&\ddots&\ddots&\vdots\\
			a_1&\cdots&a_d&0&\cdots&0&1
		\end{array}\right).$$
		It is easy to see that \(\mbT\) can be derived by deleting the first \(d\) rows and columns of \(\mfC(T+d)\). According to \cite{karner2003spectral}, the singular values of \(\mfC(T+d)\) are
		$$s_k:=\left|1+\sum_{l=1}^da_le^{-2\pi{\rm i}(k-1)l/(T+d)}\right|=\prod_{l=1}^{m_d}\big|1-\mfa_l e^{-2\pi{\rm i}(k-1)/(T+d)}\big|^{d_l}=\prod_{l=1}^{m_d}\big|1-e^{2\pi{\rm i}[\mathtt{a}_l-(k-1)/(T+d)]}\big|^{d_l},$$
		where \(k=1,\cdots,T+d\) and \(\mfa_l=\exp(2\pi{\rm i}\mathtt{a}_l)\) such that \(\mathtt{a}_l\in[0,1)\) for \(1\leq l\leq m_d\). Here, we sort all \(s_k\) as follows:
		$$\sigma_1(\mfC(T+d))=\max\{s_1,\cdots,s_{T+d}\},\quad\sigma_l(\mfC(T+d)):=\max\{s_1,\cdots,s_{T+d}\}\backslash\{\sigma_1,\cdots,\sigma_{l-1}\},$$
		where \(l=2,\cdots,T+d\). Due to \(d_1=\max_{1\leq l\leq m_d}d_l\), denote \(k_0:=\arg\min\{|\mathtt{a}_1-(k-1)/(T+d)|:1\leq k\leq T\}\), it gives that \(\sigma_{T+d}(\mfC(T+d))=s_{k_0}\). Next, we claim that \(\sigma_{T+d-k}(\mfC(T+d))\asymp\mrO((k/T)^{d_1})\) for \(1\leq k\leq T^{\delta}\), where \(\delta>0\) is a sufficiently small fixed number. Since
		\begin{align*}
			&\left|e^{{\rm i}x}-1\right|^2=(1-\cos^2x)^2+\sin^2x=2(1-\cos x)=4\sin^2(x/2),
		\end{align*}
		then \(\left|e^{{\rm i}x}-1\right|\asymp\mrO(|x|)\) for \(x\in[-\pi,\pi]\). Since all \(\mathtt{a}_l\) are different, it implies that
		\begin{align*}
			&\sigma_{T+d-k}(\mfC(T+d))\asymp\mrO(s_{k_0\pm k})\asymp\mrO\big(\big|1-e^{\pm2\pi{\rm i}k/(T+d)}\big|^{d_1}\big)\asymp\mrO((k/T)^{d_1}).
		\end{align*}
		Finally, according to Cauchy's interlacing theorem, we know that
		$$\sigma_{t+2d}(\mfC(T+d))\leq\sigma_t(\mbT)\leq\sigma_t(\mfC(T+d)),$$
		where \(1\leq t\leq T-2d\). In other words, for $d<t\leq T$, we have
		$$\mrO(((t-d)/T)^{-d_1})\asymp\sigma_{T-t+2d}(\mfC(T+d))^{-1}\leq\sigma_t(\mbT^{-1})\leq\sigma_{T-t}(\mfC(T+d))^{-1}\asymp\mrO(((t+d)/T)^{-d_1}).$$
        By Cauchy's interlacing theorem again, we also have
        $$\sigma_t(\mbT^{-1})\geq\sigma_t(\mbT^{-1}\bbM)\geq\sigma_{t-1}(\mbT^{-1})$$
        Since $d\in\mbN^+$ is a fixed integer, combining with $\sigma_1(\mbT^{-1}\bbM)\asymp\mrO(T^{d_1})$, for \(d<k\leq T^{\delta}\), we have
		\begin{align*}
			\tsigma_k(\mbT^{-1}\bbM)=\frac{\sigma_k(\mbT^{-1}\bbM)}{\sigma_1(\mbT^{-1}\bbM)}\geq\frac{\sigma_{T+1+2d-k}(\mfC(T+d))^{-1}}{\sigma_1(\mbT^{-1}\bbM)}\geq\mrO(l^{-d_1}),
		\end{align*}
		and 
		\begin{align*}
			&\tsigma_k(\mbT^{-1}\bbM)\leq \frac{\sigma_{T+2d-k}(\mfC(T+d))^{-1}}{\sigma_1(\mbT^{-1}\bbM)}\leq\mrO(l^{-d_0}),
		\end{align*}
		and this upper bound can be extended to \(1\leq l\leq d\) due to \(d\) is a fixed integer. Now, we complete our proofs.
\end{proof}
Finally, let's prove Proposition \ref{Thm of CLT multiple unit roots}.
\begin{proof}[Proof of Proposition \ref{Thm of CLT multiple unit roots}]
	Recall that \(\fe\) is a standard Gaussian random matrix, so \(\Cov(\fe_j\bbv_k,\fe_j\bbv_l)=\delta_{k,l}\) and we do not need to remove the weak correlation among all $\fe_j\bbv_1,\cdots,\fe_j\bbv_T$ as in Lemmas \ref{Lem of remove dependence} and \ref{Lem of unify variance}. Hence, let's first repeat the proofs of Lemma \ref{Lem of positive covariance}, which will conclude that
	$$\sum_{i=1}^n\Var\left(\frac{\tsigma_k^2(\bbGa_i\fe\bbv_k)^2}{\sum_{t=1}^T\tsigma_t^2(\bbGa_i\fe\bbv_t)^2}\right)\leq\Var\left(\sum_{i=1}^n\frac{\tsigma_k^2(\bbGa_i\fe\bbv_k)^2}{\sum_{t=1}^T\tsigma_t^2(\bbGa_i\fe\bbv_t)^2}\right)\leq C_{M_0,m_0}\tsigma_k^4\Vert\tilde{\bbGa}\Vert_F^2,$$
	and
	$$\Var\left(\sum_{i=1}^n\frac{\tsigma_k\tsigma_l(\bbGa_i\fe\bbv_k)(\bbGa_i\fe\bbv_l)}{\sum_{t=1}^T\tsigma_t^2(\bbGa_i\fe\bbv_t)^2}\right)\leq C_{M_0,m_0}\tsigma_k^2\tsigma_l^2\Vert\tilde{\bbGa}\Vert_F^2,$$
	i.e. the paralleling results of (\ref{Eq of finite variance}). In fact, the proofs of (\ref{Eq of positive covariance 1}) are independent of \(\tsigma_t\), so we only focus on (\ref{Eq of positive covariance 2}). And the first concern is that whether (\ref{Eq of finite integration of positive covariance}) holds under Assumptions of Proposition \ref{Thm of CLT multiple unit roots}. Based on Lemma \ref{Lem of singular values}, we know that \(\tsigma_{10}\asymp C_{d,d_0}10^{-d_0}\), so (\ref{Eq of finite integration of positive covariance}) is valid. The other concern is that whether the summations in \(\frac{d}{d\tau}H_{k,l}(\tau)\) are finite, which is equivalent to whether \(\sum_{t=1}^T\tsigma_t^2\) is finite or not, and the second term in Lemma \ref{Lem of singular values} has concluded it.
	
	Next, we will give the paralleling results of Theorem \ref{Thm of convergence in probability}, and the key step is to show Lemma \ref{Lem of alpha}. Since \(\sum_{t=1}^T\tsigma_t^2\) is finite, it is straightforward to repeat the proofs of Lemma \ref{Lem of alpha} and conclude that
	$$\lim_{n\to\infty}\sqrt{n}\mbE\left[1-\alpha_{k,k}^2\right]=0$$
	under Assumptions of Proposition \ref{Thm of CLT multiple unit roots}, so we omit details here. Similarly, we can also obtain
	$$\frac{\hla_k-\mbE[\hla_k]}{\sqrt{n}}\overset{\mbP}{\longrightarrow}\frac{1}{\sqrt{n}}\sum_{i=1}^n\frac{\tsigma_k^2(\bbGa_i\fe\bbv_k)^2}{\sum_{t=1}^T\tsigma_t^2(\bbGa_i\fe\bbv_t)^2}-\mbE\left[\frac{\tsigma_k^2(\bbGa_i\fe\bbv_k)^2}{\sum_{t=1}^T\tsigma_t^2(\bbGa_i\fe\bbv_t)^2}\right],$$
	as what we have done in Theorem \ref{Thm of convergence in probability}. Finally, given the additional conditions like \(m\)-dependence in Assumption \ref{Ap of m dependent}, we can derive that
	$$\frac{\sqrt{n}}{\bbm_{k,n}}\left(\frac{\hla_k}{n}-\mbE[\hla_k]\right)\overset{d}{\longrightarrow}\mcN(0,1)$$
	by Lemma \ref{Lem of preliminary CLT}, where $\bbm_{k,n}$ is defined in Proposition \ref{Thm of CLT multiple unit roots}. Now, similar as what we have done in Theorem \ref{Thm of CLT I1}, we only need to show that
	$$\lim_{n\to\infty}n^{-1/2}\left|\mbE\left[\hla_k/n-\mbM_{k,n}\right]\right|=0,$$
	which can be concluded by the same arguments as in Theorem \ref{Thm of CLT I1} based on \(\lim_{n\to\infty}\sqrt{n}\mbE[1-\alpha_{k,k}^2]=0\) and \(\sum_{t=1}^T\tsigma_t^2\) is finite.
\end{proof}
Finally, let's further investigate how the nonstationary roots effect the asymptotic spectral properties of the sample covariance matrix of $\bbX=[X_1,\cdots,X_T]$ generated by \eqref{Eq of AR nonstationary}.
\begin{remark}\label{Rem of why not covariance 2}
	Given the observations \(\bbX\) as in Proposition \ref{Thm of CLT multiple unit roots}, and denote \(\hSig\) to be the sample covariance matrix of \(\bbX\), i.e.
	\begin{align*}
		&\hSig=\frac{1}{n}\bbM\bbX'\bbX\bbM=\frac{1}{n}\bbM(\mbT^{-1})'\bbe'\bbe\mbT^{-1}\bbM,
	\end{align*}
	by (\ref{Eq of multiple unit roots SVD}), the SVD of \(\mbT^{-1}\bbM\) is \(\sum_{k=1}^T\sigma_k\bbv_k\bbw_k'\), so we have
	$$\Vert\hSig\Vert\geq\bbw_1'\hSig\bbw_1=\frac{\sigma_1^2}{n}\bbv_1'\bbe'\bbe\bbv_1=\frac{\sigma_1^2}{n}\bbv_1'\fe'\bbGa'\bbGa\fe\bbv_1\geq\frac{\sigma_1^2m_0}{n}\Vert\fe\bbv_1\Vert_2^2,$$
	where we use Assumption \ref{Ap of nonpanel}. Since \(\fe\) is a standard Gaussian random matrix, it implies that \(\fe\bbv_1\sim\mcN(\boldsymbol{0},\bbI_n)\), by the Gaussian concentration inequality, \(\mbP(n^{-1}\Vert\fe\bbv_1\Vert_2^2\geq1/2)\geq1-\mrO(e^{-Cn})\). By Lemma \ref{Lem of singular values}, it further deduces that \(\Vert\hSig\Vert\geq\mrO(\sigma_1^2)=\mrO(T^{2d_1})\) with probability of 1.
\end{remark}
\section{Applications}\label{Sec of applications}
\setcounter{equation}{0}
\def\theequation{\thesection.\arabic{equation}}
\setcounter{subsection}{0}
\subsection{Unit root tests}\label{sec of unit root test}
For a $n$-dimensional stochastic process $X_t$ generated by $X_t=(\bbI-\bbPi)\phi+\bbPi X_{t-1}+e_t$, the unit root test is to test
\begin{align}
    H_0:\bbPi=\bbI_n,\quad\text{versus}\quad H_1:\Vert\bbPi\Vert<1.\label{Eq of unit root test}
\end{align}
Under $H_0$, $X_t$ is a random walk, and we have established the CLT for extreme eigenvalues of the sample correlation matrix of the high-dimensional random walks in \S\ref{Sec of correlation independent} and \S\ref{Sec of correlation dependent}. In this section, we will investigate the asymptotic spectral properties of the sample correlation matrix of $X_t$ under $H_1$. To distinguish with $X_t$, let's consider a $n$-dimensional stochastic process $Y_t$ generated by
\begin{align}
	Y_t=\bbPi Y_{t-1}+e_t,\quad e_t=\bbGa\fe_t=\bbGa\sum_{k=0}^{\infty}\Psi_k\varepsilon_{t-k},\label{Eq of Yt}
\end{align}
where $\bbPi\in\mbR^{n\times n}$ such that \(\tau_0:=\Vert\bbPi\Vert\leq1\) and the cross-sectional matrix $\bbGa$ and $\{\Psi_k:k\in\mbN\}$ satisfy Assumption \ref{Ap of nonpanel} and \ref{Ap of panel lag polynomial}, respectively. For the alternative hypothesis $H_1$, the basic form is that \(\tau_0=\Vert\bbPi\Vert<1\), and we say \(Y_t\) is totally stationary under this situation. On the other hand, \cite{pesaran2012interpretation} suggested that it would be more proper to consider the following $H_1$:
\begin{align}
	H_1:\bbPi=\left(\begin{array}{cc}
		\bbI_{n_1}&\boldsymbol{0}_{n_1\times n_2}\\
		\boldsymbol{0}_{n_2\times n_1}&\widetilde{\bbPi}
	\end{array}\right),\label{Eq of partially stationary}
\end{align}
where \(\Vert\widetilde{\bbPi}\Vert<1\) and \(\lim_{n\to\infty}n_1/n\in[0,1)\). In this case, we say \(Y_t\) is partially stationary. Now, given the data matrix \(\bbY=[Y_1,\cdots,Y_T]\in\mbR^{n\times T}\) generated by \eqref{Eq of Yt}, the corresponding sample correlation matrix of $\bbY$ is 
\begin{align}
	\hat{\bbR}:=\bbD^{-1/2}(n^{-1}\bbY\bbM\bbY')\bbD^{-1/2},\label{Eq of correlation Yt}
\end{align}
where \(\bbD:=n^{-1}\diag(\bbY\bbM\bbY')\) and we abuse the notation \(\hat{\bbR}\) and \(\bbD\) here. In \S\ref{ssec of totally stationary alternatives} and \S\ref{ssec of partially stationary alternatives}, we will investigate the asymptotic properties of $n^{-1}\Vert\hR\Vert$ defined in \eqref{Eq of correlation Yt} for totally stationary $H_1$ in \eqref{Eq of unit root test} and partially stationary $H_1$ in \eqref{Eq of partially stationary}, respectively.
\subsubsection{Totally stationary alternatives}\label{ssec of totally stationary alternatives}
In this section, we will show that 
\begin{thm}\label{Thm of H1 norm}
	Under Assumptions {\rm \ref{Ap of highdimensionality}, \ref{Ap of panel lag polynomial}} and {\rm \ref{Ap of nonpanel}}, suppose $\varepsilon_t=(\varepsilon_{1,t},\cdots,\varepsilon_{n,t})'$ in \eqref{Eq of Yt} satisfies that \(\varepsilon_t=(\varepsilon_{1,t},\cdots,\varepsilon_{n,t})'\) are independent random vectors such that all \(\varepsilon_{i,t}\) are independent and \(\mbE[\varepsilon_{i,t}]=0,\mbE[\varepsilon_{i,t}^2]=1\) and \(\kappa_{10}:=\sum_{i,t\in\mbZ}\mbE[\varepsilon_{i,t}^{10}]<\infty\). Let \(\hR\) \eqref{Eq of correlation Yt} be the sample correlation matrix of \(\bbY=[Y_1,\cdots,Y_T]\) generated by \eqref{Eq of Yt}, then we have
	\(n^{-1}\Vert\hat{\bbR}\Vert\overset{\mbP}{\longrightarrow}0\).
\end{thm}
We will prove the above Theorem by the following two steps:
\begin{itemize}
	\item[1.] Prove that \(\Vert\bbD\Vert=n^{-1}\Vert\diag(\bbY\bbM\bbY')\Vert>C_{B,b,c,\tau_0}\) with probability at least of \(1-\mrO(T^{-1})\).
	\item[2.] Prove that \(n^{-2}\Vert\bbY\bbM\bbY'\Vert\leq\mrO(n^{-1/15})\) with probability at least of \(1-\mrO(T^{-1/6}\log^5(T))\).
\end{itemize}
For simplicity, we make some notations here:
\begin{itemize}
	\item Given two matrices \(\bbA,\bbB\), the notation ``\(\bbA\otimes\bbB\)'' represents the kronecker product between \(\bbA\) and \(\bbB\).
	\item By (\ref{Eq of Yt}), \(Y_t=\bbPi^tY_0+\sum_{k=1}^t\bbPi^{t-k}\bbGa\Psi(L)\varepsilon_t\), where \(\fe_t=\sum_{k=0}^{\infty}\Psi_k\varepsilon_{t-k}\). Let
	\begin{align}
		\bbV(\bbPi):=\left(\begin{array}{ccccc}
			\bbI_n&\boldsymbol{0}&\cdots&\cdots&\boldsymbol{0}\\
			\bbPi&\bbI_n&\boldsymbol{0}&\cdots&\boldsymbol{0}\\
			\vdots&\ddots&\ddots&\ddots&\vdots\\
			\bbPi^{T-2}&\cdots&\bbPi&\bbI_n&\boldsymbol{0}\\
			\bbPi^{T-1}&\bbPi^{T-2}&\cdots&\bbPi&\bbI_n
		\end{array}\right)\quad\vec{\fe}=\left(\begin{array}{c}
			\fe_1\\\fe_2\\\vdots\\\fe_T
		\end{array}\right)\quad\vec{\bby}=\left(\begin{array}{c}
			Y_1\\Y_2\\\vdots\\Y_T
		\end{array}\right),\label{Eq of bbV bbPi}
	\end{align}
	where \(\bbV(\bbPi)\in\mbR^{nT\times nT}\) and \(\vec{\fe},\vec{\bby}\in\mbR^{nT}\). Then we obtain
	\begin{align*}
		\vec{\bby}=\bbV(\bbPi)(\bbI_T\otimes\bbGa)\vec{\fe}+\diag(\bbPi,\cdots,\bbPi^T)(\boldsymbol{1}_{T\times1}\otimes Y_0),
	\end{align*}
	where \(\otimes\) is the Kronecker product and \(\diag(\bbPi,\cdots,\bbPi^T)\in\mbR^{nT\times nT}\) has \(T\times T\) blocks such that its \(t\)-th diagonal block is \(\bbPi^t\).
	\item For \(\fe_t\) in \eqref{Eq of Yt}, by Assumption \ref{Ap of panel lag polynomial}, denote
	\begin{align}
		&\fe_t=\sum_{k=0}^{\infty}\Psi_k\varepsilon_{t-k}:=\sum_{k=0}^{T-1}\Psi_k\varepsilon_{t-k}+r_t\label{Eq of rt}
	\end{align}
	and
	\begin{align}
		\bbXi:=\left(\begin{array}{cccccc}
			\boldsymbol{0}&\cdots&\boldsymbol{0}&\Psi_0&\cdots&\Psi_{T-1}\\
			\vdots&\iddots&\iddots&\iddots&&\vdots\\
			\boldsymbol{0}&\Psi_0&\cdots&&\Psi_{T-1}&\cdots\\
			\Psi_0&\cdots&&\Psi_{T-1}&\cdots&\boldsymbol{0}
		\end{array}\right)\quad\vec{\bbve}=\left(\begin{array}{c}
			\varepsilon_T\\\varepsilon_{T-1}\\\vdots\\\varepsilon_{-T+2}
		\end{array}\right)\quad\vec{\bbr}=\left(\begin{array}{c}
			r_1\\\vdots\\r_t
		\end{array}\right),\label{Eq of bbXi}
	\end{align}
	where \(\bbXi\in\mbR^{nT\times n(2T-1)}\) and \(\vec{\bbve}\in\mbR^{n(2T-1)},\vec{\bbr}\in\mbR^{nT}\) so we obtain \(\vec{\fe}=\bbXi\vec{\bbve}+\vec{\bbr}\) and
	\begin{align}
		\vec{\bby}&=\bbV(\bbPi)(\bbI_T\otimes\bbGa)\bbXi\vec{\bbve}+\bbV(\bbPi)(\bbI_T\otimes\bbGa)\vec{\bbr}+\diag(\bbPi,\cdots,\bbPi^T)(\boldsymbol{1}_{T\times1}\otimes Y_0)\notag\\
		:&=\bbV(\bbPi)(\bbI_T\otimes\bbGa)\bbXi\vec{\bbve}+\vec{\bbz}.\label{Eq of vec bby}
	\end{align} 
\end{itemize}
Before proving $\lim_{n\to\infty}\mbP(\Vert\bbD\Vert>C_{B,b,c,\tau_0})=1$ and $\lim_{n\to\infty}\mbP(n^{-1}\Vert\bbY\bbM\bbY'\Vert\leq\mrO(n^{-1/15}))=1$, we need the following Lemmas \ref{Lem of bbV} and \ref{Lem of entrywise concentration} for preliminaries. First, we cite the following results in \cite{bai2010spectral}, which plays important roles in proving Lemma \ref{Lem of entrywise concentration}:
\begin{lem}[Lemma B.26, \cite{bai2010spectral}]\label{Lem of Bai 26}
	Let \(\bbA\) be an \(n\times n\) nonrandom matrix and \(\bbx=(x_1,\cdots,x_n)'\) be a random vector of independent entries. Assume that \(\mbE[x_i]=1,\mbE[x_i^2]=1\) and \(\max_{1\leq i\leq n}\mbE[|x_i|^l]\leq\kappa_l\).Then, for any \(p\geq1\),
	\begin{align*}
		\mbE[|\bbx'\bbA\bbx-\tr(\bbA)|^p]\leq C_{p,\kappa_{2p}}\tr(\bbA\bbA')^{p/2},
	\end{align*}
	where \(C_{p,\kappa_{2p}}\) is a constant depending on \(p\) and \(\kappa_{2p}\) only.
\end{lem}
The following lemma provide the upper and lower bound for the singular values of $\bbV(\bbPi)$ and $\bbXi$ defined in \eqref{Eq of bbV bbPi} and \eqref{Eq of bbXi}, respectively.
\begin{lem}\label{Lem of bbV}
	Denote \(\sigma_{\min}(\bbV(\bbPi))\) and \(\sigma_{\max}(\bbV(\bbPi))\) to be the smallest and largest singular value of \(\bbV(\bbPi)\) in \eqref{Eq of bbV bbPi}, then
	$$(1+|\tau_0|)^{-1}\leq\sigma_{\min}(\bbV(\bbPi))\leq\sigma_{\max}(\bbV(\bbPi))\leq(1-|\tau_0|)^{-1}.$$
	Moreover, under Assumption {\rm \ref{Ap of panel lag polynomial}}, we have 
	$$b\leq\sigma_{\min}(\bbXi)\leq\sigma_{\max}(\bbXi)\leq B,$$
    where $\bbXi$ is defined in \eqref{Eq of bbXi}.
\end{lem}
\begin{proof}
	For the first term, suppose the SVD of \(\bbPi\) is \(\mcU\diag(\bbtau)\mcV'\), where \(\mcU,\mcV\) are two orthogonal matrices and \(\diag(\bbtau)=\diag(\tau_1,\cdots,\tau_n)\) such that \(\max_{1\leq i\leq n}|\tau_i|<\tau_0<1\). It is easy to see that
	\begin{align*}
		\bbV(\bbPi)^{-1}=(\mcU\otimes\bbI_T)\left(\begin{array}{ccccc}
			\bbI_n&\boldsymbol{0}&\boldsymbol{0}&\cdots&\boldsymbol{0}\\
			-\diag(\bbtau)&\bbI_n&\boldsymbol{0}&\cdots&\vdots\\
			\vdots&&\ddots&\ddots&\vdots\\
			\boldsymbol{0}&\cdots&-\diag(\bbtau)&\bbI_n&\boldsymbol{0}\\
			\boldsymbol{0}&\cdots&\boldsymbol{0}&-\diag(\bbtau)&\bbI_n
		\end{array}\right)(\mcV'\otimes\bbI_T),
	\end{align*}
	so \((\mcU\otimes\bbI_T)\bbV(\bbPi)^{-1}(\mcV\otimes\bbI_T)\) is a block toeplitz matrix. By Lemma 4.3 in \cite{gutierrez2012block}, it gives that
	\begin{align*}
		\Vert\bbV(\bbPi)^{-1}\Vert=\Vert(\mcU\otimes\bbI_T)\bbV(\bbPi)^{-1}(\mcV\otimes\bbI_T)\Vert\leq1+\max_{1\leq i\leq n}|\tau_i|<1+\tau_0,
	\end{align*}
	where we use \(\Vert\mcU\otimes\bbI_T\Vert=\Vert\mcV\otimes\bbI_T\Vert=1\). On the other hand, \((\mcU\otimes\bbI_T)\bbV(\bbPi)^{-1}(\mcV\otimes\bbI_T)\) is diagonally dominated, by Corollary 2 in \cite{varah1975lower}, we have
	$$\sigma_{\min}(\bbV(\bbPi)^{-1})\geq1-\max_{1\leq i\leq n}|\tau_i|>1-\tau_0.$$
	For the second term, since \(\bbXi\bbXi'\) is a block toeplitz matrix as follows:
	\begin{align*}
		\bbXi\bbXi'=\left(\begin{array}{cccc}
			\Xi(0)&\Xi(1)&\cdots&\Xi(T-1)\\
			\Xi(1)&\Xi(0)&\Xi(1)&\cdots\\
			\vdots&\ddots&\ddots&\vdots\\
			\Xi(T-1)&\cdots&\Xi(1)&\Xi(0)
		\end{array}\right),
	\end{align*}
	where \(\Xi(k)=\sum_{t=0}^{T-1-k}\Psi_t\Psi_{t+k}\) are all diagonal by Assumption \ref{Ap of panel lag polynomial}. By Theorem 4.4 in \cite{gutierrez2012block}, it yields that
	\begin{align*}
		&\sigma_{\max}(\bbXi)^2\leq\max_{1\leq j\leq n}\sup_{x\in[0,2\pi]}\left|\sum_{t=0}^{T-1}\varphi_{j,t}^2+2\sum_{k=1}^{T-1}\cos(kx)\sum_{t=0}^{T-1-k}\varphi_{j,t}\varphi_{j,t+k}\right|\\
		&=\max_{1\leq j\leq n}\sup_{x\in[0,2\pi]}\left|\sum_{t=0}^{T-1}\varphi_{j,t}e^{{\rm i}tx}\right|^2\leq B^2,
	\end{align*}
	where we use Assumption \ref{Ap of panel lag polynomial}. On the other hand, by Theorem 4.4 in \cite{gutierrez2012block} and Assumption \ref{Ap of panel lag polynomial} again, we have
	\begin{align*}
		&\sigma_{\min}(\bbXi)^2\geq\min_{1\leq j\leq n}\inf_{x\in[0,2\pi]}\left|\sum_{t=0}^{T-1}\varphi_{j,t}^2+2\sum_{k=1}^{T-1}\cos(kx)\sum_{t=0}^{T-1-k}\varphi_{j,t}\varphi_{j,t+k}\right|\\
		&=\min_{1\leq j\leq n}\inf_{x\in[0,2\pi]}\left|\sum_{t=0}^{T-1}\varphi_{j,t}e^{{\rm i}tx}\right|^2\geq b^2,
	\end{align*}
	which completes our proof.
\end{proof}
Next, recall that the data matrix \(\bbY=[Y_1,\cdots,Y_T]\) is generated by \eqref{Eq of Yt}, the following Lemma \ref{Lem of entrywise concentration} is to provide an entrywise approximation for $n^{-1}\bbY\bbM\bbY'$. Precisely, by \eqref{Eq of vec bby}, it gives that
\begin{align}
	&Y_{k\cdot}=[(\bbI_T\otimes\bbi_k^{(n)})'\vec{\bby}]':=(\mcI_k\vec{\bby})',\label{Eq of mcI}
\end{align}
where \(Y_{k\cdot}\) is the \(k\)-th row of \(\bbY\) and \(\bbi_k^{(n)}\in\mbR^{n\times 1}\) such that its \(k\)-th entry is 1 while others are zero. Moreover, denote 
\begin{align}
	\bbH_{k,l}:=\bbXi'(\bbI_T\otimes\bbGa)'\bbV(\bbPi)'\mcI_k'\bbM\mcI_l\bbV(\bbPi)(\bbI_T\otimes\bbGa)\bbXi,\label{Eq of bbH}
\end{align}
where \(k,l\in\{1,\cdots,n\}\), and 
\begin{align}
	\bbP:=\bbV(\bbPi)(\bbI_T\otimes\bbGa)\bbXi\bbXi'(\bbI_T\otimes\bbGa)'\bbV(\bbPi)'=[\bbP^{k,l}]_{T\times T},\quad\mathring{\bbP}:=\frac{1}{n}\sum_{t=1}^T\bbP^{t,t},\label{Eq of bbP}
\end{align}
where \(\bbP\in\mbR^{nT\times nT}\) and \(\bbP^{k,l}\in\mbR^{n\times n}\) is the \((k,l)\)-th block of \(\bbP\). Now, we will show that
\begin{lem}\label{Lem of entrywise concentration}
	Under Assumptions in Theorem {\rm \ref{Thm of H1 norm}}, let \(\mathring{P}_{k,l}\) be the \((k,l)\)-th entry of \(\mathring{\bbP}\) defined in {\rm (\ref{Eq of bbP})}, where \(k,l\in\{1,\cdots,n\}\), then
	\begin{align}
		\mbP(|n^{-1}Y_{k\cdot}\bbM Y_{l\cdot}'-\mathring{P}_{k,l}|>n^{-1/15})\leq C_{\kappa_{10},B,M_0,\tau_0,c}n^{-13/6}.\label{Eq of Yk Yl concentration 2}
	\end{align}
\end{lem}
\begin{proof}
	According to (\ref{Eq of mcI}) and (\ref{Eq of vec bby}), we have
	\begin{align}
		&n^{-1}Y_{k\cdot}\bbM Y_{l\cdot}'=n^{-1}\vec{\bby}'\mcI_k'\bbM\mcI_l\vec{\bby}=n^{-1}\vec{\bbve}'\bbH_{k,l}\vec{\bbve}+n^{-1}\vec{\bbz}'\mcI_k'\bbM\mcI_l\vec{\bbz}\notag\\
		&+n^{-1}\vec{\bbz}'\mcI_k'\bbM\mcI_l\bbV(\bbPi)(\bbI_T\otimes\bbGa)\bbXi\vec{\bbve}+n^{-1}\vec{\bbz}'\mcI_l'\bbM\mcI_k\bbV(\bbPi)(\bbI_T\otimes\bbGa)\bbXi\vec{\bbve}.\label{Eq of Yk Yl}
	\end{align}
	Let's first show that 
	\begin{align}
		\mbP(n^{-1/2}\Vert\vec{\bbz}\Vert_2>n^{-1/4})\leq C_{\kappa_{10},B,\tau_0,c}T^{-5/2}.\label{Eq of bbz concetration}
	\end{align}
	By (\ref{Eq of vec bby}), since \(\vec{\bbz}=\bbV(\bbPi)(\bbI_T\otimes\bbGa)\vec{\bbr}+\diag(\bbPi,\cdots,\bbPi^T)(\boldsymbol{1}_{T\times 1}\otimes Y_0)\). It is enough to show that
	\begin{align*}
		n^{-1/2}\Vert\diag(\bbPi,\cdots,\bbPi^T)(Y_0\otimes\bbI_T)\Vert_2=\mrO(n^{-1/2})
	\end{align*}
	and
	\begin{align}
		\mbP(n^{-1/2}\Vert\bbV(\bbPi)(\bbI_T\otimes\bbGa)\vec{\bbr}\Vert_2>n^{-1/4})\leq C_{\kappa_{10},B,\tau_0,c}T^{-5/2}.\label{Eq of bbr concentration}
	\end{align}
	For the first term, it can be derived by
	\begin{align*}
		&\Vert\diag(\bbPi,\cdots,\bbPi^T)(Y_0\otimes\bbI_T)\Vert_2^2=\sum_{t=1}^T\Vert\bbPi^t Y_0\Vert_2^2<\Vert Y_0\Vert_2^2\sum_{t=1}^T\tau_0^{2t}<(1-\tau_0)^{-1}\Vert Y_0\Vert_2^2.
	\end{align*}
	Next, Recall the definition of \(r_t\) in (\ref{Eq of rt}), by Assumption \ref{Ap of finite integration}, we know that \(\varepsilon_{j,t}\) have uniformly bounded \(10\)-th moment, then
	$$\mathbb{E}[r_{j,t}^{10}]\leq C\sum_{k_1,k_2,k_3,k_4,k_5=T}^{\infty}\mbE\left[\prod_{\alpha=1}^5\varphi_{j,k_{\alpha}}^2\varepsilon_{j,k_{\alpha}}^2\right]\leq\frac{C_{\kappa_{10}}}{T^{10}}\left(\sum_{k=T}^{\infty}k|\varphi_{j,k}|\right)^{10}\leq C_{\kappa_{10},B}T^{-10},$$
	where we use Assumption \ref{Ap of panel lag polynomial}. By the Chebyshev's inequality, it gives that
	$$\mathbb{P}(|r_{j,t}|>\epsilon)\leq C_{\kappa_{10},B}(T\epsilon)^{-10}$$
	and
	$$\mathbb{P}(\Vert\vec{\bbr}\Vert_2>n^{1/4})\leq\sum_{j=1}^n\sum_{t=1}^T\mathbb{P}(|r_{j,t}|>n^{1/4}(nT)^{-1/2})\leq C_{\kappa_{10},B,c}T^{-5/2}.$$
	Combining with Lemma \ref{Lem of bbV} and \(\Vert\bbI_T\otimes\bbGa\Vert\leq\Vert\bbGa\Vert\leq M_0^{1/2}\) by Assumption \ref{Ap of nonpanel}, we can imply (\ref{Eq of bbr concentration}). Furthermore, according to Lemma \ref{Lem of Bai 26}, we know that
	\begin{align*}
		n^{-5}\mbE[|\vec{\bbve}'\bbH_{k,l}\vec{\bbve}-\tr(\bbH_{k,l})|^5]\leq C_{\kappa_{10}}n^{-5}\tr(\bbH_{k,l}\bbH_{k,l}')^{5/2}.
	\end{align*}
	Since \(\Vert\bbM\Vert,\Vert\mcI_k\Vert,\Vert\mcI_l\Vert\leq1\), by Lemma \ref{Lem of bbV} and Assumption \ref{Ap of nonpanel}, we have
	\begin{align*}
		&\tr(\bbH_{k,l}\bbH_{k,l}')\leq T\Vert\bbH_{k,l}\Vert^2\leq T(1-\tau_0)^{-2}M_0B^2.
	\end{align*}
	Hence, we obtain that
	\begin{align}
		n^{-5}\mbE[|\vec{\bbve}'\bbH_{k,l}\vec{\bbve}-\tr(\bbH_{k,l})|^5]\leq C_{\kappa_{10},B,M_0,\tau_0,c}n^{-5/2}\notag
	\end{align}
	and
	\begin{align}
		\mbP(n^{-1}|\vec{\bbve}'\bbH_{k,l}\vec{\bbve}-\tr(\bbH_{k,l})|>\epsilon)\leq C_{\kappa_{10},B,M_0,\tau_0,c}\epsilon^{-5}n^{-5/2}.\label{Eq of bbH concentration}
	\end{align}
	As a result, combine (\ref{Eq of bbz concetration}) and (\ref{Eq of bbH concentration}), by the Cauchy's inequality, we have
	\begin{align*}
		&n^{-1}|\vec{\bbz}'\mcI_k'\bbM\mcI_l\bbV(\bbPi)(\bbI_T\otimes\bbGa)\bbXi\vec{\bbve}|\leq n^{-1/2}\Vert\vec{\bbz}\Vert_2\cdot(n^{-1}\vec{\bbve}'\bbH_{l,l}\vec{\bbve})^{1/2}\\
		&\leq2n^{-1/4}(n^{-1}\tr(\bbH_{l,l}))^{1/2}\leq C_{B,M_0,\tau_0,c}n^{-1/4}
	\end{align*}
	with probability at least of \(1-C_{\kappa_{10},B,M_0,\tau_0,c}n^{-5/2}\), so does \(n^{-1}|\vec{\bbz}'\mcI_l'\bbM\mcI_k\bbV(\bbPi)(\bbI_T\otimes\bbGa)\bbXi\vec{\bbve}|\). Therefore, by (\ref{Eq of Yk Yl}), it gives that
	\begin{align}
		\mbP(n^{-1}|Y_{k\cdot}\bbM Y_{l\cdot}'-\tr(\bbH_{k,l})|>n^{-1/15})\leq C_{\kappa_{10},B,M_0,\tau_0,c}n^{-13/6}\label{Eq of Yk Yl concentration 1}
	\end{align}
	for all \(k,l\in\{1,\cdots,n\}\). Moreover, since
	\begin{align*}
		&\tr(\bbH_{k,l})=\tr(\bbXi'(\bbI_T\otimes\bbGa)'\bbV(\bbPi)'\mcI_l'\mcI_k\bbV(\bbPi)(\bbI_T\otimes\bbGa)\bbXi)\\
		&-T^{-1}\boldsymbol{1}_{1\times T}\mcI_k\bbV(\bbPi)(\bbI_T\otimes\bbGa)\bbXi\bbXi'(\bbI_T\otimes\bbGa)'\bbV(\bbPi)'\mcI_l'\boldsymbol{1}_{T\times1},
	\end{align*}
	by Lemma \ref{Lem of bbV} and Assumption \ref{Ap of nonpanel}, we know that
	\begin{align*}
		&T^{-1}|\boldsymbol{1}_{1\times T}\mcI_k\bbV(\bbPi)(\bbI_T\otimes\bbGa)\bbXi\bbXi'(\bbI_T\otimes\bbGa)'\bbV(\bbPi)'\mcI_l'\boldsymbol{1}_{T\times1}|\\
		&\leq\Vert\mcI_k\bbV(\bbPi)(\bbI_T\otimes\bbGa)\bbXi\bbXi'(\bbI_T\otimes\bbGa)'\bbV(\bbPi)'\mcI_l'\Vert\leq(1-\tau_0)^{-2}M_0(2B^2\log T).
	\end{align*}
	and
	\begin{align*}
		&\tr(\bbXi'(\bbI_T\otimes\bbGa)'\bbV(\bbPi)'\mcI_l'\mcI_k\bbV(\bbPi)(\bbI_T\otimes\bbGa)\bbXi)\\
		&=\sum_{t=1}^T(\bbi_k^{(n)}\otimes\bbi_t^{(T)})'\bbV(\bbPi)(\bbI_T\otimes\bbGa)\bbXi\bbXi'(\bbI_T\otimes\bbGa)'\bbV(\bbPi)'(\bbi_l^{(n)}\otimes\bbi_t^{(T)})=n(\bbi_k^{(n)})'\mathring{\bbP}\bbi_l^{(n)},
	\end{align*}
	where \(\mathring{\bbP}=[\mathring{P}_{k,l}]_{n\times n}\) has been defined in (\ref{Eq of bbP}). Now, combine the above results with (\ref{Eq of Yk Yl concentration 1}), we can conclude (\ref{Eq of Yk Yl concentration 2}).
\end{proof}
Here, we first conclude a upper bound for $n^{-2}\Vert\bbY\bbM\bbY'\Vert$ as follows:
\begin{pro}\label{Thm of H1 numerator}
	Under Assumptions in Theorem {\rm \ref{Thm of H1 norm}}, we have \(n^{-2}\Vert\bbY\bbM\bbY'\Vert\overset{\mbP}{\longrightarrow}0\).
\end{pro}
\begin{proof}
	Let \(\mcE\) be an event
	\begin{align*}
		\mcE:=\{k,l\in\{1,\cdots,n\}:|n^{-1}Y_{k\cdot}\bbM Y_{l\cdot}'-\mathring{P}_{k,l}|\leq n^{-1/15}\}.
	\end{align*}
	According the (\ref{Eq of Yk Yl concentration 2}), we know that
	\begin{align*}
		\mbP(\mcE)\geq1-\sum_{k,l=1}^n\mbP(|n^{-1}Y_{k\cdot}\bbM Y_{l\cdot}'-\mathring{P}_{k,l}|>n^{-1/15})\geq1-C_{B,\kappa_{10},\tau_0,c}n^{-1/6}.
	\end{align*}
	Therefore, under \(\mcE\), let
	\begin{align*}
		\Delta:=\frac{1}{n}\bbY\bbM\bbY'-\mathring{\bbP},
	\end{align*}
	where \(\Delta=[\Delta_{k,l}]\) such that \(|\Delta_{k,l}|\leq n^{-1/15}\). By the definition of \(\mathring{\bbP}\) in (\ref{Eq of bbP}), Lemma \ref{Lem of bbV} and Assumption \ref{Ap of nonpanel}, it implies that
	$$\Vert\mathring{\bbP}\Vert\leq\frac{1}{n}\sum_{t=1}^T\Vert\bbP^{t,t}\Vert\leq C_{M_0,\tau_0,c}.$$
	Moreover, \(\Vert\Delta\Vert\leq\Vert\Delta\Vert_F\leq n^{14/15}\). Hence, we obtain that
	\begin{align*}
		\frac{1}{n^2}\Vert\bbY\bbM\bbY'\Vert\leq\frac{1}{n}(\Vert\mathring{\bbP}\Vert+\Vert\Delta\Vert)=\mrO(n^{-1/15})
	\end{align*}
	under \(\mcE\), which completes our proof.
\end{proof}
Next, we provide the lower bound for $\Vert\bbD\Vert$ as follows:
\begin{pro}\label{Thm of H1 dominator}
	Under Assumptions in Theorem {\rm \ref{Thm of H1 norm}}, let \(\bbD=\diag(n^{-1}\bbY\bbM\bbY')\), then \(\Vert\bbD\Vert\geq C_{b,c,\tau_0,m_0}\) with probability at least of \(1-C_{\kappa_{10},B,M_0,\tau_0,c}n^{-7/6}\).
\end{pro}
\begin{proof}
	It is enough to show that each \(|D_{k,k}|\geq C_{b,c,\tau_0,m_0}\) with probability at least of \(1-C_{\kappa_{10},B,M_0,\tau_0,c}n^{-13/6}\). By Lemma \ref{Lem of entrywise concentration}, it is equivalent to
	\begin{align*}
		&n^{-1}\tr(\mcI_k\bbV(\bbPi)(\bbI_T\otimes\bbGa)\bbXi\bbXi'(\bbI_T\otimes\bbGa)'\bbV(\bbPi)'\mcI_k')\geq C_{b,c,\tau_0,m_0}.
	\end{align*}
	Since
	\begin{align*}
		&\sigma_l(\mcI_k\bbV(\bbPi)(\bbI_T\otimes\bbGa)\bbXi)\geq\sigma_{\min}(\bbV(\bbPi))\sigma_{\min}(\bbGa)\sigma_{\min}(\bbXi)\sigma_l(\mcI_1)\geq(1+|\tau_0|)^{-1}bm_0\sigma_l(\mcI_k),
	\end{align*}
	where we use Lemma \ref{Lem of bbV} and Assumptions \ref{Ap of nonpanel}, then
	\begin{align*}
		&n^{-1}\tr(\mcI_k\bbV(\bbPi)(\bbI_T\otimes\bbGa)\bbXi\bbXi'(\bbI_T\otimes\bbGa)'\bbV(\bbPi)'\mcI_k')=\sum_{l=1}^n\sigma_l(\mcI_k\bbV(\bbPi)(\bbI_T\otimes\bbGa)\bbXi)^2\\
		&\geq(1+|\tau_0|)^{-2}b^2m_0^2n^{-1}\sum_{l=1}^n\sigma_l(\mcI_k)^2=C_{b,c,\tau_0,m_0},
	\end{align*}
	where we use the definition of \(\mcI_k\) in (\ref{Eq of mcI}).
\end{proof}
Finally, combining Propositions \ref{Thm of H1 numerator} and \ref{Thm of H1 dominator}, we know that
\begin{align*}
	\frac{1}{n}\Vert\hat{\bbR}\Vert\leq\Vert\bbD^{-1}\Vert\cdot\frac{1}{n^2}\Vert\bbY\bbM\bbY'\Vert\overset{\mbP}{\longrightarrow}0,
\end{align*}
which completes the proof of Theorem \ref{Thm of H1 norm}.
\subsubsection{Partially stationary alternatives}\label{ssec of partially stationary alternatives}
\begin{thm}\label{Thm of partially stationary alternatives}
	Under Assumptions in Theorem {\rm \ref{Thm of H1 norm}}, assume \(\bbPi\) satisfies {\rm (\ref{Eq of partially stationary})} such that
	\begin{align}
		c_1:=\lim_{n\to\infty}\frac{n_1}{n}\in[0,1),\label{Eq of c1}
	\end{align}
	then we have
	$$\lim_{n\to\infty}\mbP(n^{-1}\Vert\hat{\bbR}\Vert\leq c_1\mbE[\fM_{1,1}])=1,$$
    where $\hR$ and $\fM_{1,1}$ are defined in \eqref{Eq of correlation Yt} and \eqref{Eq of fM}, respectively.
\end{thm}
\begin{proof}
	For simplicity, let's define
	\begin{align*}
		Y_t^{(1)}=(Y_{1,t},\cdots,Y_{n_1,t})'\quad{\rm and}\quad Y_t^{(2)}=(Y_{n_1+1,t},\cdots,Y_{n,t})',
	\end{align*}
	where \(Y_t=(Y_{1,t},\cdots,Y_{n,t})\), \(\bbY^{(1)}:=[Y_1^{(1)},\cdots,Y_T^{(1)}]\) and \(\bbY^{(2)}:=[Y_1^{(2)},\cdots,Y_T^{(2)}]\). Further denote
	\begin{align}
		\hat{\bbR}=\left(\begin{array}{cc}
			\hat{\bbR}^{11}&\hat{\bbR}^{12}\\
			\hat{\bbR}^{21}&\hat{\bbR}^{22}
		\end{array}\right),\label{Eq of bbR divide}
	\end{align}
	where \(\hat{\bbR}^{11}\in\mbR^{n_1\times n_1},\hat{\bbR}^{22}\in\mbR^{n_2\times n_2}\) and \(\hat{\bbR}^{12}=(\hat{\bbR}^{12})'\in\mbR^{n_1\times n_2}\). By (\ref{Eq of bbR divide}), \(\hat{\bbR}^{11}\) is the sample correlation matrix of \(\bbY^{(1)}\). According to Theorem \ref{Thm of CLT}, if \(c_1>0\), it gives that
	$$\frac{\Vert\hat{\bbR}^{11}\Vert}{n_1}\overset{\mbP}{\longrightarrow}\mbE[\fM_{1,1}].$$
	Similarly, since \(\hat{\bbR}^{22}\) is the sample correlation matrix of \(\bbY^{(2)}\). By Theorem \ref{Thm of H1 norm}, we have
	$$\frac{\Vert\hat{\bbR}^{22}\Vert}{n_2}\overset{\mbP}{\longrightarrow}0.$$
	For \(\hat{\bbR}^{12}\), notice that
	$$\hat{\bbR}^{12}=n^{-1}(\bbD^{(1)})^{-1/2}\bbY^{(1)}\bbM(\bbY^{(2)})'(\bbD^{(2)})^{-1/2},$$
	where
	$$\bbD^{(1)}:=n^{-1}\diag(\bbY^{(1)}\bbM(\bbY^{(1)})')\quad{\rm and}\quad\bbD^{(2)}:=n^{-1}\diag(\bbY^{(2)}\bbM(\bbY^{(2)})').$$
	Then
	\begin{align*}
		&\Vert\hat{\bbR}^{12}\Vert=n^{-1}\Vert(\bbD^{(1)})^{-1/2}\bbY^{(1)}\bbM^2(\bbY^{(2)})'(\bbD^{(2)})^{-1/2}\Vert\\
		&\leq n^{-1}\Vert(\bbD^{(1)})^{-1/2}\bbY^{(1)}\bbM\Vert\cdot\Vert\bbM(\bbY^{(2)})'(\bbD^{(2)})^{-1/2}\Vert=n^{-1}\Vert\hat{\bbR}^{11}\Vert^{1/2}\cdot\Vert\hat{\bbR}^{22}\Vert^{1/2}\overset{\mbP}{\longrightarrow}0.
	\end{align*}
	Finally, if \(c_1=0\), since all entries of \(\hat{\bbR}^{11}\) is no more than \(1\), then \(\Vert\hat{\bbR}^{11}\Vert\leq n_1\). Now, notice that
	$$\frac{\Vert\hat{\bbR}\Vert}{n}\leq\frac{\Vert\hat{\bbR}^{11}\Vert}{n}+\frac{\Vert\hat{\bbR}^{22}\Vert}{n}+\frac{\Vert\hat{\bbR}^{12}\Vert}{n}=\frac{n_1}{n}\frac{\Vert\hat{\bbR}^{11}\Vert}{n_1}+\frac{n_2}{n}\frac{\Vert\hat{\bbR}^{22}\Vert}{n_2}+\frac{\Vert\hat{\bbR}^{12}\Vert}{n},$$
	by our previous arguments, the later two terms will converge to 0 in probability. For the first one, if \(c_1>0\), we conclude that
	$$\frac{\Vert\hat{\bbR}\Vert}{n}\leq\frac{n_1}{n}\frac{\Vert\hat{\bbR}^{11}\Vert}{n_1}+\frac{n_2}{n}\frac{\Vert\hat{\bbR}^{22}\Vert}{n_2}+\frac{\Vert\hat{\bbR}^{12}\Vert}{n}\overset{\mbP}{\longrightarrow}c_1\mbE[\fM_{1,1}].$$
	Otherwise, \(\lim_{n\to\infty}\frac{n_1}{n}=0\), since \(\Vert\hat{\bbR}^{11}\Vert\leq n_1\), we can still derive that 
	$$\frac{\Vert\hat{\bbR}\Vert}{n}\leq\frac{n_1}{n}\frac{\Vert\hat{\bbR}^{11}\Vert}{n_1}+\frac{n_2}{n}\frac{\Vert\hat{\bbR}^{22}\Vert}{n_2}+\frac{\Vert\hat{\bbR}^{12}\Vert}{n}\overset{\mbP}{\longrightarrow}0,$$
	which completes our proof.
\end{proof}
Finally, since the test statistic for our unit root test \eqref{Eq of unit root test} is
\begin{align}
    \widehat{T}_n(0)=\frac{\sqrt{n}}{\mfm_{1,1}}\left(\frac{\Vert\hR\Vert}{n}-\mbE[\fM_{1,1}]\right),\label{Eq of unit root test statistic}
\end{align}
where $\mfm_{1,1}$ is defined in \eqref{Eq of mfm kk}. Although the explicit value of $\mfm_{1,1}$ is generally unknown, we can use a bootstrap method to solve this difficulty, see \S\ref{sec of estimation variance} for details. Under $H_0$, by Theorem \ref{Thm of CLT I1}, we know that $\widehat{T}_n(0)\overset{d}{\longrightarrow}\mcN(0,1)$. On the other hand, under the totally stationary alternative $H_1$ in \eqref{Eq of unit root test}, Theorem \ref{Thm of H1 norm} implies that $\lim_{n\to\infty}\mbP(\widehat{T}_n(0)\asymp\mrO(-\sqrt{n}))=1$ (``$\asymp$'' is defined in \eqref{Eq of asymp mbP}). Under the partially stationary alternative $H_1$ in \eqref{Eq of partially stationary}, by Theorem \ref{Thm of partially stationary alternatives}, we know that $\lim_{n\to\infty}\mbP(\mbE[\fM_{1,1}]-n^{-1}\Vert\hR\Vert>(1-c_1)\mbE[\fM_{1,1}])=1$. Since $c_1<1$, it gives that $\lim_{n\to\infty}\mbP(\widehat{T}_n(0)\asymp\mrO(-\sqrt{n}))=1$. In summary, under these two kinds alternatives (totally stationary and partially stationary), we always have $\lim_{n\to\infty}\mbP(\widehat{T}_n(0)\asymp\mrO(-\sqrt{n}))=1$, which is different with $\widehat{T}_n(0)\overset{d}{\longrightarrow}\mcN(0,1)$ under $H_0$. And this difference is the key of distinguishing $H_0$ and $H_1$.
\subsection{Determine the number of unit roots in high-dimensional autoregressive processes}\label{sec of number of unit roots}
In this section, we will establish a forward sequential test to determine the number of unit roots for high-dimensional time series data. Let \(X_t\) be an \(n\)-dimensional AR process and define the following hypothesis
\begin{center}
	\(\mbH_0^{(p)}\): \(X_t\) has \(p\) unit roots, i.e. \((1-L)^pX_t=e_t\),
\end{center}
where \(p\in\mbN^+\) and \(e_t=\bbGa\sum_{k=0}^{\infty}\Psi_k\varepsilon_{t-k}\) satisfying Assumptions \ref{Ap of panel lag polynomial} and \ref{Ap of nonpanel} and \(\varepsilon_t\overset{i.i.d.}{\sim}\mcN(\boldsymbol{0},\bbI_n)\). Moreover, define
\begin{center}
	\(\mbH_0^{(0)}\): \(X_t\) is stationary,\quad and\quad\(\mbH_0^{(\infty)}\): \(X_t\) has super nonstationary roots.
\end{center}
Currently, we have established the test procedure of \(\mbH_0^{(0)}\) versus \(\mbH_0^{(1)}\) in \S\ref{sec of unit root test}, which is the unit root test. Thus, starting from testing \(\mbH_0^{(0)}\) versus \(\mbH_0^{(1)}\), suppose we reject \(\mbH_0^{(0)}\), then we will test \(\mbH_0^{(1)}\) versus \(\mbH_0^{(2)}\).
\subsubsection{Test statistics}\label{ssec of test statistics}
When \(X_t\) is a \(n\)-dimensional random walk, i.e. under \(\mbH_0^{(1)}\), given the observations \(\bbX=[X_1,\cdots,X_T]\) satisfying Assumption \ref{Ap of highdimensionality}, we have established the CLT for the largest eigenvalue of \(\bbX\)'s sample correlation matrix in \S\ref{sec of CLT bbR} and \S\ref{sec of CLT correlation dependent}. For the sake of technical issues, when testing \(\mbH_0^{(p)}\) versus \(\mbH_0^{(p+1)}\) for \(p\in\mbN^+\), we will not directly construct the test statistics based on the largest eigenvalue of the sample correlation matrix of \(\bbX\). Precisely, let's first define an operator as follows:
\begin{align}
	\mathscr{T}_p(\bbX)=\left\{\begin{array}{ll}
		\bbX\bbU^{-p}(\bbM\bbU'\bbU\bbM)^{p/2}&p\equiv0\mod2,\\
		\bbX\bbU^{-p}\bbM\bbU'(\bbU\bbM\bbU')^{(p-1)/2}&p\equiv1\mod2.
	\end{array}\right.\label{Eq of operator msT}
\end{align}
Then construct
\begin{align}
	\hR(p)=\mathscr{R}(\mathscr{T}_p(\bbX))=\mathscr{T}_p(\bbX)'\diag(\mathscr{T}_p(\bbX)\mathscr{T}_p(\bbX)')^{-1}\mathscr{T}_p(\bbX),\label{Eq of hat Rp}
\end{align}
and we will use the largest eigenvalue of \(\hR(p)\) to construct the test statistic. For preliminaries, recall that \(\sigma_1\geq\cdots\geq\sigma_T\) are singular values of \(\bbM\bbU'\) in (\ref{Eq of SVD of MU}), as a generalization of $\fM_{k,l}$ defined in \eqref{Eq of fM}, given any \(x\in[1,\infty)\) and \(\{Z_t:t\in\mbN^+,z_t\overset{i.i.d.}{\sim}\mcN(0,1)\}\), define
\begin{align}
	\fM_{k,l}(x)=\frac{(kl)^{-x}Z_kZ_l}{\sum_{t=1}^{\infty}t^{-2x}Z_t^2}.\label{Eq of fMx}
\end{align}
Moreover, given \(\tilde{\bbGa}\) defined in (\ref{Eq of bbF}), let \(\{\vec{z}_t=(z_{1,t},\cdots,z_{n,t})'\sim\mcN(\boldsymbol{0},\tilde{\bbGa}):t=1,\cdots,T-1\}\) be a sequence of $n$-dimensional i.i.d. normal vectors, then define
\begin{align}
	\mfm_{k,k}^2(x)=\frac{1}{n}\sum_{i_1,i_2=1}^n\Cov\left(\frac{\sigma_k^{2x}z_{i_1,k}^2}{\sum_{t=1}^{T-1}\sigma_k^{2x}z_{i_1,t}^2},\frac{\sigma_k^{2x}z_{i_2,k}^2}{\sum_{t=1}^{T-1}\sigma_k^{2x}z_{i_2,t}^2}\right).\label{Eq of mfm}
\end{align}
Now, we will show that
\begin{thm}\label{Thm of forwards}
    Under Assumptions {\rm \ref{Ap of highdimensionality}, \ref{Ap of panel lag polynomial}, \ref{Ap of nonpanel}} and {\rm (\ref{Eq of Ap of panel lag polynomial})}, under \(\mbH_0^{(p)}:(1-L)^pX_t=e_t\), where \(e_t=\bbGa\sum_{k=0}^{\infty}\Psi_k\varepsilon_{t-k}\) and \(\varepsilon_t\overset{i.i.d.}{\sim}\mcN(\boldsymbol{0},\bbI_n)\), let \(\hla_1(p)\) be the largest eigenvalue of \(\hR(p)\) defined in {\rm (\ref{Eq of hat Rp})}, then we have
    \begin{align*}
	\left\{\begin{array}{ll}
		\frac{\sqrt{n}}{\mfm_{1,1}(p)}\left(\frac{\hla_1(p)}{n}-\mbE[\fM_{1,1}(p)]\right)\overset{d}{\longrightarrow}\mcN(0,1),&{\rm under\ }\mbH_0^{(p)},\\
		\frac{\sqrt{n}}{\mfm_{1,1}(p+1)}\left(\frac{\hla_1(p)}{n}-\mbE[\fM_{1,1}(p+1)]\right)\overset{d}{\longrightarrow}\mcN(0,1),&{\rm under\ }\mbH_0^{(p+1)},
	\end{array}\right.
    \end{align*}
    where \(\mfm_{1,1}(x)\asymp\mrO(1)\) for any fixed \(x\in[1,\infty)\) and ``$\asymp$'' is defined in \eqref{Eq of asymp mbP}.
\end{thm}
\begin{proof}
For simplicity, we only present the detailed proofs for
\begin{align}
	\frac{\sqrt{n}}{\mfm_{1,1}(1)}\left(\frac{\hla_1(1)}{n}-\mbE[\fM_{1,1}(1)]\right)\overset{d}{\longrightarrow}\mcN(0,1),\quad{\rm under\ }\mbH_0^{(1)},\label{Eq of I1 CLT}
\end{align}
and
\begin{align}
	\frac{\sqrt{n}}{\mfm_{1,1}(2)}\left(\frac{\hla_1(1)}{n}-\mbE[\fM_{1,1}(2)]\right)\overset{d}{\longrightarrow}\mcN(0,1),\quad{\rm under\ }\mbH_0^{(2)},\label{Eq of I2 CLT}
\end{align}
since the arguments for more general \(p\) are totally the same as those for (\ref{Eq of I1 CLT}) and (\ref{Eq of I2 CLT}). Under \(\mbH_0^{(p)}\), i.e. \(X_t\) has \(p\) unit roots, we have \(\bbX=\bbe\bbU^p\), where \(\bbe=[e_1,\cdots,e_T]\) and \(\bbU\) is an \(T\times T\) upper toeplitz matrix with the \(k\)-th sup-diagonals and main diagonals are \(1\). Thus, by (\ref{Eq of hat Rp}), it gives that
$$\left\{\begin{array}{ll}
     \hR(1)=\bbU\bbM\bbe'\diag(\bbe\bbM\bbU'\bbU\bbM\bbe')^{-1}\bbe\bbM\bbU',&{\rm under\ } \mbH_0^{(1)},\\
     \hR(1)=\bbU\bbM\bbU'\bbe'\diag(\bbe\bbU\bbM\bbU'\bbU\bbM\bbU'\bbe')^{-1}\bbe\bbU\bbM\bbU',&{\rm under\ }\mbH_0^{(2)}.
\end{array}\right.$$
Let's first prove (\ref{Eq of I1 CLT}). By the SVD of \(\bbM\bbU'\) in (\ref{Eq of SVD of MU}), denote $\Sigma=\diag(\sigma_1,\cdots,\sigma_T)$, where $\sigma_1,\cdots,\sigma_T$ are singular values of \(\bbM\bbU'\). Then under \(\mbH_0^{(1)}\), we have 
$$\hR(1)=\bbV\Sigma\bbW'\bbe'\diag(\bbGa\bbe\bbW\Sigma^2\bbW'\bbe'\bbGa')^{-1}\bbe\bbW\Sigma\bbV'.$$
Similar as (\ref{Eq of la1}), let \(\tF_1(1)\) be the eigenvector of \(\hla_1(1)\) such that \(\tF_1=\sum_{t=1}^{T-1}\talpha_{1,t}\bbv_t\), where \(\sum_{t=1}^{T-1}\talpha_{1,t}^2=1\), then
$$\hla_1(1)=\sum_{k,l=1}^{T-1}\talpha_{1,k}\talpha_{1,l}\sum_{i=1}^n\frac{\sigma_k\sigma_l(\bbe_i\bbw_k)(\bbe_i\bbw_l)}{\sum_{t=1}^{T-1}\sigma_t^2(\bbe_i\bbw_t)^2}.$$
Since the proof procedures of (\ref{Eq of I1 CLT}) are nearly the same as what we have done in Theorem \ref{Thm of CLT I1}, we only briefly present the key steps and omit the details to save space. 
\begin{enumerate}
	\item To distinguish \(\bbe_j\), denote \(\fe_j\) to be the \(j\)-th {\bf row} of \(\sum_{k=0}^{\infty}\Psi_k\varepsilon_{t-k}\). Paralleling with Lemma \ref{Lem of covariance adjust}, we have
	\begin{align}
		\left|\Cov(\fe_j\bbw_k,\fe_j\bbw_l)\right|\leq C_BT^{3\delta-2},\quad\left\{\begin{array}{cc}
			1\leq k\leq T^{\delta}&T^{1-\delta}<l<T-T^{1-\delta}\\
			1\leq l\leq T^{\delta}&T^{1-\delta}<k<T-T^{1-\delta}
		\end{array}\right..\label{Eq of Lem of covariance adjust}
	\end{align}
	Without loss of generality, assume \(1\leq k\leq T^{\delta}\ll T^{1-\delta}<l<T-T^{1-\delta}\). Notice that
	\begin{align*}
		&\fe_j\bbw_k=-\sqrt{\frac{2}{T}}\sum_{t=1}^Te_{j,t}\cos\left(\pi k(2t-1)/(2T)\right)=-\sqrt{\frac{2}{T}}\sum_{t=1}^Te_{j,t}\Re\left(\exp\left({\rm i}\pi k(2t-1)/(2T)\right)\right)\\
		&=-\sqrt{\frac{1}{2T}}\sum_{t=1}^Te_{j,t}\left(\exp\left({\rm i}\pi k(2t-1)/(2T)\right)+\exp\left(-{\rm i}\pi k(2t-1)/(2T)\right)\right)\\
		&=-\sqrt{\pi}\left(d_j(-\theta_k/2)e^{-{\rm i}\theta_k/4}+d_j(\theta_k/2)e^{{\rm i}\theta_k/4}\right),
	\end{align*}
	where \(d_j(\theta)\) is defined in (\ref{Eq of dj}). Hence,
	\begin{align*}
		&\Cov(\fe_j\bbw_k,\fe_j\bbw_l)=\pi\big(e^{{\rm i}(\theta_l-\theta_k)/4}\mbE\left[d_j(-\theta_k/2)d_j(\theta_l/2)\right]+e^{-{\rm i}(\theta_k+\theta_l)/4}\mbE\left[d_j(-\theta_k/2)d_j(-\theta_l/2)\right]\\
		&+e^{{\rm i}(\theta_k+\theta_l)/4}\mbE\left[d_j(\theta_k/2)d_j(\theta_l/2)\right]+e^{{\rm i}(\theta_k-\theta_l)/4}\mbE\left[d_j(\theta_k/2)d_j(-\theta_l/2)\right]\big).
	\end{align*}
	According to the proof of Lemma \ref{Lem of covariance adjust}, when \(k\not\equiv l\mod2\), we know that
	\begin{align*}
		&e^{{\rm i}(\theta_l-\theta_k)/4}\mbE[d_j(-\theta_k/2)\cdot d_j(\theta_l/2)]\\
		&=\frac{-e^{{\rm i}(\theta_l-\theta_k)/4}}{2\pi^2T\left[1-e^{{\rm i}\theta_l/2}\right]}\left(\phi_0^j+\sum_{r=1}^{T-1}\phi_r^j\left(\cos\left(r\theta_k/2\right)+\cos\left(r\theta_l/2\right)\right)\right)+C_BT^{3\delta-2},
	\end{align*}
	and
	\begin{align*}
		&e^{{\rm i}(\theta_k-\theta_l)/4}\mbE[d_j(\theta_k/2)\cdot d_j(-\theta_l/2)]\\
		&=\frac{-e^{{\rm i}(\theta_k-\theta_l)/4}}{2\pi^2T\left[1-e^{-{\rm i}\theta_l/2}\right]}\left(\phi_0^j+\sum_{r=1}^{T-1}\phi_r^j\left(\cos\left(r\theta_k/2\right)+\cos\left(r\theta_l/2\right)\right)\right)+C_BT^{3\delta-2}.
	\end{align*}
	Since
	\begin{align*}
		&\left|\frac{e^{{\rm i}(\theta_l-\theta_k)/4}}{1-e^{{\rm i}\theta_l/2}}+\frac{e^{{\rm i}(\theta_k-\theta_l)/4}}{1-e^{-{\rm i}\theta_l/2}}\right|=\left|\frac{4\sin(\theta_k/4)\sin(\theta_l/4)}{(1-\cos(\theta_l/2))^2+\sin^2(\theta_l/2)}\right|\\
		&\leq\frac{4\pi\sin(\theta_k/4)}{\sin(\theta_l/2)}\leq\frac{\pi^2\theta_k}{\theta_l}\leq\pi^2T^{2\delta-1},
	\end{align*}
	where we use the fact that \(1\leq k\leq T^{\delta}<T^{1-\delta}\leq l\) in the last step, then
	$$\left|e^{{\rm i}(\theta_l-\theta_k)/4}\mbE[d_j(-\theta_k/2)\cdot d_j(\theta_l/2)]+e^{{\rm i}(\theta_k-\theta_l)/4}\mbE[d_j(\theta_k/2)\cdot d_j(-\theta_l/2)]\right|\leq C_BT^{3\delta-1}.$$
	Similarly, we can also derive that
	$$\left|e^{{\rm i}(\theta_k-\theta_l)/4}\mbE[d_j(\theta_k/2)\cdot d_j(-\theta_l/2)]+e^{{\rm i}(\theta_l-\theta_k)/4}\mbE[d_j(-\theta_k/2)\cdot d_j(\theta_l/2)]\right|\leq C_BT^{3\delta-1}.$$
	Then we conclude (\ref{Eq of Lem of covariance adjust}) when \(k\not\equiv l\mod2\), for the other case, the proofs are totally the same, so we omit them here.
	\item Paralleling with Lemma \ref{Lem of variance adjust}, we have
	\begin{align}
		\left|\Var(\fe_j\bbw_t)-2\pi f_j(0)\right|\leq C_Bt/T.\label{Eq of Lem of variance adjust}
	\end{align}
	Initially, we can use the same method of Lemma 10 in \cite{onatski2021spurious} to show that \(\left|\Var(\fe_j\bbw_t)-2\pi f_j(\pi t/T)\right|\leq C_B/T\), then repeat the proof of Lemma \ref{Lem of variance adjust} to obtain (\ref{Eq of Lem of variance adjust}).
	\item Based on (\ref{Eq of Lem of covariance adjust}) and (\ref{Eq of Lem of variance adjust}), we can repeat the proofs of Lemmas \ref{Lem of remove dependence} and \ref{Lem of unify variance} to conclude that
	\begin{align}
		\frac{1}{\sqrt{n}}\sum_{i=1}^n\frac{\sigma_k\sigma_l(\bbe_i\bbw_k)(\bbe_i\bbw_l)}{\sum_{t=1}^{T-1}\sigma_t^2(\bbe_i\bbw_t)^2}\overset{L^2}{\longrightarrow}\frac{1}{\sqrt{n}}\sum_{i=1}^n\frac{\sigma_k\sigma_l z_{i,k}z_{i,l}}{\sum_{t=1}^{T-1}\sigma_k^2z_{i,t}^2},\label{Eq of Lem of L2 convergence}
	\end{align}
	where \(\{\vec{z}_t=(z_{1,t},\cdots,z_{n,t})'\sim\mcN(\boldsymbol{0},\tilde{\bbGa}):t=1,\cdots,T-1\}\) is a sequence of $n$-dimensional i.i.d. normal vectors and $\tilde{\bbGa}$ is defined in \eqref{Eq of bbF}. Combining with (\ref{Eq of Lem of L2 convergence}) and Lemma \ref{Lem of positive covariance}, it further gives that
	\begin{align}
		\frac{1}{n}\Var\left(\sum_{i=1}^n\frac{\sigma_k\sigma_l z_{i,k}z_{i,l}}{\sum_{t=1}^{T-1}\sigma_k^2z_{i,t}^2}\right)\leq C_{B,b,M_0,m_0}(kl)^{-2}.\label{Eq of Lem of bounded variance}
	\end{align}
    Hence, we can obtain that $\mfm_{1,1}^2(1)\asymp\mrO(1)$ in \eqref{Eq of mfm}.
	\item Based on (\ref{Eq of Lem of bounded variance}), we can establish the asymptotic behaviors of \(\talpha_{1,k}\) as in Lemma \ref{Lem of alpha}, i.e.
	$$\lim_{n\to\infty}\sqrt{n}\mbE\left[1-\talpha_{1,1}^2\right]=0.$$
	Again, by the same method in Theorem \ref{Thm of convergence in probability}, we derive that
	$$\frac{\hla_1(1)-\mbE[\hla_1(1)]}{\sqrt{n}}\overset{\mbP}{\longrightarrow}\frac{1}{\sqrt{n}}\sum_{i=1}^n\left(\frac{\sigma_1^2(\bbe_i\bbw_1)^2}{\sum_{t=1}^{T-1}\sigma_t^2(\bbe_i\bbw_t)^2}-\mbE\left[\frac{\sigma_1^2(\bbe_i\bbw_1)^2}{\sum_{t=1}^{T-1}\sigma_t^2(\bbe_i\bbw_t)^2}\right]\right).$$
	\item Finally, by the \(m\)-dependent condition of \(\tilde{\bbGa}\) in Assumption \ref{Ap of m dependent}, we repeat the proof of Theorem \ref{Thm of CLT I1} and obtain 
	$$\frac{\sqrt{n}}{\mfm_{1,1}(1)}\left(\frac{\hla_1(1)}{n}-\mbE[\fM_{1,1}(1)]\right)\overset{d}{\longrightarrow}\mcN(0,1),$$
	where \(\mfm_{1,1}(x)\) is defined in (\ref{Eq of mfm}).
\end{enumerate}
For the proofs of (\ref{Eq of I2 CLT}), by the SVD of \(\bbM\bbU'\) in (\ref{Eq of SVD of MU}), under \(\mbH_0^{(2)}\), we have
$$\hR(1)=\bbV\Sigma^2\bbV'\bbe'\diag(\bbe\bbV\Sigma^4\bbV'\bbe')^{-1}\bbe\bbV\Sigma^2\bbV',$$
then it is easy to see that 
$$\hla_1(1)=\sum_{k,l=1}^{T-1}\talpha_{1,k}\talpha_{1,l}\sum_{i=1}^n\frac{\sigma_k^2\sigma_l^2(\bbe_i\bbv_k)(\bbe_i\bbv_l)}{\sum_{t=1}^{T-1}\sigma_t^4(\bbe_i\bbv_t)^2},$$
where \(\hF_1(1)=\sum_{t=1}^{T-1}\talpha_{1,t}\bbv_t\) is the eigenvector of \(\hla_1(1)\). Since the asymptotic behaviors of \(\bbe_i\bbv_k\) have been well studied in Lemmas \ref{Lem of covariance adjust}, \ref{Lem of variance adjust}, \ref{Lem of remove dependence}, \ref{Lem of unify variance}, the proofs of (\ref{Eq of I2 CLT}) are totally the same as what we have done in Theorem \ref{Thm of CLT I1}, so we omit details here to save space.
\end{proof}
\subsubsection{Test procedures}\label{sec of multiple unit roots tests}
In this part, we will construct the test statistic for $\mbH_0^{(p)}$ versus $\mbH_0^{(p+1)}$. Recall $\hla_1(p)$
Here, let's construct
\begin{align}
	\widehat{T}_n(p):=\sqrt{n}\left(\frac{\hla_1(p)}{n}-\mbE[\fM_{1,1}(p)]\right),\label{Eq of hat Tp}
\end{align}
by Theorem \ref{Thm of forwards}, we have
\begin{align}
	\widehat{T}_n(p)/\mfm_{1,1}(p)\overset{d}{\longrightarrow}\mcN(0,1),\quad{\rm under\ }\mbH_0^{(p)}.\label{Eq of CLT hat Tp}
\end{align}
On the other hand, under \(\mbH_0^{(p+1)}\), Theorem \ref{Thm of forwards} implies that 
\begin{align*}
    \widehat{T}_n(p)/\mfm_{1,1}(p+1)+\sqrt{n}\big(\mbE[\fM_{1,1}(p)]-\mbE[\fM_{1,1}(p+1)]\big)/\mfm_{1,1}(p+1)\overset{d}{\longrightarrow}\mcN(0,1),
\end{align*}
where $\mfm_{1,1}(p+1)\asymp\mrO(1)$. By the definition of $\fM_{1,1}(x)$ in \eqref{Eq of fMx}, we know that $\fM_{1,1}(p+1)>\fM_{1,1}(p)$, so it yields that \(\widehat{T}_n(p)\asymp\mrO(\sqrt{n})\) under \(\mbH_0^{(p+1)}\). Therefore, for testing 
\begin{center}
	\(\mbH_0^{(p)}\): \(X_t\) has \(p\) unit roots.\quad versus\quad\(\mbH_0^{(p+1)}\): \(X_t\) has \(p+1\) unit roots,
\end{center}
we will reject \(\mbH_0^{(p)}\) if \(\widehat{T}_n(p)>\log(n)\) and the asymptotic power is
\begin{align}
	\lim_{n\to\infty}\mbP\left(\widehat{T}_n(p)>\log(n)\big|\mbH_0^{(p+1)}\right)=1.\label{Eq of power hat Tp}
\end{align}
In summary, we test the number of unit roots by the following inductively procedures:
\begin{enumerate}
	\item Given \(\bbX\), we first compute the all eigenvalues \(\hla_1\geq\cdots\geq\hla_T\) of the sample correlation matrix of \(\bbX\), if it shows the low rank structure, i.e. there exists a \(K\in\mbN\) such that \(\hla_k=0\) for \(k\geq K\), then we accept 
	\begin{center}
		\(\mbH_0^{(\infty)}\): \(X_t\) has totally nonstationary roots.
	\end{center}
	Otherwise, let's construct \(\widehat{T}_n(0)\) by (\ref{Eq of hat Tp}), (\ref{Eq of hat Rp}) and (\ref{Eq of operator msT}), if \(\widehat{T}_n(0)<-\log(n)\), we reject \(\mbH_0^{(0)}\); otherwise, we accept \(\mbH_0^{(0)}\) and stop.
	\item Suppose our current test is
	\begin{center}
		\(\mbH_0^{(p)}\): \(X_t\) has \(p\) unit roots,\quad versus\quad\(\mbH_0^{(p+1)}\): \(X_t\) has \(p+1\) unit roots,
	\end{center}
	where \(p\geq1\). Let's construct \(\widehat{T}_n(p)\) by (\ref{Eq of hat Tp}). If \(\widehat{T}_n(p)>\log(n)\), we reject \(\mbH_0^{(p)}\) and move to test \(\mbH_0^{(p+1)}\) versus \(\mbH_0^{(p+2)}\). Otherwise, we accept \(\mbH_0^{(p)}\) and stop.
\end{enumerate}
\subsection{Estimation of the asymptotic variance of the statistic for unit root tests}\label{sec of estimation variance}
Recall the statistic $\widehat{T}_n(0)$ in \eqref{Eq of unit root test statistic} for the unit root test \eqref{Eq of unit root test}, since the asymptotic variance $\mfm_{1,1}$ is generally unknown, we will estimate this $\mfm_{1,1}$ in this section. Actually, by the definition of \(\mfm_{1,1}(x)\) in (\ref{Eq of mfm}), we see that $\mfm_{1,1}$ in \eqref{Eq of mfm kk} is indeed a special case of \(\mfm_{1,1}(x)\) in (\ref{Eq of mfm}). Now, given the noise matrix \(\bbe=[e_1,\cdots,e_T]\), where $e_t$ is defined in \eqref{Eq of nonpanel Xt} satisfying Assumptions \ref{Ap of panel lag polynomial}, \ref{Ap of nonpanel} and \ref{Ap of m dependent}, we will estimate \(\mfm_{1,1}(x)\) in (\ref{Eq of mfm}) for \(x\in[1,\infty)\). Recall that 
$$\mfm_{1,1}^2(x)=\frac{1}{n}\Var\left(\sum_{i=1}^n\frac{\sigma_1^{2x}z_{i,1}^2}{\sum_{t=1}^{T-1}\sigma_t^{2x}z_{i,t}^2}\right),$$
where $\{\vec{z}_t=(z_{1,t},\cdots,z_{n,t})'\overset{i.i.d.}{\sim}\mcN(\boldsymbol{0},\tilde{\bbGa}):t=1,\cdots,T-1\}$ and $\tilde{\bbGa}$ is defined in \eqref{Eq of bbF} and $\sigma_t$ are singular values of $\bbM\bbU'$ in \eqref{Eq of SVD of MU}. Hence, to estimate $\mfm_{1,1}(x)$, we need to estimate $\tilde{\bbGa}$ first. Here, let
$$\vec{f}(0):=\frac{1}{2\pi}\bbGa\left(\sum_{l=0}^{\infty}\Psi_l^2\right)\bbGa'+\frac{1}{\pi}\sum_{k=1}^{\infty}\bbGa\left(\sum_{l=0}^{\infty}\Psi_l\Psi_{k+l}\right)\bbGa',$$
where $\bbGa$ is the cross-sectional matrix in \eqref{Eq of nonpanel Xt}, by (\ref{Eq of bbF}), we know that  
$$\tilde{\bbGa}=\diag(\vec{f}(0))^{-1/2}\vec{f}(0)\diag(\vec{f}(0))^{-1/2}.$$
Hence, to estimate \(\tilde{\bbGa}\), it suffices to estimate $\vec{f}(0)$. Here, we will use the hard thresholding method in \cite{sun2018large} and \cite{zhang2021convergence} to estimate \(\tilde{\Gamma}\). Precisely, \cite{sun2018large} proposed a estimation method for \(\vec{f}(0)\), one can find the extension version for the coherence \(\tilde{\bbGa}\) in Corollary 5 of \cite{zhang2021convergence}.
\begin{enumerate}
	\item First, define
	$$\bbH:=\frac{1}{2\pi T}\sum_{t=1}^Te_te_t'+\frac{1}{2\pi}\sum_{l=1}^{[T^{1/2}]}\frac{1}{T-l}\sum_{t=l+1}^T\left(e_te_{t-l}'+e_{t-l}e_t\right)',$$
    and \(\tilde{\bbH}:=\diag(\bbH)^{-1/2}\bbH\diag(\bbH)^{-1/2}\).
	\item According to Assumptions \ref{Ap of highdimensionality}, \ref{Ap of panel lag polynomial} and \ref{Ap of nonpanel}, we define
	$${\rm ess}\Vert\vec{f}\Vert:={\rm ess}\inf_{x\in[0,2\pi]}\Vert\vec{f}(x)\Vert\leq\frac{M_0B^2}{2\pi},$$
	and
	$$\Omega_T(\vec{f}):=\max_{1\leq s,t\leq n}\sum_{k=0}^T(1+k)\left|\Gamma_s\left(\sum_{l=0}^{\infty}\Psi_l\Psi_{k+l}\right)\Gamma_t'\right|\leq M_0B^2,$$
	and
	$$L_T(\vec{f}):=\max_{1\leq s,t\leq n}\sum_{k>T}^{\infty}\left|\Gamma_s\left(\sum_{l=0}^{\infty}\Psi_l\Psi_{k+l}\right)\Gamma_t'\right|\leq M_0B^2T^{-2},$$
	then given a sufficiently large constant \(R\) and \(M:=M_T\) such that \(\big({\rm ess}\Vert\vec{f}\Vert\big)^2\log(T)\leq M\leq T/\Omega_T(\vec{f})\), define a threshold
	$$\nu:=2R\cdot{\rm ess}\Vert\vec{f}\Vert\sqrt{\frac{\log n}{M}}+\frac{2M+\pi^{-1}}{T}\Omega_T(\vec{f})+\frac{L_T(\vec{f})}{\pi}.$$
	\item Next, construct the hard thresholding operator as follows:
	$$T_{\nu}(x):=\left\{\begin{array}{cc}
		x&|x|\geq\nu\\
		0&|x|<0
	\end{array}\right.,$$
	and define \(T_{\nu}(\tilde{\bbH}):=\big[T_{\nu}(\tilde{H}_{s,t})\big]_{s,t}\in\mbR^{n\times n}\). Combining Proposition 3.6 in \cite{sun2018large} and Corollary 5 of \cite{zhang2021convergence}, we can conclude that
	\begin{align}
		\mbP\big(n^{-1}\Vert T_{\nu}(\tilde{\bbH})-\tilde{\bbGa}\Vert_F^2\geq m\nu^2\big)\leq C_Rn^{-1},\label{Eq of positive definite 1}
	\end{align}
	where \(m=m_n=\mro(\sqrt{n})\) is defined in Assumption \ref{Ap of m dependent}.
	\item Since \(T_{\nu}(\tilde{\bbH})\) may not be positive semi-definite, then we will apply the method in \S2.2 of \cite{chen2013covariance} to remove the negative eigenvalues of \(T_{\nu}(\tilde{\bbH})\). Precisely, given the threshold \(\nu\) as above, \cite{chen2013covariance} suggested a new threshold as follows:
	$$\mu=\left(\frac{\nu^2}{n}\sum_{s,t=1}^n1_{|\tilde{H}_{s,t}|\geq\nu}\right)^{1/2}.$$
	For simplicity, we set \(\nu=T^{-1/3}\log T\) in our numerical experiments. Suppose the SVD of \(T_{\nu}(\tilde{\bbH})\) is
	$$T_{\nu}(\tilde{\bbH})=\sum_{i=1}^n\gamma_i\bbq_i\bbq_i',\quad\gamma_1\geq\cdots\geq\gamma_n,$$
	let's construct
	$$T_{\nu,\mu}(\tilde{\bbH})=\sum_{i=1}^n\max\{\gamma_i,\mu\}\bbq_i\bbq_i',$$
	it is easy to see that \(T_{\nu,\mu}(\tilde{\bbH})\) is positive semi-definite, and we show that
	$$\mbP\left(n^{-1}\Vert T_{\nu,\mu}(\tilde{\bbH})-\tilde{\bbGa}\Vert_F^2\geq Cm\nu^2\right)\leq C_Rn^{-1}$$
    by the method in \S2.2 of \cite{chen2013covariance}.
\end{enumerate}
In practice, when $X_t$ is a random walk generated by \eqref{Eq of nonpanel Xt}, we can recover $e_t$ by $X_t-X_{t-1}$. Consequently, we can repeat the above procedures to estimate \(T_{\nu,\mu}(\tilde{\bbH})\). Next, given a significance level \(\alpha\in(0,1)\), we can estimate the upper and lower \(\alpha\) quantile of \(\widehat{T}_n(0)\) in \eqref{Eq of unit root test statistic} as follows:
\begin{enumerate}
	\item Given \(T_{\nu,\mu}(\tilde{\bbH}),n,T\) and the number of simulations \(B\in\mbN^+\), simulate \(B\) standard Gaussian random matrices \(\bbe^{(1)},\cdots,\bbe^{(B)}\in\mbR^{n\times T}\), i.e. all \(\bbe^{(b)}\)'s entries are i.i.d. \(\mcN(0,1)\). And denote \(\bbOme\) to be the square root of \(T_{\nu,\mu}(\tilde{\bbH})\), i.e. \(\bbOme^2:=T_{\nu,\mu}(\tilde{\bbH})\).
	\item For each \(\bbe^{(b)}\), construct 
	$$\bbX^{(b)}:=\bbOme\bbe^{(b)}\diag(1,\cdots,T^{-1})\quad{\rm and}\quad\bbR^{(b)}:=(\bbX^{(b)})'\diag(\bbX^{(b)}(\bbX^{(b)})')^{-1}\bbX^{(b)},$$
	then let
	$$\widehat{T}_n^{(b)}(0):=\sqrt{n}\left(\frac{\Vert\bbR^{(b)}\Vert}{n}-\mbE[\fM_{1,1}(1)]\right).$$
	\item Sorting all \(\widehat{T}_n^{(b)}(0)\) in ascending order, i.e. \(\widehat{T}_n^{(1*)}(0)\leq\cdots\leq \widehat{T}_n^{(B*)}(0)\), we accept \(H_0\) \eqref{Eq of unit root test} if
	$$\widehat{T}_n(0)\in\left[T_n^{([\alpha B]*)}(0),T_n^{([(1-\alpha)B]*)}(0)\right],$$
	where \(\alpha\in(0,1)\) is the significance level.
\end{enumerate}

\end{document}